\documentclass{llncs}

\usepackage{microtype}%if unwanted, comment out or use option "draft"
\usepackage{azalea}
\usepackage{ori}

\allowdisplaybreaks

\begin{document}

\title{On the Semantics of Snapshot Isolation}

%\titlerunning{Dummy short title}%optional, please use if title is longer than one line

\author{Azalea Raad \inst{1}  \and Ori Lahav \inst{2} \and Viktor Vafeiadis \inst{1}}
\institute{
	MPI-SWS, Germany
	\and Tel Aviv University, Israel \\
	\email{\{azalea,viktor\}@mpi-sws.org} \\
	\email{orilahav@tau.ac.il}
}

%\keywords{Software transactional memory, snapshot isolation, weak isolation}%mandatory

%%%%%%%%%%%%%%%%%%%%%%%%%%%%%%%%%%%%%%%%%%%%%%%%%%%%%%

%%% For ArXiV version %%%%
%\nolinenumbers
%\hideLIPIcs
%\makeatletter\def\@oddfoot{}\makeatother
%%% End of ArXiV version

%\newtheorem{defn}{Definition}{\bfseries}{}
%\crefname{defn}{Definition}{Definitions}

\crefformat{appendix}{#2\S{}#1#3}
%\Crefname{appendix}{Appendix}{Appendices}
\Crefformat{appendix}{Appendix #2#1#3}
%\crefformat{definition}{Def.~#2#1#3}

\maketitle

\begin{abstract}
Snapshot isolation (SI) is a standard transactional consistency model used in databases, distributed systems and software transactional memory (STM).
Its semantics is formally defined both declaratively as an acyclicity axiom, 
and operationally as a concurrent algorithm with memory bearing timestamps.

We develop two simpler equivalent operational definitions of SI
as lock-based reference implementations that do not use timestamps.
Our first locking implementation is prescient in that requires \textit{a priori} knowledge of the data accessed by a transaction and carries out transactional writes eagerly (in-place).
Our second implementation is non-prescient and performs transactional writes lazily by recording them in a local log and propagating them to memory at commit time.
Whilst our first implementation is simpler and may be better suited for developing a program logic for SI transactions, 
our second implementation is more practical due to its non-prescience.
We show that both implementations are sound and complete against the declarative SI specification and thus yield equivalent operational definitions for SI.

We further consider, for the first time \emph{formally}, the use of SI in a context with racy non-transactional accesses,
as can arise in STM implementations of SI.
We introduce \emph{robust snapshot isolation} (RSI), an adaptation of SI with similar semantics and guarantees in this mixed setting.
We present a declarative specification of RSI as an acyclicity axiom and analogously develop two operational models as lock-based reference implementations (one eager, one lazy). 
We show that these operational models are both sound and complete against the declarative RSI model.

\end{abstract}

\renewcommand{\paragraph}[1]{
\medskip
\noindent\textbf{#1}\;\;
}

\section{Introduction}
\label{sec:intro}
Transactions are the \emph{de facto} synchronisation mechanism in databases and geo-replicated distributed systems, 
%~\cite{???},
and are thus gaining adoption in the shared-memory setting via \emph{software transactional memory} (STM)~\cite{stm,tm}.
In contrast to other synchronisation mechanisms, transactions readily provide atomicity, isolation, and consistency guarantees for sequences of operations,
allowing programmers to focus on the high-level design of their systems.

However, providing these guarantees comes at a significant cost.
As such, various transactional consistency models in the literature trade off consistency guarantees for better performance.
At nearly the one end of the spectrum, we have \emph{serialisability}~\cite{Papadimitriou79}, 
which requires transactions to appear to have been executed in some total order consistent with the program order.
Serialisability provides strong guarantees, but is widely considered too expensive to implement.
The main problem is that two conflicting transactions (e.g.~one reading from and one updating the same datum) 
cannot both execute and commit in parallel.

Consequently, most major databases,
both centralised (e.g.~Oracle and MS SQL Server) and distributed~\cite{SI-distributed-lazy-replication,SI-distributed-replication,SI-distributed-notifications}, 
have opted for a slightly weaker model called \emph{snapshot isolation} (SI)~\cite{SI} 
as their default consistency model.
SI has much better performance than serialisability
by allowing conflicting transactions to execute concurrently and commit successfully as long as they do not have a write-write conflict.
This in effect allows reads of SI transactions to read from an earlier memory snapshot than the one affected by their writes, 
and permits the \emph{write skew anomaly}~\cite{SI-Cerone} depicted in \cref{fig:litmus_tests}.
Besides this anomaly, however, SI is essentially the same as serialisability: 
Cerone et al.~\cite{SI-Cerone} provide a widely applicable condition under which SI and serialisability coincide for a given set of transactions.
For these reasons, SI has also started gaining adoption in the generic programming language setting via STM implementations~\cite{SI-STM-clojure,SI-Java,SI-STM-hindsight,SI-STM-anomalies,SI-STM-aborts} that provide SI semantics for their transactions.

The formal study of SI, however, has so far not accounted for the more general STM setting
in which both transactions and uninstrumented non-transactional code can access the same memory locations.
%As a result, there is no suitable definition of SI  
%with \emph{mixed-mode} (i.e., both transactional and non-transactional) accesses to the same locations. 
While there exist two equivalent definitions of SI---%
one declarative in terms of an acyclicity constraint~\cite{cerone15,SI-Cerone}
and one operational in terms of an optimistic multi-version concurrency control algorithm~\cite{SI}---%
neither definition supports 
\emph{mixed-mode} (i.e.\ both transactional and non-transactional) accesses to the same locations. 
Extending the definitions to do so is difficult for two reasons:
(1) the operational definition attaches a timestamp to every memory location, 
which heavily relies on the absence of non-transactional accesses; and
(2) there are subtle interactions between the transactional implementation
and the weak memory model underlying the non-transactional accesses.

In this article, we address these limitations of SI.
We develop two simple lock-based reference implementations for SI that do not use timestamps. 
Our first implementation is \emph{prescient}~\cite{tm-book} in that it requires \textit{a priori} knowledge of the data accessed by a transaction, and performs transactional writes \emph{eagerly} (in-place).
Our second implementation is non-prescient and carries out transactional writes \emph{lazily} by first recording them in a local log and subsequently propagating them to memory at commit time.
Our first implementation is simpler and may be better suited for understanding and developing a program logic for SI transactions,
whilst our second implementation is more practical due to its non-prescience.
We show that both implementations are sound and complete against the declarative SI specification and thus yield equivalent operational definitions for SI.

%\azalea{Rephrase when second RSI imp. is done.}
%In \cref{sec:rsi}, we then extend our eager implementation to make it robust under uninstrumented non-transactional accesses,
%and characterise declaratively the semantics we obtain.
%We call this extended model \emph{robust snapshot isolation} (RSI) and show that it gives reasonable semantics with mixed-mode accesses.

We then extend both our eager and lazy implementations to make them robust under uninstrumented non-transactional accesses,
and characterise declaratively the semantics we obtain.
We call this extended model \emph{robust snapshot isolation} (RSI) and show that it gives reasonable semantics with mixed-mode accesses.

To provide SI semantics, instead of timestamps, our implementations use \emph{multiple-readers-single-writer} (MRSW) locks.
They acquire locks in reader mode to take a snapshot of the memory locations accessed by a transaction
and then promote the relevant locks to writer mode to enforce an ordering on transactions with write-write conflicts.
As we shall discuss, the equivalence of the RSI implementation and its declarative characterisation depends heavily upon the axiomatisation of MRSW locks:
here, we opted for the weakest possible axiomatisation that does not order any concurrent reader lock operations 
and present an MRSW lock implementation that achieves this.

\paragraph{Outline}
In \cref{sec:ideas} we present an overview of our contributions by describing our reference implementations for both SI and RSI.
In \cref{sec:framework} we define the declarative framework for specifying STM programs. 
In \cref{sec:si} we present the declarative SI specification against which we demonstrate the soundness and completeness of our SI implementations.
In \cref{sec:rsi} we formulate a declarative specification for RSI and demonstrate the soundness and completeness of our RSI implementations.
We discuss related and future work in \cref{sec:related_and_future_work}.%
%\footnote{
%A full version of this article is available at~\cite{appendix}.
%}

\newcommand{\rlock}{\code{lock\_r}\xspace}
\newcommand{\runlock}{\code{unlock\_r}\xspace}
\newcommand{\wlock}{\code{lock\_w}\xspace}
\newcommand{\wunlock}{\code{unlock\_w}\xspace}

\section{Background and Main Ideas}
\label{sec:ideas}

\makeatletter
\def\namedlabel#1#2{\begingroup(#2)\def\@currentlabel{#2}\phantomsection\label{#1}\endgroup}
\makeatother

\newcommand{\CtransI}[2]{\textbf{\smaller{#1:}}\left[\inarr{#2}\right.}
\newcommand{\Ctrans}[2]{\textbf{\smaller{#1:}}\\\,\left[\inarr{#2}\right.}
\newcommand{\Anomaly}[6]{%
\begin{minipage}[b]{#1\textwidth}\centering\vspace{1ex}%Initially, \ensuremath{#5}.\\[.5ex]%
\ensuremath{#6}\\[1ex]\namedlabel{#3}{#2} #4\\[1ex]\end{minipage}}

\newcommand\denot[1]{\ensuremath{\llbracket{#1}\rrbracket}}
\newcommand\snapshotx{\ensuremath{\code{s}_{\code{x}}}}

As noted earlier, the key challenge in specifying STM transactions lies 
in accounting for the interactions between mixed-mode accesses to the same data.
One simple approach is to treat each non-transactional access as a singleton mini-transaction
and to provide \emph{strong isolation}~\cite{Blundell-isolation,Blundell-isolation2}, 
i.e.\ full isolation between transactional and non-transactional code.
This, however, requires \emph{instrumenting} non-transactional accesses to adhere to same access policies as transactional ones (e.g.\  acquiring the necessary locks), 
which incurs a substantial performance penalty for non-transactional code.
A more practical approach is to enforce isolation only amongst transactional accesses,
an approach known as \emph{weak isolation}~\cite{Blundell-isolation,Blundell-isolation2}, adopted by the relaxed transactions of C++~\cite{C++}.
%Under weak isolation, however, transactions with explicit aborts are problematic 
%as their intermediate state may be observed by non-transactional code. 
%As such, weakly isolated STMs (e.g.\ C++ relaxed transactions~\cite{C++}) often forbid explicit abort instructions altogether.

As our focus is on STMs with SI guarantees, instrumenting non-transactional accesses is not feasible. % and thus our STM guarantees weak isolation.
In particular, as we expect many more non-transactional accesses than transactional ones,
we do not want to incur any performance degradation on non-transactional code when executed in parallel with transactional code. 
As such, we opt for an STM with SI guarantees under \emph{weak isolation}. 
Under weak isolation, however, transactions with explicit abort instructions are problematic 
as their intermediate state may be observed by non-transactional code. 
As such, weakly isolated STMs (e.g.\ C++ relaxed transactions~\cite{C++}) often forbid explicit aborts altogether.
%For simplicity, throughout our development we assume that transactions are not nested.
Throughout our development we thus make two simplifying assumptions: 
(1) transactions are not nested; and 
(2) there are no explicit abort instructions, following the example of weakly isolated relaxed transactions of C++. 
%(3) the locations accessed by a transaction can be statically determined.
%Specifically, we assume that each transaction $\code T$ is supplied with its \emph{read set}, $\readset$, and \emph{write set}, $\writeset$, containing those locations read and written by $\code T$, respectively (a static over-approximation of these sets suffices for soundness.).
%
%\azalea{
%(3) will not hold of the new implementation.
%} 
As we describe later in \cref{subsec:ideas_si_alt_implementation}, it is straightforward to lift the latter restriction (2) for our lazy implementations. 

For non-transactional accesses, we naturally have to pick some consistency model.
For simplicity and uniformity, we pick the release/acquire (RA) subset of the C++ memory model~\cite{C11,SRA},
a well-behaved platform-independent memory model, whose compilation to x86 requires no memory fences.

\paragraph{Snapshot Isolation (SI)}
The initial model of SI in~\cite{SI} is described informally in terms of a multi-version concurrent algorithm as follows. 
A transaction $\code T$ proceeds by taking a \emph{snapshot} $S$ of the shared objects.
The execution of $\code T$ is then carried out locally: read operations query $S$ and write operations update $S$. 
Once $\code T$ completes its execution, it attempts to \emph{commit} its changes and succeeds \emph{only if} it is not \emph{write-conflicted}. 
Transaction $\code T$ is write-conflicted if another \emph{committed} transaction $\code T'$  has written to a location also written to by $\code T$, since \code{T} recorded its snapshot. 
If $\code T$ fails the conflict check it aborts and may restart; otherwise, it commits its changes, and its changes become visible to all other transactions that take a snapshot thereafter.

To realise this, the shared state is represented as a series of \emph{multi-versioned} objects: each object is associated with a history of several versions at different \emph{timestamps}.
In order to obtain a snapshot, a transaction $\code T$ chooses a \emph{start-timestamp} $t_0$,
and reads data from the committed state as of $t_0$, ignoring updates after $t_0$.
That is, updates committed after $t_0$ are invisible to $\code T$.
In order to commit, $\code T$ chooses a \emph{commit-timestamp} $t_c$ larger than any existing start- or commit-timestamp. 
Transaction $\code T$ is deemed write-conflicted if another transaction $\code T'$ has written to a location also written to by $\code T$ \emph{and} the commit-timestamp of $\code T'$ is in the execution interval of \code T ($[t_0, t_c]$).

\begin{figure}[t]
\centering
  \begin{tabular}{@{}|@{}c@{}|@{}c@{}|@{}c@{}|@{}}
    \hline
    \Anomaly{.33}{LU}{ex:LU}{Lost Update\\ SI: \xmark}{x=0}{\inarrII
		{ \Ctrans{T1}{ a := x; \comment{0} \!\!\\ x := a+1;\!\!} }
		{ \!\!\Ctrans{T2}{ b := x; \comment{0}\\ x := b+1;} }}
    & 
    \Anomaly{.32}{WS}{ex:WS}{Write Skew\\ SI: \cmark}{x=y=0}{\inarrII
		{ \Ctrans{T1}{ a := x; \comment{0}\\ y := 1;} }
		{ \Ctrans{T2}{ b := y; \comment{0}\\ x := 1;} }}
    & 
    \Anomaly{.33}{WS2}{ex:WS2}{Write Skew Variant \\ SI: \cmark}{x=y=0}{\inarrII
		{ \CtransI{T1}{ y := 1; }\\ \Ctrans{T3}{a := x; \comment{0}} }
		{ \Ctrans{T2}{ b := y; \comment{0}\\ x := 1;} }}
    \\\hline
	\Anomaly{.34}{LU2}{ex:LU2}{Lost Update Variant \\ SI: \xmark}{x=y=0}{\inarrII
		{ \CtransI{T1}{ x := 1;\\ y := 1;} \\[9pt] \Ctrans{T3}{ a:=y; \comment{2} }}
		{ \Ctrans{T2}{ b := x; \comment{0}\\ y := 2;} }}
	&
	\Anomaly{.32}{SBT}{ex:SBT}{Store Buffering \\ RSI: \cmark}{x=y=z=0}{\inarrII
		{ x:= 1; \\ \Ctrans{T1}{ a := z; } \\[1ex] b:=y; \comment{0} }
		{ y:= 1; \\ \Ctrans{T2}{ c := z; } \\[1ex] d:=x; \comment{0} }}			
	&
	\Anomaly{.34}{MPT}{ex:MPT}{Message Passing \\ RSI: \xmark}{x=y=0}{\inarrII
		{ x:= 1; \\ y:=1; \\[1ex]}
		{ \Ctrans{T2}{ a := y; \comment{1}\\ b := x; \comment{0}} \\[3ex]}}	
	\\\hline
  \end{tabular}
\caption{Litmus tests illustrating transaction anomalies and their admissibility under SI and RSI.
In all tests, initially, $x=y=z=0$.
The {\!\comment{v}} annotation next to a read records the value read.}
\label{fig:litmus_tests}
\end{figure}

\subsection{Towards an SI Reference Implementation without Timestamps}\label{subsec:ideas_implementations}
\label{sec:ideas_lock}

While the SI description above is suitable for understanding SI, it is not useful for integrating the SI model in a language such as C/C++ or Java. 
%\todo: why not? something about the complexity of timestamps, especially w.r.t. non-transactional code? 
From a programmer's perspective, in such languages the various threads directly access the \emph{uninstrumented} (single-versioned) shared memory; they do not access their own instrumented snapshot at a particular timestamp, which is loosely related to the snapshots of other threads. 
Ideally, what we would therefore like is an equivalent description of SI in terms of accesses to uninstrumented shared memory and a synchronisation mechanism such as locks.

In what follows, we present our first lock-based reference implementation for SI that does not rely on timestamps. 
To do this, we assume that the locations accessed by a transaction can be statically determined.
Specifically, we assume that each transaction $\code T$ is supplied with its \emph{read set}, $\readset$, and \emph{write set}, $\writeset$, containing those locations read and written by $\code T$, respectively (a static over-approximation of these sets suffices for soundness.).
As such, our first reference implementation is \emph{prescient}~\cite{tm-book} in that it requires \textit{a priori} knowledge of the locations accessed by the transaction.
Later in \cref{subsec:ideas_si_alt_implementation} we lift this assumption and develop an SI reference implementation that is \emph{non-prescient} and similarly does not rely on timestamps. 
%

%As mentioned earlier, we explore SI semantics for STMs with weak isolation guarantees~\cite{Blundell-isolation,Blundell-isolation2}. 
%However, it is well-known in the literature that under weak isolation, 
%transactions with explicit aborts are problematic as their intermediate state may be observed by non-transactional code. 
%As such, weakly isolated STMs (e.g.\ C++ relaxed transactions~\cite{C++}) often forbid explicit abort instructions altogether.
%Following the example of C++ relaxed transactions, in our first SI implementation here we assume transactions to contain no explicit abort instructions. 
%We later lift this restriction in our second SI implementation in \cref{subsec:ideas_si_alt_implementation} and account for the presence of explicit abort instructions in transactions. 

Conceptually, a candidate implementation of transaction $\code T$ would 
(1) obtain a snapshot of the locations read by $\code T$;
(2) lock those locations written by $\code T$;
(3) execute $\code T$ \emph{locally}; and
(4) unlock the locations written. 
The snapshot is obtained via \code{snapshot(\readset)} in \cref{fig:si_implementation} where the values of locations in \readset are recorded in a local array \code{s}.
The local execution of \code{T} is carried out by executing $\denot{\code{T}}$ in \cref{fig:si_implementation}, which is obtained from \code{T} by 
(i) modifying read operations to read locally from the snapshot in \code{s}, and 
(ii) updating the snapshot after each write operation. 
Note that the snapshot must be obtained \emph{atomically} to reflect the memory state at a particular instance (\textit{cf.}~start-timestamp). 
An obvious way to ensure the snapshot atomicity is to lock the locations in the read set, obtain a snapshot, and unlock the read set. 
However, as we must allow for two transactions \emph{reading} from the same location to execute in parallel, 
we opt for \emph{multiple-readers-single-writer} (MRSW) locks.

\newcommand{\Cforin}[2]{\code{{for}\,({#1}\,$\in$\,\ensuremath{#2})}}

\begin{figure}[t]
\centering
  \begin{tabular}{@{}|@{\hspace{3pt}}l@{\hspace{3pt}}|@{\hspace{3pt}}l@{\hspace{3pt}}|@{\hspace{3pt}}l@{\hspace{3pt}}|@{}}
    \hline
    \begin{subfigure}[b]{0.3\textwidth}
    \small
			\begin{enumerate}[label={\color{grey} \arabic*.}, ref=\theenumi, leftmargin=*, itemsep=0pt, labelsep=2pt]
				\item \code{\Cforin{x}{\readset}\,\rlock x\!\!}
				\label{lin:imp1-snapshot1}
				\item \code{snapshot(\readset);} \label{lin:imp1-snapshot2}
				\item \code{\Cforin{x}{\readset}\,\runlock \!x\!\!\!}
				\label{lin:imp1-snapshot3}
				\item \code{\Cforin{x}{\writeset}\,\wlock x}\label{lin:imp1-update1}
		 	    \item \denot{\code{T}};\label{lin:imp1-update2}
				\item \code{\Cforin{x}{\writeset}\,\wunlock \!x\!\!\!\!\!\!}
				\label{lin:imp1-update3}
			\end{enumerate}
    		\caption{}
    		\label{subfig:imp1}
    \end{subfigure} 
    & 
	 \begin{subfigure}[b]{0.335\textwidth}
	 \small
			\begin{enumerate}[label={\color{grey} \arabic*.}, ref=\theenumi, leftmargin=*, itemsep=0pt, labelsep=2pt]
				\item \code{\Cforin{x}{\writeset} \wlock x;}
				\item \code{\Cforin{x}{\readset{\setminus}\writeset}\,\rlock x}
				\item \code{snapshot(\readset);} 
				\item \code{\Cforin{x}{\readset{\setminus}\writeset}\,\runlock \!x\!\!\!\!}
		 	    \item \denot{\code{T}};
				\item \code{\Cforin{x}{\writeset} \wunlock x}
			\end{enumerate}
    		\caption{}
    		\label{subfig:imp2}
    \end{subfigure}     
    & 
	 \begin{subfigure}[b]{0.32\textwidth}
	 \small
		\begin{enumerate}[label={\color{grey} \arabic*.}, ref=\theenumi, leftmargin=*, itemsep=0pt, labelsep=2pt]
			\item \code{\Cforin{x}{\readset\,{\cup}\,\writeset}\,\rlock\!x\!\!}
			\item \code{snapshot(\readset);} 
			\item \code{\Cforin{x}{\readset\,{\cup}\,\writeset} \{}\label{lin:imp3_promote1}
			\item \quad \code{if\,(x\,$\in$\,\writeset)\,promote x}\label{lin:imp3_promote2}
			\item \quad \code{else \runlock x;\quad\}}\label{lin:imp3_promote3} 
		 	\item \denot{\code{T}};
			\item \code{\Cforin{x}{\writeset}\,\wunlock\!x} 
		\end{enumerate}
		\vspace{-10pt}
    		\caption{}
    		\label{subfig:imp3}
    \end{subfigure}         
    \\
\hline
 	Sound: \xmark 
 	& Sound: \cmark 
 	& Sound: \cmark   \\
 	
 	\;\; allows (\ref{ex:LU}), (\ref{ex:LU2})
 	& Complete: \xmark
 	& Complete: \xmark \\
 	
 	\phantom{complete}
 	& \;\; disallows (\ref{ex:WS})
 	& \;\; disallows (\ref{ex:WS2})\\
\hline 	
  \end{tabular}
\caption{Candidate SI implementations of transaction \code{T} given read/write sets \readset,\writeset}
\label{fig:si_candidate_implementations}
\end{figure}

%In effect, we want a definition similar in spirit to global lock atomicity
%(GLA) [27], which is arguably the simplest TM model, and models committed
%transactions as acquiring a global mutual exclusion lock, then accessing
%and updating the data in place, and finally releasing the global lock. Naturally,
%however, the implementation of PSI cannot be that simple.

Let us now try to make this general pattern more precise.
As a first attempt, consider the implementation in \cref{subfig:imp1},
which releases all the reader locks at the end of the snapshot phase before acquiring any writer locks.
This implementation is unsound as it admits the lost update \eqref{ex:LU} anomaly in \cref{fig:litmus_tests} disallowed under SI~\cite{SI-Cerone}.
To understand this, consider a scheduling where T2 runs between lines \ref{lin:imp1-snapshot3} and \ref{lin:imp1-update1} of T1 in \cref{subfig:imp1},
which would result in T1 having read a stale value.
The problem is that the writer locks on \writeset are acquired too late, allowing two conflicting transactions to run concurrently. 
To address this, writer locks must be acquired early enough to pre-empt the concurrent execution of write-write-conflicting transactions. 
Note that locks have to be acquired early even for locations only written by a transaction
to avoid exhibiting a variant of the lost update anomaly \eqref{ex:LU2}.

As such, our second candidate implementation in \cref{subfig:imp2} brings forward the acquisition of writer locks. 
Whilst this implementation is sound (and disallows lost update), it nevertheless disallows behaviours deemed valid under SI such as the write skew anomaly \eqref{ex:WS} in \cref{fig:litmus_tests}, and is thus incomplete.
The problem is that such early acquisition of writer locks not only pre-empts concurrent execution of \emph{write-write-conflicting} transactions, but also those of \emph{read-write-conflicting} transactions (e.g.~\ref{ex:WS}) due to the exclusivity of writer locks.% in MRSW locks. 

To remedy this, in our third candidate implementation in \cref{subfig:imp3} we first acquire weaker reader locks on all locations in \readset or \writeset, and later \emph{promote} the reader locks on \writeset to exclusive writer ones, while releasing the reader locks on \readset.
The promotion of a reader lock signals its intent for exclusive ownership and awaits the release of the lock by other readers before claiming it exclusively as a writer. 
To avoid deadlocks, we further assume that $\readset\,{\cup}\,\writeset$ is ordered so that locks are promoted in the same order by all threads.

Although this implementation is ``more complete'' than the previous one, 
it is still incomplete as it disallows certain behaviour admitted by SI. 
In particular, consider a variant of the write skew anomaly \eqref{ex:WS2} depicted in \cref{fig:litmus_tests}, which is admitted under SI, but not admitted by this implementation. 

To understand why this is admitted by SI, recall the operational SI model using timestamps.%
%\footnote{
%Later in \cref{sec:si} we present the formal declarative SI specification from~\cite{SI-Cerone} and demonstrate the admissibility of  (\ref{ex:WS2}) under this declarative model.
%} 
 Let the domain of timestamps be that of natural numbers $\Nats$.
The behaviour of \eqref{ex:WS2} can be achieved by assigning the following execution intervals for T1: $[t_0^{\code T_1}{=}2, t_c^{\code T_1}{=}2]$; T2: $[t_0^{\code T_2}{=}1, t_c^{\code T_2}{=}4]$; and T3: $[t_0^{\code T_3}{=}3, t_c^{\code T_3}{=}3]$. 
To see why the implementation in \cref{subfig:imp3} does not admit the behaviour in  \eqref{ex:WS2},
let us assume without loss of generality that \x is ordered before \y. 
Upon executing lines \ref{lin:imp3_promote1}-\ref{lin:imp3_promote3}, 
a) T1 promotes \y; 
b) T2 promotes \x and then c) releases the reader lock on \y; and 
d) T3 releases the reader lock on \x. 
To admit the behaviour in  \eqref{ex:WS2}, the release of \y in (c) must occur before the promotion of \y in (a) since otherwise T2 cannot read 0 for \y. 
Similarly, the release of \x in (d) must occur before its promotion in (b). 
On the other hand, since T3 is executed by the same thread after T1, we know that (a) occurs before (d). 
This however leads to circular execution: (b)$\relarrow{}$(c)$\relarrow{}$(a)$\relarrow{}$(d)$\relarrow{}$(b), which cannot be realised.

\begin{figure}[t]
\centering
%\hrule\vspace{5pt}
\begin{tabular}{l|l}
	%\hline 
	\begin{minipage}[]{0.5\textwidth}
		\small
		\begin{enumerate}[label={\color{grey} \arabic*.}, ref=\theenumi, leftmargin=*, itemsep=-2pt, labelsep=6pt]
			\setcounter{enumi}{-1}
			\item \code{\color{blue}LS\,:= $\emptyset$;}
			\label{lin:si_init0}
		 	\item \code{\Cforin{x}{\readset\,{\cup}\,\writeset} \rlock \x}
		 	\item \code{snapshot(\readset);}  \label{lin:si_snapshot}
	
			\item \code{\Cforin{x}{\readset{\setminus}\writeset} \runlock \x} \label{lin:si_runlock}
			\item \code{\Cforin{x}{\writeset} \{ } \label{lin:si_plock}
			\item \quad \code{\textcolor{blue}{if (}can-promote x\color{blue}) LS.add(x) } \label{lin:si_plock_success}
			\item \quad \code{\color{blue}else \{}\label{lin:si_plock_fail1}
			\item \qquad \code{\color{blue}\Cforin{x}{\code{LS}} \wunlock \x} 
	 		\item \qquad \code{\color{blue}\Cforin{x}{\writeset\setminus\code{LS}} \runlock \x} 
	 		\item \qquad \code{\color{blue}goto line \ref{lin:si_init0} \}}
	 		\item \code{\} } \label{lin:si_plock_fail5}
		 	\item \denot{\code{T}}; 
		 	\item \code{\Cforin{x}{\writeset} \wunlock \x}
		 \end{enumerate}
	\vspace{-5pt}
	\end{minipage} &
	$\begin{array}{@{\hspace{5pt}} r @{\hspace{2pt}} l @{\hspace{5pt}}}
		\code{snapshot(\readset)} \eqdef & \code{\Cforin{x}{\readset} \snapshotx:=\,x} \\[2ex]

		\denot{\code{a:=x}} \eqdef & \code{a:=\snapshotx} \\[1ex]
		\denot{\code{x:=a}} \eqdef & \code{x:=a;\,\snapshotx:=a} \\[1ex]
		\denot{S_1\code{;}S_2} \eqdef & \denot{S_1}\code{;}\denot{S_2} \\[1ex]
		\denot{\code{while(e)}\,S} \eqdef & \code{while(e)}\,\denot{S} \\
		& \hspace{-40pt} \ldots\; \text{and so on} \ldots \\[2ex]

		\code{snapshot}_{\mathrm{RSI}}\code{(\readset)} \eqdef &\\
		& \hspace{-60pt} \codelabel{start}  
		\begin{array}[t]{@{} l @{}}
			\code{\Cforin{x}{\readset} \snapshotx:=\,x} \\
			\code{\Cforin{x}{\readset} \{ }\\
			\quad \code{if\,(\snapshotx!=\,x) goto start} \\
			\code{\}}
		\end{array}
	\end{array}$
	%\\\hline 
	\end{tabular}
\caption{SI implementation of transaction \code{T} given \readset, \writeset; the code in {\color{blue}blue} ensures deadlock avoidance. 
The RSI implementation (\cref{sec:rsi}) is obtained by replacing \code{snapshot} on line \ref{lin:si_snapshot} with $\code{snapshot}_{\mathrm{RSI}}$.}
\label{fig:si_implementation}
\end{figure}

To overcome this, in our final candidate execution in \cref{fig:si_implementation} (ignoring the code in {\color{blue} blue}), after obtaining a snapshot, we \emph{first} release the reader locks on \readset, and \emph{then} promote the reader locks on \writeset, rather than simultaneously in one pass.
As we demonstrate in \cref{sec:si}, the implementation in \cref{fig:si_implementation} is both \emph{sound and complete} against its declarative SI specification.

\paragraph{Avoiding Deadlocks}
%To facilitate our locking implementation, we assume that each location \code x is associated with an MRSW (multiple-readers-single-writer) lock.
%Additional care is needed to avoid deadlocks in our implementation. 
As two distinct reader locks on \x may simultaneously attempt to promote their locks, promotion is done on a `first-come-first-served' basis to avoid \emph{deadlocks}. 
A call to \code{can-promote \x} by reader $r$ thus returns a boolean denoting either (i) successful promotion (true); or (ii) failed promotion as another reader $r'$ is currently promoting a lock on \x (false). 
In the latter case, $r$ must release its reader lock on \x to ensure the successful promotion of \xl by $r'$ and thus avoid deadlocks.
To this end, our implementation in \cref{fig:si_implementation} includes a deadlock avoidance mechanism (code in {\color{blue}blue}) as follows.  
We record a list \code{LS} of those locks on the write set that have been successfully promoted so far.
When promoting a lock on \x succeeds (line~\ref{lin:si_plock_success}), the \code{LS} is extended with \x.
On the other hand, when promoting \x fails (line~\ref{lin:si_plock_fail1}), all those locks promoted so far (i.e.~in \code{LS}) as well as those yet to be promoted (i.e.~in \writeset$\setminus$\code{LS}) are released and the transaction is restarted.

\begin{remark}
Note that the deadlock avoidance code in {\color{blue} blue} does not influence the correctness of the implementation in \cref{fig:si_implementation}, 
and is merely included to make the reference implementation more realistic. 
In particular, the implementation without the deadlock avoidance code is both sound and complete against the SI specification, provided that the conditional \code{can-promote} call on line \ref{lin:si_plock_success} is replaced by the blocking \code{promote} call. 
\end{remark}

\subsubsection{Avoiding Over-Synchronisation due to MRSW Locks}
Consider the store buffering program \eqref{ex:SBT} shown in \cref{fig:litmus_tests}.
If, for a moment, we ignore transactional accesses, 
our underlying memory model (RA)---as well as all other weak memory models---allows the annotated weak behaviour.
Intuitively, placing the two transactions that only \emph{read} $z$ in \eqref{ex:SBT} should still allow the weak behaviour
since the two transactions do not need to synchronise in any way.
Nevertheless, most MRSW lock implementations forbid this outcome because 
they use a single global counter to track the number of readers that have acquired the lock, 
which inadvertently also synchronises the readers with one another.
As a result, the two read-only transactions act as memory fences forbidding the weak outcome of \eqref{ex:SBT}\@.
To avoid such synchronisation, in the technical appendix (\cref{app:lock_implementations}) we provide a different MRSW implementation
using a separate location for each thread so that reader lock acquisitions do not synchronise.

To keep the presentation simple, we henceforth assume an abstract specification of a MRSW lock library
providing operations for acquiring/releasing reader/writer locks, as well as promoting reader locks to writer ones.
We require that (1) calls to writer locks (to acquire, release or promote) \emph{synchronise} with all other calls to the lock library; and (2) writer locks provide \emph{mutual exclusion} while held. We formalise these notions in \cref{sec:si}.
These requirements do not restrict synchronisation between \emph{two read} lock calls: two read lock calls may or may not synchronise. 
Synchronisation between read lock calls is relevant only for the completeness of our RSI implementation (handling mixed-mode code);
for that result, we further require that (3) read lock calls not synchronise. 

\subsection{Handling Racy Mixed-Mode Accesses}
Let us consider what happens when data accessed by a transaction is modified concurrently by an uninstrumented atomic non-transactional write.
Since such writes do not acquire any locks, the snapshots taken may include values written by non-transactional accesses.
The result of the snapshot then depends on the order in which the variables are read.
Consider the \eqref{ex:MPT} example in \cref{fig:litmus_tests}.  %, admitted by PSI:
In our implementation, if in the snapshot phase $y$ is read before $x$, then the annotated weak behaviour is not possible 
because the underlying model (RA) disallows this weak ``message passing'' behaviour.
If, however, $x$ is read before $y$, then the weak behaviour is possible.
In essence, this means that the SI implementation described so far is of little use
when there are races between transactional and non-transactional code.
Technically, our SI implementation violates \emph{monotonicity} with respect to wrapping code inside a transaction.
The weak behaviour of the \eqref{ex:MPT} example is disallowed by RA if we remove the transaction block T2,
and yet it is exhibited by our SI implementation with the transaction block.

To get monotonicity under RA, it suffices for the snapshots to read the variables in the same order they are accessed by the transactions.
Since a static calculation of this order is not always possible, 
following~\cite{raad1:psi}, we achieve this by reading each variable twice.
In more detail, our $\code{snapshot}_{\mathrm{RSI}}$ implementation in \cref{fig:si_implementation}
takes \emph{two} snapshots of the locations read by the transaction, 
and checks that they both return the same values for each location.
This ensures that every location is read both before and after every other location in the transaction,
and hence all the high-level happens-before orderings in executions of the transactional program 
are also respected by its implementation.
As we demonstrate in \cref{sec:rsi}, our RSI implementation is both \emph{sound and complete} against our proposed declarative semantics for RSI. 
% Note that the sole difference of the RSI implementation from that of SI lies in obtaining a snapshot (\code{snapshot\_rsi(\readset)}), where each location is read twice as discussed.
There is however one caveat: 
since equality of values is used to determine whether the two snapshots agree, 
we will miss cases where different non-transactional writes to a location write the same value.
In our formal development (see \cref{sec:rsi}), we thus assume that 
if multiple non-transactional writes write the same value to the same location, they
cannot race with the same transaction.
Note that this assumption cannot be lifted without instrumenting non-transactional writes, and thus impeding performance substantially.
That is, to lift this restriction we must instead replace every non-transactional \emph{write} \code{x:=\,v} with \code{\wlock x;\,x:=\,v;\,\wunlock x}. 
\subsection{Non-Prescient  Reference Implementations without Timestamps}
\label{subsec:ideas_si_alt_implementation}
Recall that the SI and RSI implementations in \cref{subsec:ideas_implementations} are prescient in that they require knowledge of the read and write sets of  transactions beforehand.
In what follows we present alternative SI and RSI implementations that are \emph{non-prescient}. 
%As we demonstrate later in \cref{sec:si}, \ref{sec:rsi}, both implementations presented here are sound against the declarative SI and RSI specifications, respectively.
%
%

\subsubsection{Non-Prescient SI Reference Implementation}
In  \cref{fig:si_alt_implementation} we present a \emph{lazy} lock-based reference implementation for SI.
This implementation is \emph{non-prescient} and does not require \textit{a priori} knowledge of the read set \readset and the write set \writeset. 
Rather, the \readset and \writeset are computed on the fly as the execution of the transaction unfolds.
%Moreover, it does not preclude the inclusion of explicit abort instructions in transactions.
As with the SI implementation in \cref{fig:si_implementation}, this implementation does not rely on timestamps and uses MRSW locks to synchronise concurrent accesses to shared data. 
As before, the implementation consults a local \emph{snapshot} at \code s for read operations. 
However, unlike the eager implementation in \cref{fig:si_implementation} where transactional writes are performed \emph{in-place}, 
the implementation in \cref{fig:si_alt_implementation} is \emph{lazy} in that it logs the writes in the local array \code{s} and propagates them to memory at commit time, as we describe shortly.
\begin{figure}[t]
\centering
%\hrule\vspace{5pt}
\begin{tabular}{l|l}
	%\hline 
	\begin{minipage}[]{0.47\textwidth}
		\small
		\begin{enumerate}[label={\color{grey} \arabic*.}, ref=\theenumi, leftmargin=15pt, itemsep=-2pt, labelsep=6pt]
			\setcounter{enumi}{-1}
			\item {\color{blue}\code{LS:=\,$\emptyset$;}}
	 		\label{lin:si_alt_init0}
			\item \code{\readset:=\,$\emptyset$; \writeset:=\,$\emptyset$;}
			\label{lin:si_alt_init1}
			\item \code{\Cforin{x}{\code{Locs}} s[x]:=\,$\bot$}
			\label{lin:si_alt_init2}
		 	\item \denot{\code{T}}; 
		 	\label{lin:si_alt_body}
		 	\item \code{\Cforin{x}{\readset{\setminus}\writeset} \runlock \x}
		 	\label{lin:si_alt_runlock}
		 	\item \code{\Cforin{x}{\writeset} \{}
		 	\item \code{\textcolor{blue}{\quad if (}can-promote x\color{blue}) LS.add(x)} 
		 	\label{lin:si_alt_deadlock_begin}
			\item \code{\color{blue}\quad else \{}
			\item \code{\color{blue}\qquad \Cforin{x}{\code{LS}} \wunlock \x} 
	 		\item \code{\color{blue}\qquad \Cforin{x}{\writeset\,{\setminus}\,\code{LS}} \runlock \x} 
	 		\item \code{\color{blue}\qquad goto line \ref{lin:si_alt_init0} \}} \code{ \}}
	 		\label{lin:si_alt_deadlock_end}
		 	\item \code{\Cforin{x}{\writeset} x\,:=\,s[x]}
		 	\label{lin:si_alt_write}
		 	\item \code{\Cforin{x}{\writeset} \wunlock \x}
		 	\label{lin:si_alt_wunlock}
		 \end{enumerate}
	\vspace{-5pt}
	\end{minipage} &
	$\begin{array}{@{\hspace{5pt}} r @{\hspace{2pt}} l @{\hspace{5pt}}}
%		\vspace{-10pt}\\
		\denot{\code{a:=x}} \eqdef & 
		\begin{array}[t]{@{} l @{}}
			\code{%
				if\,(x\,{$\not\in$}\,{\readset}\,{$\cup$}\,$\writeset$)\,\{
			}\\
			\quad 
			\code{%
					\rlock \x;
					\readset.add(\x);
			} \\
			\quad 
			\code{%
				s[x]:=\,x; 
			} \\
			\code{\}} \\
			\code{%
%				\readset:=\,{\readset}\,$\cup$\,\{x\}; 
%				\readset.add(\x);
				a:=\,s[x];
			}	
		\end{array} \vspace{5pt}\\
		\denot{\code{x:=a}} \eqdef & 
		\begin{array}[t]{@{} l @{}}
			\code{%
				if\,(x\,{$\not\in$}\,{\readset}\,{$\cup$}\,$\writeset$)
					\rlock \x;
			} \\
			\code{%
%				\writeset:=\,{\writeset}\,$\cup$\,\{x\}; 
				\writeset.add(\x);
				s[x]:=\,a;
			} 
		\end{array}  \vspace{5pt}\\	
%		
%		\denot{\code{abort}} \eqdef & 
%		\begin{array}[t]{@{} l @{}}
%			\code{%
%				\Cforin{x}{\readset\,{\cup}\,\writeset} \runlock \x;
%			} \\
%			\code{%
%				return;
%			} 
%		\end{array}  \vspace{5pt}\\		
%		
		\denot{S_1\code{;}S_2} \eqdef & \denot{S_1}\code{;}\denot{S_2}  \vspace{5pt}\\	
		\denot{\code{while(e)}\,S} \eqdef & \code{while(e)}\,\denot{S}  \vspace{5pt}\\	
		& \hspace{-40pt} \ldots\; \text{and so on} \ldots 
	\end{array}$
	%\\\hline 
	\end{tabular}
\caption{Non-prescient SI implementation of transaction \code{T} with \readset and \writeset computed on the fly; 
the code in {\color{blue} blue} ensures deadlock avoidance.
%The RSI implementation (\cref{sec:rsi}) is obtained by replacing \code{snapshot} on line \ref{lin:si_snapshot} with $\code{snapshot}_{\mathrm{RSI}}$.
}
\label{fig:si_alt_implementation}
\end{figure}

Ignoring the code in {\color{blue} blue}, the implementation in \cref{fig:si_alt_implementation} proceeds with initialising \readset and \writeset with $\emptyset$ (line \ref{lin:si_alt_init1}); it then populates the local snapshot array at \code{s} with initial value $\bot$ for each location \x (line \ref{lin:si_alt_init2}).
%Recall that SI transactions operate on a local \emph{snapshot} of the shared state.
%As such, the local snapshot array at \code{s} is populated with initial value $\bot$ for each location \x (line \ref{lin:si_alt_init2}).
%As we describe shortly, the implementation in \cref{fig:si_alt_implementation} is optimistic in that it logs the writes performed by the transaction in the local array \code{s} and propagates them to memory at commit time, rather than performing the writes \emph{in-place} as with its pessimistic counterpart in \cref{fig:si_implementation}.
%
It then executes $\denot{\code T}$ which is obtained from \code T as follows. 
For each \emph{read} operation \code{a:=\,x} in \code T, first the value of \code{s[x]} is inspected to ensure it contains a snapshot of \x.
If this is not the case (i.e.\ \x$\not\in {\readset}\,\cup\writeset$), a reader lock on \x is acquired, a snapshot of \x is recorded in \code{s[x]}, and  
the read set \readset is extended with \x. 
The snapshot value in \code{s[x]} is subsequently returned in \code a. 
Analogously, for each write operation \code{x:=\,a}, the  \writeset is extended with \x, and the written value is lazily logged in \code{s[x]}. 
Recall from our candidate executions in \cref{fig:si_candidate_implementations} that to ensure implementation correctness, for each written location \x, the implementation must first acquire a reader lock on \x, and subsequently promote it to a writer lock.
As such, for each write operation in \code T, the implementation first checks if a reader lock for \x has been acquired (i.e.\ \x$\in {\readset}\,\cup\writeset$) and obtains one if this is not the case.
%Lastly, upon reaching an explicit \code{abort} instruction in \code{T}, the reader locks acquired thus far are released, and the execution of \code{T} returns and terminates, as expected.

Once the execution of $\denot{\code T}$ is completed, the implementation proceeds to \emph{commit} the transaction. 
To this end, the reader locks on \readset are released (line \ref{lin:si_alt_runlock}), 
reader locks on \writeset are promoted to writer ones (line \ref{lin:si_alt_deadlock_begin}), 
the writes logged in \code s are propagated to memory (line \ref{lin:si_alt_write}), 
and finally the writer locks on \writeset are released (line \ref{lin:si_alt_wunlock}).
As we demonstrate later in \cref{sec:si}, the implementation in \cref{fig:si_alt_implementation} is both sound and complete against the declarative SI specification.

Note that the implementation in \cref{fig:si_alt_implementation} is optimistic in that it logs the writes performed by the transaction in the local array \code{s} and propagates them to memory at commit time, rather than performing the writes \emph{in-place} as with its pessimistic counterpart in \cref{fig:si_implementation}.
As before, the code in {\color{blue} blue} ensures deadlock avoidance and is identical to its counterpart in \cref{fig:si_implementation}.
As before, this deadlock avoidance code does not influence the correctness of the implementation and is merely included to make the reference implementation more practical. 
%That is, the implementation without the deadlock avoidance code is both sound and complete against the SI specification, provided that \code{can-promote} on line \ref{lin:si_alt_deadlock_begin} is replaced by \code{promote x}.

\subsubsection{Non-Prescient RSI Reference Implementation}
In  \cref{fig:rsi_alt_implementation} we present a \emph{lazy} lock-based reference implementation for RSI.
As with its SI counterpart, this implementation is non-prescient and computes the \readset and \writeset on the fly.
%Moreover, it does not preclude the inclusion of explicit abort instructions in transactions.
As before, the implementation does not rely on timestamps and uses MRSW locks to synchronise concurrent accesses to shared data. 
Similarly, the implementation consults the local \emph{snapshot} at \code s for read operations, whilst logging write operations lazily in a \emph{write sequence} at \code{wseq}, as we describe shortly.
\begin{figure}[t]
\centering
%\hrule\vspace{5pt}
\begin{tabular}{l|l}
	%\hline 
	\begin{minipage}[]{0.46\textwidth}
		\small
		\begin{enumerate}[label={\color{grey} \arabic*.}, ref=\theenumi, leftmargin=15pt, itemsep=-2pt, labelsep=6pt]
			\setcounter{enumi}{-1}
			\item {\color{blue}\code{LS:=\,$\emptyset$;}}
	 		\label{lin:rsi_alt_init0}
			\item \code{\readset:=\,$\emptyset$; \writeset:=\,$\emptyset$; wseq:=\,[];}
			\label{lin:rsi_alt_init1}
			\item \code{\Cforin{x}{\code{Locs}} s[x]:=\,($\bot$,$\bot$)}
			\label{lin:rsi_alt_init2}
		 	\item \denot{\code{T}}; 
		 	\label{lin:rsi_alt_body}
		 	\item \code{\Cforin{x}{\readset} \{(r,-):=\,s[x];}
		 	\label{lin:rsi_alt_validate_begin}
		 	\item \code{\quad if (x!=r) \{ \comment{read \x again}}
		 	\label{lin:rsi_alt_validate_failure}
			\item \code{\qquad \Cforin{x}{\readset\,{\cup}\,\writeset} \runlock \x}
			\label{lin:rsi_alt_validate_release}
			\item \code{\qquad goto line \ref{lin:rsi_alt_init0} \} \}}
			\label{lin:rsi_alt_validate_end}
		 	\item \code{\Cforin{x}{\readset{\setminus}\writeset} \runlock \x}
		 	\label{lin:rsi_alt_runlock}
		 	\item \code{\Cforin{x}{\writeset} \{}
		 	\item \code{\textcolor{blue}{\quad if (}can-promote x\color{blue}) LS.add(x)\!\!} 
		 	\label{lin:rsi_alt_plock}
		 	\label{lin:rsi_alt_deadlock_begin}
			\item \code{\color{blue}\quad else \{}
			\item \code{\color{blue}\qquad \Cforin{x}{\code{LS}} \wunlock \x} 
	 		\item \code{\color{blue}\qquad \Cforin{x}{\writeset\,{\setminus}\,\code{LS}} \runlock \x} 
	 		\item \code{\color{blue}\qquad goto line \ref{lin:rsi_alt_init0} \}} \code{ \}}
	 		\label{lin:rsi_alt_deadlock_end}
		 	\item \code{\Cforin{(x,v)}{\code{wseq}} x\,:=\,v}
		 	\label{lin:rsi_alt_write}
		 	\item \code{\Cforin{x}{\writeset} \wunlock \x}
		 	\label{lin:rsi_alt_wunlock}
		 \end{enumerate}
	\vspace{-5pt}
	\end{minipage} &
	$\begin{array}{@{\hspace{5pt}} r @{\hspace{2pt}} l @{\hspace{5pt}}}
		\vspace{-10pt}\\
		\denot{\code{a:=x}} \eqdef & 
		\begin{array}[t]{@{} l @{}}
			\code{%
				%if\,(s[x]==($\bot$,$\bot$))\,\{
				if\,(x\,{$\not\in$}\,{\readset}\,{$\cup$}\,$\writeset$)\,\{
			}\\
			\quad 
			\code{%
					\rlock \x;
					\readset.add(\x);
			} \\
			\quad \code{%
				r:=\,x;
				s[x]:=\,(r,r);	
			} \\
			\code{%
				\}
				(-,c):=\,s[x];
				a:=\,c;
			}	
		\end{array} \vspace{5pt}\\
		\denot{\code{x:=a}} \eqdef & 
		\begin{array}[t]{@{} l @{}}
			\code{%
				%if\,(s[x]==($\bot$,$\bot$))
				if\,(x\,{$\not\in$}\,{\readset}\,{$\cup$}\,$\writeset$)
					\rlock \x
			} \\
			\code{%
				\writeset.add(\x);
			} \\
			\code{%
				(r,-):=\,s[x];\,s[x]:=\,(r,a);	
			}	\\
			\code{%			
				wseq:=\,wseq++[(x,a)];
			}
		\end{array}  \vspace{5pt}\\	
%		
%		\denot{\code{abort}} \eqdef & 
%		\begin{array}[t]{@{} l @{}}
%			\code{%
%				\Cforin{x}{\readset\,{\cup}\,\writeset} \runlock \x;
%			} \\
%			\code{%
%				return;
%			} 
%		\end{array}  \vspace{5pt}\\		
%		
		\denot{S_1\code{;}S_2} \eqdef & \denot{S_1}\code{;}\denot{S_2}  \vspace{5pt}\\	
		\denot{\code{while(e)}\,S} \eqdef & \code{while(e)}\,\denot{S}  \vspace{5pt}\\	
		& \hspace{-40pt} \ldots\; \text{and so on} \ldots 
	\end{array}$
	%\\\hline 
	\end{tabular}
\caption{Non-prescient RSI implementation of transaction \code{T} with \readset and \writeset computed on the fly; 
the code in {\color{blue} blue} ensures deadlock avoidance.
%The RSI implementation (\cref{sec:rsi}) is obtained by replacing \code{snapshot} on line \ref{lin:si_snapshot} with $\code{snapshot}_{\mathrm{RSI}}$.
}
\label{fig:rsi_alt_implementation}
\end{figure}

Recall from the RSI implementation in \cref{subsec:ideas_implementations} that to ensure snapshot validity, each location is read twice to preclude intermediate non-transactional writes. 
As such, when writing to a location \x, the initial value read (recorded in \code s) must not be \emph{overwritten} by the transaction to allow for subsequent validation of the snapshot. 
To this end, for each location \x, the snapshot array \code s contains a \emph{pair} of values, $(r, c)$, where $r$ denotes the snapshot value (initial value read), and $c$ denotes the current value which may have overwritten the snapshot value. 

Recall that under weak isolation, the intermediate values written by a transaction may be observed by non-transactional reads. 
For instance, given the $\CtransI{T}{ x:=1; x:=2} \big|\big|\, a:=x$ program, 
the non-transactional read $a:=x$, may read either $1$ or $2$ for $x$. 
As such, at commit time, it is not sufficient solely to propagate the last written value (in program order) to each location (e.g.~to propagate only the $x:=2$ write in the example above).
Rather, to ensure implementation completeness, one must propagate all written values to memory, in the order they appear in the transaction body.
To this end, we track the values written by the transaction as a (FIFO) \emph{write sequence} at location \code{wseq}, containing items of the form $(x, v)$, denoting the location written ($x$) and the associated value ($v$).

Ignoring the code in {\color{blue} blue}, the implementation in \cref{fig:rsi_alt_implementation} initialises \readset and \writeset with $\emptyset$, initialises \code{wseq} as an empty sequence \code{[]} (line \ref{lin:rsi_alt_init1}), 
and populates the local snapshot array \code{s} with initial value $(\bot, \bot)$ for each location \x (line \ref{lin:rsi_alt_init2}). 
It then executes $\denot{\code T}$, obtained from \code T in an analogous manner to that in \cref{fig:si_alt_implementation}.
For every read \code{a:=\,x} in $\denot{\code T}$, the current value recorded for \x in \code s (namely \code c when \code{s[x]} holds \code{(-,c)}) is returned in \code a.
Dually, for every write \code{x:=\,a} in $\denot{\code T}$, the current value recorded for \x in \code s is updated to \code a,
and the write is logged in the write sequence \code{wseq} by appending \code{(x,a)} to it. 

Upon completion of $\denot{\code T}$, the snapshot in \code s is \emph{validated} (lines \ref{lin:rsi_alt_validate_begin}-\ref{lin:rsi_alt_validate_end}). 
Each location \x in \readset is thus read again and its value is compared against the snapshot value in \code{s[x]}. 
If validation fails (line \ref{lin:rsi_alt_validate_failure}), the locks acquired are released (line \ref{lin:rsi_alt_validate_release}) and the transaction is restarted (line \ref{lin:rsi_alt_validate_end}).

If validation succeeds, the transaction is committed:
the reader locks on \readset are released (line \ref{lin:rsi_alt_runlock}), 
the reader locks on \writeset are promoted (line \ref{lin:rsi_alt_plock}), 
the writes in \code{wseq} are propagated to memory in FIFO order (line \ref{lin:rsi_alt_write}), 
and finally the writer locks on \writeset are released (line \ref{lin:rsi_alt_wunlock}). 

As we show in \cref{sec:rsi}, the implementation in \cref{fig:rsi_alt_implementation} is both sound and complete against our proposed declarative specification for RSI.
As before, the code in {\color{blue} blue} ensures deadlock avoidance; it does not influence the implementation correctness and is merely included to make the implementation more practical.

\paragraph{Supporting Explicit Abort Instructions}
It is straightforward to extend the lazy implementations in \cref{fig:si_alt_implementation} and \cref{fig:rsi_alt_implementation} to handle transactions containing explicit abort instructions. 
More concretely, as the effects (writes) of a transaction are logged locally and are not propagated to memory until commit time, upon reaching an \code{abort} in $\denot{T}$ no roll-back is necessary, 
and one can simply release the locks acquired so far and return.
That is, one can extend $\denot{.}$ in \cref{fig:si_alt_implementation} and \cref{fig:rsi_alt_implementation}, and define 
$
	\denot{\code{abort}} 
	\eqdef 
	\code{%
		\Cforin{x}{{\readset}\,{\cup}\,\writeset}\,\runlock \x;
		return%
	}
$.

\section{A Declarative Framework for STM}
\label{sec:framework}

We present the notational conventions used in the remainder of this article, 
and describe a general framework for declarative concurrency models.
Later in this article, we present SI, its extension with non-transactional accesses,
and their lock-based implementations as instances of this general definition.

%the \emph{release-acquire} (RA) fragment~\cite{SRA} of the C11 memory model~\cite{C11}, in which we implement our STM. 
%In \cref{subsec:framework_specification} we describe how we extend this formal model to specify the behaviour of STM programs.

\paragraph{Notation}
%\sloppy
Given a relation $\makerel r$ on a set $A$, we write $\refC{\makerel{r}}$,  $\transC{\makerel{r}}$ and $\reftransC{\makerel{r}}$ for the  reflexive, transitive and reflexive-transitive closure of $\makerel{r}$, respectively. 
We write $\inv{\makerel r}$ for the inverse of $\makerel r$;
$\coerce{\makerel r}{A}$ for $\makerel r \cap (A\times A)$; 
$[A]$ for the identity relation on $A$, i.e.~$\setcomp{(a, a)}{a \in A}$;
$\irr{r}$ for $\nexsts{a} (a, a) \in r$;
and $\acyc{\makerel r}$ for $\irr{\transC{\makerel r}}$.
Given two relations $\makerel r_1$ and $\makerel r_2$, we write $\makerel r_1; \makerel r_2$ for their (left) relational composition, i.e.~$\setcomp{(a, b)}{\exsts{c} (a, c) \in \makerel r_1 \land (c, b) \in \makerel r_2}$.
Lastly, when $\makerel r$ is a strict partial order, we write $\imm{\makerel r}$ for the \emph{immediate} edges in $\makerel r$: $\setcomp{(a, b) \in \makerel r}{\nexsts{c} (a, c) \in \makerel{r} \allowbreak \land (c, b) \in \makerel{r}}$.

%\medskip
%The RA model is given by the fragment of the C11 memory model, 
%where all read accesses are acquire ($\acq$) reads, all writes are release ($\rel$) writes, 
%and all atomic updates (i.e. RMWs) are acquire-release ($\acqrel$) updates. 
%The semantics of a program under RA is defined as a set of \emph{consistent executions}. 
%%We now proceed  with the formal definition of executions.\footnote{When introducing a type $A$, we typically write $a \in A$ to denote that we use $a$ (and its variants $a', a_1$ and so forth) as a meta variable of type $A$.}
%%We present the notational conventions used in the remainder of this article and then proceed with the formal definition of executions.

Assume finite sets of \emph{locations} $\Locs$; 
\emph{values} $\Vals$; \emph{thread identifiers} $\TIDs$,
and  \emph{transaction identifiers} $\TXIDs$.
We use $x, y, z$ to range over locations, $v$ over values, $\tau$ over thread identifiers,
and $\txid$ over transaction identifiers.

\begin{definition}[Events]
An event is a tuple $\tup{n,\tau,\txid,l}$, where $n\in\Nats$ is an event identifier, 
$\tau\in\TIDs \uplus \set{0}$ is a thread identifier ($0$ is used for initialisation events),
$\txid\in\TXIDs \uplus \set{0}$ is a transaction identifier ($0$ is used for non-transactional events),
and $l$ is an event label that takes one of the following forms:
\begin{itemize}%[itemsep=0pt]
\item A memory access label: $\readE{\acq}{x}{v}$ for \emph{reads}; 
	 $\writeE{\rel}{x}{v}$ for \emph{writes}; 
	and $\updateE{\acqrel}{x}{v_r}{v_w}$ for \emph{updates}.
\item A lock label: 	$\rlockE{x}$ for \emph{reader lock acquisition};
$\runlockE{x}$ for \emph{reader lock release};
$\wlockE{x}$ for \emph{writer lock acquisition};
$\wunlockE{x}$ for \emph{writer lock release};
and $\plockE{x}$ for \emph{reader to writer lock promotion}.
\end{itemize}	
We typically use $a$, $b$, and $e$ to range over events.
The functions $\lTHRD$, $\lTX$, $\lLAB$, $\lTYP$, $\lLOC$, $\lVALR$ and $\lVALW$
respectively project the thread identifier, transaction identifier, label, type 
(in $\simpleset{\lR,\lW,\lU,\lRL,\lRU,\lWL,\lWU,\lPL}$), 
location, and read/written values of an event, where applicable. 
We assume only reads and writes are used in transactions 
($\tx{a}\neq 0\implies \typ{a}\in\simpleset{\lR,\lW}$).
\end{definition}

Given a relation $\makerel r$ on events, we write $\perloc{\makerel r}$ for $\setcomp{(a, b) \in \makerel r}{\loc{a} = \loc{b}}$.
Analogously, given a set $A$ of events, we write $A_{x}$ for $\setcomp{a\in A}{\loc a {=} x}$.

\begin{definition}[Execution graphs]
An \emph{execution graph}, $G$, is a tuple of the form 
$(\Events, \po, \rf, \co, $ $\lo)$, where:
\begin{itemize}%[itemsep=0pt]
\item $\Events$ is a set of events, assumed to contain a set $\Events_0$ of initialisation events,
consisting of a write event with label $\writeE{\rel}{x}{0}$ for every $x\in\Locs$.
The sets of \emph{read events} in $\Events$ is denoted by $\Reads \eqdef \setcomp{e \in \Events}{\typ{e} \in\set{\lR,\lU}}$;
\emph{write events} by $\Writes \eqdef \setcomp{e \in \Events}{\typ{e} \in\set{\lW,\lU}}$;
\emph{update events} by $\Updates \eqdef \Reads \cap \Writes$;
and \emph{lock events} by $\LEvents \eqdef \setcomp{e \in \Events}{\typ{e} \in \simpleset{\lRL,\lRU,\lWL,\lWU,\lPL}}$. 
The sets of reader lock acquisition and release events, $\RLocks$ and $\RUnlocks$, 
writer lock acquisition and release events, $\WLocks$ and $\WUnlocks$, 
and lock promotion events $\PLocks$ are defined analogously.
  The set of transactional events in $\Events$ is denoted by $\Transactions$
  ($\Transactions \eqdef  \setcomp{e \in \Events}{\tx{e}\neq 0}$); 
  and the set of non-transactional events is denoted by $\NT$ ($\NT \eqdef \Events \setminus \Transactions$).
	\item $\po \suq \Events \times \Events$ denotes the \emph{`program-order' relation},  defined as a disjoint
	union of strict total orders, each ordering the events of one thread,
	together with $\Events_0\times (\Events \setminus \Events_0)$ that places the initialisation events
	before any other event.
	We assume that events belonging to the same transaction are ordered by $\po$, 
	and that any other event $\po$-between them also belongs to the same transaction.
	\item $\rf \suq \Writes \times \Reads$ denotes the \emph{`reads-from' relation}, 
	defined as a relation between write and read events of the same location with matching read and written values; 
	it is total and functional on reads, 
	i.e.~every read event is related to exactly one write event.
	\item $\co \suq \Writes \times \Writes$ denotes the \emph{`modification-order' relation}, 
defined as a disjoint	union of strict total orders, each of which ordering the write events to one location.
	\item $\lo \suq \LEvents \times \LEvents$ denotes the \emph{`lock-order' relation}, 
	defined as a disjoint union of strict orders, each of which (partially) ordering the lock events to one location.
\end{itemize}
\end{definition}

In the context of an execution graph $G{=}(\Events, \po, \rf, \co, \lo)$---we often use ``$G.$'' as a prefix to make this explicit---the \emph{`same-transaction' relation}, $\st \in \Transactions \times \Transactions$, 
is the equivalence relation given by
$
		\st \eqdef 
		\setcomp{
			(a, b) \in \Transactions \times \Transactions
		}{
			\tx{a} = \tx{b}
		}	
$.
We write $\tlift{\makerel r}$ for lifting a relation $\makerel r \subseteq \Events \times \Events$ to transaction classes:
$
	\tlift{\makerel r} \eqdef \st ; (\makerel r \setminus \st); \st
$.
Analogously, we write $\tin{\makerel r}$ to restrict $\makerel r$ to its \emph{intra-transactional} edges (within a transaction): $\tin{\makerel r} \eqdef \makerel r \cap \st$;
and write $\tout{\makerel r}$ to restrict $\makerel r$ to its \emph{extra-transactional} edges (outside a transaction):  $\tout{\makerel r} \eqdef \makerel r \setminus \st$. 
Lastly,  the `reads-before' relation is defined by $\fr \eqdef (\inv{\rf}; \co) \setminus [\Events]$.

\smallskip
Execution graphs of a given program represent 
traces of shared memory accesses generated by the program.
%We only consider ``partitioned'' programs of the form $\parallel_{\tau\in\TIDs} c_\tau$, 
%where $\parallel$ denotes parallel composition and each $c_i$ is a sequential program.
The set of execution graphs associated with a program can be 
straightforwardly defined by induction over the structure of programs
(see e.g.\ \cite{rsl}).
%We omit this definition here as it depends on the syntax of the implementation programming language.
Each execution of a program $P$ has a particular program \emph{outcome}, 
prescribing the final values of local variables in each thread.
In this initial stage, 
the execution outcomes are almost unrestricted as there are very few constraints on the 
$\rf$, $\co$ and $\lo$ relations.
Such restrictions and thus the permitted outcomes of a program are determined by defining the set of \emph{consistent} executions, 
which is defined separately for each model we consider.
Given a program $P$ and a model $M$, the set $\outcomes P M$ collects the outcomes of every $M$-consistent execution of $P$.

%\begin{remark}\label{rem:general_model}
%Our definitions follow closely C11's semantics~\cite{C11} and 
%include all ingredients needed for our specification and implementation models.
%For simplicity, we have left out memory fences and access modes:
%we assume that all accesses are C11's release/acquire ones.
%(Extending our model to support those outside transactions is straightforward.)
%A key difference with C11 is that our lock events are used to specify reader-writer locks, 
%whereas in C11 they model standard mutexes.
%\viktor{Delete this remark or move to related work?}
%\end{remark}

%\newcommand{\absGSI}{\ensuremath{\absG_{\mathsf{si}}}}
%\newcommand{\impGSI}{\ensuremath{\impG_{\mathsf{si}}}}
\newcommand{\absGSI}{\ensuremath{\impG}}
\newcommand{\impGSI}{\ensuremath{\impG}}

\newcommand{\silor}{{\color{colorLO}\makerel{rlo}}}
\newcommand{\silo}{\lo}
\newcommand{\siloi}{\tin[colorLO]{\silo}}
\newcommand{\silot}{\tlift[colorLO]{\silo}}

\newcommand{\sicon}[1][\absGSI]{\ensuremath{\predFont{si-consistent}(#1)}}

\newcommand{\sihb}{\makerel[colorHB]{si-hb}}
\newcommand{\sipo}{\makerel[colorPO]{si-po}}
\newcommand{\sirf}{\makerel[colorRF]{si-rf}}
\newcommand{\sico}{\makerel[colorMO]{si-mo}}
\newcommand{\sifr}{\makerel[colorFR]{si-rb}}

\newcommand{\EReads}{\tout{\Reads}}

\section{Snapshot Isolation (SI)}\label{sec:si}
We present a declarative specification of SI and demonstrate that the SI implementations presented in \cref{fig:si_implementation} and \cref{fig:si_alt_implementation} are both sound and complete with respect to the SI specification. 

%% \paragraph{A declarative specification of SI STMs}
In~\cite{SI-Cerone} Cerone and Gotsman developed a declarative specification for SI 
using dependency graphs~\cite{Adya-graphs,Adya-thesis}. 
Below we adapt their specification to the notation of \cref{sec:framework}.
As with~\cite{SI-Cerone}, throughout this section, we take \emph{SI execution graphs} to be those in which $\Events = \Transactions \suq (\Reads \cup \Writes) \setminus \Updates$.
That is, the SI model handles transactional code only, consisting solely of read and write events (excluding updates).

\begin{definition}[SI consistency\textnormal{~\cite{SI-Cerone}}]\label{def:si_consistency}
An SI execution $\absGSI = (\Events, \po, \rf, \co, \lo)$ is \emph{SI-consistent}
%, written $\sicon$, 
if the following conditions hold:
\begin{itemize}
	\item $\rfi \cup \coi \cup \fri \subseteq \po $ \labelAxiom{int}{ax:si_int}
	\item $\acyc{(\pot \cup \rft \cup \cot); \refC{\frt}}$ \labelAxiom{ext}{ax:si_ext}
\end{itemize}
%\ori{do we actually need the notation $\sicon$?}
%\azalea{It doesn't take an extra line and we use it in the appendix, so I'd rather keep it.}
%\begin{center}
%\begin{minipage}{.35\textwidth}
%	$\rfi \cup \coi \cup \fri \subseteq \po $ \labelAxiom{int}{ax:si_int}
%\end{minipage}\hfil and\hfil
%\begin{minipage}{.5\textwidth}
%	$\acyc{(\pot \cup \rft \cup \cot); \refC{\frt}}$ \labelAxiom{ext}{ax:si_ext}
%\end{minipage}
%\end{center}
\end{definition}
Informally, \eqref{ax:si_int} ensures the consistency of each transaction internally,
while \eqref{ax:si_ext} provides the synchronisation guarantees among transactions.
In particular, we note that the two conditions together ensure that if two read events in the same transaction read from the same location $x$, and no write to $x$ is $\po$-between them,
then they must read from the same write (known as `internal read consistency').
%
%Note that \ref{ax:ext} states that the only cycles admitted by the SI execution graphs are those that contain two adjacent $frt$ edges. The graph for the anomalous execution of the write-skew example above is given below and indeed the only cycle present in the graph includes two adjacent $\frt$ edges. 

Next, we provide an alternative formulation of SI-consistency which will serve as the basis of our extension with non-transactional accesses in \cref{sec:rsi}. 
In the technical appendix we prove that the two formulations are equivalent.

\begin{proposition}\label{prop:si_consistency}
An SI execution  $\absGSI = (\Events, \po, \rf, \co, \lo)$ is \emph{SI-consistent} if and only if $\ref{ax:si_int}$ holds and 
the `SI-happens-before' relation $\sihb \eqdef \transC{(\pot \cup \rft \cup \cot \cup \sifr)}$ is irreflexive, 
where
$\sifr \eqdef [\EReads]; \frt; [\Writes]$ and $\EReads \eqdef \set{ r \mid \exsts{w} (w,r)\in\rfe }$. 
\end{proposition}
\begin{proof}
The full proof is given in  the technical appendix (see \cref{prop:SI-consistency-full} in \cref{app:aux_lemmas}). 
\end{proof}

Intuitively, SI-happens-before orders events of different transactions
due to either the program order ($\pot$), or synchronisation enforced by the implementation ($\rft \cup \cot \cup \sifr$).
By contrast, events of the same transaction are unordered,
as the implementation may well execute them in a different order 
(in particular, by taking a snapshot, it executes external reads before the writes).

In more detail, the $\rft$ corresponds to transactional synchronisation due to \emph{causality},
i.e.\ when one transaction $\code T_2$ observes an effect of an earlier transaction $\code T_1$.
% (e.g.\ when $\code T_2$ reads a value of \code x written by $\code T_1$).
%A transaction $\code T_1$ is causally-ordered before $\code T_2$, if $\code T_1$ writes to \code x and $\code T_2$ reads \code x observing the write of $\code T_1$. 
The inclusion of $\rft$ ensures that $\code T_2$ cannot read from $\code T_1$ without observing its \emph{entire} effect. 
 This in turn ensures that transactions exhibit `all-or-nothing' behaviour:
 they cannot mix-and-match the values they read. 
For instance, if $\code T_1$ writes to both $x$ and $y$, transaction $\code T_2$ may not read  $x$ from $\code T_1$ but read $y$ from an earlier (in `happens-before' order) transaction $\code T_0$.

The $\cot$ corresponds to transactional synchronisation due to \emph{write-write conflicts}.
Its inclusion enforces write-conflict-freedom of SI transactions: if $\code T_1$ and $\code T_2$ both write to $x$ via events $w_1$ and $w_2$ such that $(w_1,w_2)\in\co$, then $\code T_1$ must commit before $\code T_2$, and thus its entire effect must be visible to $\code T_2$.

To understand $\sifr$, first note that $\EReads$ denotes the \emph{external} transactional reads (i.e.~those reading a value written by another transaction).
That is, the $\EReads$ are the read events that get their values from the transactional snapshot phases.
By contrast, internal reads (those reading a value written by the same transaction) happen only after the snapshot is taken.
Now let there be an $\frt$ edge between two transactions, $\code T_1$ and $\code T_2$. 
This means there exist a read event $r$ of $\code T_1$ and a write event $w$ of $\code T_2$ such that $(r,w)\in\fr$;
i.e.\ there exists $w'$ such that $(w', r) \in \rf$ and $(w', w) \in \co$.
If $r$ reads internally (i.e.~$w'$ is an event in $\code T_1$), then $\code T_1$ and $\code T_2$ are conflicting transactions and as accounted by $\cot$ described above, all events of $\code T_1$ happen before those of $\code T_2$. 
Now, let us consider the case when $r$ reads externally ($w'$ is not in $\code T_1$).
From the timestamped model of SI, there exists a start-timestamp $t_0^{\code T_1}$ as of which the $\code T_1$ snapshot (all its external reads including $r$) is recorded.
Similarly, there exists a commit-timestamp $t_c^{\code T_2}$ as of which the updates of $\code T_2$ (including $w$) are committed.
Moreover, since $(r, w) \in \fr$ we know $t_0^{\code T_1} < t_c^{\code T_2}$ (otherwise $r$ must read the value written by $w$ and not $w'$). 
That is, we know all events in the snapshot of $\code T_1$ (i.e.~all external reads in $\code T_1$) happen before all writes of $\code T_2$.\footnote{By taking 
 $\frt$ instead of $\sifr$ in \cref{prop:si_consistency} one obtains a characterisation of \emph{serialisability}.}
%\ori{Something in this explanation may be confusing: it seems to explain $[\EReads];([\EReads]; \fr)_T; [\Writes]$
%rather than $[\EReads]; \frt; [\Writes]$}
%\viktor{Updated explanation. Better?}

%\subsection{A Lock-based SI Implementation}\label{subsec:si_implementation}

%% \paragraph{A lock-based implementation of SI STMs}

We use the  declarative framework in \cref{sec:framework} to formalise the semantics of our implementation.
Here, our programs include only non-transactional code,
and thus \emph{implementation execution graphs} are taken as those in which $\Transactions=\emptyset$. 
Furthermore, we assume that locks in implementation programs are used in a \emph{well-formed} manner: the sequence of
lock events for \emph{each location}, in each thread (following $\po$), should match (a prefix of) the regular expression
$(\lRL\cdot\lRU \;|\;\lWL\cdot\lWU \;|\; \lRL\cdot\lPL\cdot\lWU)^*$.\label{par:lock_traces}
For instance, a thread never releases a lock, without having acquired it earlier in the program.
As a consistency predicate on execution graphs, we use the C11 release/acquire consistency 
augmented with certain constraints on lock events.

\begin{definition}\label{def:si_implementation_consistency}
An implementation execution graph $\impGSI = (\Events, \po, \rf, \co, \lo)$ is \emph{RA-consistent}  if the  following hold, 
where $\hb \eqdef \transC{(\po \cup \rf \cup \lo)}$ denotes the `RA-happens-before' relation:
\begin{itemize}%[itemsep=0pt]
\item 
$\for{x}\for{a \in \WLocks_x \cup \WUnlocks_x \cup \PLocks_x, b \in \LEvents_x} a=b \lor (a, b) \in \lo \lor (b, a) \in \lo$ \labelAxiom{WSync}{ax:lock_wsync}
\item 
$[\WLocks \cup \PLocks]; (\silo \setminus \po); [\LEvents] \suq \po; [\WUnlocks]; \silo$
\labelAxiom{WEx}{ax:lock_wex}
\item 
$[\RLocks]; (\silo \setminus \po); [\WLocks \cup \PLocks] \suq \po;[\RUnlocks \cup \PLocks]; \silo$
\labelAxiom{RShare}{ax:lock_rshared}
\item 
$\acyc{\hbloc \cup \co \cup \fr}$
\labelAxiom{Acyc}{ax:ra_acyc}
\end{itemize}
\end{definition}

The \eqref{ax:lock_wsync} states that write lock calls (to acquire, release or promote) \emph{synchronise} with all other calls to the same lock.

The next two constraints ensure the `single-writer-multiple-readers' paradigm.
In particular, \eqref{ax:lock_wex} states that write locks provide \emph{mutual exclusion} while held: any lock event $l$ of thread $\tau$  $\lo$-after a write lock acquisition or promotion event $l'$ of another thread $\tau'$, is $\lo$-after a subsequent write lock release event $u$ of $\tau'$ (i.e.~$(l', u) \in \po$ and $(u, l) \in \lo$). 
As such, the lock cannot be acquired (in read or write mode) by another thread until it has been released by its current owner.

The \eqref{ax:lock_rshared} analogously states that once a thread acquires a lock in read mode, 
the lock cannot be acquired in write mode by other threads until it has either been released, or promoted to a writer lock (and subsequently released) by its owner.
Note that this does not preclude other threads from simultaneously acquiring the lock in read mode.
In the technical appendix we present two MRSW lock implementations that satisfy the conditions outlined above. 

The last constraint (\ref{ax:ra_acyc}) is that of C11 RA consistency~\cite{SRA}, 
with the $\hb$ relation extended with $\lo$.

\begin{remark}\label{rem:sihb_general_model}
Our choice of implementing the SI STMs on top of the RA fragment is purely for presentational convenience. 
Indeed, it is easy to observe that execution graphs of $\denot{P}$ are data race free,
and thus, \ref{ax:ra_acyc} could be replaced by any condition that implies 
$\for{x} ([\Writes_x];\transC{(\po \cup \lo)};[\Writes_x];\transC{(\po \cup \lo)};[\Reads_x]) \cap \rf=\emptyset$ 
and that is implied by 
$\acyc{\po \cup \rf \cup \lo \cup \co \cup \fr}$.
In particular, the C11 non-atomic accesses or sequentially consistent accesses may be used.
\end{remark}

We next demonstrate that our SI implementations in \cref{fig:si_implementation} and \cref{fig:si_alt_implementation} are both \emph{sound and complete} with respect to the declarative specification given above. 
The proofs are non-trivial and the full proofs are given in the technical appendix.

\begin{theorem}[Soundness and completeness]
\label{thm:si_snd_cmp}
%Let $P$ be a transactional program
%and let $\denot{P}$ denote its implementation as given in \cref{fig:si_implementation}.
%Then, $\outcomes{P}{\mathrm{SI}} = \outcomes{\denot{P}}{\mathrm{RA}}$.
Let $P$ be a transactional program;
let $\denot{P}_{\textsc e}$ denote its eager implementation as given in \cref{fig:si_implementation} 
and $\denot{P}_{\textsc l}$ denote its lazy implementation as given in \cref{fig:si_alt_implementation}.
Then:
\[
	\outcomes{P}{\mathrm{SI}} = \outcomes{\denot{P}_{\textsc e}}{\mathrm{RA}} = \outcomes{\denot{P}_{\textsc l}}{\mathrm{RA}}
\]	
\end{theorem}
\begin{proof}
The full proofs of both implementations is given in the technical appendix (\cref{app:si} and \cref{app:si_alternative}, respectively).
%The full proof for the eager implementation is given in \cref{app:si}. 
%The full proof of the lazy implementation is given in \cref{app:si_alterative}.
\end{proof}

%\begin{theorem}[Soundness and completeness]
%\label{thm:si_alt_snd_cmp}
%Let $P$ be a transactional program
%and let $\denot{P}$ denote its implementation as given in \cref{fig:si_alt_implementation}.
%Then, $\outcomes{P}{\mathrm{SI}} = \outcomes{\denot{P}}{\mathrm{RA}}$.
%\end{theorem}

\paragraph{Stronger MRSW Locks} 
As noted in \cref{sec:ideas}, for both (prescient and non-prescient) SI implementations 
our soundness and completeness proofs show that 
the same result holds for a stronger lock specification,
in which reader locks synchronise as well. Formally, this specification is obtained by adding the constraint:
\begin{itemize}[itemsep=0pt]
\item 
$\for{x}\for{a, b \in \RLocks_x \cup \RUnlocks_x} a=b \lor (a, b) \in \lo \lor (b, a) \in \lo$ \labelAxiom{RSync}{ax:lock_rsync}
\end{itemize}
Soundness of this stronger specification ($\outcomes{\denot{P}_{\textsc x}}{\mathrm{RA}} \suq \outcomes{P}{\mathrm{SI}}$ for $\textsc x \in \{\textsc e, \textsc l\}$) follows immediately from \cref{thm:si_snd_cmp}.
% and \cref{thm:si_alt_snd_cmp}. 
Completeness ($\outcomes{P}{\mathrm{SI}} \suq \outcomes{\denot{P}_{\textsc x}}{\mathrm{RA}}$ for $\textsc x \in \{\textsc e, \textsc l\}$), however, 
is more subtle, as we need to additionally satisfy (\ref{ax:lock_rsync}) when constructing $\lo$. 
While we can do so for SI, 
%as we describe shortly in \cref{sec:rsi}, 
it is essential for the completeness of our RSI implementations that reader locks not synchronise, as shown by \eqref{ex:SBT} in \cref{sec:ideas}.

In the technical appendix we present two MRSW lock implementations (see  \cref{app:lock_implementations}).
Both implementations are \emph{sound}  against the $\lo$ conditions in \cref{def:si_implementation_consistency}.
Additionally, the first implementation is \emph{complete} against the conditions of \cref{def:si_implementation_consistency} augmented with \eqref{ax:lock_rsync},
whilst the second is complete against the conditions of \cref{def:si_implementation_consistency} alone.
%\azalea{Why is cref printing Thm. (\cref{def:si_implementation_consistency}) and not Def.?}

%\begin{remark}\label{rem:sihb_general_model}
%Our choice of implementing the SI STMs on top of the RA fragment is purely for presentational convenience.
%Had we chosen a different model, the implementation consistency in \cref{def:si_implementation_consistency} would change minimally with $\hb \eqdef \transC{(\po \cup \sw \cup \lo)}$, where $\rf$ is replaced with $\sw$, denoting the C11 `synchronises-with' relation.
%As in the RA model $\sw \eqdef \rf$, we have inlined this in \cref{def:si_implementation_consistency}.
%
%\ori{do we want to mention that one can use even non-atomics? or even SC-accesses?
%It feels to me that this remark somehow misses the point... perhaps something like:}
%

\newcommand{\absGRSI}{\ensuremath{\impG}}
\newcommand{\impGRSI}{\ensuremath{\impG}}

\newcommand{\rsilo}{\lo}
\newcommand{\rsiloi}{\tin[colorLO]{\rsilo}}
\newcommand{\rsilot}{\tlift[colorLO]{\rsilo}}

\newcommand{\rsicon}[1][\absGRSI]{\ensuremath{\predFont{rsi-consistent}(#1)}}

\newcommand{\rsihb}{\makerel[colorHB]{rsi-hb}}
\newcommand{\rsipo}{\makerel[colorPO]{rsi-po}}
\newcommand{\rsirf}{\makerel[colorRF]{rsi-rf}}
\newcommand{\rsico}{\makerel[colorMO]{rsi-mo}}
\newcommand{\rsifr}{\makerel[colorFR]{rsi-rb}}

\newcommand{\rsipoi}{\tin{\makerel{rsi-po}}}
\newcommand{\rsihbloc}{{\perloc{\rsihb}}}

\section{Robust Snapshot Isolation (RSI)}\label{sec:rsi}
We explore the semantics of SI STMs in the presence of non-transactional code with \emph{weak isolation} guarantees
(see \cref{sec:ideas}). We refer to this model as \emph{robust snapshot isolation} (RSI), due to its ability to provide SI guarantees between transactions even in the presence of non-transactional code.
We propose the first declarative specification of RSI programs and develop two lock-based reference implementations that are both \emph{sound and complete} against our proposed specification.
%
%\subsection{A Declarative Specification of RSI STMs}

\begin{figure}[t]
\begin{center}
\begin{tabular}{|@{}c@{}|@{}c@{}|@{}c@{}|}
%\hrule \vspace{3pt}
 \hline
\begin{subfigure}[b]{0.29\textwidth}\small\centering
\begin{tikzpicture}[yscale=1,xscale=1.05]
  \draw[draw=black,rounded corners,dotted,fill=blue!10] (-1.65,1.4) rectangle (-0.15,-0.3);
  \node[draw=black,rounded corners,dotted,fill=blue!20] at (-1.55,1.4) {$\code T$};
  \node (01)  at (-0.75,2) {$\evlab{\wlab}{}{x}{0}$};
  \node (02)  at (0.75,2) {$\evlab{\wlab}{}{y}{0}$};
  \node (21)  at (-0.9,1) {$w_1\colon\evlab{\wlab}{}{y}{1}$ };
  \node (22)  at (-0.9,0) {$w_2\colon\evlab{\wlab}{}{x}{1}$ };
  \node (31)  at (0.8,1) {$\evlab{\rlab}{}{x}{1}$ };
  \node (32)  at (0.8,0) {$\evlab{\rlab}{}{y}{0}$ };
  \draw[po] (21) edge (22);
  \draw[po] (31) edge (32);
  \draw[po] (01) edge (21) edge (31);
  \draw[po] (02) edge (21) edge (31);
  \draw[rf,bend left=40] (02) edge node[right]{$\!\rf$} (32);
%  \draw[mo,bend right=85] (01) edge node[left]{$\!\!\!\co\!$} (22);
  \draw[rf] (22) edge node[above]{$\rf$} (31);
  \draw[fr] (32) edge node[below]{$\fr$}  (21);
  \draw[mo,pos=.3,bend right=10] (02) edge node[above=2pt]{$\co$} (21); 
%  \draw[rf,bend left=5] (02) edge (22);
\end{tikzpicture} 
%\vspace{-8pt}
\caption{}
\label{subfig:rsi_po}
\end{subfigure}
&
\begin{subfigure}[b]{0.29\textwidth}\small\centering
\begin{tikzpicture}[yscale=1,xscale=1.05]
  \draw[draw=black,rounded corners,dotted,fill=blue!10] (-1.7,1.4) rectangle (-0.25,-0.3);
  \node[draw=black,rounded corners,dotted,fill=blue!20] at (-1.6,1.4) {$\code T$};
  \node (01)  at (-0.8,2) {$\evlab{\wlab}{}{x}{0}$};
  \node (02)  at (0.8,2) {$\evlab{\wlab}{}{y}{0}$};
  \node (21)  at (-1,1) {$r'\!\colon\!\evlab{\rlab}{}{y}{0}$ };
  \node (22)  at (-1,0) {$r\colon\!\evlab{\rlab}{}{x}{1}$ };
  \node (31)  at (0.9,1) {$\evlab{\wlab}{}{y}{1}$ };
  \node (32)  at (0.9,0) {$w\colon\!\evlab{\wlab}{}{x}{1}$ };
  \draw[po] (21) edge (22);
  \draw[po] (31) edge (32);
  \draw[po] (01) edge (21) edge (31);
  \draw[po] (02) edge (21) edge (31);
  \draw[mo,bend left=20] (02) edge node[right]{$\co$} (31);
  \draw[rf] (32) edge node[above]{$\rf$} (22);
  \draw[fr] (21) edge node[below]{$\fr$}  (31);
  \draw[rf,bend right=10] (02) edge node[above=2pt]{$\rf$} (21);
%  \draw[rf,bend left=5] (02) edge (22);
\end{tikzpicture}
%\vspace{-15pt}
\caption{}
\label{subfig:rsi_ntrf}
\end{subfigure}
%~\hfill
&
\begin{subfigure}[b]{0.41\textwidth}\small\centering
\begin{tikzpicture}[yscale=1,xscale=1]
  \draw[draw=black,rounded corners,dotted,fill=blue!10] (-2.55,0.4) rectangle (-1.05,-0.2);
  \node[draw=black,rounded corners,dotted,fill=blue!20] at (-2.4,0.5) {$\code T_1$};
  \draw[draw=black,rounded corners,dotted,fill=orange!10] (2.1,1.35) rectangle (0.7,0.7);
  \node[draw=black,rounded corners,dotted,fill=orange!20] at (2,1.5) {$\code T_2$};
  \node (01)  at (-0.8,1.8) {$\evlab{\wlab}{}{x}{0}$};
  \node (02)  at (0.8,1.8) {$\evlab{\wlab}{}{y}{0}$};
  \node (21)  at (-1.8,1) {$\evlab{\wlab}{}{y}{1}$ };
  \node (22)  at (-1.8,0.1) {$w\colon\evlab{\wlab}{}{x}{1}$ };
  \node (23)  at (-1.8,-0.8) {$\evlab{\rlab}{}{x}{2}$ };
%  \node (31)  at (0,1) {$r_1:\evlab{\rlab}{}{x}{1}$ };
  \node (31)  at (-0.2,1) {$\evlab{\wlab}{}{x}{2}$ };
  \node (41)  at (1.4,1) {$r\colon\evlab{\rlab}{}{x}{2}$ };
  \node (42)  at (1.4,0) {$\evlab{\rlab}{}{y}{0}$ };
  \draw[po] (21) edge (22);
  \draw[po] (22)  edge (23);
%  \draw[po] (31) edge (32);
  \draw[po] (41) edge (42);
  \draw[po] (01) edge (21) edge (31) edge (41);
  \draw[po] (02) edge (21) edge (31) edge (41);
%  \draw[mo,bend left=20] (02) edge node[right]{$\co$} (31);
  \draw[rf] (31) edge node[below]{$\rf$} (41);
  \draw[mo] (22) edge node[below]{$\co$}  (31);
%  \draw[rf] (21) edge node[below]{$\rf$}  (31);
  \draw[rf,bend right=40] (02) edge (42);
  \draw[fr] (42) edge node[below]{$\fr$} (21);
\end{tikzpicture} 
\vspace{-17pt}
\caption{}
\label{subfig:rsi_trf}
\end{subfigure}\\
\hline
%\vspace{3pt}
%\hrule \vspace{3pt}
\end{tabular}
\end{center}\vspace{-10pt}
\caption{RSI-inconsistent executions due to (\subref{subfig:rsi_po}) $\rsipo$; 
(\subref{subfig:rsi_ntrf}) $[\NT];\rf;\st$;
(\subref{subfig:rsi_trf}) $\tlift{(\co;\rf)}$
}
\label{fig:rsi_inconsistent_executions}
\end{figure}

\paragraph{A Declarative Specification of RSI STMs}
We formulate a declarative specification of RSI semantics by adapting the SI semantics in \cref{prop:si_consistency} to account for non-transactional accesses. 
To specify the abstract behaviour of RSI programs, \emph{RSI execution graphs} are taken to be those in which $\LEvents=\emptyset$. 
Moreover, as with SI graphs, RSI execution graphs are those in which $\Transactions \suq (\Reads \cup \Writes) \setminus \Updates$.
That is, RSI transactions comprise solely read and write events, excluding updates. 

\begin{definition}[RSI consistency]\label{def:rsi_consistency}
An execution $\absGRSI = (\Events, \po, \rf, \co, \lo)$ is \emph{RSI-consistent} iff $\ref{ax:si_int}$ holds and 
$\acyc{\rsihbloc \cup \co \cup \fr}$, 
where $\rsihb \eqdef \transC{(\rsipo \cup \rsirf \cup \cot \cup \sifr)}$ is the `RSI-happens-before' relation, with 
$\rsipo \eqdef (\po \setminus \poi) \cup [\Writes]; \poi;[\Writes]$ and 
$\rsirf \eqdef  (\rf; [\NT]) \cup ([\NT]; \rf; \st) \cup \tlift{\rf} \cup  \tlift{(\co; \rf)}$.
\end{definition}

As with SI and RA, we characterise the set of executions admitted by RSI as graphs that lack cycles of certain shapes.
To account for non-transactional accesses, similar to RA, we require $\rsihbloc \cup \co \cup \fr$ to be acyclic
(recall that $\rsihbloc\defeq\setcomp{(a, b) \in \rsihb}{\loc{a} = \loc{b}}$).
The RSI-happens-before relation $\rsihb$ includes both the synchronisation edges enforced by the transactional implementation (as in $\sihb$),
and those due to non-transactional accesses (as in $\hb$ of the RA consistency).
The $\rsihb$ relation itself is rather similar to $\sihb$. 
In particular, the $\cot$ and $\sifr$ subparts can be justified as in $\sihb$;
the difference between the two lies in $\rsipo$ and $\rsirf$. 

To justify $\rsipo$, recall from \cref{sec:si} that $\sihb$ includes $\pot$.
The $\rsipo$ is indeed a strengthening of $\pot$ to account for non-transactional events: it additionally includes 
(i) $\po$ to and from non-transactional events; and 
(ii) $\po$ between two \emph{write} events in a transaction.  
We believe (i) comes as no surprise to the reader; for (ii), consider the execution graph in \cref{subfig:rsi_po}, where transaction $\code T$ is denoted by the dashed box labelled $\code T$, comprising the write events $w_1$ and $w_2$.
Removing the $\code T$ block (with $w_1$ and $w_2$ as non-transactional writes), this execution is deemed inconsistent, as this weak ``message passing'' behaviour is disallowed in the RA model.
We argue that the analogous transactional behaviour in \cref{subfig:rsi_po} must be similarly disallowed 
to maintain monotonicity with respect to wrapping non-transactional code in a transaction (see \cref{thm:rsi-mono}). 
As in SI, we cannot include the entire $\po$ in $\rsihb$ because the write-read order in transactions is not preserved by the implementation.

Similarly, $\rsirf$ is a strengthening of $\rft$ to account for non-transactional events: 
in the absence of non-transactional events $\rsirf$ reduces to $\rft \cup \tlift{(\co; \rf)}$ which is contained in $\sihb$.
The $\rf; [\NT]$ part is required to preserve the `happens-before' relation for non-transactional code. 
That is, as $\rf$ is included in the $\hb$ relation of underlying memory model (RA), it is also included in $\rsihb$.

The $[\NT]; \rf; \st$ part asserts that in an execution 
where a read event $r$ of transaction $\code T$ reads from a non-transactional write $w$,
the snapshot of $\code T$ reads from $w$ and so all events of $\code T$ happen after $w$.
Thus, in \cref{subfig:rsi_ntrf}, $r'$ cannot read from the overwritten initialisation write to $y$.

For the $\tlift{(\co; \rf)}$ part, 
consider the execution graph in \cref{subfig:rsi_trf} where there is a write event $w$ of transaction $\code T_1$ and a read event $r$ of transaction $\code T_2$ such that $(w,r) \in \co;\rf$.
Then, transaction $\code T_2$ must acquire the read lock of $\loc w$ after $\code T_1$ releases the writer lock,
which in turn means that every event of $\code T_1$ happens before every event of $\code T_2$.

\begin{remark}\label{rem:rsihb_general_model}
Recall that our choice of modelling SI and RSI STMs in the RA fragment is purely for presentational convenience (see \cref{rem:sihb_general_model}). Had we chosen a different model, the RSI consistency definition (\cref{def:rsi_consistency}) would largely remain unchanged, with the exception of $\rsirf \eqdef \shade{(\sw; [\NT]) \cup ([\NT]; \sw; \st)} \cup \rft \cup \tlift{(\co; \rf)}$, 
where in the \shadetext{highlighted} changes the $\rf$ relation is replaced with $\sw$, denoting the `synchronises-with' relation. As in the RA model $\sw \eqdef \rf$, we have inlined this in \cref{def:rsi_consistency}.
\end{remark}

\paragraph{SI and RSI Consistency}
We next demonstrate that in the absence of non-transactional code, the definitions of SI-consistency (\cref{prop:si_consistency}) and RSI-consistency (\cref{def:rsi_consistency}) coincide. 
That is, for all executions $G$, if $G.\NT = \emptyset$, then $G$ is SI-consistent if and only if $G$ is RSI-consistent.
%$\outcomes{P}{\mathrm{SI}} = \outcomes{P}{\mathrm{RSI}}$.
This is captured in the following theorem with its full proof given in the technical appendix. 

\begin{theorem}
For all executions $G$, if $G.\NT = \emptyset$, then:
\[
	G \text{ is SI-consistent}
	\iff
	G \text{ is RSI-consistent}
\]
\end{theorem}
\begin{proof}
The full proof is given in the technical appendix (see \cref{thm:SI_RSI_Equiv} in  \cref{app:aux_lemmas}).
\end{proof}

Note that the above theorem implies that for all transactional programs $P$, if $P$ contains no non-transactional accesses, then 
$\outcomes{P}{\mathrm{SI}} = \outcomes{P}{\mathrm{RSI}}$.

\paragraph{RSI Monotonicity}
We next prove the \emph{monotonicity} of RSI when wrapping non-transactional events into a transaction.
That is, wrapping a block of non-transactional code inside a new transaction does not introduce additional behaviours.
More concretely, given a program $\prog$, when a block of non-transactional code in $\prog$ is wrapped inside a new transaction to obtain a new program $\prog_{\code T}$,  then $\outcomes{P_{\code T}}{\mathrm{RSI}} \suq \outcomes{P}{\mathrm{RSI}}$.
This is captured in the theorem below, with its full proof given in the technical appendix.

%(\cref{app:aux_lemmas} and \cref{app:rsi} respectively).
%\azalea{Why is cref printing Thm. (\cref{def:rsi_consistency}) and not Def.?}
%
\begin{theorem}[Monotonicity]\label{thm:rsi-mono}
Let $P_{\code T}$ and $P$ be RSI programs such that 
$P_{\code T}$ is obtained from $P$ by wrapping a block of non-transactional code inside a new transaction.
Then: 
\[
	\outcomes{P_{\code T}}{\mathrm{RSI}} \suq \outcomes{P}{\mathrm{RSI}}
\]	
%Let $\absGRSI_T$ be an RSI execution graph obtained from an RSI execution graph $\absGRSI$ by wrapping some non-transactional events 
%inside a new transaction.
%If $\absGRSI_T$ is RSI-consistent, then so is $\absGRSI$.
\end{theorem}
\begin{proof}
The full proof is given in the technical appendix (see \cref{thm:rsi-mono-full} in \cref{app:aux_lemmas}).
\end{proof}

Lastly, we show that our RSI implementations in \cref{sec:ideas} (\cref{fig:si_implementation} and \cref{fig:rsi_alt_implementation}) are sound and complete with respect to \cref{def:rsi_consistency}.
This is captured in the theorem below. 
The soundness and completeness proofs are non-trivial; the full proofs are given in the technical appendix.

\begin{theorem}[Soundness and completeness]\label{thrm:rsi_sound_complete}
%Let $P$ be a program that possibly mixes transactional and non-transactional code and let $\denot{P}$ denote its RSI implementation as given in \cref{fig:si_implementation}.
%If for every location $x$ and value $v$, 
%every RSI-consistent execution of $P$ contains either i) at most one non-transactional write of $v$ to $x$; or
%ii) all non-transactional writes of $v$ to $x$ are happens-before-ordered with respect to all transactions accessing $x$, 
%then $\outcomes{P}{\mathrm{RSI}} = \outcomes{\denot{P}}{\mathrm{RA}}$.
Let $P$ be a program that possibly mixes transactional and non-transactional code.
Let $\denot{P}_{\textsc e}$ denote its eager RSI implementation as given in \cref{fig:si_implementation}
and $\denot{P}_{\textsc l}$ denote its lazy RSI implementation as given in \cref{fig:rsi_alt_implementation}.

If for every location $x$ and value $v$, 
every RSI-consistent execution of $P$ contains either (i) at most one non-transactional write of $v$ to $x$; or
(ii) all non-transactional writes of $v$ to $x$ are happens-before-ordered with respect to all transactions accessing $x$, 
then:
\[
	\outcomes{P}{\mathrm{RSI}} = \outcomes{\denot{P}_{\textsc e}}{\mathrm{RA}} = \outcomes{\denot{P}_{\textsc l}}{\mathrm{RA}}
\]	
\end{theorem}
\begin{proof}
The full proofs of both implementations is given in the technical appendix (\cref{app:rsi} and \cref{app:rsi_alternative}, respectively).
%The full proof for the eager implementation is given in \cref{app:rsi}. 
%The full proof of the lazy implementation is given in \cref{app:rsi_alternative}.
\end{proof}

\section{Related and Future Work}\label{sec:related_and_future_work}

Much work has been done in formalising the semantics of weakly consistent \emph{database transactions}~\cite{SI,PSI,Adya-thesis,Adya-graphs,cerone15,PSI-Cerone,SI-Cerone,Cerone-Algebraic,Crooks-isolation,Gotsman-reasoning}, both operationally and declaratively. 
On the operational side, Berenson et al.~\cite{SI} gave an operational model of SI as a multi-version concurrent algorithm. 
Later, Sovran et al.~\cite{PSI} described and operationally defined the \emph{parallel snapshot isolation} model (PSI), as a close relative of SI with weaker guarantees. 

On the declarative side, Adya et al.~\cite{Adya-thesis,Adya-graphs} introduced \emph{dependency graphs} (similar to execution graphs of our framework in \cref{sec:framework}) for specifying transactional semantics and formalised several ANSI isolation levels.
Cerone et al.~\cite{cerone15,PSI-Cerone} introduced \emph{abstract executions} and formalised several isolation levels including SI and PSI.
Later in~\cite{SI-Cerone}, they used dependency graphs of Adya to develop equivalent SI and PSI semantics; 
recently in~\cite{Cerone-Algebraic}, they provided a set of algebraic laws for connecting these two declarative styles.

To facilitate client-side reasoning about the behaviour of database transactions,
Gotsman et al.~\cite{Gotsman-reasoning} developed a proof rule for proving invariants of client applications under a number of consistency models. 

Recently, Kaki et al.~\cite{Kaki-popl18} developed a program logic to reason about transactions under ANSI SQL isolation levels (including SI).
To do this, they formulated an operational model of such programs (parametric in the isolation level).
They then proved the soundness of their \emph{logic} with respect to their proposed operational model.
However, the authors did not establish the \emph{soundness} or \emph{completeness} of their \emph{operational model} against existing formal semantics, e.g.~\cite{SI-Cerone}. 
The lack of the completeness result means that their proposed operational model may exclude behaviours deemed valid by the corresponding declarative models. 
This is a particular limitation as possibly many valid behaviours cannot be shown correct using the logic and is thus detrimental to its usability. 

By contrast, the semantics of transactions in the STM setting with mixed-mode (both transactional and non-transactional) accesses is under-explored on both operational and declarative sides.
Recently, Dongol at al.~\cite{Dongol-popl18} applied execution graphs~\cite{cats-Alglave} to specify the behaviour of \emph{serialisable} STM programs under weak memory models. 
Raad et al.~\cite{raad1:psi} formalised the semantics of PSI STMs both declaratively (using execution graphs) and operationally (as lock-based reference implementations).
Neither work, however, handles the semantics of SI STMs under weaker isolation guarantees.

Finally, Khyzha et al.~\cite{Khyzha18} formalise the sufficient conditions on STMs and the programs running on them
that together ensure strong isolation. 
That is, non-transactional accesses can be viewed as singleton transactions (transactions containing single instructions).
However, their conditions require \emph{serialisability} for fully transactional programs, and as such, RSI transactions do not meet their conditions. 
Nevertheless, we conjecture that a DRF guarantee for strong atomicity, similar to the one in~\cite{Khyzha18}, may be established for RSI.
That is, if all executions of a given fully transactional program have no races between singleton and non-singleton transactions, 
then it is safe to replace all singleton transactions by non-transactional accesses.

In the future, we plan to build on the work presented here by developing reasoning techniques that would allow us to verify properties of STM programs.
This can be achieved by either extending existing program logics for weak memory, or developing new ones for currently unsupported models.
In particular, we can reason about the SI models presented here by developing custom proof rules in the existing program logics for RA such as~\cite{ogra,rsl}.

\paragraph{Acknowledgements}
This research was supported in part by a European Research Council
(ERC) Consolidator Grant for the project ``RustBelt", under the European
Union's Horizon 2020 Framework Programme (grant agreement no. 683289).
%%
%% Bibliography
%%

%% Please use bibtex, 
\bibliographystyle{splncs04}
\bibliography{biblio}

\begin{thebibliography}{10}
\providecommand{\url}[1]{\texttt{#1}}
\providecommand{\urlprefix}{URL }
\providecommand{\doi}[1]{https://doi.org/#1}

\bibitem{SI-STM-clojure}
The {C}lojure {L}anguage: {R}efs and {T}ransactions,
  \url{http://clojure.org/refs}

\bibitem{C++}
Technical specification for {C++} extensions for transactional memory (2015),
  \url{http://www.open-std.org/jtc1/sc22/wg21/docs/papers/2015/n4514.pdf}

\bibitem{Adya-thesis}
Adya, A.: Weak consistency: A generalized theory and optimistic implementations
  for distributed transactions. Ph.D. thesis, MIT (1999)

\bibitem{Adya-graphs}
Adya, A., Liskov, B., O'Neil, P.: Generalized isolation level definitions. In:
  Proceedings of the 16th International Conference on Data Engineering. pp.
  67--78 (2000)

\bibitem{cats-Alglave}
Alglave, J., Maranget, L., Tautschnig, M.: Herding cats: Modelling, simulation,
  testing, and data mining for weak memory. ACM Trans. Program. Lang. Syst.
  \textbf{36}(2),  7:1--7:74 (2014)

\bibitem{C11}
Batty, M., Owens, S., Sarkar, S., Sewell, P., Weber, T.: Mathematizing {C}++
  concurrency. In: Proceedings of the 38th Annual ACM SIGPLAN-SIGACT Symposium
  on Principles of Programming Languages. pp. 55--66 (2011)

\bibitem{SI}
Berenson, H., Bernstein, P., Gray, J., Melton, J., O'Neil, E., O'Neil, P.: A
  critique of {ANSI SQL} isolation levels. In: Proceedings of the 1995 ACM
  SIGMOD International Conference on Management of Data. pp. 1--10 (1995)

\bibitem{SI-STM-hindsight}
Bieniusa, A., Fuhrmann, T.: Consistency in hindsight: A fully decentralized
  {STM} algorithm. In: Proceedings of the 2010 IEEE International Symposium on
  Parallel and Distributed Processing, IPDPS 2010. pp. 1 -- 12 (2010)

\bibitem{Blundell-isolation2}
Blundell, C., C.~Lewis, E., M.~K.~Martin, M.: Deconstructing transactions: The
  subtleties of atomicity. In: 4th Annual Workshop on Duplicating,
  Deconstructing, and Debunking (2005)

\bibitem{cerone15}
Cerone, A., Bernardi, G., Gotsman, A.: A framework for transactional
  consistency models with atomic visibility. In: Proceedings of the 26th
  International Conference on Concurrency Theory. pp. 58--71 (2015)

\bibitem{SI-Cerone}
Cerone, A., Gotsman, A.: Analysing snapshot isolation. In: Proceedings of the
  2016 ACM Symposium on Principles of Distributed Computing. pp. 55--64 (2016)

\bibitem{PSI-Cerone}
Cerone, A., Gotsman, A., Yang, H.: Transaction chopping for parallel snapshot
  isolation. In: Proceedings of the 29th International Symposium on Distributed
  Computing - Volume 9363. pp. 388--404 (2015)

\bibitem{Cerone-Algebraic}
Cerone, A., Gotsman, A., Yang, H.: Algebraic laws for weak consistency. In:
  CONCUR (2017)

\bibitem{Crooks-isolation}
Crooks, N., Pu, Y., Alvisi, L., Clement, A.: Seeing is believing: A
  client-centric specification of database isolation. In: Proceedings of the
  ACM Symposium on Principles of Distributed Computing. pp. 73--82. PODC '17,
  ACM, New York, NY, USA (2017). \doi{10.1145/3087801.3087802},
  \url{http://doi.acm.org/10.1145/3087801.3087802}

\bibitem{SI-distributed-lazy-replication}
Daudjee, K., Salem, K.: Lazy database replication with snapshot isolation. In:
  Proceedings of the 32Nd International Conference on Very Large Data Bases.
  pp. 715--726 (2006)

\bibitem{SI-Java}
Dias, R.J., Distefano, D., Seco, J.a.C., Louren\c{c}o, J.a.M.: Verification of
  snapshot isolation in transactional memory java programs. In: Proceedings of
  the 26th European Conference on Object-Oriented Programming. pp. 640--664.
  ECOOP'12, Springer-Verlag, Berlin, Heidelberg (2012).
  \doi{10.1007/978-3-642-31057-7-28},
  \url{http://dx.doi.org/10.1007/978-3-642-31057-7-28}

\bibitem{Dongol-popl18}
Dongol, B., Jagadeesan, R., Riely, J.: Transactions in relaxed memory
  architectures. Proc. ACM Program. Lang.  \textbf{2}(POPL),  18:1--18:29 (Dec
  2017). \doi{10.1145/3158106}, \url{http://doi.acm.org/10.1145/3158106}

\bibitem{Gotsman-reasoning}
Gotsman, A., Yang, H., Ferreira, C., Najafzadeh, M., Shapiro, M.: 'cause i'm
  strong enough: Reasoning about consistency choices in distributed systems.
  In: Proceedings of the 43rd Annual ACM SIGPLAN-SIGACT Symposium on Principles
  of Programming Languages. pp. 371--384. POPL '16, ACM, New York, NY, USA
  (2016). \doi{10.1145/2837614.2837625},
  \url{http://doi.acm.org/10.1145/2837614.2837625}

\bibitem{tm-book}
Harris, T., Larus, J., Rajwar, R.: Transactional Memory, 2Nd Edition. Morgan
  and Claypool Publishers, 2nd edn. (2010)

\bibitem{tm}
Herlihy, M., Moss, J.E.B.: Transactional memory: Architectural support for
  lock-free data structures. In: Proceedings of the 20th Annual International
  Symposium on Computer Architecture. pp. 289--300 (1993)

\bibitem{Kaki-popl18}
Kaki, G., Nagar, K., Najafzadeh, M., Jagannathan, S.: Alone together:
  Compositional reasoning and inference for weak isolation. Proc. ACM Program.
  Lang.  \textbf{2}(POPL),  27:1--27:34 (Dec 2017). \doi{10.1145/3158115},
  \url{http://doi.acm.org/10.1145/3158115}

\bibitem{Khyzha18}
Khyzha, A., Attiya, H., Gotsman, A., Rinetzky, N.: Safe privatization in
  transactional memory. In: Proceedings of the 23rd {ACM} {SIGPLAN} Symposium
  on Principles and Practice of Parallel Programming. pp. 233--245 (2018)

\bibitem{SRA}
Lahav, O., Giannarakis, N., Vafeiadis, V.: Taming release-acquire consistency.
  In: Proceedings of the 43rd Annual ACM SIGPLAN-SIGACT Symposium on Principles
  of Programming Languages. pp. 649--662 (2016)

\bibitem{ogra}
Lahav, O., Vafeiadis, V.: {O}wicki-{G}ries reasoning for weak memory models.
  In: Proceedings, Part II, of the 42Nd International Colloquium on Automata,
  Languages, and Programming - Volume 9135. pp. 311--323 (2015)

\bibitem{SI-STM-aborts}
Litz, H., Cheriton, D., Firoozshahian, A., Azizi, O., Stevenson, J.P.: {SI-TM}:
  Reducing transactional memory abort rates through snapshot isolation. SIGPLAN
  Not. pp. 383--398 (2014)

\bibitem{SI-STM-anomalies}
Litz, H., Dias, R.J., Cheriton, D.R.: Efficient correction of anomalies in
  snapshot isolation transactions. ACM Trans. Archit. Code Optim.
  \textbf{11}(4),  65:1--65:24 (Jan 2015). \doi{10.1145/2693260}

\bibitem{Blundell-isolation}
Martin, M., Blundell, C., Lewis, E.: Subtleties of transactional memory
  atomicity semantics. IEEE Comput. Archit. Lett.  \textbf{5}(2),  17--17
  (2006)

\bibitem{Papadimitriou79}
Papadimitriou, C.H.: The serializability of concurrent database updates. J. ACM
   \textbf{26}(4),  631--653 (Oct 1979). \doi{10.1145/322154.322158},
  \url{http://doi.acm.org/10.1145/322154.322158}

\bibitem{SI-distributed-notifications}
Peng, D., Dabek, F.: Large-scale incremental processing using distributed
  transactions and notifications. In: Proceedings of the 9th USENIX Conference
  on Operating Systems Design and Implementation. pp. 251--264 (2010)

\bibitem{raad1:psi}
Raad, A., Lahav, O., Vafeiadis, V.: On parallel snapshot isolation and
  release/acquire consistency. In: Proceedings of the 27th European Symposium
  on Programming (2018), to appear

\bibitem{SI-distributed-replication}
Serrano, D., Patino-Martinez, M., Jimenez-Peris, R., Kemme, B.: Boosting
  database replication scalability through partial replication and
  1-copy-snapshot-isolation. In: Proceedings of the 13th Pacific Rim
  International Symposium on Dependable Computing. pp. 290--297 (2007)

\bibitem{stm}
Shavit, N., Touitou, D.: Software transactional memory. In: Proceedings of the
  Fourteenth Annual ACM Symposium on Principles of Distributed Computing. pp.
  204--213 (1995)

\bibitem{PSI}
Sovran, Y., Power, R., Aguilera, M.K., Li, J.: Transactional storage for
  geo-replicated systems. In: Proceedings of the Twenty-Third ACM Symposium on
  Operating Systems Principles. pp. 385--400 (2011)

\bibitem{rsl}
Vafeiadis, V., Narayan, C.: Relaxed separation logic: A program logic for {C11}
  concurrency. In: Proceedings of the 2013 ACM SIGPLAN International Conference
  on Object Oriented Programming Systems Languages \& Applications. pp.
  867--884 (2013)

\end{thebibliography}

\appendix
\AtAppendix{\counterwithin{lemma}{section}}
\AtAppendix{\counterwithin{proposition}{section}}
\AtAppendix{\counterwithin{notation}{section}}
\AtAppendix{\counterwithin{assumption}{section}}
\AtAppendix{\counterwithin{definition}{section}}
\AtAppendix{\counterwithin{remark}{section}}

\newpage
%\section{Technical Appendix}
%\paragraph{Notation} In the remainder of this article, given an execution graph $(\Events, \po, \rf, \co, \lo)$ we write $\TClasses$ for the set of equivalence classes of $\Transactions$ induced by $\st$; $\class{a}{\st}$ for the equivalence class that contains $a$; and $\Transactions_\txid$ for the equivalence class of transaction $\txid\in \TXIDs$: $\Transactions_\txid \eqdef \setcomp{a}{\tx{a} {=} \txid}$.
%We write $\sicon$ to denote that $\absGSI$ is SI-consistent; write $\rsicon$ to denote that $\absGRSI$ is RSI-consistent; and write $\consistent{\impG}$ to denote that $\impG$ is RA-consistent.

\section{MRSW Lock Implementations}\label{app:lock_implementations}
We consider two different MRSW library implementations, both satisfying the lock synchronisation guarantees required by \cref{def:si_implementation_consistency}. As we demonstrate shortly, our first implementation (\cref{subsec:full_sync_locks}) offers additional synchronisation guarantees by ensuring that any two calls to the library (including those of read locks) synchronise. 
That is, our first implementation satisfies the (\ref{ax:lock_rsync}) axiom on page \pageref{ax:lock_rsync}.
By contrast, the synchronisation guarantees of our second implementation (\cref{subsec:write_sync_locks}) are exactly those required by \cref{def:si_implementation_consistency}, where a library call to a \emph{write} lock synchronises with all other lock library calls.
As we discussed earlier in \cref{sec:ideas}, whilst \emph{both} library implementations can be used in our \emph{SI} implementation, \emph{only the weaker} (second) implementation can be used in our \emph{RSI} implementation.

\subsection{Fully Synchronising MRSW Lock Implementation}\label{subsec:full_sync_locks}
Our first MRSW lock library is implemented in the RA fragment of C11~\cite{SRA}, and is given in \cref{fig:MRSW_full_sync}.
In this implementation, the lock associated with each location \x resides at location $\x{+}1$, written \xl. 
The state of a lock \xl is represented by an integer value. A lock \xl may hold either:
\begin{enumerate}[label={\roman*)}]
	\item value $0$, denoting that the lock is free (not held in read or write mode); or
	\item value $1$, denoting that the lock is held (exclusively) in write mode; or
	\item an even value $2n$ with $n > 0$, denoting that the lock is held in (shared) read mode by $n$ readers; or
	\item an odd value $2n{+}1$ with $n > 0$, denoting that the lock is currently being promoted, awaiting the release of $n$ readers. 
\end{enumerate}
As such, the implementation of \code{\wlock x} simply spins until it can atomically update (via \code{CAS}) the value of \xl from zero (free) to one (acquired in write mode).
The \code{CAS(xl,\,v\,,v')} denotes the atomic `compare-and-set' operation, where \emph{either} \xl currently holds value $v$ in which case it is atomically updated to $v'$ and true is returned; \emph{or} \xl currently holds a value other than $v$ in which case it is left unchanged and false is returned. 
Dually, the implementation of \code{\wunlock x} simply releases the write lock by atomically assigning \xl to zero. 

The implementation of \code{can-promote x} is more involved. As multiple readers may attempt to promote their reader locks simultaneously, promotion is granted on a `first-come-first-served' bases. As such, the implementation of \code{can-promote x} first reads the value of \xl. If \xl holds an odd value, then another reader is currently promoting \xl and thus promotion of \xl fails by returning false.  On the other hand, if \xl holds an even value, then its value is atomically decremented (to an odd value) to signal the intention to promote.
The implementation then proceeds by spinning until all other readers have released their locks on \x (i.e.~\xl == 1), at which point true is returned to denote the successful acquisition of \x in write mode. 
Note that once a reader has signalled its intention to promote \x (by decrementing \xl to an odd value), any other such attempt to promote the lock on \x, as well as calls to acquire it in read mode will fail thereafter until such time that \x is released by its current promoter.

The implementation of \code{\rlock x} is similar. It first checks whether \xl is odd (held in write mode or being promoted). If so then the implementation spins until \xl is even (free or held in read mode), at which point its value is incremented by two (to increase the number of readers by one) using the atomic `fetch-and-add' (\code{FAA}) operation, and \x is successfully acquired in read mode. 
Dually, the implementation of \code{\runlock x} atomically decrements the value of \xl by two to decrease the reader count by one.

\begin{figure}[t]
\hrule
\[
\begin{array}[t]{@{} l @{\hspace{25pt}} l @{}}
	\begin{array}[t]{@{} l @{}}
		\code{\rlock x} \eqdef  \\
		\quad \begin{array}[t]{@{} l @{\hspace{3pt}} l @{}}
			\codelabel{start}& \code{a\,:=\,\xl;}\\
			& \code{if (is-odd a) goto start;} \\
			& \code{if (!CAS(xl,\,a,\,a+2)) } \\
			& \quad \code{goto start;}
		\end{array} \\\\
		
		\code{\runlock x} 
		\eqdef 
		\code{FAA(xl,\,-2);} \\\\

		\code{\wlock x} \eqdef  \code{while\,(!CAS(xl,0,1))\,skip;}\\
	\end{array}
	& 
	\begin{array}[t]{@{} l @{}}
		\code{can-promote x} \eqdef \\
		\quad \begin{array}[t]{@{} r @{\hspace{3pt}} l @{}}
			\codelabel{start}& \code{a\,:=\,xl;}\\
			& \code{if\,(is-odd a)\,return false;} \\
			& \code{if\,(!CAS(xl,\,a,\,a-1)) } \\
			& \quad \code{goto start;} \\
%			\codelabel{wait}
%			& \code{if (xl\,!=\,1) goto wait} \\
			& \code{while\,(xl\,!=\,1) skip;} \\
			& \code{return true;}
		\end{array} \\\\
		
		\code{\wunlock x} 
		\eqdef  
		\code{xl\,:=\,0;} 
		
	\end{array} 
\end{array}	
\]
%
%
%\vspace{-10pt}
%\begin{tabular}{p{10cm} p{10cm}}
%	\begin{enumerate}[label={\color{grey} \arabic*.}, ref=\theenumi, leftmargin=*, itemsep=-2pt, labelsep=4pt]
%		\setcounter{enumi}{-1}
%		\item   \code{\rlock x} $\eqdef$  
%		\item \quad \code{a:=\,xl;}
%		\item \quad \code{if (is-odd a) goto line 1;} 
%		\item \quad \code{if (!CAS(xl,\,a,\,a+2)) } 
%		\item \qquad \code{goto line 1;} \newline
%		
%		\setcounter{enumi}{-1}
%		\item \code{\runlock x} 
%		$\eqdef $
%		\code{FAA(xl,\,-2);} \newline
%		
%		\setcounter{enumi}{-1}
%		\item \code{\wlock x} $\eqdef$  
%		\code{while\,(!CAS(xl,\,0,\,1)) skip;}
%	\end{enumerate}
%	& 
%	\begin{enumerate}[label={\color{grey} \arabic*.}, ref=\theenumi, leftmargin=*, itemsep=-2pt, labelsep=4pt]
%		\setcounter{enumi}{-1}
%		\item \code{plock x} $\eqdef$
%		\item \quad \code{a\,:=\,xl;}
%		\item \quad \code{if (is-odd a) return false;} 
%		\item \quad \code{if (!CAS(xl,\,a,\,a-1)) } 
%		\item \qquad \code{goto line 1;} 
%		\item \quad \code{while\,(xl\,!=\,1) skip;} 
%		\item \quad \code{return true;}
%		\newline
%
%		
%		\setcounter{enumi}{-1}
%		\item \code{\wunlock x} 
%		$\eqdef $ 
%		\code{xl\,:=\,0;} 
%	\end{enumerate}	
%
%\end{tabular}
\hrule\vspace{5pt}
\caption{Fully synchronising MRSW lock implementation in the RA fragment of C11}
\label{fig:MRSW_full_sync}
\end{figure}
\subsubsection{Synchronisation Guarantees}
In what follows we demonstrate that our implementation satisfies the (\ref{ax:lock_wsync}), (\ref{ax:lock_wex}) and (\ref{ax:lock_rshared}) conditions in \cref{def:si_implementation_consistency}, as well as the stronger (\ref{ax:lock_rsync}) condition discussed on page \pageref{ax:lock_rsync}.
%Observe that i) a successful acquisition of a reader lock is done via an atomic update (read and write) operation (i.e.~when the \code{CAS} succeeds); 
%ii) a successful acquisition of a writer lock is done via an atomic update operation (when the \code{CAS} succeeds);
%iii) a read lock is released via an atomic update operation (the \code{FAA} operation);
%iv) a write lock is released via an atomic write operation; and
%v) a reader lock is successfully promoted via an atomic update operation (when the \code{CAS} succeeds), followed by an atomic read operation once the last reader has released its hold on the lock (the last iteration of the while loop).

Observe that a successful acquisition of a writer lock is done via an atomic update operation (when the \code{CAS} succeeds).
That is, a call to \code{\wlock x} returns only when the \code{CAS} is successful, i.e.~when no other thread holds a lock on \x. 
Similarly, a call to \code{\wunlock x} returns after an atomic write to \xl assigning it to zero. 
Moreover, once \x is acquired in write mode by a thread $\tid$, no other thread can acquire it (all other calls to \code{\rlock x} and \code{\wlock x} spin until \x is released by $\tid$).
As such, any RA-consistent execution graph of a program $P$ containing a call by thread $\tid$ to \code{\wlock x} followed by its subsequent release via \code{\wunlock x} includes a trace of the following form, where $wl$ and $wu$ are events of thread $\tid$, $wl$ denotes the update event associated with the successful \code{CAS} acquiring the writer lock, and $wu$ denotes the write event associated with its release:
\begin{align}
	\relarrow{\rf}\, wl:\updateE{}{\xl}{0}{1} \relarrow{\imm{\po_\xl}} wu:\writeE{}{\xl}{0}
	\label{trace:wlock}
\end{align}
Note that if the use of locks in $P$ is well-formed (i.e.~the sequence of lock events in each thread following $\po$ matches (a prefix of) the regular expression $(\lRL\cdot\lRU \;|\;\lWL\cdot\lWU \;|\; \lRL\cdot\lPL\cdot\lWU)^*$ -- see page \pageref{par:lock_traces}), no other write event on \xl can happen between $wl$ and $wu$. 
This is because  i) well-formedness of lock traces ensures that the lock is released by the acquiring thread $\tid$ itself, and ii) the lock on \x cannot be acquired after $wl$ and before $wu$: calls to \code{\wlock x} fail because \xl is non-zero; and calls to \code{can-promote x} and \code{\rlock x} fail because $\xl=1$ holds an odd value.
This in turn ensures the mutual exclusion property of writer locks, as required by the (\ref{ax:lock_wex}) condition in \cref{def:si_implementation_consistency}.

Similarly, a call to \code{can-promote x} returns after the \code{CAS} succeeds and subsequently the condition of the while loop amounts to false and no other thread owns a lock on \x (in read or write mode). 
As such, any RA-consistent execution graph of a program $P$ containing a call by thread $\tid$ to \code{can-promote x} followed by its subsequent release via \code{\wunlock x} includes a trace of the following form, where $sp$, $pl$ and $wu$ are events of thread $\tid$, $sp$ denotes the update event associated with the successful \code{CAS} signalling promotion, $pl$ denotes the atomic read event in the final iteration of the while loop denoting successful promotion, and $wu$ denotes the write event associated with its release:
\begin{align}
	\hspace*{-10pt}
	\relarrow{\rf}\, sp:\updateE{}{\xl}{2i {+} 2}{2i {+} 1} 
	\relarrow{\rf} b_1 \relarrow{\rf} \cdots \relarrow{\rf} b_i
	\relarrow{\rf} pl: \readE{}{\xl}{1}	
	\relarrow{\imm{\po_\xl}} wu:\writeE{}{\xl}{0}
	\label{trace:plock}
\end{align}
for some $i$ denoting the reader count on \x, with $b_1 \cdots b_m \in \RUnlocks_{\x}$, $rl \relarrow{\imm{\po_\xl}} ru$.
In other words, once $\tid$ has signalled its intention to promote \x, no other lock on \x can be acquired -- calls to \code{\wlock x} fail as the value of \xl is non-zero; calls to \code{\rlock x} fail as \xl holds an odd value. 
However, existing reader locks must be released before \x can be successfully promoted.
As with writer locks above, if the use of locks in $P$ is well-formed, no other event on \xl can happen between $pl$ and $wu$. 
This once again ensures the mutual exclusion property of promoted locks, as required by the (\ref{ax:lock_wex}) condition in \cref{def:si_implementation_consistency}.

Analogously, a call to \code{\rlock x} returns when the \code{CAS} succeeds; and a call to \code{\runlock x} returns after the atomic \code{FAA} operation. 
Moreover, once \x is acquired in reader mode by a thread $\tid$, no other thread can acquire it in write mode (all other calls to \code{\wlock x} spin) until it has been released by $\tid$ (and potentially other readers). 
As such, any RA-consistent execution graph of a program $P$ containing a call to \code{\rlock x} followed by its subsequent release via \code{\runlock x} includes a trace of the following form, where $rl$ and $ru$ are events of thread $\tid$, $rl$ denotes the update event associated with the successful \code{CAS} acquiring the reader lock, and $ru$ denotes the update event associated with its release via \code{FAA}:
\small
\begin{align}
	& \hspace*{-10pt} \relarrow{\rf} rl:\updateE{}{\xl}{2i}{2i {+} 2} 
	\underbrace{\relarrow{\rf} a_1 \relarrow{\rf} \cdots \relarrow{\rf} a_n}_{\text{zero or more }}
	\underbrace{\relarrow{\rf} sp: \updateE{}{\xl}{2k {+} 2}{2k{+} 1}}_{\text{zero or one }}
	\underbrace{\relarrow{\rf} b_1 \relarrow{\rf} \cdots \relarrow{\rf} b_m}_{\text{zero or more }} 
	\nonumber \\
	&\hspace*{-10pt}
	\relarrow{\rf} ru:\updateE{}{\xl}{j{+} 2}{j}
	\label{trace:rlock}
\end{align}
\normalsize
for some $i, j, k>0$, where $a_1 \cdots a_n \in \RLocks_{\x} \cup \RUnlocks_{\x}$,  $b_1 \cdots b_m \in \RUnlocks_{\x}$, $rl \relarrow{\imm{\po_\xl}} ru$.
In other words, other threads may also acquire \x in read mode (denoted by events of $a_1 \cdots a_i$) and their calls may be interleaved between the read lock and unlock of $\tid$. 
Moreover, another thread may signal to promote its reader lock on \x (denoted by $sp$)  in between the acquisition and release of the reader lock on \x by $\tid$. 
Lastly, once the lock has been signalled for promotion, no other reader lock on \x can be acquired, though existing reader locks may be released (denoted by  $b_1 \cdots b_m$).
Note that if the use of locks in $P$ is well-formed, no thread can acquire \x in write mode (as \x holds a non-zero value) or successfully promote it before \x is released by $\tid$ via $ru$. 
More concretely, threads may signal their intention to promote (see $sp$ above). However, they cannot successfully promote it before $\tid$ has released it in $ru$ as $\xl$ holds a value greater than $1$ (see $pl$ in \eqref{trace:plock}). 
This ensures the \eqref{ax:lock_rshared} condition in \cref{def:si_implementation_consistency}.

Observe that given an RA-consistent execution graph involving calls to the above MRSW lock library above, for each location \x and its lock at \xl, the trace of events on \xl comprises $\po$ and $\rf$ edges, as demonstrated by the traces in \eqref{trace:wlock}, \eqref{trace:plock} and \eqref{trace:rlock}. 
In other words, any two events on \xl are related by $\transC{\perloc{(\po \cup \rf)}}$, i.e.\ $\transC{\perloc{(\po \cup \rf)}}$ is total for each \xl. 
Moreover, from the RA-consistency of our execution we know that $\transC{\perloc{(\po \cup \rf)}}$ is acyclic. As such, since $\transC{\perloc{(\po \cup \rf)}}$ is transitively closed, we know that it is a strict total order. 
For each location \x, we thus define $\lo_\x$ as the strict total order given by $\lo_\x \eqdef \transC{(\po_{\xl} \cup \rf_\xl)}$.
As such, our implementation satisfies both the (\ref{ax:lock_wsync}) condition in \cref{def:si_implementation_consistency} and the stronger (\ref{ax:lock_rsync}) condition discussed on page \pageref{ax:lock_rsync}.
%As discussed, our implementation satisfies the (\ref{ax:lock_wex}) and (\ref{ax:lock_rshared}) conditions required by \cref{def:si_implementation_consistency}.

\subsection{Write Synchronising MRSW Lock Implementation}\label{subsec:write_sync_locks}

\newcommand{\onearray}{\textbf{1}}
\newcommand{\zeroarray}{\textbf{0}}
\newcommand{\aplus}{\cdot}
Our second MRSW lock implementation is similarly implemented in the RA fragment of C11~\cite{SRA} and is given in \cref{fig:MRSW_write_sync}. 
In this implementation, each lock \x is represented as an \emph{ordered} map at location $\x {+} 1$, written \xl. The map at \xl contains one entry per thread as follows. 
For each thread with identifier $\tid$, the $\code{xl[\tid]}$ map entry records the current locking privileges of $\tid$ on \x. 
More concretely, when $\code{xl[\tid]}=0$, then $\tid$ does not hold the \x lock; when $\code{xl[\tid]}=2$, then $\tid$ holds \x in read mode; and when $\code{xl[\tid]}=1$; then \emph{some} thread (either $\tid$ or another thread) either holds \x in write mode, or it is in the process of acquiring \x in write mode.
The \x lock is held in write mode only when all entries in \xl are mapped to one. As we describe shortly, for thread $\tid$ to acquire \x in write mode, it must inspect each entry in \xl (in order), wait for it be free (zero) and then set it to one. 
As we discuss shortly, this in-order inspection of entries allows us to avoid deadlocks. 
In our implementation, we assume that the thread identifier can be obtained by calling \code{getTID}.
We identify the top-most thread by $\tid=0$; as such, the entry of top-most thread in each map is ordered before all other threads.

\begin{figure}[t]
\hrule
\[
\begin{array}[t]{@{} l @{\hspace{25pt}} l @{}}
	\begin{array}[t]{@{} l @{}}
		\code{\rlock x} \eqdef  \\
		\quad \begin{array}[t]{@{} l @{}}
			\code{t:=\,getTID;} \\
			\code{while\,(!CAS(xl[t],\,0,\,2)) skip;} \\
		\end{array} \\\\
		
		\code{\runlock x} 
		\eqdef 
		\code{t:=\,getTID; } 
		\code{xl[t]:=\,0} \\\\

		\code{\wlock x} \eqdef  \\
		\quad \begin{array}[t]{@{} l @{}}
			\code{for\,(t\,$\in$\,dom(xl))} \\
			\quad \code{while\,(!CAS(xl[t],0,1))\,skip;} \\
		\end{array} \\\\

		\code{\wunlock x} 
		\eqdef  \\
		\quad \begin{array}[t]{@{} l @{}}
			\code{for\,(t\,$\in$\,dom(xl)) xl[t]:=\,0; }
		\end{array} 
	\end{array}
	& 
	\begin{array}[t]{@{} l @{}}
		\code{can-promote x} \eqdef \\
		\quad \begin{array}[t]{@{} r @{\hspace{3pt}} l @{}}
			& \code{t:=\,getTID; }  \\
			& \code{if\,(t\,==\,0) xl[t]:=\,1;} \\
			& \code{else \{} \\
			\codelabel{retry}& \quad \code{a:=\,xl[0];}\\
			& \quad \code{if (a\,==\,1) return false;} \\
			& \quad \code{if (!CAS(xl[0],\,0,\,1)) } \\
			& \qquad \code{goto retry; } \\
			& \code{\} xl[t]:=\,1;} \\
			& \code{for\,(i\,$\in$\,dom(xl) \&\& i\,$\not\in$\{0,t\})} \\
			& \quad \code{while\,(!CAS(xl[i],0,1))\,skip;} \\
			& \code{return true;}
		\end{array}

	\end{array} 
\end{array}
\]
\hrule\vspace{5pt}
\caption{Write synchronising MRSW lock implementation in the RA fragment of C11}
\label{fig:MRSW_write_sync}
\end{figure}

We proceed with a more detailed explanation of our implementation after introducing our map notation.
%
%
%\paragraph{Map notation} 
We write $\onearray$ to denote a map where all entries have value $1$; similarly, we write $\zeroarray$ to denote a map where all entries have value $0$. 
%We write $v^{i}$ to denote a map of size $i$ where all entries hold value $v$.
%We use `$.$' to denote map concatenation and write e.g.~$a^{i}.b^{j}$ to denote a map of size $i {+} j$, where the first $i$ entries hold value $a$ and the last $j$ entries hold value $b$. 
Lastly, we write $S \subseteq \xl$, to denote that the values held in map \xl are a superset of $S$.
%to denote that the values held in map $\xl$ are contained in the set $S$; \textit{mutatis mutandis} for  $S \subseteq \xl$.
The lock map \xl associated with location \x can be in one of the following states: 
\begin{itemize}
	\item $\xl=\zeroarray$ when \x is free;
	\item $\xl = \onearray$ when \x is held in write mode;
	\item $\{2\} \suq \xl$ when \x is held in read mode (by those threads $\tau$ where $\code{xl[\tid]} = 2$).
	%\item $\{1, 0\} \suq \xl$ or $\{1, 2\} \suq \xl$, when \x is \emph{being} acquired in write mode (via \code{\wlock} or \code{can-promote}).
\end{itemize}

When thread $\tid$ calls \code{\rlock \x}, it simply spins until the lock is free ($\xl = \zeroarray$ and thus $\code{xl[\tid]}=0$), at which point it acquires it in read mode by setting $\code{xl[\tid]}$ to two. 
Dually, when $\tid$ calls \code{\runlock \x} it simply sets $\code{xl[\tid]}$ to zero. 

Analogously, when $\tid$ calls \code{\wlock x}, it traverses the \xl map in order, spinning on each entry until it is free ($0$) and subsequently acquiring it (by setting it to $1$).
Conversely, when $\tid$ calls \code{\wunlock x}, it releases \x by traversing \xl in order and setting each entry to one. 

%As before, the promotion of reader locks are done on a `first-come-first-served' basis. 
To understand the implementation of lock promotion, first consider the case where \code{can-promote x} is called by $\tid \ne 0$, i.e.~a thread other than the top-most thread.  
The implementation of \code{can-promote \x} then inspects the \emph{first} entry in the map ($\code{xl[0]}$), i.e.~that of the top-most thread. If $\code{xl[0]}=1$, then \x is currently being acquired by another thread; the promotion thus fails and false is returned. 
If on the other hand $\code{xl[0]} \ne 1$ (i.e.~$\code{xl[0]}=0$ or $\code{xl[0]}=2$), the implementation spins until it is zero and atomically updates it to one, signalling its intention to promote \x. 
This pre-empts the promotion of \x by other threads: any such attempt would fail as now $\code{xl[0]}=1$.
The implementation then sets its own entry (\code{xl[\tid]}) to one, traverses the map in order, and spins on each entry until they too can be set to one. At this point the lock is successfully promoted and true is returned. 
Note that it is safe for $\tid$ to update its own entry \code{xl[\tid]} to one: at this point in execution no thread holds the writer lock on \x, no thread can promote its lock on \x, and those threads with a reader lock on \x never access the \code{xl[\tid]} entry -- the read lock calls of another thread \tid' solely accesses \code{xl[\tid']}.

Let us now consider the case when the top-most thread with $\tid=0$ calls \code{can-promote x}. Since prior to a \code{can-promote x} call $\tid$ owns a reader lock on \x, i.e.~$\code{xl[\tid]}=2$, no other thread can promote its \x lock. As such, $\tid$ successfully sets \code{xl[\tid]} to one, signalling its intention to promote \x. 
In other words, the promotion is skewed in favour of the top-most thread: if a thread races against the top-most thread to promote \x, the top-most thread always wins.
With the exception of the top-most thread, promotion is done on a `first-come-first-served' basis. 
The rest of the implementation is then carried out as before: the map \xl is traversed in turn and each entry is set to one.

\subsubsection{Synchronisation Guarantees} 
In what follows we demonstrate that our implementation satisfies the (\ref{ax:lock_wsync}), (\ref{ax:lock_wex}) and (\ref{ax:lock_rshared}) conditions in \cref{def:si_implementation_consistency}, while it does not satisfy the stronger (\ref{ax:lock_rsync}) condition discussed on page \pageref{ax:lock_rsync}.

Observe that a successful acquisition of a writer lock is done via several atomic update operations (via successful \code{CAS} operations in the for loop).
That is, a call to \code{\wlock x} returns only once the \code{CAS} on all \xl entries succeeds, i.e.~when no other thread holds a lock on \x. 
Similarly, a call to \code{\wunlock x} returns after atomic writes on each entry in \xl, assigning them to zero. 
Moreover, once \x is acquired in write mode by a thread $\tid$, no other thread can acquire it: all other calls to \code{\rlock x} spin until the relevant \xl entry is set to zero by $\tid$; all calls to \code{\wlock x} spin until all entries in \xl are set to zero by $\tid$).
As such, any RA-consistent execution graph of a program $P$ containing a call by thread $\tid$ to \code{\wlock x} or \code{can-promote x}, followed by its subsequent release via \code{\wunlock x} includes a trace of the following form, where the domain of the \xl map is $0 \cdots n$ for some $n$, the $wl_0, \cdots wl_n, wu_0, \cdots, wu_n$ are events of thread $\tid$, each $wl_i$ denotes the update event associated with the successful \code{CAS} (in \code{\wlock x}) on the \code{xl[$i$]} entry, and each $wu_i$ denotes the write event setting the \code{xl[$i$]} entry to zero (in \code{\wunlock x}):
\begin{align}
%	\begin{array}{@{} r @{\hspace{2pt}} c @{\hspace{2pt}} c @{\hspace{2pt}} c @{}}
%		\relarrow{\rf}
%		& wl_0:\updateE{}{\xl[0]}{0}{1} 
%		& \relarrow{\imm{\po_\xl}} 
%		& wu_0:\writeE{}{\xl[0]}{0} \\
%		& \downarrow \po &&\downarrow \po \\
%		& \vdots  && \vdots \\
%		& \downarrow \po &&\downarrow \po \\
%		\relarrow{\rf}
%		& wl_n:\updateE{}{\xl[n]}{0}{1} 
%		& \relarrow{\imm{\po_\xl}} 
%		& wu_n:\writeE{}{\xl[n]}{0} \\
%	\end{array}
	\begin{array}{@{} r @{\hspace{2pt}} c @{\hspace{2pt}} l @{}}
		\relarrow{\rf}
		& wl_0:\updateE{}{\xl[0]}{0}{1} \\
		& \downarrow \po   \\
		& \vdots  \\
		& \downarrow \po  \\
		\relarrow{\rf}
		& wl_n:\updateE{}{\xl[n]}{0}{1} 
		& \relarrow{\po}
		wu_0:\writeE{}{\xl[0]}{0}
		\relarrow{\po}
		\cdots
		\relarrow{\po}
		wu_n:\writeE{}{\xl[0]}{0}
	\end{array}
	\label{trace:wlock2}
\end{align}
As before, note that if the use of locks in $P$ is well-formed (see page \pageref{par:lock_traces}), no other write event on \code{xl[$i$]} can happen between $wl_i$ and $wu_i$. 
This is because no thread can acquire the write lock on \xl after $wl_0$ and before $wu_n$; and thread $\tid_i$ cannot acquire the read lock on \xl after $wl_i$ and before $wu_i$.
This in turn ensures the mutual exclusion property of writer locks, as required by the (\ref{ax:lock_wex}) condition in \cref{def:si_implementation_consistency}.

Analogously, a call to \code{\rlock x} returns when the \code{CAS} succeeds; and a call to \code{\runlock x} returns after the atomic write operation setting \code{xl[\tid]} to zero. 
%Moreover, once \x is acquired in reader mode by a thread $\tid$, no other thread can acquire it in write mode (all other calls to \code{\wlock x} spin) until \code{xl[\tid]} is set to zero.
As such, any RA-consistent execution graph of a program $P$ containing a call to \code{\rlock x} by thread $\tid$, followed by its subsequent release via \code{\runlock x} includes a trace of the following form, where $rl$ and $ru$ are events of thread $\tid$, $rl$ denotes the update event associated with the successful \code{CAS} acquiring the reader lock, and $ru$ denotes the write event associated with its release by assigning it to zero:
\small
\begin{align}
	\relarrow{\rf} rl:\updateE{}{\xl[\tid]}{0}{2} 
	\relarrow{\po} ru:\writeE{}{\xl[\tid]}{0}
	\label{trace:rlock2}
\end{align}
\normalsize
Note that each read lock call on \x by thread \tid, accesses \code{xl[\tid]} alone and no other entry in \xl. As such, two threads may simultaneously acquire the reader lock on \x, ensuring the (\ref{ax:lock_rshared}) condition in \cref{def:si_implementation_consistency}.
Moreover, as read lock calls by two distinct threads \tid, \tid' access \emph{disjoint} memory locations (\code{xl[\tid]} and \code{xl[\tid']}), they never synchronise. That is, the MRSW lock implementation in \cref{fig:MRSW_write_sync} does not satisfy the (\ref{ax:lock_rsync}) axiom on page \pageref{ax:lock_rsync}.

Lastly, we demonstrate that our implementation satisfies the \eqref{ax:lock_wsync} condition in \cref{def:si_implementation_consistency}, when accessed by an arbitrary (finite) number of threads $n$.\\
Observe that given an RA-consistent execution graph involving calls to the MRSW lock library above, for each location \x, its lock at \xl, and each thread $\tid$, the trace of events on \code{xl[\tid]} comprises $\po$ and $\rf$ edges, as demonstrated by the traces in (\ref{trace:wlock2}) and (\ref{trace:rlock2}). 
In other words, any two event on \code{xl[\tid]} are related by $\transC{\perloc{(\po \cup \rf)}}$, i.e.~$\transC{\perloc{(\po \cup \rf)}}$ is total for each \code{xl[\tid]}. 
We thus use this total order to determine the $\lo$ (synchronisation order) between two lock events, where at least one of them is a write lock event. 
Let us now pick two distinct lock events, $w, l \in \LEvents_\x$, where at least one of them $w$ is a write lock event, i.e.~$w \in \WLocks_\x \cup \PLocks_\x \cup \WUnlocks_\x$. 
Either 1) $w, l$ are events of the same thread; or 2) $w$ and $l$ are events of distinct threads $\tid$, $\tid'$, respectively.
In the first case the two events are related by $\po$ one way or another, and as $\po \subseteq \hb$, we know that the two events synchronise.
In the second case, there are two additional cases to consider: either i) $l \in \WLocks_\x \cup \PLocks_\x \cup \WUnlocks_\x$; or $l \in \RLocks_\x \cup \RUnlocks_\x$. 
In case (i), as we discussed above we know that each call amounts to a trace akin to that in \eqref{trace:wlock2}. 
That is, we know the trace of $\tid$ contains $wl: \updateE{}{\xl[0]}{0}{1} \relarrow{\po} wu: \writeE{}{\xl[0]}{0}$, and the trace of $\tid'$ contains $wl': \updateE{}{\xl[0]}{0}{1} \relarrow{\po} wu': \writeE{}{\xl[0]}{0}$.
Moreover, as we discussed above, we know that no other write event on \code{xl[0]} can happen between $wl$ and $wu$, and between $wl'$ and $wu'$.
As  $\transC{\perloc{(\po \cup \rf)}}$ is total for each \code{xl[\tid]}, we then know that we either have $wl \relarrow{\po} wu \relarrow{\transC{\perloc{(\po \cup \rf)}}} wl' \relarrow{\po} wu'$, or we have $wl' \relarrow{\po} wu' \relarrow{\transC{\perloc{(\po \cup \rf)}}} wl \relarrow{\po} wu$. As such, in both cases we know that the two events synchronise. 

In case (ii), we know that the write call results in a trace akin to that in \eqref{trace:wlock2}, while the read call results in a trace similar to that in  \eqref{trace:rlock2}. 
That is, we know the trace of $\tid$ contains $wl: \updateE{}{\xl[\tid']}{0}{1} \relarrow{\po} wu: \writeE{}{\xl[\tid']}{0}$, and the trace of $\tid'$ contains $rl: \updateE{}{\xl[\tid']}{0}{2} \relarrow{\po} ru: \writeE{}{\xl[\tid']}{0}$.
Moreover, as we discussed above, we know that no other write event on \code{xl[$\tid'$]} can happen between $wl$ and $wu$.
As  $\transC{\perloc{(\po \cup \rf)}}$ is total for \code{xl[$\tid'$]}, we then know that we either have $wl \relarrow{\po} wu \relarrow{\transC{\perloc{(\po \cup \rf)}}} rl \relarrow{\po} ru$, or we have $rl \relarrow{\po} ru \relarrow{\transC{\perloc{(\po \cup \rf)}}} wl \relarrow{\po} wu$. As such, in both cases we know that the two events synchronise.

%As discussed, our implementation satisfies the (\ref{ax:lock_wex}) and (\ref{ax:lock_rshared}) conditions required by \cref{def:si_implementation_consistency}.

%Note that each \emph{read} lock call on \x by thread \tid, accesses \code{xl[\tid]} alone and no other entry in \xl. 
%As such, read lock calls by two distinct threads \tid, \tid' access \emph{disjoint} memory locations and thus never synchronise. That is, the MRSW lock implementation in \cref{fig:MRSW_write_sync} does not satisfy the (\ref{ax:lock_rsync}) axiom on page \pageref{ax:lock_rsync}.

\newpage
\section{Auxiliary Results}
\label{app:aux_lemmas}

\begin{proposition}[SI-consistency]
\label{prop:SI-consistency-full}
An execution graph $G$ is SI-consistent if and only if 
$\ref{ax:si_int}$ holds and 
the `SI-happens-before' relation $\sihb \eqdef \transC{(\pot \cup \rft \cup \cot \cup \sifr)}$ is irreflexive, 
where
$\sifr \eqdef [\EReads]; \frt; [\Writes]$ and $\EReads \eqdef \codom(\rfe)$. 
\begin{proof}
Pick an arbitrary execution graph $G$. 
We are then required to show: 
\[
	\acyc{(\pot \cup \rft \cup \cot); \refC{\frt}} \iff 
	\irr{\sihb}
\]
\textbf{The $\Rightarrow$ direction}\\
We proceed by contradiction. 
Assume that $\acyc{(\pot \cup \rft \cup \cot); \refC{\frt}}$ and $\neg\irr{\sihb}$ both hold.
We then know there exists $e$ such that $(e, e) \in \sihb = \transC{(\pot \cup \rft \cup \cot \cup \sifr)}$. 
Note that $\sifr;\sifr=\emptyset$ because $\EReads \cap \Writes = \emptyset$.
Hence, the $\sihb$ cycle cannot have adjacent $\sifr$ edges,
which means that we have a cycle in $\transC{((\pot\cup\rft \cup \cot);\refC{\sifr})}$,
which contradicts our first assumption as $\sifr \suq \frt$.
\\

\noindent\textbf{The $\Leftarrow$ direction}\\
We proceed by contradiction. Let us assume $\irr{\sihb}$ and $\neg\acyc{(\pot \cup \rft \cup \cot); \refC{\frt}}$.
We then know there exists $e$ such that $(e, e) \in \transC{((\pot \cup \rft \cup \cot); \refC{\frt})}$, i.e.~there exists $e$ such that $(e, e) \in \transC{(\refC{\frt}; (\pot \cup \rft \cup \cot))}$.
There are now two cases to consider: either
1) $(e, e) \in \transC{(\pot \cup \rft \cup \cot)}$; or 
2) $\exsts{a, d} (e, a) \in \reftransC{(\refC{\frt}; (\pot \cup \rft \cup \cot))}$,  $(a, d) \in \frt; (\pot \cup \rft \cup \cot)$ and $(d, e) \in \reftransC{(\refC{\frt}; (\pot \cup \rft \cup \cot))}$, i.e.~$\exsts{a} (a, a) \in \frt; (\pot \cup \rft \cup \cot); \reftransC{(\refC{\frt}; (\pot \cup \rft \cup \cot))}$.
In case (1) from the definition of $\sihb$ we then have $(e, e) \in \sihb$, contradicting the assumption that $\irr{\sihb}$ holds. 

In case (2) we then know there exists $b, c$ such that $\class b \st \ne \class c \st$, $(a, b) \in \frt$, $(b, c) \in \pot \cup \rft \cup \cot$ and $(c, a) \in \reftransC{(\refC{\frt}; (\pot \cup \rft \cup \cot))}$. 
From the definitions of $\pot$, $\rft$ and $\cot$ we then have $\class b \st \times \class c \st \subseteq \pot \cup \rft \cup \cot$ and thus from the definition of $\sihb$ we have $\class b \st \times \class c \st \subseteq \sihb$. 
As such, from \cref{lem:si_equivalent_aux} below we have $\class b \st \times \class a \st \subseteq \sihb$. 
On the other hand, from $(a, b) \in \frt$ we know there exist $r, w, w'$ such that $r \in \class a \st$, $w \in \class b \st$, $\class a \st \ne \class b \st$, $(r, w) \in \fr$, $(w', r) \in \rf$ and $(w', w) \in \co$. 
There are now two cases to consider: i) $w' \in \class a \st$; or ii) $w' \not\in \class a \st$.

In case (2.i) we then have $\class a \st \times \class b \st \subseteq \cot \subseteq \sihb$.
As such, we have $(a, b) \in \sihb$. 
Since we have also established $\class b \st \times \class a \st \subseteq \sihb$, we have $(b, a) \in \sihb$. 
By transitivity, we have $(a,a)\in\sihb$, contradicting the assumption that $\irr{\sihb}$ holds. 

In case (2.ii) we then know $r \in \EReads$. As such we have $(r, w) \in \sifr \subseteq \sihb$.
Since we have also established $\class b \st \times \class a \st \subseteq \sihb$, we have $(w, r) \in \sihb$. 
By transitivity, we have $(r,r)\in \sihb$, contradicting the assumption that $\irr{\sihb}$ holds. 

\end{proof} 
\end{proposition}

\newpage
\begin{theorem}
\label{thm:SI_RSI_Equiv}
For all executions $G$, if $G.\NT = \emptyset$, then:
\[
	G \text{ is SI-consistent}
	\iff
	G \text{ is RSI-consistent}
\]
\begin{proof}[the $\Leftarrow$ direction]
Pick an arbitrary $G$ such that $G.\NT = \emptyset$ and $G$ is RSI-consistent. 
Let us proceed by contradiction and assume $G$ is not SI-consistent. 
That is, there exists $a$ such that $(a, a) \in G.\sihb$. 
From the definition of $\sihb$ we then have $G.\sihb \suq G.\rsihb$. 
As such, we have $(a, a) \in G.\rsihb$, contradicting the assumption that $G$ is RSI-consistent. 
\end{proof}
\begin{proof}[the $\Rightarrow$ direction]
Pick an arbitrary $G$ such that $G.\NT = \emptyset$ and $G$ is SI-consistent. 
As $G.\NT = \emptyset$, we then have $G.\rsihb = \transC{(\pot \cup \rft \cup \cot \cup \sifr \cup A)}$ with $A \eqdef G.([\Writes]; \poi; [\Writes])$. 
That is, $G.\rsihb = \transC{(G.\sihb \cup A)}$.
In what follows we demonstrate:
\begin{align}
	A; \sihb \suq \sihb
	\quad \text{and} \quad 
	\sihb; A \suq \sihb
	\label{subgoal:A_sihb_in_sihb}
\end{align}
As such, since $\sihb$ is transitively closed, we have $G.\rsihb = G.\sihb \cup A$.

Moreover, we have $G.\co \suq G.\cot \cup G.\coi$. 
As such, since from \eqref{ax:si_int} we have $G.\coi \suq G.\poi$, we have $G.\coi \suq A$ and thus $G.\co \suq G.\rsihbloc$.

Let $B \eqdef G.\fr \cap G.\poi$. 
We also have $G.\fr \suq G.\sifr \cup G.\fri$. 
As such, since from \eqref{ax:si_int} we have $G.\fri \suq G.\poi$, we have $G.\fri \suq B$ and thus $G.\fr \suq G.\rsihbloc \cup B$.
Consequently, we have 
$G.(\rsihbloc \cup \co \cup \fr) 
= G.(\rsihbloc \cup B) 
= G.(\perloc{\sihb} \cup \perloc{A} \cup B)$. 
Therefore, $G.\transC{(\rsihbloc \cup \co \cup \fr)} = G.\transC{(\perloc{\sihb} \cup \perloc{A} \cup B)}$.

Observe that $A; B = \emptyset$. 
In what follows we demonstrate that:
\begin{align}
	& B; A \suq B
	\label{subgoal:B_A_in_B} \\
	& \sihb; B \suq \sihb 
	\quad \text{and} \quad
	B; \sihb \suq \sihb
	\label{subgoal:B_sihb_in_sihb}
\end{align}
Note that since $G$ is SI-consistent we know that $A$ is irreflexive. 
As such, to show that $G.\transC{(\rsihbloc \cup \co \cup \fr)} = G.\transC{(\perloc{\sihb} \cup \perloc{A} \cup B)}$ is irreflexive, from \eqref{subgoal:A_sihb_in_sihb}, \eqref{subgoal:B_A_in_B}, \eqref{subgoal:B_sihb_in_sihb}, and since $B$ is irreflexive, it suffices to show that $G.\sihb$ is irreflexive, which follows immediately from the SI-consistency of $G$.\\

%\begin{align}
%	\for{n \in \Nats} 
%	C^n \text{ is irreflexive} 
%	\quad \text{where} \quad
%	C^0 \eqdef \sihb; B 
%	\quad \text{and} \quad
%	C^{i {+} 1} \eqdef C^i; C^0 
%	\text{ for }i \in \Nats
%	\label{subgoal:C_irr}
%\end{align}
%%

\noindent \textbf{TS. \eqref{subgoal:A_sihb_in_sihb}}\\
To show $A; \sihb \suq \sihb$, pick an arbitrary $(a, b) \in A; \sihb$.
We then know there exists $c$ such that $a, c \in \Writes$, $(a, c) \in A$ and $(c, b) \in \sihb$.
From the definition of $\sihb$ and since $c \in \Writes$, we know there exists $d$ such that 
$(c, d) \in \pot \cup \rft \cup \cot$ and $(d, b) \in \sihb^*$.
As $(a, c) \in A \suq \st$, we thus know that $(a, d) \in \pot \cup \rft \cup \cot \suq \sihb$. 
Consequently, since $(d, b) \in \sihb^*$ and $\sihb$ is transitively closed, 
we have $(a, b) \in \sihb$, as required. 

To show $\sihb; A \suq \sihb$, pick an arbitrary $(a, b) \in \sihb; A$.
We then know there exists $c$ such that $b, c \in \Writes$, $(a, c) \in \sihb$ and $(c, b) \in A$.
From the definition of $\sihb$, we know there exists $d$ such that 
$(a, d) \in \sihb^*$ and $(d, c) \in \pot \cup \rft \cup \cot \cup \sifr$.
As $(c, b) \in A \suq \st$ and since $b, c \in \Writes$ we thus know that $(d, b) \in \pot \cup \rft \cup \cot \cup \sifr \suq \sihb$. 
Consequently, since $(a, d) \in \sihb^*$ and $\sihb$ is transitively closed, 
we have $(a, b) \in \sihb$, as required.\\

\noindent\textbf{TS. \eqref{subgoal:B_A_in_B}}\\
Pick an arbitrary $a, b$ such that $(a, b) \in B; A$.
We then know there exists $c$ such that $b, c \in \Writes$, $(a, c) \in B$, $(c, b) \in A$ and $\tx{a} = \tx{b} = \tx{c}$.
As $b, c \in \Writes$, we know either $(b, c) \in \coi$ or $(c, b) \in \coi$. 
However, as $G$ is SI-consistent, from \eqref{ax:si_int} we know that $\coi \suq \poi$ and thus we have $(c, b) \in \coi \cap \poi$. 
As such, since $(a, c) \in \fr \cap \poi$, $(c, b) \in \coi \cap \poi$ and $\poi$ is transitively closed, from the definition of $\fr$ we have  $(a, b) \in \fr \cap \poi = B$, as required.\\

\noindent \textbf{TS. \eqref{subgoal:B_sihb_in_sihb}}\\
To show $B; \sihb \suq \sihb$, pick an arbitrary $(a, b) \in B; \sihb$.
We then know there exists $c$ such that $a \in \Reads$, $c \in \Writes$, $(a, c) \in B$ and $(c, b) \in \sihb$.
From the definition of $\sihb$ and since $c \in \Writes$, we know there exists $d$ such that 
$(c, d) \in \pot \cup \rft \cup \cot$ and $(d, b) \in \sihb^*$.
As $(a, c) \in B \suq \st$, we thus know that $(a, d) \in \pot \cup \rft \cup \cot \suq \sihb$. 
Consequently, since $(d, b) \in \sihb^*$ and $\sihb$ is transitively closed, 
we have $(a, b) \in \sihb$, as required. 

To show $\sihb; B \suq \sihb$, pick an arbitrary $(a, b) \in \sihb; B$.
We then know there exists $c$ such that $c \in \Reads$, $b \in \Writes$, $(a, c) \in \sihb$ and $(c, b) \in B$.
From the definition of $\sihb$, and since $c \in \Reads$, we know there exists $d$ such that 
$(a, d) \in \sihb^*$ and $(d, c) \in \pot \cup \rft \cup \cot$.
As $(c, b) \in B \suq \st$, we thus know that $(d, b) \in \pot \cup \rft \cup \cot \suq \sihb$. 
Consequently, since $(a, d) \in \sihb^*$ and $\sihb$ is transitively closed, 
we have $(a, b) \in \sihb$, as required.
\end{proof}
\end{theorem}

\newpage
\begin{theorem}[Monotonicity]\label{thm:rsi-mono-full}
Let $\absGRSI_T$ be an RSI execution graph obtained from an RSI execution graph $\absGRSI$ by wrapping some non-transactional events 
inside a new transaction.
If $\absGRSI_T$ is RSI-consistent, then so is $\absGRSI$.
\end{theorem}
\begin{proof}
First, we show that $[\Writes]; \absGRSI.\rsihb  \suq  \absGRSI_T.\poi  \cup  \absGRSI_T.\rsihb$.
Let $A$ denote  the events of the new transactions in $\absGRSI_T$.
Let $\tup{a,b}\in [\Writes]; \absGRSI.\rsihb$,
and consider a  
 $$(\po \setminus \poi) \cup [\Writes]; \poi;[\Writes] \cup 
\rf; [\NT] \cup [\NT]; \rf; \st \cup \tlift{\rf} \cup  \tlift{(\co; \rf)} \cup
\cot  \cup \sifr$$ path from $a$ to $b$ of minimal length (all relations are in $\absGRSI$).
Note that the only possible edges on this path that do not appear in a corresponding relation in $\absGRSI_T$
are $[A];\rf;[A]$ and $[A];\po;[A] \setminus (\Writes \times \Writes)$ edges.
Suppose first that $[A];\rf;[A]$ is used on this path.
Immediately  before or after $[A];\rf;[A]$, we can only have $\po \setminus \poi$, but then the
two edges can be replaced by a shorter path that uses only $\po \setminus \poi$.
Thus, in this case it follows that  $\tup{a,b}\in [A];\rf;[A] \suq \absGRSI_T.\poi$.
Next, suppose that $[A];\po;[A] \setminus (\Writes \times \Writes)$  is used on this path
(the minimality ensures such edge is used only once).
First, if this edge is the first edge on the path, then (again) it is the only edge on the path
(immediately after $[A];\po;[A]$, we can only have $\rf; [\NT\setminus A]$ or $[\NT]; \rf;\st$, but
both start with a write).
Hence, in this case we have  $\tup{a,b}\in [A];\po;[A] \suq \absGRSI_T.\poi$.
Second, consider the case that the $[A];\po;[A]$  edge is not the first on the path.
Immediately before $[A];\po;[A]$, we can only have a $[\Events\setminus A];\rf$ edge.
Then, in $\absGRSI_T$, the two edges can be joined into either $[\NT]; \absGRSI_T.\rf; \absGRSI_T.\st$
or $\absGRSI_T.\tlift{\rf}$.
Hence, we obtain $\tup{a,b}\in  \absGRSI_T.\rsihb$.

Now, suppose that $\absGRSI$ is not RSI-consistent.
If $\ref{ax:si_int}$ does not hold, then $\absGRSI_T$ is also not RSI-consistent and we are done.
Otherwise, $\absGRSI.\rsihbloc \cup \absGRSI.\co \cup \absGRSI.\fr$ is cyclic.
Since $\absGRSI.\co$ is total on writes to each location, 
it follows that 
 $\tup{a,a}\in \absGRSI.\rsihb;(\absGRSI.\co \cup \absGRSI.\fr)$
 for some $a\in \absGRSI.\Writes$.
 Since $\absGRSI.\co=\absGRSI_T.\co$, $\absGRSI.\fr=\absGRSI_T.\fr$, our claim above entails that 
$\tup{a,a}\in (\absGRSI_T.\poi  \cup  \absGRSI_T.\rsihb);(\absGRSI_T.\co \cup \absGRSI_T.\fr)$.
Now, 
$\tup{a,a}\in \absGRSI_T.\poi;(\absGRSI_T.\co \cup \absGRSI_T.\fr)$
implies that $\ref{ax:si_int}$ does not hold for $\absGRSI_T$,
while $\tup{a,a}\in \absGRSI_T.\rsihb;(\absGRSI_T.\co \cup \absGRSI_T.\fr)$
implies that  $\absGRSI_T.\rsihbloc \cup \absGRSI_T.\co \cup \absGRSI_T.\fr$
is cyclic.
In both cases, $\absGRSI_T$ is not RSI-consistent.
\end{proof}

\newpage
\begin{lemma}\label{lem:si_equivalent_aux}
For all $a, b, c$:
\[
%	\for{\txid_a, \txid_b, b, c} 
	\class a \st \ne \class b \st 
	\land \class a \st \times \class b \st \subseteq \sihb 
	\land (b, c) \in \reftransC{(\refC{\frt}; (\pot \cup \rft \cup \cot))} 
	\Rightarrow
	\class a \st \times \class c \st \subseteq \sihb
\]
\begin{proof}
Note that $\reftransC{(\refC{\frt}; (\pot \cup \rft \cup \cot))} = \bigcup_{n \in \Nats} S_{n}$, where $S_0 \eqdef \makerel{id}$ and for all $n \in \Nats$ we have $S_{n {+} 1} \eqdef (\refC{\frt}; (\pot \cup \rft \cup \cot)); S_n$. We thus demonstrate instead:
\[
	\for{a, b, c}  \for{n \in \Nats}
	\class a \st \ne \class b \st 
	\land \class a \st \times \class b \st \subseteq \sihb 
	\land (b, c) \in S_n
	\Rightarrow
	\class a \st \times \class c \st \subseteq \sihb
\]
We proceed by induction on $n$.\\

\noindent \textbf{Base case $n = 0$}\\
Follows immediately from the assumptions of the lemma and the definition of $S_0$.\\

\noindent \textbf{Inductive case $n = m {+} 1$}
\begin{align}
	& \for{a, b, c}  \for{i \leq m} \nonumber \\
	& \quad
	\class a \st \ne \class b \st 
	\land \class a \st \times \class b \st \subseteq \sihb 
	\land (b, c) \in S_i
	\Rightarrow
	\class a \st \times \class c \st \subseteq \sihb
	\tag{I.H.}
	\label{IH:aux_consistency}
\end{align}
Pick arbitrary $a, b, c$ such that $\class a \st \ne \class b \st$, $\class a \st \times \class b \st \subseteq \sihb $, and $(b, c) \in S_n$. From the definition of $S_n$ we then know there exists $d$ such that $(b, d) \in \refC{\frt}; (\pot \cup \rft \cup \cot)$ and $(d, c) \in S_m$. There are now two cases to consider: 1) $(b, d) \in \pot \cup \rft \cup \cot$; or 2) $(b, d) \in \frt; (\pot \cup \rft \cup \cot)$.

In case (1) from the definitions of $\pot$, $\rft$ and $\cot$ we have $\class b \st \times \class d \st \subseteq \pot \cup \rft \cup \cot$ and thus from the definition of $\sihb$ we have $\class b \st \times \class d \st \subseteq \sihb$. 
Since we also have $\class a \st \times \class b \st \subseteq \sihb$ and $\sihb$ is transitively closed, we have $\class a \st \times \class d \st \subseteq \sihb$. As $(d, c) \in S_m$, from (\ref{IH:aux_consistency}) we then have $\class a \st\times \class c \st \subseteq \sihb$, as required.

In case (2) we then know there exists $e$ such that $(b, e) \in \frt$ and $(e, d) \in \pot \cup \rft \cup \cot$.
From the definitions of $\pot$, $\rft$ and $\cot$ we then have $\class e \st \times \class d \st \subseteq \pot \cup \rft \cup \cot$ and thus from the definition of $\sihb$ we have $\class e \st \times \class d \st \subseteq \sihb$. 
On the other hand, from $(b, e) \in \frt$ we know there exist $r, w, w'$ such that $r \in \class b \st$, $w \in \class e \st$, $\class b \st \ne \class e \st$, $\class d \st \ne \class e \st$, $(r, w) \in \fr$, $(w', r) \in \rf$ and $(w', w) \in \co$. 
There are now two cases to consider: i) $w' \in \class b \st$; or ii) $w' \not\in \class b \st$.

In case (2.i) we then have $\class b \st \times \class e \st \subseteq \cot \subseteq \sihb$. Since we also have $\class a \st \times \class b \st \subseteq \sihb$ and $\class e \st \times \class d \st \subseteq \sihb$, from the transitivity of $\sihb$ we have $\class a \st \times \class d \st \subseteq \sihb$. As $(d, c) \in S_m$, from (\ref{IH:aux_consistency}) we then have $\class a \st \times \class c \st \subseteq \sihb$, as required.

In case (2.ii) we then know $r \in \EReads$. As such we have $(r, w) \in \sifr \subseteq \sihb$. As $\class a \st \times \class b \st \subseteq \sihb$ and $r \in \class b \st$ we then have $\class a \st \times \{w\} \subseteq \sihb$. 
Similarly, as $\class e \st \times \class d \st \subseteq \sihb$ and $w \in \class e \st$ we then have $\{w\} \times \class d \st \subseteq \sihb$. 
As such, from the transitivity of $\sihb$ we have $\class a \st \times \class d \st \subseteq \sihb$. As $(d, c) \in S_m$, from (\ref{IH:aux_consistency}) we then have $\class a \st \times \class c \st \subseteq \sihb$, as required.
\end{proof}

\end{lemma}

\begin{lemma}[Lock ordering]\label{lem:lock-ordering}
Given an RA-consistent execution graph $\impG = (\Events, \po, \rf, \co, \lo)$ of the SI or RSI implementations in \cref{fig:si_implementation}, for all $wl, wl' \in \WLocks_{\x}$, $pl, pl' \in \PLocks_\x$, $wu, wu' \in \WUnlocks_{\x}$, $rl, rl' \in \RLocks_\x$, and $ru \in \RUnlocks_\x$:
\begin{align}
%	& \for{wl, wl' \in \WLocks_{\xl}} \for{pl, pl' \in \PLocks_\xl} \for{wu, wu' \in \WUnlocks_{\xl}} \for{rl, rl' \in \RLocks_\xl}  \nonumber \\
	& 
	\begin{array}{r @{\hspace{1pt}} l  @{}}
		(wl, wu), (wl', wu') \in \imm{\po_\x} 
		& \Rightarrow 
		(wu, wl') \in \silo_\x \lor (wu', wl) \in \silo_\x \\
		(wl, wu), (rl, pl, wu') \in \imm{\po_\x} 
		& \Rightarrow 
		(wu, rl) \in \silo_\x \lor (wu', wl) \in \silo_\x \\
		(rl, pl, wu), (rl', pl', wu') \in \imm{\po_\x} 
		& \Rightarrow 
		(wu, rl') \in \silo_\x \lor (wu', rl) \in \silo_\x \\
	\end{array} 
	\tag{\textsc{WWSync}}
	\label{ax:lock_wwaxiom} \\[5pt]
%	
%
%	& \for{wl \in \WLocks_{\xl}} \for{pl \in \PLocks_\xl} \for{wu \in \WUnlocks_{\xl}} \for{rl, rl' \in \RLocks_\xl} \for{ru \in \RUnlocks_\xl}  \nonumber \\
%	&\quad 
	& \begin{array}{r @{\hspace{1pt}} l  @{}}
		(wl, wu), (rl, ru) \in \imm{\po_\x}
		& \Rightarrow 
		(wu, rl) \in \silo_\x \lor (ru, wl) \in \silo_\x   \\
		(rl, pl, wu), (rl', ru) \in \imm{\po_\x}
		& \Rightarrow 
		(wu, rl') \in \silo_\x \lor (ru, pl) \in \silo_\x  
	\end{array}	
	\tag{\textsc{RWSync}}
	\label{ax:lock_rwaxiom}	
\end{align}
where given a relation $r$ we write $(a, b, c) \in r$ as a shorthand for $(a, b), (b, c) \in r$. 
\begin{proof}[Proof (\ref{ax:lock_wwaxiom})]
Pick an arbitrary RA-consistent execution graph $\impG = (\Events, \po, \rf, \co, \lo)$ of the SI or RSI implementations in \cref{fig:si_implementation}, and pick arbitrary $wl, wl' \in \WLocks_{\x}$, $pl, pl' \in \PLocks_\x$, $wu, wu' \in \WUnlocks_{\x}$ and $rl, rl' \in \RLocks_\x$. We are then required to show: 
\begin{align} 
		(wl, wu), (wl', wu') \in \imm{\po_\x} 
		& \Rightarrow 
		(wu, wl') \in \silo_\x \lor (wu', wl) \in \silo_\x
		\label{goal:wwlock1}\\
		(wl, wu), (rl, pl, wu') \in \imm{\po_\x} 
		& \Rightarrow 
		(wu, rl) \in \silo_\x \lor (wu', wl) \in \silo_\x
		\label{goal:wwlock2}\\
		(rl, pl, wu), (rl', pl', wu') \in \imm{\po_\x} 
		& \Rightarrow 
		(wu, rl') \in \silo_\x \lor (wu', rl) \in \silo_\x
		\label{goal:wwlock3}
\end{align}
\textbf{RTS. (\ref{goal:wwlock1})}\\
We proceed by contradiction. Let $(wl, wu), (wl', wu') \in \imm{\po_\x}$ and $(wu, wl') \not\in \lo_\x \land (wu', wl) \not\in \lo_\x$. 
Since $\lo_\x$ is totally ordered w.r.t. write lock events, we then have $(wl', wu) \in \silo_\x$ and $(wl, wu') \in \silo_\x$. 

From the mutual exclusion \eqref{ax:lock_wex} afforded by lock events, and since $(wl', wu) \in \silo_\x$ we know there exists $wu'' \in \WUnlocks_{\x}$ such that $wl' \relarrow{\po} wu''$ and $wu'' \relarrow{\lo} wu$. Moreover, since $wl' \relarrow{\imm{\po_\x}} wu'$ we have $wu' \relarrow{\reftransC{\po}} wu''$. 
Once again, since $\lo$ is total w.r.t. write lock events and agrees with $\po$ (otherwise we would have a cycle contradicting the assumption that $\impG$ is RA-consistent), we have $wu' \relarrow{\reftransC{\lo}} wu''$. We then have $wu' \relarrow{\reftransC{\lo}} wu'' \relarrow{\lo}wu$ and thus since $\lo$ is an order (i.e.~is transitive), we have $wu' \relarrow{\lo}wu$. 

Following a similar argument symmetrically, we get $wu \relarrow{\lo}wu'$. We then have $wu' \relarrow{\lo}wu \relarrow{\lo}wu'$, contradicting the assumption that $\lo$ is a strict order.\\

\noindent \textbf{RTS. (\ref{goal:wwlock2})}\\
We proceed by contradiction. Let $(wl, wu), (rl, pl, wu') \in \imm{\po_\x}$ and $(wu, rl) \not\in \lo_\x \land (wu', wl) \not\in \lo_\x$. 
Since $\lo_\x$ is totally ordered w.r.t. write lock events, we then have $(rl, wu) \in \silo_\x$ and $(wl, wu') \in \silo_\x$. 

From \eqref{ax:lock_rshared} and since $(rl, wu) \in \silo_\x$ we know there exists $l \in \RUnlocks_{\x} \cup \PLocks_{\x}$ such that $rl \relarrow{\po} l$ and $l \relarrow{\lo} wu$. Moreover, since $rl \relarrow{\imm{\po_\x}} pl$ we have $pl \relarrow{\reftransC{\po}} l$. 
Once again, since $\lo$ is total w.r.t. write lock events and agrees with $\po$ (otherwise we would have a cycle contradicting the assumption that $\impG$ is RA-consistent), we have $pl \relarrow{\reftransC{\lo}} l$. We then have $pl \relarrow{\reftransC{\lo}} l \relarrow{\lo}wu$ and thus since $\lo$ is an order (i.e.~is transitive), we have $pl \relarrow{\lo}wu$. 
Similarly, from the mutual exclusion \eqref{ax:lock_wex} of write locks and since $pl \relarrow{\lo}wu$, we know there exists $wu'' \in \WUnlocks_{\x}$ such that $pl \relarrow{\po} wu''$ and $wu'' \relarrow{\lo} wu$. Moreover, since $pl \relarrow{\imm{\po_\x}} wu'$ we have $wu' \relarrow{\reftransC{\po}} wu''$. 
Again, since $\lo$ is total w.r.t. write lock events and agrees with $\po$ (otherwise we would have a cycle contradicting the assumption that $\impG$ is RA-consistent), we have $wu' \relarrow{\reftransC{\lo}} wu''$. We then have $wu' \relarrow{\reftransC{\lo}} wu'' \relarrow{\lo}wu$ and thus since $\lo$ is an order (i.e.~is transitive), we have $wu' \relarrow{\lo}wu$.

Analogously, from \eqref{ax:lock_wex} and since $(wl, wu') \in \lo$, we know there exists $wu''' \in \WUnlocks_{\x}$ such that $wl \relarrow{\po} wu'''$ and $wu''' \relarrow{\lo} wu'$. Moreover, since $wl \relarrow{\imm{\po_\x}} wu$ we have $wu \relarrow{\reftransC{\po}} wu'''$. 
Once again, since $\lo$ is total w.r.t. write lock events and agrees with $\po$ (otherwise we would have a cycle contradicting the assumption that $\impG$ is RA-consistent), we have $wu \relarrow{\reftransC{\lo}} wu'''$. We then have $wu \relarrow{\reftransC{\lo}} wu''' \relarrow{\lo}wu'$ and thus since $\lo$ is an order (i.e.~is transitive), we have $wu \relarrow{\lo}wu'$.

We then have $wu' \relarrow{\lo}wu \relarrow{\lo}wu'$, contradicting the assumption that $\lo$ is a strict order.\\

\noindent \textbf{RTS. (\ref{goal:wwlock3})}\\
We proceed by contradiction. Let $(rl, pl, wu), (rl', pl', wu') \in \imm{\po_\x}$ and $(wu, rl') \not\in \lo_\x \land (wu', rl) \not\in \lo_\x$. 
Since $\lo_\x$ is totally ordered w.r.t. write lock events, we then have $(rl, wu') \in \silo_\x$ and $(rl', wu) \in \silo_\x$. 

From \eqref{ax:lock_rshared} and since $(rl, wu') \in \silo_\x$ we know there exists $l \in \RUnlocks_{\x} \cup \PLocks_{\x}$ such that $rl \relarrow{\po} l$ and $l \relarrow{\lo} wu'$. 
Moreover, since $rl \relarrow{\imm{\po_\x}} pl$ we have $pl \relarrow{\reftransC{\po}} l$. 
Once again, since $\lo$ is total w.r.t. write lock events and agrees with $\po$ (otherwise we would have a cycle contradicting the assumption that $\impG$ is RA-consistent), we have $pl \relarrow{\reftransC{\lo}} l$. We then have $pl \relarrow{\reftransC{\lo}} l \relarrow{\lo}wu'$ and thus since $\lo$ is an order (i.e.~is transitive), we have $pl \relarrow{\lo}wu'$. 
Similarly, from \eqref{ax:lock_wex} and since $pl \relarrow{\lo}wu'$, we know there exists $wu'' \in \WUnlocks_{\x}$ such that $pl \relarrow{\po} wu''$ and $wu'' \relarrow{\lo} wu'$. Moreover, since $pl \relarrow{\imm{\po_\x}} wu$ we have $wu \relarrow{\reftransC{\po}} wu''$. 
Again, since $\lo$ is total w.r.t. write lock events and agrees with $\po$ (otherwise we would have a cycle contradicting the assumption that $\impG$ is RA-consistent), we have $wu \relarrow{\reftransC{\lo}} wu''$. We then have $wu \relarrow{\reftransC{\lo}} wu'' \relarrow{\lo}wu'$ and thus since $\lo$ is an order (i.e.~is transitive), we have $wu \relarrow{\lo}wu'$.

Following a similar argument symmetrically, we get $wu' \relarrow{\lo}wu$. We then have $wu \relarrow{\lo}wu' \relarrow{\lo}wu$, contradicting the assumption that $\lo$ is a strict order.

\renewcommand{\qed}{}
\end{proof}
\begin{proof}[Proof (\ref{ax:lock_rwaxiom})]
Pick an arbitrary RA-consistent execution graph $\impG = (\Events, \po, \rf, \co, \lo)$ of the SI or RSI implementations in \cref{fig:si_implementation}, and pick arbitrary $wl \in \WLocks_{\x}$, $pl \in \PLocks_\x$, $wu \in \WUnlocks_{\x}$, $rl, rl' \in \RLocks_\x$ and $ru \in \RUnlocks_\x$. We are then required to show: 
\begin{align}
	(wl, wu), (rl, ru) \in \imm{\po_\x}
	& \Rightarrow 
	(wu, rl) \in \silo_\x \lor (ru, wl) \in \silo_\x   
	\label{goal:rwlock1}\\
	(rl, pl, wu), (rl', ru) \in \imm{\po_\x}
	& \Rightarrow 
	(wu, rl') \in \silo_\x \lor (ru, pl) \in \silo_\x  
	\label{goal:rwlock2}
\end{align}
\textbf{RTS. (\ref{goal:rwlock1})}\\
We proceed by contradiction. Let $(wl, wu), (rl, ru) \in \imm{\po_\x}$ and $(wu, rl) \not\in \lo_\x \land (ru, wl) \not\in \lo_\x$. 
Since $\lo_\x$ is totally ordered w.r.t. write lock events, we then have $(rl, wu) \in \lo_\x$ and $(wl, ru) \in \lo_\x$. 

From \eqref{ax:lock_wex} and since $(wl, ru) \in \lo_\x$ we know there exists $wu' \in \WUnlocks_{\x}$ such that $wl \relarrow{\po} wu'$ and $wu' \relarrow{\lo} ru$. Moreover, since $wl \relarrow{\imm{\po_\x}} wu$ we have $wu \relarrow{\reftransC{\po}} wu'$. 
Once again, since $\lo$ is total w.r.t. write lock events and agrees with $\po$ (otherwise we would have a cycle contradicting the assumption that $\impG$ is RA-consistent), we have $wu \relarrow{\reftransC{\lo}} wu'$. We then have $wu \relarrow{\reftransC{\lo}} wu' \relarrow{\lo}ru$ and thus since $\lo$ is an order (i.e.~is transitive), we have $wu \relarrow{\lo}ru$. 

Analogously, from \eqref{ax:lock_rshared} and since $(rl, wu) \in \lo_\x$ we know there exists $l \in \RUnlocks_{\x} \cup \PLocks_{\x}$ such that $rl \relarrow{\po} l$ and $l \relarrow{\lo} wu$. Moreover, since $rl \relarrow{\imm{\po_\x}} ru$ we have $ru \relarrow{\reftransC{\po}} l$. 
We then have $l \relarrow{\lo} wu \relarrow{\lo} ru \relarrow{\reftransC{\po}} l$, contradicting the assumption that $\impG$ is RA-consistent. \\

\noindent\textbf{RTS. (\ref{goal:rwlock2})}\\
We proceed by contradiction. Let $(rl, pl, wu), (rl', ru) \in \imm{\po_\x}$ and $(wu, rl') \not\in \lo_\x \land (ru, pl) \not\in \lo_\x$. 
Since $\lo_\x$ is totally ordered w.r.t. write lock events, we then have $(rl', wu) \in \lo_\x$ and $(pl, ru) \in \lo_\x$. 

From \eqref{ax:lock_wex} and since $(pl, ru) \in \lo_\x$ we know there exists $wu' \in \WUnlocks_{\x}$ such that $pl \relarrow{\po} wu'$ and $wu' \relarrow{\lo} ru$. Moreover, since $pl \relarrow{\imm{\po_\x}} wu$ we have $wu \relarrow{\reftransC{\po}} wu'$. 
Once again, since $\lo$ is total w.r.t. write lock events and agrees with $\po$ (otherwise we would have a cycle contradicting the assumption that $\impG$ is RA-consistent), we have $wu \relarrow{\reftransC{\lo}} wu'$. We then have $wu \relarrow{\reftransC{\lo}} wu' \relarrow{\lo}ru$ and thus since $\lo$ is an order (i.e.~is transitive), we have $wu \relarrow{\lo}ru$. 

Analogously, from \eqref{ax:lock_rshared} and since $(rl', wu) \in \lo_\x$ we know there exists $l \in \RUnlocks_{\x} \cup \PLocks_{\x}$ such that $rl \relarrow{\po} l$ and $l \relarrow{\lo} wu$. Moreover, since $rl \relarrow{\imm{\po_\x}} ru$ we have $ru \relarrow{\reftransC{\po}} l$. 
We then have $l \relarrow{\lo} wu \relarrow{\lo} ru \relarrow{\reftransC{\po}} l$, contradicting the assumption that $\impG$ is RA-consistent. 
\end{proof}
\end{lemma}

\newcommand{\TMO}[1][]{\makerel{TMO}_{#1}}

\renewcommand{\absGSI}{\ensuremath{\impG'}}

\newcommand{\PO}{\makerel{PO}}
\newcommand{\RF}[1][\code x]{\makerel{RF}^{#1}}
\newcommand{\CO}[1][\code x]{\makerel{MO}^{#1}}

\newcommand{\FR}{\makerel{FR}}
\newcommand{\HB}{\makerel{HB}}

\newcommand{\phase}[2][\code x]{\ensuremath{\mathit P^{#1}_{#2}}}
\newcommand{\rphase}[2][\code x]{\ensuremath{\mathit{RP}^{#1}_{#2}}}
\newcommand{\sphase}[2][\code x]{\ensuremath{\mathit S^{#1}_{#2}}}
\newcommand{\srphase}[2][\code x]{\ensuremath{ \mathit{RS}^{#1}_{#2}}}

\newcommand{\writer}[1]{\ensuremath{#1\mathsf{.writer}}}

\newcommand{\stage}[1]{\func{stg}{#1}}

\newcommand{\TCO}{\ensuremath{\makerel{\uppercase{TCO}}}\xspace}
\newcommand{\TCOA}[1][\code x]{\ensuremath{\makerel{\uppercase{TCO}}_{#1}}\xspace}

\newcommand{\LO}{{\color{red} \makerel{LO}}}
\newcommand{\RLO}{{\color{red} \makerel{RLO}}}

\newcommand{\wsync}[1][\lo_{\xl}]{\ensuremath{\pred{wsync}{#1}}}
\newcommand{\rsync}[1][\silor_{\xl}]{\ensuremath{\pred{rsync}{#1}}}

\newcommand{\siloe}{{\color{colorLO} \tout{\silo}}}

~\newpage
\section{Soundness and Completeness of the Eager SI Implementation}\label{app:si}
\paragraph{Notation} Given an execution graph $\absGSI=(\Events, \po, \rf, \co, \lo)$ we write $\TClasses$ for the set of equivalence classes of $\Transactions$ induced by $\st$; $\class{a}{\st}$ for the equivalence class that contains $a$; and $\Transactions_\txid$ for the equivalence class of transaction $\txid\in \TXIDs$: $\Transactions_\txid \eqdef \setcomp{a}{\tx{a} {=} \txid}$.
We write $\sicon$ to denote that $\absGSI$ is SI-consistent; and write $\consistent{\impG}$ to denote that $\impG$ is RA-consistent.

Given an execution graph $\impGSI$ of the SI implementation in \cref{fig:si_implementation}, let us assign a transaction identifier to each transaction executed by the program; and given a transaction $\txid$, let $\readset_{\txid}$ and $\writeset_{\txid}$ denote its read and write sets, respectively.
Observe that given a transaction $\txid$ of the SI implementation in \cref{fig:si_implementation} with $\readset_{\txid} \cup \writeset_{\txid} = \simpleset{\code{x}_1, \cdots, \code{x}_i}$, the trace of $\txid$, written $\trace_{\txid}$, is of the following form: 
\[
	\trace_{\txid} = 
	\mathit{FS}^{*} 
	\relarrow{\imm \po} \mathit{Rs} 
	\relarrow{\imm \po} \mathit{RUs}
	\relarrow{\imm \po} \mathit{PLs}
	\relarrow{\imm \po} \mathit{Ts}
	\relarrow{\imm \po} \mathit{Us}
\]
where $\mathit{FS}^{*}$ denotes the sequence of events attempting but failing to obtain a snapshot, and
\begin{itemize}
	\item $\mathit{Rs}$ denotes the sequence of events acquiring the reader locks (on all locations accessed) and capturing a snapshot of the read set, and is of the form $\mathit{rl}_{\code x_1} \relarrow{\imm \po} \cdots \relarrow{\imm \po} \mathit{rl}_{\code x_i} \relarrow{\imm \po} \mathit{S}_{\code x_1} \relarrow{\imm \po} \cdots \relarrow{\imm \po} \mathit{S}_{\code x_i}$, where for all $n \in \{1 \cdots i\}$:
%	\small
%	\[
%	\begin{array}{l}
%		\mathit{R}_{\code x_n} = 
%		\mathit{rl}_{\code x_n} 
%		\relarrow{\imm \po} 
%		\mathit{S}_{\code x_n}		
%	\end{array}	
%	\]
%%	\normalsize
%	where 
	\[
		\begin{array}{l}
			\mathit{rl}_{\code x_n} = \rlockE{\xl_n} 
			\qquad
			\mathit{S}_{\code x_n} =
			\begin{cases}
				\mathit{rs}_{\code x_n}
				\relarrow{\imm \po}
				\mathit{ws}_{\code x_n}
				& \text{if } \code x \in \readset_{\txid} \\
				
				\emptyset & \text{otherwise}
			\end{cases} 
		\end{array}
	\]	
	with $\mathit{rs}_{\code x_n} \eqdef \readE{\acq}{\code x_n}{v_n}$ 
	and $\mathit{ws}_{\code x_n} \eqdef \writeE{\rel}{\code{s[x}_n\code ]}{v_n} $, for some $v_n$.
	\item $\mathit{RUs}$ denotes the sequence of events releasing the reader locks (when the given location is in the read set only) and is of the form $\mathit{ru}_{\code x_1} \relarrow{\imm \po} \cdots \relarrow{\imm \po} \mathit{ru}_{\code x_i}$, where for all $n \in \{1 \cdots i\}$:
	\[
	\begin{array}{l}
		\mathit{ru}_{\code x_n} = 
		\begin{cases}
			\runlockE{\xl_n}
			& \text{ if } \code x_n \not\in \writeset_{\txid}  \\
			\emptyset
			& \text{ otherwise}
		\end{cases}
	\end{array}	
	\]
	\item $\mathit{PLs}$ denotes the sequence of events promoting the reader locks to writer ones (when the given location is in the write set), and is of the form $\mathit{pl}_{\code x_1} \relarrow{\imm \po} \cdots \relarrow{\imm \po} \mathit{pl}_{\code x_i}$, where for all $n \in \{1 \cdots i\}$:
	\[
	\begin{array}{l}
		\mathit{pl}_{\code x_n} = 
		\begin{cases}
			\plockE{\xl_n}
			& \text{if }  \code x_n \in \writeset_{\txid} \\
			\emptyset
			& \text{ otherwise } 
		\end{cases}
	\end{array}	
	\]
	\item $\mathit{Ts}$ denotes the sequence of events corresponding to the execution of \denot{\code{T}} in \cref{fig:si_implementation} and is of the form $\mathit{t}_1 \relarrow{\imm \po} \cdots \relarrow{\imm \po} \mathit{t}_k$, where for all $m \in \{1 \cdots k\}$:
	\[
	\mathit{t}_m = 
	\begin{cases}
		\readE{-}{\code{s[x}_n \code{]}}{v_n} & \text{if } O_m {=} \readE{-}{\code x_n}{v_n} \\
		\writeE{\rel}{\code x_n}{v_n} \relarrow{\imm \po} \writeE{-}{\code{s[x}_n \code ]}{v_n}
		& \text{if } O_m {=} \writeE{\rel}{\code x_n}{v_n} \\
	\end{cases}
	\]
	where $O_m$ denotes the $m$th event in the trace of the original $\code T$;
	\item $\mathit{Us}$ denotes the sequence of events releasing the locks on the write set. That is, the events in $\mathit{Us}$ correspond to the execution of the last line of the implementation in \cref{fig:si_implementation}, and is of the form $\mathit{wu}_{\code x_1} \relarrow{\imm \po} \cdots \relarrow{\imm \po} \mathit{wu}_{\code x_i}$, where for all $n \in \{1 \cdots i\}$:
	\[
		\mathit{wu}_{\code x_n} = 
		\begin{cases}
			\wunlockE{\code{yl}_n} & \text{if } \code x_n \in \writeset_{\txid} \\
			\emptyset & \text{otherwise}
		\end{cases}				
	\]
\end{itemize}

%and $\accset = \simpleset{\code z_1, \cdots, \code z_k}$, 
%where $\accset$ denotes the set of locations both read from and written to by \code T; that is, $\accset \subseteq \readset \land \accset \subseteq \writeset$. 

Given a transaction trace $\trace_{\txid}$, we write e.g.~$\txid.\mathit{Ls}$ to refer to its constituent $\mathit{Ls}$ sub-trace and write $\mathit{Ls}.\Events$ for the set of events related by \po in $\mathit{Ls}$. Similarly, we write $\txid.\Events$ for the set of events related by \po in $\trace_{\txid}$.
Note that $\impGSI.\Events = \bigcup\limits_{\txid \in \sort{Tx}}  \txid.\Events $.
%As such, the $\sti$ relation induces a set of equivalence classes on $\absGSI.\Events$, written $\absGSI.\Events/\sti$. As before, we write $\class{a}{\sti}$ for the equivalence class in $\absGSI.\Events/\sti$ that contains $a$.

%
\subsection{Implementation Soundness}
\label{subapp:si_soundness}
In order to establish the soundness of our implementation, it suffices to show that given an RA-consistent execution graph of the implementation $\impGSI$, we can construct a corresponding SI-consistent execution graph $\absGSI$ with the same outcome.
%
%\[
%	\for{\impGSI} \consistent{\impGSI} \Rightarrow \exsts{\absGSI} \sicon
%\]
%%

Given a transaction $\txid \in \sort{Tx}$ with $\readset_{\txid} \cup \writeset_{\txid}=\{\code x_1 \cdots \code x_i\}$
%, $\writeset_{\txid_t}=\{\code y_1 \cdots \code y_j\}$ 
and trace 
$ \trace_{\txid} = \mathit{Fs}^{*}
	\relarrow{\imm \po} \mathit{Rs}
	\relarrow{\imm \po} \mathit{RUs}
	\relarrow{\imm \po} \mathit{PLs}
	\relarrow{\imm \po} \mathit{Ts}
	\relarrow{\imm \po} \mathit{Us}$, 
with $\mathit{Ts}= \mathit{t}_1 \relarrow{\imm \po} \cdots \relarrow{\imm \po} \mathit{t}_k$, 
we construct the corresponding implementation trace $\trace'_{\txid}$ as follows:
\[
	\trace'_{\txid} \eqdef  \mathit{t}'_1 \relarrow{\po} \cdots \relarrow{\po} \mathit{t}'_k
%	b
%	\relarrow{\imm \po}
%	 \mathit{Ts}'
%	\relarrow{\imm \po}
%	e  
\]
where  for all $m \in \{1 \cdots k\}$:
\[
\begin{array}{l c l}
	\mathit{t}'_m {=} \readE{\acq}{\code x_n}{rb_n}
	& \text{when} & 
	\mathit{t}_m = \readE{-}{\code{s[x}_n \code{]}}{rb_n} \\
	
	\mathit{t}'_m {=} \writeE{\rel}{\code x_n}{rb_n}
	& \text{when} & 
	\mathit{t}_m = \writeE{\rel}{\code x_n}{rb_n} \relarrow{\imm \po} \writeE{-}{\code{s[x}_n \code{]}}{rb_n} 
\end{array}	
\]
such that in the first case the identifier of $\mathit{t}'_m$ is that of $\mathit{t}_m$; and in the second case the identifier of $\mathit{t}'_m$ is that of the first event in $\mathit{t}_m$.
We then define:
\[
	\mathsf{RF}_{\txid} \eqdef
	\setcomp{
		(w, t'_j)	
	}{
		t'_j \in \mathit{Ts}' \land \exsts{\code x, v} t'_j {=} \readE{\acq}{\code x}{v} \land w {=} \writeE{\rel}{\code x}{v} \\
		\land (w \in \txid.\Events \Rightarrow 
		\begin{array}[t]{@{} l @{}}
			w \relarrow{\po} t'_j \,\land\\
			(\for{e \in \txid.\Events } w \relarrow{\po} e \relarrow{\po} t'_j \Rightarrow (\loc e {\ne} \code x \lor e {\not\in} \Writes)))
		\end{array} \\				
		
		\land (w \not \in \txid.\Events \Rightarrow 
		\begin{array}[t]{@{} l @{}}
			(\for{e \in \txid.\Events} (e \relarrow{\po} t'_j \Rightarrow (\loc e \ne \code x \lor e \not\in \Writes)) \\
			\land\, (w, \txid.\mathit{rs}_{\x})  \in \impGSI.\rf)
		\end{array}		
	}
\]

We are now in a position to demonstrate the soundness of our implementation. Given an RA-consistent execution graph $\impGSI$ of the implementation, we construct an SI execution graph $\absGSI$ as follows and demonstrate that $\sicon$ holds.

\begin{itemize}
	\item $\absGSI.\Events = \bigcup\limits_{\txid \in \sort{Tx}} \trace'_{\txid}.\Events$, with the $\tx{.}$ function defined as:
	\[
		\tx{a} \eqdef \txid \quad \text{ where } \quad a \in \trace'_{\txid}
	\]
	\item $\absGSI.\po = \coerce{\impGSI.\po}{\absGSI.\Events}$
	\item $\absGSI.\rf = \bigcup_{\txid \in \textsc{Tx}} \mathsf{RF}_{\txid}$%\absGSI.\Writes \times \absGSI.\Reads \cap \setcomp{(a, b)}{\src{b} = a}$
	\item $\absGSI.\co = \coerce{\impGSI.\co}{\absGSI.\Events}$
	\item $\absGSI.\silo = \emptyset$
%	\item $\absGSI.\Transactions = \absGSI.\Events$
%	with the $\tx{.}$ function defined as:
%	\[
%		\tx{a} \eqdef \txid \quad \text{ where } \quad a \in \trace'_{\txid}
%	\]
\end{itemize}
Observe that the events of each $\trace'_{\txid}$ trace coincides with those of the equivalence class  $\Transactions_{\txid}$ of $\absGSI$. That is,  $\trace'_{\txid}.\Events = \Transactions_{\txid}$. 

\begin{lemma}\label{lem:si_lock_hb}
Given an RA-consistent execution graph $\impGSI$ of the implementation and its corresponding SI execution graph $\absGSI$ constructed as above, for all $a, b, \txid_a, \txid_b, \x$:
\small
\begin{align}
	& \hspace*{-15pt}
	\txid_a \ne \txid_b
	\land a \in \txid_a.\Events 
	\land b \in \txid_b.\Events 
	\land \loc a = \loc b = \x
	\Rightarrow \nonumber \\
	&  \hspace*{-15pt} \;\; ((a, b) \in \absGSI.\rf \Rightarrow  \txid_a.\mathit{wu}_{\x} \relarrow{\impGSI.\hb} \txid_b.\mathit{rl}_{\x} ) 
	\label{lem:si_lock_hb_rf} \\
	& \hspace*{-15pt} \;\; \land ((a, b) \in \absGSI.\co \Rightarrow  \txid_a.\mathit{wu}_{\x} \relarrow{\impGSI.\hb} \txid_b.\mathit{rl}_{\x} ) 
	\label{lem:si_lock_hb_co}\\
	& \hspace*{-15pt} \;\; \land ((a, b) \in \absGSI.\co;\rf \Rightarrow  \txid_a.\mathit{wu}_{\x} \relarrow{\impGSI.\hb} \txid_b.\mathit{rl}_{\x} ) 
	\label{lem:si_lock_hb_corf}\\
	& \hspace*{-15pt} \;\; \land \big((a, b) \in \absGSI.\fr \Rightarrow  
		(\x \in \writeset_{\txid_a} \land \txid_a.\mathit{wu}_{\x} \relarrow{\impGSI.\hb} \txid_b.\mathit{rl}_{\x} ) 
		\lor 
		(\x \not\in \writeset_{\txid_a} \land \txid_a.\mathit{ru}_{\x} \relarrow{\impGSI.\hb} \txid_b.\mathit{pl}_{\x} ) 
	\big)
	\label{lem:si_lock_hb_fr}
\end{align}
\normalsize	
\begin{proof}
Pick an arbitrary RA-consistent execution graph $\impGSI$ of the implementation and its corresponding SI execution graph $\absGSI$ constructed as above. Pick an arbitrary $a, b, \txid_a, \txid_b, \x$ such that $\txid_a \ne \txid_b$, $a \in \txid_a.\Events$, $a \in \txid_a.\Events$, and $\loc a = \loc b = \x$.\\

\noindent \textbf{RTS. (\ref{lem:si_lock_hb_rf})}\\
Assume $(a, b) \in \absGSI.\rf$. From the definition of $\absGSI.\rf$ we then know $(a, \txid_a.\mathit{rs}_{\x}) \in \impGSI.\rf$.
On the other hand, from \cref{lem:lock-ordering} we  know that either i) $\x \in \writeset_{\txid_b}$ and $\txid_b.\mathit{wu}_{x} \relarrow {\impGSI.\hb} \txid_a.\mathit{rl}_{x}$; or ii)  $\x \not\in \writeset_{\txid_b}$ and  $\txid_b.\mathit{ru}_{x} \relarrow {\impGSI.\hb} \txid_a.\mathit{pl}_{x}$; or iii) $\txid_a.\mathit{wu}_{x} \relarrow {\impGSI.\hb} \txid_b.\mathit{rl}_{x}$.
In case (i) we then have $a \relarrow{\impGSI.\rf} \txid_a.\mathit{rs}_{\x} \relarrow{\impGSI.\po} \txid_b.\mathit{wu}_{x}  \relarrow{\impGSI.\hb} \txid_a.\mathit{rl}_{x}  \relarrow{\impGSI.\po} a$. That is, we have $a \relarrow{\impGSI.\hbloc} a$, contradicting the assumption that $\impGSI$ is RA-consistent. 
Similarly in case (ii) we have $a \relarrow{\impGSI.\rf} \txid_a.\mathit{rs}_{\x} \relarrow{\impGSI.\po} \txid_b.\mathit{ru}_{x}  \relarrow{\impGSI.\hb} \txid_a.\mathit{pl}_{x}  \relarrow{\impGSI.\po} a$.  That is, we have $a \relarrow{\impGSI.\hbloc} a$, contradicting the assumption that $\impGSI$ is RA-consistent. 
In case (iii) the desired result holds trivially.\\

\noindent \textbf{RTS. (\ref{lem:si_lock_hb_co})}\\
Assume $(a, b) \in \absGSI.\co$. From the definition of $\absGSI.\co$ we then know $(a, b) \in \impGSI.\co$.
On the other hand, from \cref{lem:lock-ordering}  we  know that either i) $\txid_b.\mathit{wu}_{x} \relarrow {\impGSI.\hb} \txid_a.\mathit{rl}_{x}$; or ii) $\txid_a.\mathit{wu}_{x} \relarrow {\impGSI.\hb} \txid_b.\mathit{rl}_{x}$.
In case (i) we then have $a \relarrow{\impGSI.\co} b \relarrow{\impGSI.\po} \txid_b.\mathit{wu}_{x}  \relarrow{\impGSI.\hb} \txid_a.\mathit{rl}_{x}  \relarrow{\impGSI.\po} a$. That is, we have $a \relarrow{\impGSI.\hbloc} a$, contradicting the assumption that $\impGSI$ is RA-consistent. 
In case (ii) the desired result holds trivially.\\

\noindent \textbf{RTS. (\ref{lem:si_lock_hb_corf})}\\
Assume $(a, b) \in \absGSI.\co; \rf$. We then know there exists $w$ such that $(a, w) \in \absGSI.\co$ and $(w, b) \in \absGSI.\rf$. From the definition of $\absGSI.\co$ we then know $(a, w) \in \impGSI.\co$.
There are now three cases to consider: 1) $w \in \txid_a$; or 2) $w \in \txid_b$; or 3) $w \in \txid_c \land \txid_c \ne \txid_a \land \txid_c \ne \txid_b$. 
In case (1) the desired result follows from part \ref{lem:si_lock_hb_rf}.
In case (2) since $(a, w) \in \absGSI.\co$ the desired result follows from part \ref{lem:si_lock_hb_co}.

In case (3) from the proof of part \ref{lem:si_lock_hb_co} we have $\txid_a.\mathit{wu}_{x} \relarrow {\impGSI.\hb} \txid_c.\mathit{rl}_{x}$.
Moreover, from the shape of $\impGSI$ traces we have $\txid_c.\mathit{rl}_{x} \relarrow {\impGSI.\po} \txid_c.\mathit{wu}_{x}$.
On the other hand, from the proof of part \ref{lem:si_lock_hb_rf} we have $\txid_c.\mathit{wu}_{x} \relarrow {\impGSI.\hb} \txid_b.\mathit{rl}_{x}$.
We thus have 
$\txid_a.\mathit{wu}_{x} \relarrow {\impGSI.\hb} \txid_c.\mathit{rl}_{x} \relarrow {\impGSI.\po} \txid_c.\mathit{wu}_{x} \relarrow {\impGSI.\hb} \txid_b.\mathit{rl}_{x}$.
As $\impGSI.\po \subseteq \impGSI.\hb$ and $\impGSI.\hb$ is transitively closed, we have $\txid_a.\mathit{wu}_{x} \relarrow {\impGSI.\hb} \txid_b.\mathit{rl}_{x}$, as required. \\

%In case (3) from the definition of $\absGSI.\rf$ we know that $(w, \txid_b.\mathit{rs}_{\x}) \in \impGSI.\rf$. 
%From \cref{lem:lock-ordering} we  know that either i) $\x \in \writeset_{\txid_b}$ and $\txid_b.\mathit{wu}_{x} \relarrow {\impGSI.\hb} \txid_a.\mathit{rl}_{x}$; or ii)  $\x \not\in \writeset_{\txid_b}$ and  $\txid_b.\mathit{ru}_{x} \relarrow {\impGSI.\hb} \txid_a.\mathit{pl}_{x}$; or iii) $\txid_a.\mathit{wu}_{x} \relarrow {\impGSI.\hb} \txid_b.\mathit{rl}_{x}$.
%In case (i) we then have $a \relarrow{\impGSI.\co} w \relarrow{\impGSI.\rf} \txid_a.\mathit{rs}_{\x} \relarrow{\impGSI.\po} \txid_b.\mathit{wu}_{x}  \relarrow{\impGSI.\hb} \txid_a.\mathit{rl}_{x}  \relarrow{\impGSI.\po} a$. That is, we have $a \relarrow{\impGSI.\co} w\relarrow{\impGSI.\hbloc} a$, contradicting the assumption that $\impGSI$ is consistent. 
%Similarly in case (ii) we have $a  \relarrow{\impGSI.\co} w \relarrow{\impGSI.\rf} \txid_a.\mathit{rs}_{\x} \relarrow{\impGSI.\po} \txid_b.\mathit{ru}_{x}  \relarrow{\impGSI.\hb} \txid_a.\mathit{pl}_{x}  \relarrow{\impGSI.\po} a$.  That is, we have $a \relarrow{\impGSI.\co} w\relarrow{\impGSI.\hbloc} a$, contradicting the assumption that $\impGSI$ is consistent. 
%In case (iii) the desired result holds trivially.\\

\noindent \textbf{RTS. (\ref{lem:si_lock_hb_fr})}\\
Assume $(a, b) \in \absGSI.\fr$. From the definition of $\absGSI.\fr$ we then know $(\txid_a.\mathit{rs}_{\x}, b), (\txid_a.\mathit{vs}_{\x}, b) \in \impGSI.\fr$.
On the other hand, from \cref{lem:lock-ordering} we  know that either i) $\txid_b.\mathit{wu}_{x} \relarrow {\impGSI.\hb} \txid_a.\mathit{rl}_{x}$; or ii)  $\x \not\in \writeset_{\txid_a}$ and $\txid_a.\mathit{ru}_{x} \relarrow {\impGSI.\hb} \txid_a.\mathit{pl}_{x}$; or iii) $\x \in \writeset_{\txid_a}$ and  $\txid_a.\mathit{wu}_{x} \relarrow {\impGSI.\hb} \txid_b.\mathit{rl}_{x}$.
In case (i) we then have $b \relarrow{\impGSI.\po} \txid_b.\mathit{wu}_{x}  \relarrow{\impGSI.\hb} \txid_a.\mathit{rl}_{x}  \relarrow{\impGSI.\po} \txid_a.\mathit{rs}_{\x} \relarrow{\impGSI.\fr} b$. That is, we have $b \relarrow{\impGSI.\hbloc} \txid_a.\mathit{rs}_{\x} \relarrow{\impGSI.\fr} b$, contradicting the assumption that $\impGSI$ is RA-consistent. 
In cases (ii-iii) the desired result holds trivially.\\
\end{proof}
\end{lemma}

\begin{lemma}\label{lem:si_soundness}
For all RA-consistent execution graphs $\impGSI$ of the implementation and their counterpart SI execution graphs $\absGSI$ constructed as above, 
\begin{enumerate}
	\item $(\absGSI.\pot \subseteq \impGSI.\po) \land (\absGSI.\pot; \absGSI.\frt \subseteq \impGSI.\hb)$
	\label{lem:si_soundness_po}
	
	\item $(\absGSI.\cot \subseteq \impGSI.\hb)
	\land (\absGSI.\cot; \absGSI.\frt \subseteq \impGSI.\hb)$
	\label{lem:si_soundness_co}
	
	\item $(\absGSI.\rft  \subseteq \impGSI.\hb) \land (\absGSI.\rft; \absGSI.\frt \subseteq \impGSI.\hb)$
	\label{lem:si_soundness_rf}	

\end{enumerate}
\begin{proof}
Pick an arbitrary RA-consistent execution graph $\impGSI$ of the implementation and its counterpart SI execution graph $\absGSI$ constructed as above.\\

\noindent \textbf{RTS. (Part \ref{lem:si_soundness_po})}\\
The proof of the first conjunct is immediate from the definitions of $\absGSI.\po$ and $\impGSI.\po$.
For the second conjunct, pick arbitrary $(a, b) \in  \absGSI.(\pot; \frt)$. We then know there exist $c$ such that $(a, c) \in \absGSI.\pot$ and $(c, b) \in \absGSI.\frt$.
Since $(a, c) \in \absGSI.\pot$, from the definition of $\absGSI.\po$ we also have $(a, c) \in \impGSI.\po$ and thus $(a, c) \in \impGSI.\hb$. 
Moreover, from the definition of $\absGSI.\frt$ we know there exist $\txid_1, \txid_2, r, w$ such that $\txid_1 \ne \txid_2$, $c, r \in \trace'_{\txid_1}$, $b, w \in \trace'_{\txid_2}$ and $(r, w) \in \absGSI.\fr$. 
Let $\loc r = \loc w = \code x$. 
We then know that $\code x \in \readset_{\txid_1} \cup \writeset_{\txid_1}$, $\code x \in \writeset_{\txid_2}$,
and that there exists $w_x$ such that $(w_x, r) \in \absGSI.\rf$ and $(w_x, w) \in \absGSI.\co$.
From the construction of $\absGSI.\co$ we then have $(w_x, w) \in \impGSI.\co$.
From the construction of $\absGSI.\rf$ there are now two cases to consider:
1) $(w_x, \txid_1.\mathit{rs}_{\code x}) \in \impGSI.\rf$; or 2) $\code x \in \writeset_{\txid_1}$, $w_x \in \txid_1$ and $(w_x, r) \in \impGSI.\po$.

In case (1) we then have $(\txid_1.\mathit{rs}_{\code x}, w) \in  \impGSI.\fr$. 
Moreover, from \cref{lem:si_lock_hb} we have that either: 
i) $\x \not\in \writeset_{\txid_1} \land \txid_1.\mathit{ru}_{\code x} \relarrow{\impGSI.\hb} \txid_2.\mathit{pl}_{\code x}$; or ii) $\x \in \writeset_{\txid_1} \land \txid_1.\mathit{wu}_{\code x} \relarrow{\impGSI.\hb} \txid_2.\mathit{rl}_{\code x}$.

%
%In case (1.i) we then have $w \relarrow{\impGSI.\po} \txid_2.\mathit{wu}_{\code x} \relarrow{\impGSI.\hb} \txid_1.\mathit{rl}_{\code x} \relarrow{\impGSI.\po} \txid_1.\mathit{rs}_{\code x} \relarrow{\impGSI.\fr} w$. That is, as $\impGSI.\po \subseteq \impGSI.\hb$ and $\impGSI$ is transitively closed, we have $w \relarrow{\impGSI.\hbloc} \txid_1.\mathit{rs}_{\code x} \relarrow{\impGSI.\fr} w$, contradicting the assumption that $\impGSI$ is consistent. 

In case (1.i), since we have $(a, c) \in \absGSI.\pot$, from the construction of $\absGSI.\po$ we have $(a, \txid_1.\mathit{ru}_{\code x}) \in \impGSI.\po$. We thus have $a \relarrow{\impGSI.\po} \txid_1.\mathit{ru}_{\code x} \relarrow{\impGSI.\hb} \txid_2.\mathit{pl}_{\code x} \relarrow{\impGSI.\po} b$. That is, as $\impGSI.\po \subseteq \impGSI.\hb$ and $\impGSI$ is transitively closed, we have $(a, b) \in \impGSI.\hb$, as required.
Similarly, in case (1.ii), since we have $(a, c) \in \absGSI.\pot$, from the construction of $\absGSI.\po$ we have $(a, \txid_1.\mathit{wu}_{\code x}) \in \impGSI.\po$. We thus have $a \relarrow{\impGSI.\po} \txid_1.\mathit{wu}_{\code x} \relarrow{\impGSI.\hb} \txid_2.\mathit{rl}_{\code x} \relarrow{\impGSI.\po} b$. That is, as $\impGSI.\po \subseteq \impGSI.\hb$ and $\impGSI$ is transitively closed, we have $(a, b) \in \impGSI.\hb$, as required.

In case (2), from \cref{lem:si_lock_hb} we have $\txid_1.\mathit{wu}_{\code x} \relarrow{\impGSI.\hb} \txid_2.\mathit{rl}_{\code x}$. 
%from (\ref{ax:si_lock_wwaxiom}) and the definition of $\hb$ we have either: i) $\txid_2.\mathit{wu}_{\code x} \relarrow{\impGSI.\hb} \txid_1.\mathit{rl}_{\code x}$; or ii) $\txid_1.\mathit{wu}_{\code x} \relarrow{\impGSI.\hb} \txid_2.\mathit{rl}_{\code x}$. 
%
%In case (2.i) we then have $w \relarrow{\impGSI.\po} \txid_2.\mathit{wu}_{\code x} \relarrow{\impGSI.\hb} \txid_1.\mathit{rl}_{\code x} \relarrow{\impGSI.\po} w_x \relarrow{\impGSI.\co} w$. That is, as $\impGSI.\po \subseteq \impGSI.\hb$ and $\impGSI$ is transitively closed, we have $w \relarrow{\impGSI.\hbloc} w_x \relarrow{\impGSI.\co} w$, contradicting the assumption that $\impGSI$ is consistent. 
%
%In case (2.ii), 
Moreover, 
we have $(c, \txid_1.\mathit{wu}_{\code x}) \in \impGSI.\po$ and. $(\txid_2.\mathit{rl}_{\code x}, b) \in \impGSI.\po$.
We thus have, $a \relarrow{\impGSI.\hb} c  \relarrow{\impGSI.\po} \txid_1.\mathit{wu}_{\code x} \relarrow{\impGSI.\hb}\txid_2.\mathit{rl}_{\code x} \relarrow{\impGSI.\po} b$. That is, as $\impGSI.\po \subseteq \impGSI.\hb$ and $\impGSI$ is transitively closed, we have $(a, b) \in \impGSI.\hb$, as required.\\

\noindent \textbf{RTS. (Part \ref{lem:si_soundness_co})}\\
For the first conjunct, pick an arbitrary $(a, b) \in \absGSI.\cot$; we are then required to show that $(a, b) \in \impGSI.\hb$.

From the definition of  $\absGSI.\cot$ and the construction of $\absGSI$ we know there exist $\txid_1, \txid_2, d, e$ such that $\txid_1 \ne \txid_2$, $(d, e) \in \absGSI.\co$, $a, d \in \Transactions_{\txid_1}$ and $b, e \in  \Transactions_{\txid_2}$.  Let $\loc{d} = \loc{e} = \code y$. 
We then know 
$a \relarrow{\impGSI.\po} \txid_1.\mathit{wu}_{\code y}$, 
$\txid_1.\mathit{rl}_{\code y} \relarrow{\impGSI.\po} d \relarrow{\impGSI.\po} \txid_1.\mathit{wu}_{\code y}$,
$\txid_2.\mathit{rl}_{\code y} \relarrow{\impGSI.\po} b$
and $\txid_2.\mathit{rl}_{\code y} \relarrow{\impGSI.\po} e \relarrow{\impGSI.\po} \txid_2.\mathit{wu}_{\code y}$.

%From (\ref{ax:si_lock_wwaxiom}) and the definition of $\hb$ we then know that either $\txid_1.\mathit{wu}_{\code y} \relarrow{\impGSI.\hb} \txid_2.\mathit{rl}_{\code y}$, or $\txid_2.\mathit{wu}_{\code y} \relarrow{\impGSI.\hb} \txid_1.\mathit{rl}_{\code y}$.
%Let us assume that the latter holds. We then have $e \relarrow{\impGSI.\po} \txid_2.\mathit{wu}_{\code y} \relarrow{\impGSI.\hb} \txid_1.\mathit{rl}_{\code y} \relarrow{\impGSI.\po} d \relarrow{\impGSI.\co} e$. That is, since $\impGSI.\po \in \impGSI.\hb $ and $\impGSI.\hb$ is transitively closed, we have $e \relarrow{\impGSI.\hb} d \relarrow{\impGSI.\co} e$, contradicting the assumption that $\impGSI$ is consistent. 
%We thus know that $\txid_1.\mathit{wu}_{\code y} \relarrow{\impGSI.\hb} \txid_2.\mathit{rl}_{\code y}$.
From \cref{lem:si_lock_hb} we then know $\txid_1.\mathit{wu}_{\code y} \relarrow{\impGSI.\hb} \txid_2.\mathit{rl}_{\code y}$.
 As such, we have $a \relarrow{\impGSI.\po}  \txid_1.\mathit{wu}_{\code y} \relarrow{\impGSI.\hb} \txid_2.\mathit{rl}_{\code y} \relarrow{\impGSI.\po}  b$. As $\impGSI.\po \in \impGSI.\hb $ and $\impGSI.\hb$ is transitively closed, we have $a \relarrow{\impGSI.\hb} b$, as required.

For the second conjunct,  pick an arbitrary $c$ such that $(b, c) \in \absGSI.\frt$. We are then required to show that $(a, c) \in \impGSI.\hb$.
From the definition of $\absGSI.\frt$ we then know there exist $\txid_3, r, w$ such that $\txid_3 \ne \txid_2$, $r \in  \Transactions_{\txid_2}$, $c, w \in  \Transactions_{\txid_3}$ and $(r, w) \in \absGSI.\fr$. 
Let $\loc r = \loc w = \code x$. 
We then know that $\code x \in \readset_{\txid_2} \cup \writeset_{\txid_2}$, $\code x \in \writeset_{\txid_3}$,
and that there exists $w_x$ such that $(w_x, r) \in \absGSI.\rf$ and $(w_x, w) \in \absGSI.\co$.
From the construction of $\absGSI.\co$ we then have $(w_x, w) \in \impGSI.\co$.
From the construction of $\absGSI.\rf$ there are now two cases to consider:
1) $(w_x, \txid_2.\mathit{rs}_{\code x}) \in \impGSI.\rf$; or 2) $\code x \in \writeset_{\txid_2}$, $w_x \in \txid_2$ and $(w_x, r) \in \impGSI.\po$.

In case (1) we then have $(\txid_2.\mathit{rs}_{\code x}, w) \in  \impGSI.\fr$. 
Moreover, from \cref{lem:si_lock_hb} we have either i) $\x \not\in \writeset_{\txid_2} \land \txid_2.\mathit{ru}_{\code x} \relarrow{\impGSI.\hb} \txid_3.\mathit{pl}_{\code x}$; or ii) $\x \in \writeset_{\txid_2} \land \txid_2.\mathit{wu}_{\code x} \relarrow{\impGSI.\hb} \txid_3.\mathit{rl}_{\code x}$.
%Moreover, from (\ref{ax:si_lock_wwaxiom}), (\ref{ax:si_lock_rwaxiom}) and the definition of $\hb$ we have that either: i) $\txid_3.\mathit{wu}_{\code x} \relarrow{\impGSI.\hb} \txid_2.\mathit{rl}_{\code x}$; or ii) $\txid_2.\mathit{ru}_{\code x} \relarrow{\impGSI.\hb} \txid_3.\mathit{pl}_{\code x}$; or iii) $\txid_2.\mathit{wu}_{\code x} \relarrow{\impGSI.\hb} \txid_3.\mathit{rl}_{\code x}$.
%
%%
%In case (1.i) we then have $w \relarrow{\impGSI.\po} \txid_3.\mathit{wu}_{\code x} \relarrow{\impGSI.\hb} \txid_2.\mathit{rl}_{\code x} \relarrow{\impGSI.\po} \txid_2.\mathit{rs}_{\code x} \relarrow{\impGSI.\fr} w$. That is, as $\impGSI.\po \subseteq \impGSI.\hb$ and $\impGSI$ is transitively closed, we have $w \relarrow{\impGSI.\hbloc} \txid_2.\mathit{rs}_{\code x} \relarrow{\impGSI.\fr} w$, contradicting the assumption that $\impGSI$ is consistent. 
%
%
In case (1.i), from the proof of the first conjunct recall that we have $a \relarrow{\impGSI.\po} \txid_1.\mathit{wu}_{\code y} \relarrow{\impGSI.\hb} \txid_2.\mathit{rl}_{\code y}$. Also, from the shape of $\impGSI$ traces we know that $\txid_2.\mathit{rl}_{\code y} \relarrow{\impGSI.\po} \txid_2.\mathit{ru}_{\code x}$.  As such, we have 
$a \relarrow{\impGSI.\po}  \txid_1.\mathit{wu}_{\code y} \relarrow{\impGSI.\hb}
\txid_2.\mathit{rl}_{\code y} \relarrow{\impGSI.\po} \txid_2.\mathit{ru}_{\code x} \relarrow{\impGSI.\hb} \txid_3.\mathit{pl}_{\code x} \relarrow{\impGSI.\po} c$. That is, as $\impGSI.\po \subseteq \impGSI.\hb$ and $\impGSI$ is transitively closed, we have $(a, c) \in \impGSI.\hb$, as required.

Similarly, in case (1.ii), from the proof of the first conjunct recall that we have $a \relarrow{\impGSI.\po}  \txid_1.\mathit{wu}_{\code y} \relarrow{\impGSI.\hb} \txid_2.\mathit{rl}_{\code y}$. Also, from the shape of $\impGSI$ traces we know that $\txid_2.\mathit{rl}_{\code y} \relarrow{\impGSI.\po} \txid_2.\mathit{wu}_{\code x}$.  As such, we have 
$a \relarrow{\impGSI.\po}  \txid_1.\mathit{wu}_{\code y} \relarrow{\impGSI.\hb}
\txid_2.\mathit{rl}_{\code y} \relarrow{\impGSI.\po} \txid_2.\mathit{wu}_{\code x} \relarrow{\impGSI.\hb} \txid_3.\mathit{rl}_{\code x} \relarrow{\impGSI.\po} c$. That is, as $\impGSI.\po \subseteq \impGSI.\hb$ and $\impGSI$ is transitively closed, we have $(a, c) \in \impGSI.\hb$, as required.

%In case (2), from (\ref{ax:si_lock_wwaxiom}) and the definition of $\hb$ we have either: i) $\txid_3.\mathit{wu}_{\code x} \relarrow{\impGSI.\hb} \txid_2.\mathit{rl}_{\code x}$; or ii) $\txid_2.\mathit{wu}_{\code x} \relarrow{\impGSI.\hb} \txid_3.\mathit{rl}_{\code x}$. 
%%
%In case (2.i) we then have $w \relarrow{\impGSI.\po} \txid_3.\mathit{wu}_{\code x} \relarrow{\impGSI.\hb} \txid_2.\mathit{rl}_{\code x} \relarrow{\impGSI.\po} w_x\relarrow{\impGSI.\co} w$. That is, as $\impGSI.\po \subseteq \impGSI.\hb$ and $\impGSI$ is transitively closed, we have $w \relarrow{\impGSI.\hbloc} w_x \relarrow{\impGSI.\co} w$, contradicting the assumption that $\impGSI$ is consistent. 
%
%In case (2.ii), 
In case (2) from \cref{lem:si_lock_hb}  we have $\txid_2.\mathit{wu}_{\code x} \relarrow{\impGSI.\hb} \txid_3.\mathit{rl}_{\code x}$. Recall that 
with the first conjunct we demonstrated that $(a, b) \in \impGSI.\hb$. 
Moreover, we have $(b, \txid_2.\mathit{wu}_{\code x}) \in \impGSI.\po$ and $(\txid_3.\mathit{rl}_{\code x}, c) \in \impGSI.\po$. 
We thus have, $a \relarrow{\impGSI.\hb} b  \relarrow{\impGSI.\po} \txid_2.\mathit{wu}_{\code x} \relarrow{\impGSI.\hb}\txid_3.\mathit{rl}_{\code x} \relarrow{\impGSI.\po} c$. That is, as $\impGSI.\po \subseteq \impGSI.\hb$ and $\impGSI$ is transitively closed, we have $(a, b) \in \impGSI.\hb$, as required.\\

\noindent \textbf{RTS. (Part \ref{lem:si_soundness_rf})}\\
For the first conjunct, pick an arbitrary $(a, b) \in \absGSI.\rft$; we are then required to show that $(a, b) \in \impGSI.\hb$.

From the definition of  $\absGSI.\rft$ and the construction of $\absGSI$ we know there exist $\txid_1, \txid_2, w_y, r_y$ such that $\txid_1 \ne \txid_2$, $(w_y, r_y) \in \absGSI.\rf$, $a, w_y \in \txid_1$ and $b, r_y \in \txid_2$.  Let $\loc{w} = \loc{r} = \code y$. 
We then know $\txid_1.\mathit{rl}_{\code y} \relarrow{\impGSI.\po} \txid_1.\mathit{pl}_{\code y} \relarrow{\impGSI.\po} w_y \relarrow{\impGSI.\po} \txid_1.\mathit{wu}_{\code y}$ and $a \relarrow{\impGSI.\po} \txid_1.\mathit{wu}_{\code y}$.

Let $w_y = \writeE{\rel}{\code y}{v}$ and $r_y = \readE{\acq}{\code y}{v}$. From the construction of $\absGSI$ we know that $(w_y, \txid_2.\mathit{rs}_{\code y}) \in \impGSI.\rf$.
On the other hand, from \cref{lem:si_lock_hb} we have $\txid_1.\mathit{wu}_{\code y} \relarrow{\impGSI.\hb} \txid_2.\mathit{rl}_{\code y}$. 
We then have $a \relarrow{\impGSI.\po} \txid_1.\mathit{wu}_{\code y}   \relarrow{\impGSI.\hb} \txid_2.\mathit{rl}_{\code y}  \relarrow{\impGSI.\po}  b$. That is, as $ \impGSI.\po \subseteq  \impGSI.\hb$, we have $a \relarrow{\impGSI.\hb} b$, as required.

%In case (1), we know $\txid_2.\mathit{rl}_{\code y} \relarrow{\impGSI.\reftransC \po} b$.
%On the other hand, from \cref{lem:si_lock_hb} we have $\txid_1.\mathit{wu}_{\code y} \relarrow{\impGSI.\hb} \txid_2.\mathit{rl}_{\code y}$. 
%We then have $a \relarrow{\impGSI.\reftransC \po} \txid_1.\mathit{wu}_{\code y}   \relarrow{\impGSI.\hb} \txid_2.\mathit{rl}_{\code y}  \relarrow{\impGSI.\reftransC \po}  b$. That is, as $ \impGSI.\po \subseteq  \impGSI.\hb$, we have $a \relarrow{\impGSI.\hb} b$, as required.
%
%
%In the latter case (2) 
%we then know $\txid_2.\mathit{rl}_{\code y} \relarrow{\impGSI.\po} \txid_2.\mathit{rs}_{\code y} \relarrow{\impGSI.\po} \txid_2.\mathit{ru}_{\code y} \relarrow{\impGSI.\reftransC \po} b$. 
%%
%On the other hand, from \cref{lem:si_lock_hb} we have $\txid_1.\mathit{wu}_{\code y} \relarrow{\impGSI.\hb} \txid_2.\mathit{rl}_{\code y}$.  
%We then have $a \relarrow{\impGSI.\reftransC \po} \txid_1.\mathit{wu}_{\code y}   \relarrow{\impGSI.\hb} \txid_2.\mathit{rl}_{\code y}  \relarrow{\impGSI.\po}  b$. That is, as $ \impGSI.\po \subseteq  \impGSI.\hb$ and $\impGSI$ is transitively closed, we have $a \relarrow{\impGSI.\hb} b$, as required.\\

For the second conjunct, pick an arbitrary $c$ such that $(b, c) \in \absGSI.\frt$. We are then required to show that $(a, c) \in \impGSI.\hb$.
From the definition of $\absGSI.\frt$ we then know there exist $\txid_3, r, w$ such that $\txid_3 \ne \txid_2$, $r \in  \Transactions_{\txid_2}$, $c, w \in  \Transactions_{\txid_3}$ and $(r, w) \in \absGSI.\fr$. 
Let $\loc r = \loc w = \code x$. 
We then know that $\code x \in \readset_{\txid_2} \cup \writeset_{\txid_2}$, $\code x \in \writeset_{\txid_3}$,
and that there exists $w_x$ such that $(w_x, r) \in \absGSI.\rf$ and $(w_x, w) \in \absGSI.\co$.
From the construction of $\absGSI.\co$ we then have $(w_x, w) \in \impGSI.\co$.
From the proof of the first conjunct recall that we have $a \relarrow{\impGSI.\po}  \txid_1.\mathit{wu}_{\code y} \relarrow{\impGSI.\hb} \txid_2.\mathit{rl}_{\code y}$.

From the construction of $\absGSI.\rf$ there are now two cases to consider:
1) $(w_x, \txid_2.\mathit{rs}_{\code x}) \in \impGSI.\rf$; or 2) $\code x \in \writeset_{\txid_2}$, $w_x \in \txid_2$ and $(w_x, r) \in \impGSI.\po$.

In case (1) we then have $(\txid_2.\mathit{rs}_{\code x}, w) \in  \impGSI.\fr$. 
Moreover, from \cref{lem:si_lock_hb} we have either: i) $\txid_2.\mathit{ru}_{\code x} \relarrow{\impGSI.\hb} \txid_3.\mathit{pl}_{\code x}$; or ii) $\txid_2.\mathit{wu}_{\code x} \relarrow{\impGSI.\hb} \txid_3.\mathit{rl}_{\code x}$.

In case (1.i), 
%from the proof of the first conjunct recall that we have $a \relarrow{\reftransC{\impGSI.\po}}  \txid_1.\mathit{wu}_{\code y} \relarrow{\impGSI.\hb} \txid_2.\mathit{rl}_{\code y}$.
from the shape of $\impGSI$ traces we know that $\txid_2.\mathit{rl}_{\code y} \relarrow{\impGSI.\po} \txid_2.\mathit{ru}_{\code x}$.  
As such, we have 
$a \relarrow{\impGSI.\po}  \txid_1.\mathit{wu}_{\code y} \relarrow{\impGSI.\hb}
\txid_2.\mathit{rl}_{\code y} \relarrow{\impGSI.\po} \txid_2.\mathit{ru}_{\code x} \relarrow{\impGSI.\hb} \txid_3.\mathit{pl}_{\code x} \relarrow{\impGSI.\po} c$. That is, as $\impGSI.\po \subseteq \impGSI.\hb$ and $\impGSI$ is transitively closed, we have $(a, c) \in \impGSI.\hb$, as required.

Similarly, in case (1.ii), from the shape of $\impGSI$ traces we know that $\txid_2.\mathit{rl}_{\code y} \relarrow{\impGSI.\po} \txid_2.\mathit{wu}_{\code x}$.  As such, we have 
$a \relarrow{\impGSI.\po}  \txid_1.\mathit{wu}_{\code y} \relarrow{\impGSI.\hb}
\txid_2.\mathit{rl}_{\code y} \relarrow{\impGSI.\po} \txid_2.\mathit{wu}_{\code x} \relarrow{\impGSI.\hb} \txid_3.\mathit{rl}_{\code x} \relarrow{\impGSI.\po} c$. That is, as $\impGSI.\po \subseteq \impGSI.\hb$ and $\impGSI$ is transitively closed, we have $(a, c) \in \impGSI.\hb$, as required.

In case (2), from \cref{lem:si_lock_hb} we have $\txid_2.\mathit{wu}_{\code x} \relarrow{\impGSI.\hb} \txid_3.\mathit{rl}_{\code x}$. 
Recall that 
with the first conjunct we demonstrated that $(a, b) \in \impGSI.\hb$. 
Moreover, we have $(b, \txid_2.\mathit{wu}_{\code x}) \in \impGSI.\po$ and $(\txid_3.\mathit{rl}_{\code x}, c) \in \impGSI.\po$. 
We thus have, $a \relarrow{\impGSI.\hb} b  \relarrow{\impGSI.\po} \txid_2.\mathit{wu}_{\code x} \relarrow{\impGSI.\hb}\txid_3.\mathit{rl}_{\code x} \relarrow{\impGSI.\po} c$. That is, as $\impGSI.\po \subseteq \impGSI.\hb$ and $\impGSI$ is transitively closed, we have $(a, b) \in \impGSI.\hb$, as required.

\end{proof}

\end{lemma}

\begin{theorem}[Soundness]
For all execution graphs $\impGSI$ of the implementation and their counterpart SI execution graphs $\absGSI$ constructed as above,
\[
	\consistent{\impGSI} \Rightarrow \sicon
\]
\begin{proof}
Pick an arbitrary execution graph $\impGSI$ of the implementation such that $\consistent{\impGSI}$ holds, and its associated SI execution graph $\absGSI$ constructed as described above. \\

\noindent \textbf{RTS. $\acyc{\absGSI.((\pot \cup \rft \cup \cot); \refC{\frt})} $}\\ 
We proceed by contradiction. Let us assume 
\begin{align}
	\consistent{\impGSI} \land \neg \acyc{\absGSI.((\pot \cup \rft \cup \cot); \refC{\frt})} 
	\label{der:si_soundness_ass}
\end{align}	
From the definition of $\consistent{.}$ we then know that there exists $a$ such that $(a, a) \in \transC{\big((\absGSI.\pot \cup \absGSI.\rft \cup \absGSI.\cot); \absGSI.\frt \big)}$. Consequently, from \cref{lem:si_soundness} we have $(a, a) \in \impGSI.\transC \hb$. That is, since $\impGSI.\hb$ is transitively closed, we have $(a, a) \in \impGSI.\hb$, contradicting our assumption that $\impGSI$ is RA-consistent.\\ 

\noindent \textbf{RTS. $\rfi \cup \coi \cup \fri \subseteq \po$}\\
Follows immediately from the construction of $\absGSI$.

\end{proof}

\end{theorem}

\subsection{Implementation Completeness}
\label{subapp:si_completeness}
In order to establish the completeness of our implementation, it suffices to show that given an SI-consistent execution graph $\absGSI = (\Events, \po, \rf, \co, \silo)$, we can construct a corresponding RA-consistent execution graph $\impGSI$ of the implementation.
%%
%\[
%	\for{\absGSI} \sicon[\absGSI] \Rightarrow \exsts{\impGSI} \consistent{\impGSI}
%\]
%%
Before proceeding with the construction of a corresponding implementation graph, we describe several auxiliary definitions.

Given an abstract transaction class $\Transactions_{\txid} \in \absGSI.\Transactions/\st$, we write $\writeset_{\txid}$ for the set of locations written to by $\Transactions_{\txid}$: $\writeset_{\txid} = \bigcup_{e \in \Transactions_{\txid}\cap \Writes} \loc{e}$.
Similarly, we write $\readset_{\txid}$ for the set of locations read from by $\Transactions_{\txid}$,
\emph{prior to} being written by $\Transactions_{\txid}$. 
For each location \code x read from by $\Transactions_{\txid}$, we additionally record the first read event in $\Transactions_{\txid}$ that retrieved the value of \code x.
That is, 
\[
\readset_{\Transactions_{\txid}} \eqdef
\setcomp{
	(\code x, r)
}{
	r \in \Transactions_i \cap \Reads_{\code x}
	\land \neg\exsts{e \in \Transactions_{\txid} \cap \Events_{\code x}}
	e \relarrow{\po} r
}
\]
Note that the execution trace for each transaction $\Transactions_{\txid} \in \absGSI.\Transactions/\st$ is of the form 
$\trace'_{\txid} = \mathit{t}'_1 \relarrow{\imm \po} \cdots \relarrow{\imm \po} \mathit{t}'_k$ for some $k$, where each $\mathit{t}'_i$ is a read or write event.
As such, we have $\absGSI.\Events = \absGSI.\Transactions = \bigcup_{\Transactions_{\txid} \in \absGSI.\Transactions/\st} \Transactions_{\txid} = \trace'_{\txid}.\Events$.

Let $\readset_{\Transactions_{\txid}} \cup \writeset_{\Transactions_{\txid}} = \{\code x_1 \cdots \code x_n\}$.
%and $\writeset_{\Transactions_{\txid}} = \{\code y_1 \cdots \code y_q\}$. 
We then construct the implementation trace $\trace_{\txid}$ as:
\[
	\trace_{\txid} = 
	\mathit{Rs}
	\relarrow{\imm \po} \mathit{RUs}
	\relarrow{\imm \po} \mathit{PLs}
	\relarrow{\imm \po} \mathit{Ts}
	\relarrow{\imm \po} \mathit{Us}
\]
where
\begin{itemize}
	\item $\mathit{Rs} = \mathit{rl}_{\code x_1} \relarrow{\imm \po} \cdots \relarrow{\imm \po} \mathit{rl}_{\code x_n} \relarrow{\imm \po} \mathit{S}_{\code x_1} \relarrow{\imm \po} \cdots \relarrow{\imm \po} \mathit{S}_{\code x_n}$, where the identifiers of all constituent events of $\mathit{Rs}$ are picked fresh, and
\[
\begin{array}{@{} r @{\hspace{2pt}} l @{}}
	\mathit{rl}_{\code x_j} = & 
	\rlockE{\x_j} 
%	\updateE{\acqrel}{\x_j}{a}{a{+} 2}
%	\quad\text{s.t. } \exsts{w} (w, \mathit{rl}_{\code x_j}) \in \RF[]{} \land \wval{w} = a \\
%			
	\qquad
	\mathit{S}_{\code x_j} = 
	\begin{cases}
		\mathit{rs}_{\code x_j}
		\relarrow{\imm \po}
		\mathit{ws}_{\code x_j}
		& \text{if } \exsts{r} (\code x_j, r) \in \readset_{\txid} \land \rval{r} = v_j \\
		
		\emptyset & \text{otherwise}
	\end{cases} 
\end{array}
\]		
with $\mathit{rs}_{\code x_j} \eqdef \readE{\acq}{\code x_j}{v_j}$
and $\mathit{ws}_{\code x_j} \eqdef \writeE{\rel}{\code{s[x}_j\code ]}{v_j}$.
	\item $\mathit{RUs} = \mathit{ru}_{\code x_1} \relarrow{\imm \po} \cdots \relarrow{\imm \po} \mathit{ru}_{\code x_n}$, where the identifiers of all constituent events of $\mathit{RUs}$ are picked fresh, and for all $j \in \{1 \cdots n\}$:
	\[
	\begin{array}{l}
		\mathit{ru}_{\code x_j} = 
		\begin{cases}
			\begin{array}{@{} l @{}}
				\runlockE{\x_j}
			\end{array}
			& \text{ if } \code x_j \not\in \writeset_{\txid}  \\\\
			\emptyset
			& \text{ otherwise}
		\end{cases}
	\end{array}	
	\]
	\item $\mathit{PLs} = \mathit{pl}_{\code x_1} \relarrow{\imm \po} \cdots \relarrow{\imm \po} \mathit{pl}_{\code x_n}$, where the identifiers of all constituent events of $\mathit{PLs}$ are picked fresh, and for all $j \in \{1 \cdots n\}$:
	\[
	\begin{array}{l}
		\mathit{pl}_{\code x_j} = 
		\begin{cases}
			\plockE{\x_j}
			&  \text{ if } \code x_j \in \writeset_{\txid} \\\\
			\emptyset
			& \text{otherwise}			
		\end{cases}
	\end{array}	
	\]
	\item $\mathit{Ts} = \mathit{t}_1 \relarrow{\imm \po} \cdots \relarrow{\imm \po} \mathit{t}_k$, where for all $m \in \{1 \cdots k\}$:
	\[
	\mathit{t}_m = 
	\begin{cases}
		\readE{-}{\code{s[x}_n \code{]}}{v_n} & \text{if } \mathit{t}'_m {=} \readE{-}{\code x_n}{v_n} \\
		\writeE{\rel}{\code x_n}{v_n} \relarrow{\imm \po} \writeE{-}{\code{s[x}_n \code ]}{v_n}
		& \text{if } \mathit{t}'_m {=} \writeE{\rel}{\code x_n}{v_n} \\
	\end{cases}
	\]
	such that in the first case the identifier of $\mathit{t}_m$ is that of $\mathit{t}'_m$; and in the second case the identifier of the first event in $\mathit{t}_m$ is that of $\mathit{t}'_m$ and the identifier of the second event is picked fresh.
	\item $\mathit{Us} = \mathit{wu}_{\code x_1} \relarrow{\imm \po} \cdots \relarrow{\imm \po} \mathit{wu}_{\code x_n}$, where the identifiers of all constituent events of $\mathit{Us}$ are picked fresh, and 
	\[
		\mathit{wu}_{\code x_j} = 
		\begin{cases}
%			\writeE{\rel}{\x_j}{0}
			\wunlockE{\x_j}
			& \text{ if } \code x_j \in \writeset_{\txid}  \\
			\emptyset
			& \text{ otherwise}
		\end{cases}
	\]
\end{itemize}
We use the $\txid.$ prefix to project the various events of the implementation trace $\trace_\txid$ (e.g.~$\txid.\mathit{rl}_{\code x_j}$). 

Given the transaction classes $\TClasses$ of $\absGSI$, we construct a strict total order $\TCO: \TClasses \times \TClasses$ as an extension of $\absGSI.\transC{(\refC{\frt}; (\pot \cup \cot \cup \rft))}$. That is, construct $\TCO$ as a total order such that:
\[
	\for{e, e'} (e, e') \in \absGSI.\transC{(\refC{\frt}; (\pot \cup \cot \cup \rft))} \Rightarrow (\class e \st, \class{e'} \st) \in \TCO
\]
Recall that since $\absGSI$ is SI-consistent, we know $\acyc{(\pot \cup \cot \cup \rft); \refC{\frt}}$ holds, i.e.\ $\acyc{\refC{\frt}; (\pot \cup \cot \cup \rft)}$.
As such, it is always possible to extend $\transC{(\refC{\frt}; (\pot \cup \cot \cup \rft)}$ to a total order as described above. 

For each location \code x, let $\mathit{WT}_{\code x} \eqdef \setcomp{\Transactions_{\txid} \in \TClasses}{\Writes_{\code x} \cap \Transactions_{\txid} \ne \emptyset}$ denote those transactions that write to \code x. 
We then define $\TMO[\code x] \eqdef \coerce{\TCO}{\mathit{WT}_{\code x}}$ as a strict total order on $\mathit{WT}_{\code x}$.
Given a strict total order $\makerel r$, we write $\itemAt{\makerel r}{i}$ for the $i$\textsuperscript{th} item of $\makerel r$, indexed from 0 (e.g.~$\itemAt{\TMO[\code x]}{i}$). 
For $i \in \Nats$, we then define:
\[
\begin{array}{@{} r @{\hspace{2pt}} l @{}}
	\sphase[\code x]{0} \eqdef 
	&
	\setcomp{
		\Transactions_{\txid}
	}{
		\itemAt{\TMO[\code x]}{0} {=} \Transactions_{\txid}
		\lor
%		(\Transactions_{\txid} \not\in \mathit{WT}_{\code x} \land 
		\exsts{r \in \Transactions_{\txid}} (\mathit{init}_{\code x}, r) \in \absGSI.\rf
	} \\
	\sphase[\code x]{i {+} 1} \eqdef 	
	&
	\setcomp{
		\Transactions_{\txid}
	}{
		\itemAt{\TMO[\code x]}{i {+} 1} {=} \Transactions_{\txid}
		\lor 
%		(\Transactions_{\txid} \not\in \mathit{WT}_{\code x} \land 
		\exsts{r \in \Transactions_{\txid}} \exsts{w \in \itemAt{\TMO[\code x]}{i}} (w, r) \in \absGSI.\rf \setminus \rfi
	} 
\end{array}	
\]
where $\mathit{init}_{\x}$ denotes the write event initialising the value of \x (with zero). 

For all locations $\code x$ and $i \in \Nats$, let $\phase{i} \eqdef \coerce{\TCO}{\sphase i}$. 
Note that for each $i \in \Nats$, the $\sphase{i}$ contains at most one transaction that writes to \code x. We denote this transaction by $\writer{\phase{i}}$. 
For each lock \x we then define:

\[
\begin{array}{@{} r @{\hspace{2pt}} l @{}}
\LO_{\x} \eqdef
& \begin{array}[t]{@{} l @{\hspace{1pt}} l @{}} 
	& \setcomp{
		(\writer{\phase{i}}.\mathit{pl}_{\code x}, \writer{\phase{i}}.\mathit{wu}_{\code x}), \\
		(\writer{\phase{i}}.\mathit{pl}_{\code x}, \itemAt{\phase{k}}{j}.\mathit{rl}_{\code x}), \\
		(\writer{\phase{i}}.\mathit{pl}_{\code x}, \itemAt{\phase{k}}{j}.\mathit{ru}_{\code x}),\\
		(\writer{\phase{i}}.\mathit{pl}_{\code x}, \writer{\phase{k}}.\mathit{pl}_{\code x}),\\
		(\writer{\phase{i}}.\mathit{pl}_{\code x}, \writer{\phase{k}}.\mathit{wu}_{\code x}),\\
		
		(\writer{\phase{i}}.\mathit{wu}_{\code x}, \itemAt{\phase{k}}{j}.\mathit{rl}_{\code x}),\\
		(\writer{\phase{i}}.\mathit{wu}_{\code x}, \itemAt{\phase{k}}{j}.\mathit{ru}_{\code x}),\\
		(\writer{\phase{i}}.\mathit{wu}_{\code x}, \writer{\phase{k}}.\mathit{pl}_{\code x}),\\
		(\writer{\phase{i}}.\mathit{wu}_{\code x}, \writer{\phase{k}}.\mathit{wu}_{\code x})
	} 
	{ 
		i \in \Nats
		\land k > i 
		\land 0 \leq j < \length{\sphase{k}} 
	} \\
	\cup 
	& \setcomp{
		(\itemAt{\phase{i}}{j}.\mathit{rl}_{\code x}, \writer{\phase{k}}.\mathit{pl}_{\code x}), \\
		(\itemAt{\phase{i}}{j}.\mathit{ru}_{\code x}, \writer{\phase{k}}.\mathit{pl}_{\code x}), \\
		(\itemAt{\phase{i}}{j}.\mathit{rl}_{\code x}, \writer{\phase{k}}.\mathit{wu}_{\code x}), \\
		(\itemAt{\phase{i}}{j}.\mathit{ru}_{\code x}, \writer{\phase{k}}.\mathit{wu}_{\code x}),
	} 
	{ 
		i \in \Nats
		\land k \geq i 
		\land 0 \leq j < \length{\sphase{i}} 
%		\land \itemAt{\phase{i}}{j} \ne \writer{\phase i}
	} \\
\end{array}  \\\\
	\RLO_{\x} \eqdef
	& \LO_{\x}
	\cup 
	\setcomp{
		(\itemAt{\phase{i}}{j}.\mathit{rl}_{\code x}, \itemAt{\phase{i}}{j}.\mathit{ru}_{\code x}), \\
		(\itemAt{\phase{i}}{j}.\mathit{rl}_{\code x}, \itemAt{\phase{i'}}{j'}.\mathit{rl}_{\code x}), \\
		(\itemAt{\phase{i}}{j}.\mathit{rl}_{\code x}, \itemAt{\phase{i'}}{j'}.\mathit{ru}_{\code x}), \\
		(\itemAt{\phase{i}}{j}.\mathit{ru}_{\code x}, \itemAt{\phase{i'}}{j'}.\mathit{rl}_{\code x}), \\
		(\itemAt{\phase{i}}{j}.\mathit{ru}_{\code x}, \itemAt{\phase{i'}}{j'}.\mathit{ru}_{\code x}), \\
	} 
	{ 
		i, i' \in \Nats
		\land 0 \leq j < \length{\sphase{i}} 
		\land 0 \leq j' < \length{\sphase{i'}} \\
		\land\, \big((i = i' \land j < j') \lor (i < i')\big)
%		\land \itemAt{\phase{i}}{j} \ne \writer{\phase i}
	} 
\end{array}
\]
\begin{remark}
Let $\lo_1 \eqdef \bigcup_{\x \in \Locs} \LO_\x$ and $\lo_2 \eqdef  \bigcup_{\x \in \Locs} \RLO_\x$. Note that both $\lo_1$ and $\lo_2$ satisfy the conditions stated in \cref{def:si_implementation_consistency}.
%Note that both $\wsync[\silo_{\x}]$ and $\wsync[\silor_\x]$ hold. 
The $\lo_2$ additionally satisfies the `read-read-synchronisation' property in (\ref{ax:lock_rsync}).
In what follows we demonstrate that given an SI-consistent execution graph $\absGSI$, it is always possible to construct an RA-consistent execution graph $\impGSI$ of the implementation with its lock order defined as $\lo_2$. Note that as $\lo_1 \subseteq \lo_2$, it is straightforward to show that replacing $\lo_2$ in such a $\impGSI$ with $\lo_1$, preserves the RA-consistency of $\impGSI$, as defined in \cref{def:si_implementation_consistency}. In other words, as $\lo_1 \subseteq \lo_2$, we have:
\[
\begin{array}{@{} l @{}}
	\acyc{\hbloc \cup \co \cup \fr} \text{ with } \hb \eqdef \transC{(\po \cup \rf \cup \lo_2)} 
	 \Rightarrow \\
	\hspace{100pt}
	\acyc{\hbloc \cup \co \cup \fr} \text{ with } \hb \eqdef \transC{(\po \cup \rf \cup \lo_1)} 
\end{array}	
\]
As such, by establishing the completeness of our implementation with respect to $\lo_2$, we also establish its completeness with respect to $\lo_1$. 
In other words, we demonstrate the completeness of our implementation with respect to both lock implementations presented earlier in \cref{app:lock_implementations}.
\end{remark}
We now demonstrate the completeness of our implementation. 
Given an SI-consistent graph,
we construct an implementation graph $\impGSI$ as follows and demonstrate that it is RA-consistent. 
\begin{itemize}
	\item $\impGSI.\Events = \bigcup\limits_{\Transactions_\txid \in \absGSI.\Transactions/\st} \trace_\txid.\Events$, with the $\tx{e} = 0$, for all $e \in \impGSI.\Events$. \\
	Observe that 
	$\absGSI.\Events \subseteq \impGSI.\Events$.
	\item $\impGSI.\po$ is defined as $\absGSI.\po$ extended by the $\po$ for the additional events of $\impGSI$, given by each $\trace_\txid$ trace defined above. Note that $\impGSI.\po$ does not introduce additional orderings between events of $\absGSI.\Events$. That is, $\for{a, b \in \absGSI.\Events} (a, b) \in \absGSI.\po \Leftrightarrow (a, b) \in \impGSI.\po$.
	\item 
	$\impGSI.\rf = 
	\bigcup_{\code x \in \textsc{Locs}} 
	\setcomp{
		(w, \txid.\mathit{rs}_{\code x})	
	}{
		\exsts{r} (\code x, r) \in \readset_{\txid}
		\land (w, r) \in \absGSI.\rf	
	}$.
	\item $\impGSI.\co = \absGSI.\co$ %\coerce{\absGSI.\co}{\absGSI.\Events_\code x}$.
	% \CO[]$, as defined above.
	
	\item $\impGSI.\silo = \bigcup_{\x \in \Locs} \RLO_\x$, with $\RLO_\x$ as defined above.
\end{itemize}
\paragraph{Notation} 
Given an implementation graph $\impGSI$ as constructed above (with $\impGSI.\Transactions = \emptyset$), and a relation $\makerel r \subseteq \impGSI.\Events \times \impGSI.\Events$, we override the $\tout{\makerel r}$ notation and write $\tout{\makerel r}$ for:
\[
	\setcomp{(a, b) \in \makerel r}{\exsts{\txid_a, \txid_b} a \in \trace_{\txid_a}.\Events \land b \in \trace_{\txid_b}.\Events \land \txid_a \ne \txid_b}
\]
Analogously, we write $\tin{\makerel r}$ for $\setcomp{(a, b) \in \makerel r}{\exsts{\txid} a, b \in \trace_{\txid}.\Events}$.
\begin{lemma}\label{lem:si_completeness_TCO}
For all SI-consistent execution graphs $\absGSI = (\Events, \po, \rf, \co, \silo)$, and for all $\code x \in \Locs$, $\txid_a, \txid_b \in \TXIDs$, and $i \in \Nats$:
\[
	\txid_a \in \sphase{i} \land  \txid_b \in \sphase{i {+} 1} \Rightarrow (\txid_a,  \txid_b) \in \TCO
\]
\begin{proof}
Pick arbitrary SI-consistent execution graph $\absGSI = (\Events, \po, \rf, \co, \Transactions)$, $\code x \in \Locs$, $\txid_a, \txid_b \in \TXIDs$, and $i \in \Nats$. There are now four cases to consider: 1) $\txid_a = \writer{\phase i}$ and $\txid_b = \writer{\phase{i {+} 1}}$; 2) $\txid_a = \writer{\phase i}$ and $\txid_b \ne \writer{\phase{i {+} 1}}$; 3) $\txid_a \ne \writer{\phase i}$ and $\txid_b = \writer{\phase{i {+} 1}}$; and 4) $\txid_a \ne \writer{\phase i}$ and $\txid_b \ne \writer{\phase{i {+} 1}}$.

In case (1) from the definitions of $\sphase i$ and $\writer{\phase i}$ we have $\Transactions_{\txid_a} \times \Transactions_{\txid_b} \subseteq \cot$ and thus from the definition of $\TCO$ we have $(\txid_a, \txid_b) \in \TCO$, as required.
Similarly, in case (2) from the definitions of $\sphase i$ and $\writer{\phase i}$ we have $\Transactions_{\txid_a} \times \Transactions_{\txid_b} \subseteq \in \rft$  and thus from the definition of $\TCO$ we have $(\txid_a, \txid_b) \in \TCO$, as required.

In case (3), from the definition of $\sphase i$ we know there exist $r \in \Transactions_{\txid_a} \cap \Reads_{\code x}$ and $w \in \writer{\phase i} \cap \Writes_{\code x}$ such that $(r, w) \in \fr$ and thus $(a, w) \in \frt$. 
On the other hand from the definitions of $\sphase i$ and $\sphase{i {+}1}$ we know there exists $w' \in \Transactions_{\txid_b} \cap \Writes_{\code x}$ such that $(w, w') \in \co$ and thus $(w, w') \in \cot$. As such we have $(a, w') \in \frt; \cot$ and thus from the definition of $\TCO$ we have $(\txid_a, \txid_b) \in \TCO$, as required.

In case (4), from the definition of $\sphase i$ we know there exist $r \in \Transactions_{\txid_a} \cap \Reads_{\code x}$ and $w \in \writer{\phase i} \cap \Writes_{\code x}$ such that $(r, w) \in \fr$ and thus $(a, w) \in \frt$. 
On the other hand from the definitions of $\sphase i$ and $\sphase{i {+}1}$ we know there exists $r' \in \Transactions_{\txid_b} \cap \Reads_{\code x}$ such that $(w, r') \in \rf$ and thus $(w, r') \in \rft$. As such we have $(a, r') \in \frt; \rft$ and thus from the definition of $\TCO$ we have $(\txid_a, \txid_b) \in \TCO$, as required.
\end{proof}
\end{lemma}
\begin{corollary}\label{cor:si_completeness_TCO}
For all SI-consistent execution graphs $\absGSI = (\Events, \po, \rf, \co, \silo)$, and for all $\code x \in \Locs$, $\txid_a, \txid_b \in \TXIDs$, and $i, j \in \Nats$:
\[
	\txid_a \in \sphase{i} \land  \txid_b \in \sphase{j} \land i < j \Rightarrow (\txid_a,  \txid_b) \in \TCO
\]
\begin{proof}
Follows by induction from \cref{lem:si_completeness_TCO}.
\end{proof}
\end{corollary}
%
%
%Given a relation $\makerel r \in \impGSI.\Events \times \impGSI.\Events$, we write $\makerel{r}_{\textsc I}$ for $\makerel r \cap \setcomp{(a, b)}{\exsts{\txid} a, b \in \trace_\txid.\Events}$ and write $\makerel{r}_{\textsf T}$ for $\makerel r \cap \setcomp{(a, b)}{\exsts{\txid_1, \txid_2} a \in \trace_{\txid_1}.\Events \land b \in \trace_{\txid_2}.\Events \land \txid_1 \ne \txid_2}$.
%

Given an implementation graph $\impGSI = (\Events, \po, \rf, \co, \silo)$, let
\[
	\hb' \eqdef \transC{\left(\po \cup \rf \cup \silo \setminus 
	\setcomp{
		(\txid_1.\mathit{rl}_{\code x}, \txid_2.\mathit{pl}_{\code x}) \in \silo_\x, \\
		(\txid_1.\mathit{ru}_{\code x}, \txid_2.\mathit{pl}_{\code x}) \in \silo_\x, \\
		(\txid_1.\mathit{rl}_{\code x}, \txid_2.\mathit{wu}_{\code x}) \in \silo_\x, \\
		(\txid_1.\mathit{ru}_{\code x}, \txid_2.\mathit{wu}_{\code x}) \in \silo_\x 
	}{
		\exsts{i, j, k} 
		\itemAt{\phase i}{j} {=} \txid_1 
		\land 
		\itemAt{\phase i}{k} {=} \txid_2  
		\land 
		j > k
	}
	\right)}
\]
\begin{lemma}\label{lem:si_completeness_aux2}
For all implementation graphs $\impGSI = (\Events, \po, \rf, \co, \silo)$ constructed as above, 
%and for all $a,b \in \Events$:
%
\[
\begin{array}{@{} l @{\hspace{2pt}} l @{}}
%	\for{i \in \Nats^{+}}
	\for{a, b, \txid_a, \txid_b} 
	& a \in \trace_{\txid_a}.\Events
	\land
	b \in \trace_{\txid_b}.\Events
	\land
	a \relarrow{\hb'} b  \\
	& \qquad 
	\Rightarrow (a, b) \in \impGSI.\poi \lor (\txid_a, \txid_b) \in \TCO
\end{array}
\]
\begin{proof} 
Pick an arbitrary implementation graph $\impGSI = (\Events, \po, \rf, \co, \silo)$ constructed as above.
Since $\hb'$ is a transitive closure, it is straightforward to demonstrate that $\hb' = \bigcup\limits_{i \in \Nats} \hb'_i$, where:\\
$\hb'_0 =\po \cup \rf \cup \silo \setminus 
	\setcomp{
		(\txid_1.\mathit{rl}_{\code x}, \txid_2.\mathit{pl}_{\code x}) \in \silo_\x , \\
		(\txid_1.\mathit{ru}_{\code x}, \txid_2.\mathit{pl}_{\code x}) \in \silo_\x, \\
		(\txid_1.\mathit{rl}_{\code x}, \txid_2.\mathit{wu}_{\code x}) \in \silo_\x, \\
		(\txid_1.\mathit{ru}_{\code x}, \txid_2.\mathit{wu}_{\code x}) \in \silo_\x 
	}{
		\exsts{i, j, k} 
		\itemAt{\phase i}{j} {=} \txid_1 
		\land 
		\itemAt{\phase i}{k} {=} \txid_2  
		\land 
		j > k
	}$ 
and $\hb'_{i{+}1} = \hb'_0 ; \hb'_i$. 
It thus suffices to show:
\[
\begin{array}{@{} l @{\hspace{2pt}} l @{}}
	\for{i \in \Nats}
	\for{a, b, \txid_a, \txid_b} 
	& a \in \trace_{\txid_a}.\Events
	\land
	b \in \trace_{\txid_b}.\Events
	\land
	a \relarrow{\hb'_i} b \\ 
	& \qquad 
	\Rightarrow (a, b) \in \impGSI.\poi \lor (\txid_a, \txid_b) \in \TCO
\end{array}
\]
We proceed by induction over $i$.\\

\noindent \textbf{Base case $i = 0$}\\
Pick arbitrary $a, b, \txid_a, \txid_b$ such that $a \in \trace_{\txid_a}.\Events$ and $b \in \trace_{\txid_b}.\Events$, $a \relarrow{\hb'_0} b$.
There are then four cases to consider: 1) $a \relarrow{\impGSI.\poi \cup \rfi \cup \siloi} b$; or 2) $a \relarrow{\impGSI.\poe} b$; 3) $a \relarrow{\impGSI.\rfe} b$; or 4) $a \relarrow{\impGSI.\siloe} b$. 

In case (1), from the construction of $\impGSI.\rfi$ we have $\impGSI.\rfi \subseteq \impGSI.\poi$; moreover, from the construction of $\impGSI.\siloi$ we have $\impGSI.\siloi \subseteq \impGSI.\poi$. Consequently we have $a \relarrow{\impGSI.\poi} b$, as required.

In case (2) from the construction of $\impGSI.\po$ we then know that $\Transactions_{\txid_a} \times  \Transactions_{\txid_b} \subseteq \absGSI.\po$  and thus from the definition of $\TCO$ we have $(\txid_a, \txid_b) \in \TCO$, as required.

In case (3) from the construction of $\impGSI.\rf$ we then know that $\Transactions_{\txid_a} \times  \Transactions_{\txid_b} \subseteq \absGSI.\rft$  and thus from the definition of $\TCO$ we have $(\txid_a, \txid_b) \in \TCO$, as required.

In case (4) from the definition of $\hb'_0$ we know that there exists $i, j, i', j'$ such that $\txid_a {=} \itemAt{\phase{i}}{j}$, $\txid_b {=} \itemAt{\phase{i'}}{j'}$ and either i) $i < i'$; or ii) $i = i' \land j < j'$.

In case (4.i) from \cref{cor:si_completeness_TCO} we have $(\txid_a, \txid_b) \in \TCO$, as required.
In case (4.ii) from the definitions of $\phase{i}$ we have $(\txid_a, \txid_b) \in \TCO$, as required. \\

\noindent \textbf{Inductive case $i = n {+} 1$}\\
Pick arbitrary $a, b, \txid_a, \txid_b$ such that $a \in \trace_{\txid_a}.\Events$, $b \in \trace_{\txid_b}.\Events$ and $a \relarrow{\hb'_i} b$. 
\begin{align}
\begin{array}{@{} l @{\hspace{2pt}} l @{}}
	\for{j \leq n}
	\for{c, d, \txid_c, \txid_d} 
	& c \in \trace_{\txid_c}.\Events
	\land
	d \in \trace_{\txid_d}.\Events
	\land
	c \relarrow{\hb'_j} d  \\
	& \qquad 
	\Rightarrow (c, d) \in \impGSI.\poi \lor (\txid_c, \txid_d) \in \TCO
\end{array}
\tag{I.H.}
\label{IH:si_completeness_auxiliary}
\end{align}
From the definition of $\hb'_{n {+} 1}$ we then know there exists $e$ such that  $(a, e) \in\hb'_0$ and $(e, b) \in \hb'_n$. Let $e \in \txid_e.\Events$. 
Consequently, from the proof of the base case we then know that $(a, e) \in \impGSI.\poi \lor (\txid_a, \txid_e) \in \TCO$. 
Similarly, from (\ref{IH:si_completeness_auxiliary}) we have $(e, b) \in \impGSI.\poi \lor (\txid_e, \txid_b) \in \TCO$. 
There are now four cases to consider: 
1) $(a, e) \in \impGSI.\poi$ and $(e, b) \in \impGSI.\poi$; 
2) $(a, e) \in \impGSI.\poi$ and $(\txid_e, \txid_b) \in \TCO$; 
3) $(\txid_a, \txid_e) \in \TCO$ and $(e, b) \in \impGSI.\poi$; 
4) $(\txid_a, \txid_e) \in \TCO$ and $(\txid_e, \txid_b) \in \TCO$.

In case (1) from the definition of $\impGSI.\poi$ we have $(a, b) \in \impGSI.\poi$.
In cases (2-4) from the definitions of $\impGSI.\poi$, $\TCO$ we have $(\txid_a, \txid_b) \in \TCO$, as required. 

\end{proof}

\end{lemma}

Given an implementation graph $\impGSI = (\Events, \po, \rf, \co, \silo)$ constructed as above, let $\stage{.}: \impGSI.\Events \rightarrow \{1, 2, 3\}$ denote the \emph{event stage} reflecting whether it is in the snapshot phase (1), reader lock release phase (2), or the update phase (3). That is, 
\[
\stage{e} \eqdef
\begin{cases}
	1 & \exsts{\txid, \code x} e \in \simpleset{\txid.\mathit{rl}_{\code x}, \txid.\mathit{rs}_{\code x}, \txid.\mathit{ws}_{\code x}} \\
	2 & \exsts{\txid, \code x} e = \txid.\mathit{ru}_{\code x} \\
	3 & \text{otherwise}
\end{cases}
\]
\begin{lemma}\label{lem:si_completeness_aux3}
For all implementation graphs $\impGSI = (\Events, \po, \rf, \co, \silo)$ constructed as above, and for all $\code x \in \Locs$, $\txid_a, \txid_1, \txid_2 \in \TXIDs$ and $a \in  \trace_{\txid_a}.\Events$:
\[
\begin{array}{@{} l @{}}
	\left(
	\begin{array}{@{} l @{}}
		(\txid_1.\mathit{rl}_{\code x} \relarrow{\silo} \txid_2.\mathit{pl}_{\code x} 
		\land 
		\txid_2.\mathit{pl}_{\code x} \relarrow{\hb} a) \\
		\lor\,
		(\txid_1.\mathit{ru}_{\code x} \relarrow{\silo} \txid_2.\mathit{pl}_{\code x} 
		\land 
		\txid_2.\mathit{pl}_{\code x} \relarrow{\hb} a) \\
		\lor\,
		(\txid_1.\mathit{rl}_{\code x} \relarrow{\silo} \txid_2.\mathit{wu}_{\code x} 
		\land 
		\txid_2.\mathit{wu}_{\code x} \relarrow{\hb} a) \\
		\lor\,
		(\txid_1.\mathit{ru}_{\code x} \relarrow{\silo} \txid_2.\mathit{wu}_{\code x} 
		\land 
		\txid_2.\mathit{wu}_{\code x} \relarrow{\hb} a) 
	\end{array}	 	
	\right)
	\land
	\exsts{j, k, l} 
		\itemAt{\phase j}{k} {=} \txid_1 
		\land 
		\itemAt{\phase j}{l} {=} \txid_2  
		\land 
		k > l \\
	\hspace{150pt}
	\Rightarrow 
%	\qquad 
	\stage{\txid_1.\mathit{ru}_{\code x}} < \stage{a}
	\lor
	(\txid_1, \txid_a) \in \TCO
\end{array}	
\]
\begin{proof}
Pick an arbitrary implementation graph $\impGSI = (\Events, \po, \rf, \co, \silo)$ constructed as above,
and pick arbitrary $\code x \in \Locs$, $\txid_a, \txid_1, \txid_2 \in \TXIDs$ and $a \in \trace_{\txid_a}.\Events$ such that 
$\big((\txid_1.\mathit{rl}_{\code x} \relarrow{\silo} \txid_2.\mathit{pl}_{\code x} 
	\land 
	\txid_2.\mathit{pl}_{\code x} \relarrow{\hb} a) 
   \lor
   (\txid_1.\mathit{ru}_{\code x} \relarrow{\silo} \txid_2.\mathit{pl}_{\code x} 
	\land 
	\txid_2.\mathit{pl}_{\code x} \relarrow{\hb} a) 
  \lor\,
  (\txid_1.\mathit{rl}_{\code x} \relarrow{\silo} \txid_2.\mathit{wu}_{\code x} 
	\land 
	\txid_2.\mathit{wu}_{\code x} \relarrow{\hb} a) 
  \lor\,
  (\txid_1.\mathit{ru}_{\code x} \relarrow{\silo} \txid_2.\mathit{wu}_{\code x} 
	\land 
	\txid_2.\mathit{wu}_{\code x} \relarrow{\hb} a)\big)$ 
and $\exsts{j, k, l} 
		\itemAt{\phase j}{k} {=} \txid_1 
		\land 
		\itemAt{\phase j}{l} {=} \txid_2  
		\land 
		k > l$.
In case of the first disjunct from the construction of $\silo$ we then also have $\txid_1.\mathit{ru}_{\code x} \relarrow{\silo} \txid_2.\mathit{pl}_{\code x} \land \txid_2.\mathit{pl}_{\code x} \relarrow{\hb} a$.
Similarly, in case of the third and fourth disjuncts from the construction of $\silo$ we also have $\txid_1.\mathit{ru}_{\code x} \relarrow{\silo} \txid_2.\mathit{pl}_{\code x} \land	\txid_2.\mathit{wu}_{\code x} \relarrow{\hb} a$. Moreover, since we have $\txid_2.\mathit{pl}_{\code x} \relarrow{\po} \txid_2.\mathit{wu}_{\code x}$ and $\po \subseteq \hb$, we also have $\txid_1.\mathit{ru}_{\code x} \relarrow{\silo} \txid_2.\mathit{pl}_{\code x} \land	\txid_2.\mathit{pl}_{\code x} \relarrow{\hb} a$.
It thus suffices to show:
\[
\begin{array}{@{} l @{}}
	\txid_1.\mathit{ru}_{\code x} \relarrow{\silo} \txid_2.\mathit{pl}_{\code x} 
	\land 
	\txid_2.\mathit{pl}_{\code x} \relarrow{\hb} a
	\land
	\exsts{j, k, l} 
		\itemAt{\phase j}{k} {=} \txid_1 
		\land 
		\itemAt{\phase j}{l} {=} \txid_2  
		\land 
		k > l \\
	\hspace{150pt}
	\Rightarrow 
%	\qquad 
	\stage{\txid_1.\mathit{ru}_{\code x}} < \stage{a}
	\lor
	(\txid_1, \txid_a) \in \TCO
\end{array}	
\]

Since $\impGSI.\hb$ is a transitive closure, it is straightforward to demonstrate that $\impGSI.\hb = \bigcup\limits_{i \in \Nats} \hb_i$, where $\hb_0 = \impGSI.\po \cup \impGSI.\rf \cup \impGSI.\silo$ and $\hb_{i{+}1} = \hb_0 ; \hb_i$. 
It thus suffices to show:
\[
\begin{array}{@{} l @{}}
	\for{i \in \Nats}
	\for{\code x \in \Locs, \txid_a, \txid_1, \txid_2 \in \TXIDs, a \in  \trace_{\txid_a}.\Events} \\
	\qquad 
	\txid_1.\mathit{ru}_{\code x} \relarrow{\silo} \txid_2.\mathit{pl}_{\code x} 
	\land 
	\txid_2.\mathit{pl}_{\code x} \relarrow{\hb_i} a
	\land
	\exsts{j, k, l} 
		\itemAt{\phase j}{k} {=} \txid_1 
		\land 
		\itemAt{\phase j}{l} {=} \txid_2  
		\land 
		k > l \\
	\hspace{150pt}
	\Rightarrow 
	\stage{\txid_1.\mathit{ru}_{\code x}} < \stage{a}
	\lor
	(\txid_1, \txid_a) \in \TCO
\end{array}
\]
We thus proceed by induction over $i$.\\

\noindent \textbf{Base case $i = 0$}\\
Pick arbitrary $\code x \in \Locs$, $\txid_a, \txid_1, \txid_2 \in \TXIDs$ and $a \in \trace_{\txid_a}.\Events$ such that 
$\txid_1.\mathit{ru}_{\code x} \relarrow{\silo} \txid_2.\mathit{pl}_{\code x}$, $\txid_2.\mathit{pl}_{\code x} \relarrow{\hb_0} a$ 
and $\exsts{j, k, l} 
		\itemAt{\phase j}{k} {=} \txid_1 
		\land 
		\itemAt{\phase j}{l} {=} \txid_2  
		\land 
		k > l$.
From the construction of $\impGSI$ we know that there exists $r \in \absGSI.\Transactions_{\txid_1}$ and $w \in \absGSI.\Transactions_{\txid_2}$ such that $(r, w) \in \absGSI.\fr$. 
Moreover, since $k \ne l$ we know that $(r, w) \in \absGSI.\frt$. 
Since we have $\txid_2.\mathit{pl}_{\code x} \relarrow{\hb_0} a$, there are four cases to consider: 
1) $\txid_2.\mathit{pl}_{\code x} \relarrow{\impGSI.\poi \cup \rfi \cup \siloi} a$; or 
%2) $\txid_2.\mathit{wu}_{\code x} \relarrow{\impGSI.\poi \cup \rfi \cup \siloi} a$; or 
2) $\txid_2.\mathit{pl}_{\code x} \relarrow{\impGSI.\poe} a$; or 
%4) $\txid_2.\mathit{wu}_{\code x} \relarrow{\impGSI.\pot} a$; or 
3) $\txid_2.\mathit{pl}_{\code x} \relarrow{\impGSI.\rfe} a$; or 
%6) $\txid_2.\mathit{wu}_{\code x} \relarrow{\impGSI.\rft} a$; or 
4) $\txid_2.\mathit{pl}_{\code x} \relarrow{\impGSI.\siloe} a$. %; or
%8) $\txid_2.\mathit{wu}_{\code x} \relarrow{\impGSI.\silot} a$. 

In case (1), from the construction of $\impGSI.\rfi$ we have $\impGSI.\rfi \subseteq \impGSI.\poi$; moreover, from the construction of $\impGSI.\siloi$ we have $\impGSI.\siloi \subseteq \impGSI.\poi$. We thus have $\txid_2.\mathit{pl}_{\code x} \relarrow{\impGSI.\poi} a$. On the other hand, from the definition of $\stage{.}$ we know $\stage{a} \geq \stage{\txid_2.\mathit{pl}_{\code x}} > \stage{\txid_1.\mathit{ru}_{\code x}}$, and thus  $\stage{a} > \stage{\txid_1.\mathit{ru}_{\code x}}$, as required. 

In case (2), since $\txid_2.\mathit{pl}_{\code x} \relarrow{\poe} a$, we have $(w, a) \in \absGSI.\pot$. As such, we have $(r, a) \in \absGSI.(\frt; \pot)$. From the definition of $\TCO$ we thus have $(\txid_1, \txid_a) \in \TCO$, as required.

%Case (3) cannot happen as any $\rf$ edge from $\txid_2.\mathit{pl}_{\code x}$ is an internal edge included in $\impGSI.\poi$.\\

Case (3) cannot happen as there are no $\rf$ edges from lock events. 
In case (4), from the construction of $\silo$ we know there exists $m, n$ such that $m>j$, and $\txid_a = \itemAt{\phase{m}}{n}$.
As such, from \cref{cor:si_completeness_TCO} we have $(\txid_1, \txid_a) \in \TCO$, as required.\\

\noindent \textbf{Inductive case $i = n {+} 1$}\\
Pick arbitrary $\code x \in \Locs$ and $\txid_a, \txid_1, \txid_2 \in \TXIDs$ and $a \in \trace_{\txid_a}.\Events$ such that $\txid_1.\mathit{ru}_{\code x} \relarrow{\lo} \txid_2.\mathit{pl}_{\code x}$, $\txid_2.\mathit{pl}_{\code x} \relarrow{\hb_i} a$ and $\exsts{j, k, l} 
		\itemAt{\phase j}{k} {=} \txid_1 
		\land 
		\itemAt{\phase j}{l} {=} \txid_2  
		\land 
		k > l$.
\begin{align}
\hspace*{-10pt}
\begin{array}{@{} l @{}}
	\for{j \leq n}
	\for{\code y \in \Locs, \txid_{a'}, \txid_3, \txid_4 \in \TXIDs, a' \in \trace_{\txid_{a'}}.\Events} \\
	\quad 
	(\txid_3.\mathit{ru}_{\code x} \relarrow{\silo} \txid_4.\mathit{pl}_{\code x} 
	\land 
	\txid_4.\mathit{pl}_{\code x} \relarrow{\hb_j} a') 
	\land
	\exsts{j, k, l} 
		\itemAt{\phase j}{k} {=} \txid_1 
		\land 
		\itemAt{\phase j}{l} {=} \txid_2  
		\land 
		k > l \\
	\hspace{150pt}
	\Rightarrow 
	\stage{\txid_3.\mathit{ru}_{\code y}} < \stage{a'}
	\lor
	(\txid_3,  \txid_{a'}) \in \TCO
\end{array}
\tag{I.H.}
\label{IH:si_completeness_auxiliary2}
\end{align}
There are two cases to consider:\\
 1) $\txid_2.\mathit{pl}_{\code x} \relarrow{\transC{(\poi \cup \rfi \cup \siloi)}} a$; or \\
 2) $\txid_2.\mathit{pl}_{\code x} \relarrow{\reftransC{(\poi \cup \rfi \cup \siloi)}; (\poe \cup \rfe \cup \siloe); \hb_m} a$, where $m \leq n$. 

In case (1), from the construction of $\impGSI.\rfi$ we have $\impGSI.\rfi \subseteq \impGSI.\poi$; moreover, from the construction of $\impGSI.\siloi$ we have $\impGSI.\siloi \subseteq \impGSI.\poi$. 
We thus have $\txid_2.\mathit{pl}_{\code x} \relarrow{\transC{\poi }} a$, i.e. $\txid_2.\mathit{pl}_{\code x} \relarrow{\poi } a$. As such, from the proof of the base case  we have $\stage{a} > \stage{\txid_1.\mathit{ru}_{\code x}}$, as required. 

In case (2)  from the construction of $\impGSI$ we know that there exists $r \in\trace_{ \txid_1}.\Events$ and $w \in \trace_{\txid_2}.\Events$ such that $(r, w) \in \absGSI.\fr$. 
%Moreover, since $\txid_1.\mathit{ru}_{\code x} \relarrow{\rf} \txid_2.\mathit{pl}_{\code x}$, from the construction of $\impGSI$ we know that $\txid_1 \ne \txid_2$ and thus we have $(r, w) \in \absGSI.\frt$. 
Moreover, since $k \ne l$, from the construction of $\impGSI$ we know that $\txid_1 \ne \txid_2$ and thus we have $(r, w) \in \absGSI.\frt$. 
On the other hand, we know that there exist $b, c$ such that $\txid_2.\mathit{pl}_{\code x} \relarrow{\reftransC{(\poi \cup \rfi \cup \siloi)}} b \relarrow{\poe \cup \rfe \cup \siloe} c \relarrow{\hb_m} a$. 
From the construction of $\impGSI.\rfi$ we have $\impGSI.\rfi \subseteq \impGSI.\poi$; moreover, from the construction of $\impGSI.\siloi$ we have $\impGSI.\siloi \subseteq \impGSI.\poi$. 
We thus have $\txid_2.\mathit{pl}_{\code x} \relarrow{\reftransC{\poi }} b$. As such we have $\txid_2.\mathit{pl}_{\code x} \relarrow{\reftransC{\poi}} b \relarrow{\poe \cup \rfe \cup \siloe} c$. Let $c \in \trace_{\txid_c}$.
There are now three cases to consider: a) $b \relarrow{\poe} c$; or b) $b \relarrow{\rfe} c$; or c) $b \relarrow{\siloe} c$. We first demonstrate that in all three cases we have $(\txid_1, \txid_c) \in \TCO$.

In case (2.a), since $\txid_2.\mathit{pl}_{\code x} \relarrow{\poe} c$, we know there exists $c' \in \Transactions_{\txid_c}$ such that $(w, c') \in \absGSI.\pot$. As such, we have $(r, c') \in \absGSI.(\frt; \pot)$. From the definition of $\TCO$ we thus have $(\txid_1, \txid_c) \in \TCO$.
%In case (2.b), there are two additional cases to consider: i) $\loc b = \loc c = \code x$; or ii) $\loc b = \loc c = \x$. 
In case (2.b), from the definition of $\impGSI.\rf$ we know there exists $c' \in \Transactions_{\txid_c}$ such that $(w, c') \in \absGSI.\rft$. 
As such, we have $(r, c) \in \absGSI.(\frt; \rft)$. From the definition of $\TCO$ we thus have $(\txid_1, \txid_c) \in \TCO$.
In case (2.c) from the construction of $\lo$ we know there exists $z$ such that either $b = \txid_2.\mathit{pl}_z$ or $b = \txid_2.\mathit{wu}_z$,  and that there exist $p, q$ such that $\txid_2 \in \sphase[\code z]{p}$, $\txid_c \in \sphase[\code z]{q}$ and $p < q$. Given the definition of $\sphase[\code z]{}$ we then know that there exists $w_z \in \absGSI.\Transactions_{\txid_2}$ where either 
i) there exists $w'_z \in \Transactions_{\txid_c}$ such that $(w_z, w'_z) \in \absGSI.\cot$; or  
ii) there exists $r_z \in \Transactions_{\txid_c}$ such that $(w_z, r_z) \in \absGSI. \rft$; or
ii) there exists $r_z \in \Transactions_{\txid_c}$ such that $(w_z, r_z) \in \absGSI.(\cot; \rft)$.  
That is, we have either $(w, w'_z) \in \absGSI.\cot$, or $(w, r_z) \in \absGSI.\rft$, or $(w, r_z) \in \absGSI.(\cot; \rft)$. 
Moreover, since we have $(r, w) \in \absGSI.\frt$, we then have $(r, w'_z) \in \absGSI.(\frt; \cot)$, or $(r, r_z) \in \absGSI.(\frt; \rft)$, or $(r, r_z) \in \absGSI.(\frt; \cot; \rft)$. From the definition of $\TCO$ we thus have $(\txid_1, \txid_c) \in \TCO$.

Since $c \relarrow{\hb_m} a$, there are now two cases to consider: 
i) $c \relarrow{\hb'} a$; or ii) $c \not\relarrow{\hb'} a$.

In case (2.i), 
%we have $c \relarrow{\hb'} a$  and thus 
from \cref{lem:si_completeness_aux2} we have $(c, a) \in \poi  \lor (\txid_c, \txid_a) \in \TCO$. As we have $(\txid_1, \txid_c) \in \TCO$, we thus have $(\txid_1, \txid_a) \in \TCO$, as required.

In case (2.ii) let us split the path from $c$ at the first occurrence of a non-$\hb'$ edge. 
That is, pick $\txid_3, \txid_4, g, h, \code y, p, q, s, k$ such that $c  \relarrow{\hb'} \txid_3.g  \relarrow{\silo} \txid_4.h \relarrow{\hb_k} a$, $k < m$, $\itemAt{\phase[y]{s}}{p} = \Transactions_{\txid_3}$,  $\itemAt{\phase[y]{s}}{q} = \Transactions_{\txid_4}$, $p > q$
and either a) $g = \mathit{rl}_{\code y} \land h = \mathit{pl}_{\code y}$; or b) $g = \mathit{ru}_{\code y} \land h = \mathit{pl}_{\code y}$; or c) $g = \mathit{rl}_{\code y} \land h = \mathit{wu}_{\code y}$; or d) $g = \mathit{ru}_{\code y} \land h = \mathit{wu}_{\code y}$.
From \cref{lem:si_completeness_aux2} we then have $(c, \txid_3.g) \in \poi  \lor (\txid_c, \txid_3) \in \TCO$. As we have $(\txid_1, \txid_c) \in \TCO$, we also have $(\txid_1, \txid_3) \in \TCO$.
%
%We first demonstrate that in all cases (2.ii.a-2.ii.d) there exists $t$ such that $t \leq m$ and $c  \relarrow{\hb'} \txid_3.\mathit{ru}_{\code y}  \relarrow{\silo} \txid_4.\mathit{pl}_{\code y} \relarrow{\hb_t} a$.
We next demonstrate that in all cases (2.ii.a-2.ii.d) there exists $t$ such that $t \leq m$ and $\txid_3.\mathit{ru}_{\code y}  \relarrow{\silo} \txid_4.\mathit{pl}_{\code y} \relarrow{\hb_t} a$.

In case (2.ii.a) from the definition of $\silo$ we also have $\txid_3.\mathit{ru}_{\code y} \relarrow{\silo} \txid_4.\mathit{pl}_{\code x}$.  
%Moreover, since $\txid_3.\mathit{rl}_{\code y} \relarrow{\po} \txid_3.\mathit{ru}_{\code y}$ and $\po \subseteq \hb'$, we also have $c \relarrow{\hb'} \txid_3.\mathit{rl}_{\code y} \relarrow{\hb'} \txid_3.\mathit{ru}_{\code y}$, i.e.~$c \relarrow{\hb'} \txid_3.\mathit{ru}_{\code y}$. 
%As such, we have $c  \relarrow{\hb'} \txid_3.\mathit{ru}_{\code y}  \relarrow{\silo} \txid_4.\mathit{pl}_{\code y} \relarrow{\hb_k} a$.
As such, we have $\txid_3.\mathit{ru}_{\code y}  \relarrow{\silo} \txid_4.\mathit{pl}_{\code y} \relarrow{\hb_k} a$.
In case (2.ii.b) the desired result holds immediately.

In cases (2.ii.c-2.ii.d) from the construction of $\silo$ we have $\txid_3.\mathit{ru}_{\code y} \relarrow{\silo} \txid_4.\mathit{pl}_{\code y}$. Moreover, since we have $\txid_4.\mathit{pl}_{\code y} \relarrow{\po} \txid_4.\mathit{wu}_{\code y}$, we also have $\txid_4.\mathit{pl}_{\code y} \relarrow{\po} \txid_4.\mathit{wu}_{\code y}  \relarrow{\hb_k} a$. As $\po \subseteq \hb$ and $\hb$ is transitively closed, we have $\txid_4.\mathit{pl}_{\code y} \relarrow{\hb_{k{+} 1}} a$. 
As such, we have $\txid_3.\mathit{ru}_{\code y}  \relarrow{\silo} \txid_4.\mathit{pl}_{\code y} \relarrow{\hb_{k {+} 1}} a$. As $k < m$, the desired result holds immediately.

Consequently, from (\ref{IH:si_completeness_auxiliary2}) we have $\stage{\txid_3.\mathit{ru}_{\code y}} < \stage{a} \lor (\txid_3,  \txid_a) \in \TCO$. 
In the case of the first disjunct we have $\stage{\txid_3.\mathit{ru}_{\code y}} {=} \stage{\txid_1.\mathit{ru}_{\code x}} < \stage{a}$, as required.
In the case of the second disjunct, since we also have $(\txid_1, \txid_3) \in \TCO$ and $\TCO$ is transitively closed, we have $(\txid_1, \txid_a) \in \TCO$, as required.

\end{proof}

\end{lemma}

\begin{theorem}[Completeness]
For all SI execution graphs $\absGSI$ and their counterpart implementation graphs $\impGSI$ constructed as above,
\[
	\sicon[\absGSI]  \Rightarrow \consistent{\impGSI}
\]
\begin{proof}
Pick an arbitrary SI execution graph $\absGSI$ and its counterpart implementation graph $\impGSI$ constructed as above and assume $\sicon$ holds.
From the definition of $\consistent{\impGSI}$ it suffices to show: 
\begin{enumerate}
	\item $\irr{\impGSI.\hb}$ \label{goal:si_completeness_hb_irr}
	\item $\irr{\impGSI.\co ; \impGSI.\hb}$ \label{goal:si_completeness_co_hb_irr}
	\item $\irr{\impGSI.\fr ; \impGSI.\hb}$ \label{goal:si_completeness_fr_hb_irr}
\end{enumerate}
\textbf{RTS. part \ref{goal:si_completeness_hb_irr}}\\
We proceed by contradiction. Let us assume that there exists $a, \trace_\txid$ such that $a \in  \trace_\txid.\Events $ and $(a, a) \in \impGSI.\hb$.
There are now two cases to consider: 1) $(a, a) \in \hb'$; or 2) $(a, a) \not\in \hb'$. 

In case (1), from \cref{lem:si_completeness_aux2} we have $(a, a) \in \impGSI.\poi \lor (\Transactions_\txid, \Transactions_\txid) \in \TCO$. The first disjunct leads to a contradiction as the construction of $\impGSI.\po$ yields an acyclic relation. The second disjunct leads to a contradiction as $\TCO$ is a strict total order. 

In case (2), let us split the $a \relarrow{\hb} a$ at the first occurrence of a non-$\hb'$ edge. That is, pick $\txid_1, \txid_2, \code x, i, j, k, g, h$ such that $a \relarrow{\hb} \txid_1.g \relarrow{\lo \setminus \hb'}  \txid_2.h \relarrow{\hb} a$, $\itemAt{\phase[x]{i}}{j} = \Transactions_{\txid_1}$,  $\itemAt{\phase[x]{i}}{k} = \Transactions_{\txid_2}$ and $j > k$.
As we have $\txid_1.g \relarrow{\lo \setminus \hb'} \txid_2.h \relarrow{\hb} a \relarrow{\hb} \txid_1.g$, from \cref{lem:si_completeness_aux3} and the definition of $\hb$ we then have $\stage{\txid_1.g } < \txid_1.g  \lor (\txid_1, \txid_1) \in \TCO$, leading to a contradiction in both disjuncts (the second disjunct yields a contradiction as $\TCO$ is a strict total order). \\
%and  $\neg\exsts{\txid_3, \txid_4, \code y} a  \relarrow{\hb} \txid_3.\mathit{ru}_{\code y}  \relarrow{\rf} \txid_4.\mathit{pl}_{\code y}  \relarrow{\hb} \txid_1.\mathit{ru}_{\code x} $.

\noindent \textbf{RTS. part \ref{goal:si_completeness_co_hb_irr}}\\
We proceed by contradiction. Let us assume that there exists $a, \trace_{\txid_a}, b, \trace_{\txid_a}$ such that $a \in  \trace_{\txid_a}.\Events$, $b \in  \trace_{\txid_b}.\Events$, $(a, b) \in \impGSI.\hb$ and $(b, a) \in \impGSI.\co$.
Let $\loc a = \loc b = \code x$ for some shared location \code x.
There are now two cases to consider: 1) $(b, a) \in \impGSI.\coi$; or 2) $(b, a) \in \impGSI.\coe$.

In case (1) we then have $(b, a) \in \absGSI.\coi \subseteq \absGSI.\poi$. That is, we have $(b, a) \in \impGSI.\po \subseteq \impGSI.\hb$. We thus have $a \relarrow{\impGSI.\hb} b \relarrow{\impGSI.\hb} a$, contradicting our proof in part~\ref{goal:si_completeness_hb_irr}.
In case (2), from the construction of $\impGSI.\co$ we have $(b, a) \in \absGSI.\cot$ and thus 
from the construction of $\impGSI.\lo$ we then have $(\txid_{b}.\mathit{wu}_{\code x}, \txid_{a}.\mathit{rl}_{\code x}) \in \impGSI.\lo$. 
As such we have $a \relarrow{\impGSI.\hb} b \relarrow{\impGSI.\po} \txid_{b}.\mathit{wu}_{\code x} \relarrow{\impGSI.\lo} \txid_{a}.\mathit{rl}_{\code x} \relarrow{\impGSI.\po} a$. That is, we have $a \relarrow{\impGSI.\hb} a$, contradicting our proof in part~\ref{goal:si_completeness_hb_irr}.\\

\noindent \textbf{RTS. part \ref{goal:si_completeness_fr_hb_irr}}\\
We proceed by contradiction. Let us assume that there exists $a, \trace_{\txid_a}, b, \trace_{\txid_a}$ such that $a \in  \trace_{\txid_a}.\Events$, $b \in  \trace_{\txid_b}.\Events$, $(a, b) \in \impGSI.\hb$ and $(b, a) \in \impGSI.\fr$.

Let $\loc a = \loc b = \code x$ for some shared location \code x.
There are now two cases to consider: 1) $(b, a) \in \impGSI.\fri$; or 2) $(b, a) \in \impGSI.\fre$.

In case (1) we then have $(b, a) \in \absGSI.\fri \subseteq \absGSI.\poi$. That is, we have $(b, a) \in \impGSI.\po \subseteq \impGSI.\hb$. We thus have $a \relarrow{\impGSI.\hb} b \relarrow{\impGSI.\hb} a$, contradicting our proof in part~\ref{goal:si_completeness_hb_irr}.
In case (2), from the construction of $\impGSI.\fr$ we have $(b, a) \in \absGSI.\frt$ and thus 
from the construction of $\impGSI.\lo$ we then have either $(b, \txid_{b}.\mathit{wu}_{\code x}) \in \impGSI.\po$ and $(\txid_{b}.\mathit{wu}_{\code x}, \txid_{a}.\mathit{rl}_{\code x}) \in \impGSI.\lo$; or $(b, \txid_{b}.\mathit{ru}_{\code x}) \in \po$ and $(\txid_{b}.\mathit{ru}_{\code x}, \txid_{a}.\mathit{pl}_{\code x}) \in \impGSI.\lo$.
As such in both cases we have $(b, a) \in \impGSI.\hb$. 
Consequently, we have $a \relarrow{\impGSI.\hb} b \relarrow{\impGSI.\hb} a$, contradicting our proof in part~\ref{goal:si_completeness_hb_irr}.

\end{proof}
\end{theorem}

\newcommand{\rseq}[3]{\ensuremath{rd(#1, #2, #3)}}
\newcommand{\wseq}[3]{\ensuremath{wr(#1, #2, #3)}}
\newcommand{\fw}[1]{\ensuremath{\mathit{fw}_{#1}}}
\newcommand{\iw}[1]{\ensuremath{\mathit{iw}_{#1}}}
\newcommand{\ir}[1]{\ensuremath{\mathit{ir}_{#1}}}
\newpage
\section{Soundness and Completeness of the Lazy SI Implementation}
\label{app:si_alternative}
Given an execution graph $\impGSI$ of the lazy SI implementation, let us assign a transaction identifier to each transaction executed by the program; and given a transaction $\txid$. 
Let $\mathit{RS}^0_{\txid} = \mathit{WS}^0_{\txid} = \emptyset$. 
% and let $\mathit{Ts'} = t'_1 \cdots t'_k$ denote its read and write sets, respectively.
Observe that given a transaction $\txid$ of the lazy SI implementation, the trace of $\txid$, written $\trace_{\txid}$, is of the form: 
\[
	\mathit{Fs}
	\relarrow{\imm \po} \mathit{Is}
	\relarrow{\imm \po} \mathit{Ts}
	\relarrow{\imm \po} \mathit{RUs}
	\relarrow{\imm \po} \mathit{PLs}
	\relarrow{\imm \po} \mathit{Ws}
	\relarrow{\imm \po} \mathit{Us}
\]
where:
\begin{itemize}
	\item $\mathit{Fs}$ denotes the sequence of events failing to obtain the necessary locks, i.e.\ those iterations that do not succeed in promoting the writer locks;
	\item $\mathit{Is}$ denotes the sequence of events initialising the values of  \code{LS},\readset and  \writeset with $\emptyset$, and initialising \code{s[x]} with $\bot$ for each location \x;
	\item $\mathit{Ts}$ denotes the sequence of events corresponding to the execution of \denot{\code{T}} and is of the form $\mathit{t}_1 \relarrow{\imm \po} \cdots \relarrow{\imm \po} \mathit{t}_k$, where for all $m \in \{1 \cdots k\}$:
	\[
	\mathit{t}_m = 
	\begin{cases}
		\rseq{\code x_m}{v_m}{\mathit{RS}_{m {-} 1},  \mathit{WS}_{m {-} 1}}
		\relarrow{\imm \po} \mathit{lr}_{\x_m}
		& \text{if } O_m {=} \readE{-}{\code x_m}{v_m} \\
		\wseq{\code x_m}{v_m}{\mathit{RS}_{m {-} 1},  \mathit{WS}_{m {-} 1}} 
		\relarrow{\imm \po} \mathit{wws}_{\x_m}
		\relarrow{\imm \po} \mathit{lw}_{\x_m}
		& \text{if } O_m {=} \writeE{\rel}{\code x_m}{v_m} \\
	\end{cases}
	\]
	where $O_m$ denotes the $m$th event in the trace of the original $\code T$; 
	$\mathit{lr}_{\x_m} \eqdef \readE{}{\code{s[x}_m\code ]}{v_m}$; 
	$\mathit{lw}_{\x_m} \eqdef \writeE{}{\code{s[x}_m\code ]}{v_m}$; 
\[
\small
	\rseq{\code x_m}{v_m}{\mathit{RS}_{m {-} 1},  \mathit{WS}_{m {-} 1}}
	\eqdef
	\begin{cases}
		\readE{}{\code{s[x}_m\code ]}{\bot}
		& \text{if } \code x_m \not\in \mathit{RS}_{m {-} 1} \cup  \mathit{WS}_{m {-} 1} \\
		 \\
		 \relarrow{\imm \po} \mathit{fs_m} \\
		\relarrow{\imm \po} \mathit{rl}_{\x_m} \\
		\relarrow{\imm \po} \mathit{wrs}_{\x_m} \\
		\relarrow{\imm \po} \mathit{rs}_{\x_m} \\
		\relarrow{\imm \po} \mathit{ws}_{\x_m}
		& \\\\
		\emptyset
		& \text{otherwise} 
	\end{cases}
\]	
$\mathit{fs}_m$ denotes the sequence of events attempting (but failing) to acquire the read lock on $\code x_m$, 
$\mathit{rl}_{\code x_m} \eqdef \rlockE{\xl_m}$, 
$\mathit{wrs}_{\x_m} \eqdef \writeE{}{\readset}{\mathit{RS}_m}$, 
$\mathit{rs}_{\code x_m} \eqdef \readE{\acq}{\code x_m}{v_m}$, 
and $\mathit{ws}_{\code x_m} \eqdef \writeE{}{\code{s[x}_m\code ]}{v_m}$; 
and for all $m > 0$:
\[
	\mathit{RS}_{m {+} 1} \eqdef 
	\begin{cases}
		\mathit{RS}_m \cup \{\code x_m\} & \text{if } O_m {=} \readE{-}{\code x_m}{-} \\
		\mathit{RS}_m & \text{otherwise}
	\end{cases}
\]
and
\[
\small
	\wseq{\code x_m}{v_m}{\mathit{RS}_{m {-} 1},  \mathit{WS}_{m {-} 1}}
	\eqdef
	\begin{cases}
		\readE{}{\code{s[x}_m\code ]}{\bot}
		& \text{if } \code x_m \not\in \mathit{RS}_{m {-} 1} \cup  \mathit{WS}_{m {-} 1} \\
		\relarrow{\imm \po} \mathit{fs_m}  \\
		\relarrow{\imm \po} \mathit{rl}_{\x_m}
		\emptyset
		& \text{otherwise} 
	\end{cases}
\]	
$\mathit{lw}_{\x_m} {=} \writeE{}{\code{s[x}_m\code ]}{v_m}$,
$\mathit{wws}_{\x_m} \eqdef \writeE{}{\writeset}{\mathit{WS}_m}$; 
$\mathit{fs}_m$ and $\mathit{rl}_{\code x_m}$ are as defined above; and 
\[
	\mathit{WS}_{m {+} 1} \eqdef 
	\begin{cases}
		\mathit{WS}_m \cup \{\code x_m\} & \text{if } O_m {=} \writeE{-}{\code x_m}{-} \\
		\mathit{WS}_m & \text{otherwise}
	\end{cases}
\]	
Let $\readset_\txid = \mathit{RS}_m$ and  
$\writeset_\txid = \mathit{WS}_m$; 
let $\readset_\txid \cup \writeset_\txid$ be enumerated as $\{\x_1 \cdots \x_i\}$ for some $i$.
	\item $\mathit{RUs}$ denotes the sequence of events releasing the reader locks (when the given location is in the read set only) and is of the form $\mathit{ru}_{\code x_1} \relarrow{\imm \po} \cdots \relarrow{\imm \po} \mathit{ru}_{\code x_i}$, where for all $n \in \{1 \cdots i\}$:
	\[
	\begin{array}{l}
		\mathit{ru}_{\code x_n} = 
		\begin{cases}
			\runlockE{\xl_n}
			& \text{ if } \code x_n \not\in \writeset_{\txid}  \\
			\emptyset
			& \text{ otherwise}
		\end{cases}
	\end{array}	
	\]
	\item $\mathit{PLs}$ denotes the sequence of events promoting the reader locks to writer ones (when the given location is in the write set), and is of the form $\mathit{pl}_{\code x_1} \relarrow{\imm \po} \cdots \relarrow{\imm \po} \mathit{pl}_{\code x_i}$, where for all $n \in \{1 \cdots i\}$:
	\[
	\begin{array}{l}
		\mathit{pl}_{\code x_n} = 
		\begin{cases}
			\plockE{\xl_n}
			& \text{if }  \code x_n \in \writeset_{\txid} \\
			\emptyset
			& \text{ otherwise } 
		\end{cases}
	\end{array}	
	\]
	\item $\mathit{Ws}$ denotes the sequence of events committing the writes of \denot{\code{T}} and is of the form $\mathit{c}_{\x_1} \relarrow{\imm \po} \cdots \relarrow{\imm \po} \mathit{c}_{\x_i}$, where for all $n \in \{1 \cdots i\}$:
\[
	\mathit{c}_{\x_n} = 
	\begin{cases}
		\readE{}{\code{s[x}_n \code{]}}{v_n}  \relarrow{\imm \po} 
		\mathit{w}_{\x_n} {=} \writeE{}{\x_n}{v_n} 
		& \text{if } \x_n \in \writeset_\txid \\
		\emptyset 
		& \text{otherwise}
	\end{cases}
	\]
	\item $\mathit{Us}$ denotes the sequence of events releasing the locks on the write set, and is of the form $\mathit{wu}_{\code x_1} \relarrow{\imm \po} \cdots \relarrow{\imm \po} \mathit{wu}_{\code x_i}$, where for all $n \in \{1 \cdots i\}$:
	\[
		\mathit{wu}_{\code x_n} = 
		\begin{cases}
			\wunlockE{\code{xl}_n} & \text{if } \code x_n \in \writeset_{\txid} \\
			\emptyset & \text{otherwise}
		\end{cases}				
	\]
\end{itemize}

%and $\accset = \simpleset{\code z_1, \cdots, \code z_k}$, 
%where $\accset$ denotes the set of locations both read from and written to by \code T; that is, $\accset \subseteq \readset \land \accset \subseteq \writeset$. 

Given a transaction trace $\trace_{\txid}$, we write e.g.~$\txid.\mathit{Us}$ to refer to its constituent $\mathit{Us}$ sub-trace and write $\mathit{Us}.\Events$ for the set of events related by \po in $\mathit{Us}$. Similarly, we write $\txid.\Events$ for the set of events related by \po in $\trace_{\txid}$.
Note that $\impGSI.\Events = \bigcup\limits_{\txid \in \sort{Tx}}  \txid.\Events $.
%As such, the $\sti$ relation induces a set of equivalence classes on $\absGSI.\Events$, written $\absGSI.\Events/\sti$. As before, we write $\class{a}{\sti}$ for the equivalence class in $\absGSI.\Events/\sti$ that contains $a$.

Note that for each transaction $\txid$ and each location $\x$, the $\txid.\mathit{rl}_\x$, $\txid.\mathit{rs}_\x$, $\txid.\mathit{ru}_\x$, $\txid.\mathit{pl}_\x$, $\txid.\mathit{wu}_\x$ and $\txid.\mathit{w}_\x$ are uniquely identified when they exist. 

For each location $\x \in \writeset_\txid$, let $\fw \x$ denote the maximal write (in $\po$ order within $\txid$) logging a write for \x in \code{s[x]}. That is,  when $\trace_{\txid}=  t_1 \relarrow{\imm \po} \cdots \relarrow{\imm \po} t_m$, let $\fw \x = \func{wmax}{\x, [t_1 \cdots t_m]}$, where
\[
\begin{array}{@{} c @{}}
	 \func{wmax}{\x, [\,]} \text{ undefined} \\
	  \func{wmax}{\x, L {++} [t]}
	  \eqdef
	  \begin{cases}
	  	\mathit{lw}_{\x}
	  	& \text{if } t {=} \wseq{\x}{-}{-, -} \relarrow{\po} \mathit{lw}_{\x}  \relarrow{\po} \mathit{wws}_{\x} \\
	  	 \func{wmax}{\x, L}
	  	 & \text{otherwise}
	  \end{cases}
\end{array}	  
\]
Similarly, for each location $\x \in \writeset_\txid$, let $\iw \x$ denote the minimal write (in $\po$ order within $\txid$) logging a write for \x in \code{s[x]}. That is,  when $\trace_{\txid}=  t_1 \relarrow{\imm \po} \cdots \relarrow{\imm \po} t_m$, let $\iw \x = \func{wmin}{\x, [t_1 \cdots t_m]}$, where
\[
\begin{array}{@{} c @{}}
	 \func{wmin}{\x, [\,]} \text{ undefined} \\
	  \func{wmin}{\x, [t] {++} L}
	  \eqdef
	  \begin{cases}
	  	\mathit{lw}_{\x}
	  	& \text{if } t {=} \wseq{\x}{-}{-, -} \relarrow{\po} \mathit{lw}_{\x}  \relarrow{\po} \mathit{wws}_{\x} \\
	  	 \func{wmin}{\x, L}
	  	 & \text{otherwise}
	  \end{cases}
\end{array}	  
\]
\subsection{Implementation Soundness}
In order to establish the soundness of our implementation, it suffices to show that given an RA-consistent execution graph of the implementation $\impGSI$, we can construct a corresponding SI-consistent execution graph $\absGSI$ with the same outcome.
%
%\[
%	\for{\impGSI} \consistent{\impGSI} \Rightarrow \exsts{\absGSI} \sicon
%\]
%%

Given a transaction $\txid \in \sort{Tx}$ with $\readset_{\txid} \cup \writeset_{\txid}=\{\code x_1 \cdots \code x_i\}$
%, $\writeset_{\txid_t}=\{\code y_1 \cdots \code y_j\}$ 
and trace 
$ \trace_{\txid}$ as above with  $\mathit{Ts}= \mathit{t}_1 \relarrow{\imm \po} \cdots \relarrow{\imm \po} \mathit{t}_k$, 
we construct the corresponding SI execution trace $\trace'_{\txid}$ as follows:
\[
	\trace'_{\txid} \eqdef  \mathit{t}'_1 \relarrow{\imm \po} \cdots \relarrow{\imm \po} \mathit{t}'_k
\]
where  for all $m \in \{1 \cdots k\}$:
\[
\begin{array}{l c l}
	\mathit{t}'_m {=} \readE{}{\code x_m}{v_m}
	& \text{when} & 
	\mathit{t}_m = 
	\rseq{\x_m}{v_m}{S_m}
	\relarrow{\imm \po} \mathit{lr}_{\x_m} \\
	
	\mathit{t}'_m {=} \writeE{}{\code x_m}{v_m}
	& \text{when} & 
	\mathit{t}_m = 
	\wseq{\x_m}{v_m}{S_m}
	\relarrow{\imm \po} \mathit{lw}_{\x_m}
	\relarrow{\imm \po} \mathit{wws}_{\x_m}
\end{array}	
\]
such that in the first case the identifier of $\mathit{t}'_m$ is that of $\mathit{lr}_{\x_m}$; 
and in the second case the identifier of $\mathit{t}'_m$ is that of $\mathit{lw}_{\x_m}$.
We then define:
\[
	\mathsf{RF}_{\txid} \eqdef
	\setcomp{
		(w, t'_j)	
	}{
		t'_j \in \trace'_{\txid}.\Events \land \exsts{\code x, v} t'_j {=} \readE{\acq}{\code x}{v} \land w {=} \writeE{\rel}{\code x}{v} \\
		\land (w \in \txid.\Events \Rightarrow 
		\begin{array}[t]{@{} l @{}}
			w \relarrow{\po} t'_j \,\land\\
			(\for{e \in \txid.\Events } w \relarrow{\po} e \relarrow{\po} t'_j \Rightarrow (\loc e {\ne} \code x \lor e {\not\in} \Writes)))
		\end{array} \\				
		
		\land (w \not \in \txid.\Events \Rightarrow 
		\begin{array}[t]{@{} l @{}}
			(\for{e \in \txid.\Events} (e \relarrow{\po} t'_j \Rightarrow (\loc e \ne \code x \lor e \not\in \Writes)) \\
			\land\, \exsts{\txid'} 
		 	(\txid'.w_\x, \txid.\mathit{rs}_{\x})  \in \impGSI.\rf)
			\land w {=} \txid'.\fw \x
%			\land \wval w {=} \wval{\txid'.w_\x}
		\end{array}		
	}
\]
Similarly, we define: 
\[
	\mathsf{MO} \eqdef
	\transC{
	\setcomp{
		(w_1, w_2), \\
		(w_3, w_4)
	}{
		\exsts{\txid} 
		w_1, w_2 \in \txid.\Events \cap \Writes 
		\land \loc{w_1} {=} \loc{w_2} 
		\land (w_1, w_2) \in \impGSI.\po \\
		\land\, \exsts{\txid_1, \txid_2, \x} 
		w_3 {=} \txid_1.\fw \x
		\land w_4 {=} \txid_1.\iw \x
		\land (\txid_1.w_\x, \txid_2.w_\x) \in \impGSI.\co \\
	}
	}
\]
We are now in a position to demonstrate the soundness of our implementation. Given an RA-consistent execution graph $\impGSI$ of the implementation, we construct an SI execution graph $\absGSI$ as follows and demonstrate that $\sicon$ holds.

\begin{itemize}
	\item $\absGSI.\Events = \bigcup\limits_{\txid \in \sort{Tx}} \trace'_{\txid}.\Events$, with the $\tx{.}$ function defined as:
	\[
		\tx{a} \eqdef \txid \quad \text{ where } \quad a \in \trace'_{\txid}
	\]
	\item $\absGSI.\po = \coerce{\impGSI.\po}{\absGSI.\Events}$
	\item $\absGSI.\rf = \bigcup_{\txid \in \textsc{Tx}} \mathsf{RF}_{\txid}$%\absGSI.\Writes \times \absGSI.\Reads \cap \setcomp{(a, b)}{\src{b} = a}$
	\item $\absGSI.\co = \mathsf{MO}$
	\item $\absGSI.\silo = \emptyset$
%	\item $\absGSI.\Transactions = \absGSI.\Events$
%	with the $\tx{.}$ function defined as:
%	\[
%		\tx{a} \eqdef \txid \quad \text{ where } \quad a \in \trace'_{\txid}
%	\]
\end{itemize}
Observe that the events of each $\trace'_{\txid}$ trace coincides with those of the equivalence class  $\Transactions_{\txid}$ of $\absGSI$. That is,  $\trace'_{\txid}.\Events = \Transactions_{\txid}$. 

\begin{lemma}\label{lem:si_alt_lock_hb}
Given an RA-consistent execution graph $\impGSI$ of the implementation and its corresponding SI execution graph $\absGSI$ constructed as above, for all $a, b, \txid_a, \txid_b, \x$:
\small
\begin{align}
	& \hspace*{-15pt}
	\txid_a \ne \txid_b
	\land a \in \txid_a.\Events 
	\land b \in \txid_b.\Events 
	\land \loc a = \loc b = \x
	\Rightarrow \nonumber \\
	& \hspace*{-15pt}
	\;\; ((a, b) \in \absGSI.\rf \Rightarrow  \txid_a.\mathit{wu}_{\x} \relarrow{\impGSI.\hb} \txid_b.\mathit{rl}_{\x} ) 
	\label{lem:si_alt_lock_hb_rf} \\
	& \hspace*{-15pt}
	\;\; \land ((a, b) \in \absGSI.\co \Rightarrow  \txid_a.\mathit{wu}_{\x} \relarrow{\impGSI.\hb} \txid_b.\mathit{rl}_{\x} ) 
	\label{lem:si_alt_lock_hb_co}\\
%	& \quad \land ((a, b) \in \absGSI.\co;\rf \Rightarrow  \txid_a.\mathit{wu}_{\x} \relarrow{\impGSI.\hb} \txid_b.\mathit{rl}_{\x} ) 
%	\label{lem:si_alt_lock_hb_corf}\\
	& \hspace*{-15pt}
	\;\; \land \big((a, b) \in \absGSI.\fr \Rightarrow  
		(\x \in \writeset_{\txid_a} \land \txid_a.\mathit{wu}_{\x} \relarrow{\impGSI.\hb} \txid_b.\mathit{rl}_{\x} ) 
		\lor 
		(\x \not\in \writeset_{\txid_a} \land \txid_a.\mathit{ru}_{\x} \relarrow{\impGSI.\hb} \txid_b.\mathit{pl}_{\x} ) 
	\big)
	\label{lem:si_alt_lock_hb_fr}
\end{align}	
\normalsize
\begin{proof}
Pick an arbitrary RA-consistent execution graph $\impGSI$ of the implementation and its corresponding SI execution graph $\absGSI$ constructed as above. Pick an arbitrary $a, b, \txid_a, \txid_b, \x$ such that $\txid_a \ne \txid_b$, $a \in \txid_a.\Events$, $a \in \txid_a.\Events$, and $\loc a = \loc b = \x$.\\

\noindent \textbf{RTS. (\ref{lem:si_alt_lock_hb_rf})}\\
Assume $(a, b) \in \absGSI.\rf$. From the definition of $\absGSI.\rf$ we then know $(\txid_a.w_\x, \txid_b.\mathit{rs}_{\x}) \in \impGSI.\rf$.
On the other hand, from \cref{lem:lock-ordering} we  know that either i) $\x \in \writeset_{\txid_b}$ and $\txid_b.\mathit{wu}_{x} \relarrow {\impGSI.\hb} \txid_a.\mathit{rl}_{x}$; or ii)  $\x \not\in \writeset_{\txid_b}$ and  $\txid_b.\mathit{ru}_{x} \relarrow {\impGSI.\hb} \txid_a.\mathit{pl}_{x}$; or iii) $\txid_a.\mathit{wu}_{x} \relarrow {\impGSI.\hb} \txid_b.\mathit{rl}_{x}$.
In case (i) we then have $\txid_a.w_\x \relarrow{\impGSI.\rf} \txid_b.\mathit{rs}_{\x} \relarrow{\impGSI.\po} \txid_b.\mathit{wu}_{x}  \relarrow{\impGSI.\hb} \txid_a.\mathit{rl}_{x}  \relarrow{\impGSI.\po} \txid_a.w_\x$. That is, we have $\txid_a.w_\x \relarrow{\impGSI.\hbloc} \txid_a.w_\x$, contradicting the assumption that $\impGSI$ is RA-consistent. 
Similarly in case (ii) we have $\txid_a.w_\x \relarrow{\impGSI.\rf} \txid_b.\mathit{rs}_{\x} \relarrow{\impGSI.\po} \txid_b.\mathit{ru}_{x}  \relarrow{\impGSI.\hb} \txid_a.\mathit{pl}_{x}  \relarrow{\impGSI.\po} \txid_a.w_\x$.  That is, we have $\txid_a.w_\x \relarrow{\impGSI.\hbloc} \txid_a.w_\x$, contradicting the assumption that $\impGSI$ is RA-consistent. 
In case (iii) the desired result holds trivially.\\

\noindent \textbf{RTS. (\ref{lem:si_alt_lock_hb_co})}\\
Assume $(a, b) \in \absGSI.\co$. From the definition of $\absGSI.\co$ we then know $(\txid_a.w_\x, \txid_b.w_\x) \in \impGSI.\co$.
On the other hand, from \cref{lem:lock-ordering}  we  know that either i) $\txid_b.\mathit{wu}_{x} \relarrow {\impGSI.\hb} \txid_a.\mathit{rl}_{x}$; or ii) $\txid_a.\mathit{wu}_{x} \relarrow {\impGSI.\hb} \txid_b.\mathit{rl}_{x}$.
In case (i) we then have $\txid_a.w_\x \relarrow{\impGSI.\co} \txid_b.w_\x \relarrow{\impGSI.\po} \txid_b.\mathit{wu}_{x}  \relarrow{\impGSI.\hb} \txid_a.\mathit{rl}_{x}  \relarrow{\impGSI.\po} \txid_a.w_\x$. That is, we have $\txid_a.w_\x \relarrow{\impGSI.\co} \txid_b.w_\x \relarrow{\impGSI.\hbloc} \txid_a.w_\x$, contradicting the assumption that $\impGSI$ is RA-consistent. 
In case (ii) the desired result holds trivially.\\

%\noindent \textbf{RTS. (\ref{lem:si_alt_lock_hb_corf})}\\
%Assume $(a, b) \in \absGSI.\co; \rf$. We then know there exists $w$ such that $(a, w) \in \absGSI.\co$ and $(w, b) \in \absGSI.\rf$. From the definition of $\absGSI.\co$ we then know $(a, w) \in \impGSI.\co$.
%There are now three cases to consider: 1) $w \in \txid_a$; or 2) $w \in \txid_b$; or 3) $w \in \txid_c \land \txid_c \ne \txid_a \land \txid_c \ne \txid_b$. 
%In case (1) the desired result follows from part \ref{lem:si_alt_lock_hb_rf}.
%In case (2) since $(a, w) \in \absGSI.\co$ the desired result follows from part \ref{lem:si_alt_lock_hb_co}.
%
%In case (3) from the proof of part \ref{lem:si_alt_lock_hb_co} we have $\txid_a.\mathit{wu}_{x} \relarrow {\impGSI.\hb} \txid_c.\mathit{rl}_{x}$.
%Moreover, from the shape of $\impGSI$ traces we have $\txid_c.\mathit{rl}_{x} \relarrow {\impGSI.\po} \txid_c.\mathit{wu}_{x}$.
%On the other hand, from the proof of part \ref{lem:si_alt_lock_hb_rf} we have $\txid_c.\mathit{wu}_{x} \relarrow {\impGSI.\hb} \txid_b.\mathit{rl}_{x}$.
%We thus have 
%$\txid_a.\mathit{wu}_{x} \relarrow {\impGSI.\hb} \txid_c.\mathit{rl}_{x} \relarrow {\impGSI.\po} \txid_c.\mathit{wu}_{x} \relarrow {\impGSI.\hb} \txid_b.\mathit{rl}_{x}$.
%As $\impGSI.\po \subseteq \impGSI.\hb$ and $\impGSI.\hb$ is transitively closed, we have $\txid_a.\mathit{wu}_{x} \relarrow {\impGSI.\hb} \txid_b.\mathit{rl}_{x}$, as required. \\

\noindent \textbf{RTS. (\ref{lem:si_alt_lock_hb_fr})}\\
Assume $(a, b) \in \absGSI.\fr$. From the definition of $\absGSI.\fr$ we then know that either 
$(\txid_a.w_{\x}, \txid_b.w_\x) \in \impGSI.\co$
or $(\txid_a.\mathit{rs}_{\x}, \txid_b.w_\x) \in \impGSI.\fr$.
%, (\txid_a.\mathit{vs}_{\x}, b)
In the former case the desired result follows immediately from the proof of part \eqref{lem:si_alt_lock_hb_co}.
In the latter case, from \cref{lem:lock-ordering} we  know that either i) $\txid_b.\mathit{wu}_{x} \relarrow {\impGSI.\hb} \txid_a.\mathit{rl}_{x}$; or ii)  $\x \not\in \writeset_{\txid_a}$ and $\txid_a.\mathit{ru}_{x} \relarrow {\impGSI.\hb} \txid_a.\mathit{pl}_{x}$; or iii) $\x \in \writeset_{\txid_a}$ and  $\txid_a.\mathit{wu}_{x} \relarrow {\impGSI.\hb} \txid_b.\mathit{rl}_{x}$.
In case (i) we then have $\txid_b.w_\x \relarrow{\impGSI.\po} \txid_b.\mathit{wu}_{x}  \relarrow{\impGSI.\hb} \txid_a.\mathit{rl}_{x}  \relarrow{\impGSI.\po} \txid_a.\mathit{rs}_{\x} \relarrow{\impGSI.\fr} \txid_b.w_\x$. That is, we have $\txid_b.w_\x \relarrow{\impGSI.\hbloc} \txid_a.\mathit{rs}_{\x} \relarrow{\impGSI.\fr} \txid_b.w_\x$, contradicting the assumption that $\impGSI$ is RA-consistent. 
In cases (ii-iii) the desired result holds trivially.\\
\end{proof}
\end{lemma}

\begin{lemma}
\label{lem:si_alt_soundness}
For all RA-consistent execution graphs $\impGSI$ of the implementation and their counterpart SI execution graphs $\absGSI$ constructed as above:
%, for all $\txid_a, \txid_b$, $a, c \in \absGSI.\Transactions_{\txid_a}, b  \in \absGSI.\Transactions_{\txid_b}$, if $(a, b) \in  \absGSI.\transC{(\pot \cup \rft \cup \cot)}$, then there exists $d \in  \absGSI.\Transactions_{\txid_b}$ such that:
\[
\begin{array}{@{} l @{}}
	\for{\txid_a, \txid_b}
	\for{a \in \absGSI.\Transactions_{\txid_a}, b  \in \absGSI.\Transactions_{\txid_b}} \\
	\quad 
	(a, b) \in  \absGSI.\transC{(\pot \cup \rft \cup \cot)}
	\Rightarrow \\
	\hspace{50pt}
	\exsts{d \in  \absGSI.\Transactions_{\txid_b}} 
	\for{c \in \absGSI.\Transactions_{\txid_a}} (c, d) \in \impGSI.\hb \\
	\hspace{50pt}
	\land (\exsts{\y} (\txid_a.\mathit{wu}_\y, d) \in \impGSI.\hb
	\lor \for{c \in \txid_a.\Events} (c, d) \in \impGSI.\hb)
\end{array}	
\]
\begin{proof}
Let $S_0 = \absGSI.(\pot \cup \rft \cup \cot)$, and $S_{n {+} 1} = S_0; S_n$, for all $n >=0$. 
It is straightforward to demonstrate that $\absGSI.\transC{(\pot \cup \rft \cup \cot)} = \bigcup\limits_{i \in \Nats} S_i$. 
We thus demonstrate instead that:
\[
\begin{array}{@{} r l @{}}
	\for{i \in \Nats} 
	\for{\txid_a, \txid_b}
	\for{a \in \absGSI.\Transactions_{\txid_a}, b  \in \absGSI.\Transactions_{\txid_b}} & \\
	\quad 
	(a, b) \in S_i 
	\Rightarrow
	&
	\exsts{d \in  \absGSI.\Transactions_{\txid_b}} 
	\for{c \in \absGSI.\Transactions_{\txid_a}} (c, d) \in \impGSI.\hb \\ % \land (b {=} d \lor b \in \Reads) \\
	& \land (\exsts{\y} (\txid_a.\mathit{wu}_\y, d) \in \impGSI.\hb
	\lor \for{c \in \txid_a.\Events} (c, d) \in \impGSI.\hb)
\end{array}	
\]
We proceed by induction on $i$.\\

\noindent \textbf{Base case $i = 0$}\\
Pick arbitrary $\txid_a, \txid_b$, $a \in \absGSI.\Transactions_{\txid_a}, b  \in \absGSI.\Transactions_{\txid_b}$ such that $(a, b) \in  S_0$. There are now three cases to consider: 
1) $(a, b) \in \absGSI.\pot$; or 
2) $(a, b) \in \absGSI.\rft$; or 
3) $(a, b) \in \absGSI.\cot$.

In case (1), pick an arbitrary $c \in \absGSI.\Transactions_{\txid_a}$. From the definition of $\absGSI.\pot$ we have $(c, b) \in \impGSI.\po \subseteq \impGSI.\hb$, as required. 
Pick an arbitrary $c \in \txid_a.\Events$. From the definition of $\absGSI.\pot$ we have $(c, b) \in \impGSI.\po \subseteq \impGSI.\hb$, as required. 

In case (2), we then know there exists $w \in \absGSI.\Transactions_{\txid_a}$ and $r \in \absGSI.\Transactions_{\txid_b}$ such that $(w, r) \in \absGSI.\rf$. 
Let $\loc w = \loc r = \x$. 
From \cref{lem:si_alt_lock_hb} we then have $\txid_a.\mathit{wu}_\x  \relarrow{\impGSI.\hb} \txid_b.\mathit{rl}_\x$.
Pick an arbitrary $c \in \absGSI.\Transactions_{\txid_a}$.
As such we have 
$c \relarrow{\impGSI.\po} 
\txid_a.\mathit{wu}_\x  \relarrow{\impGSI.\hb} 
\txid_b.\mathit{rl}_\x  \relarrow{\impGSI.\po} \txid_b.\mathit{rs}_\x$. 
That is, we have $(c, \txid_b.\mathit{rs}_\x), (\txid_a.\mathit{wu}_\x, \txid_b.\mathit{rs}_\x) \in \impGSI.\hb$, as required. % and $\Transactions_{\txid_b}.\mathit{rs}_\x \in \Reads$. 
%There are now two cases to consider: 
%i) $b \in \Reads$; or ii) $b \in \Writes$. 
%In case (2.i) the desired result holds immediately. 
%In case (2.ii) we then have $c \relarrow{\impGSI.\hb} \Transactions_{\txid_b}.\mathit{rs}_\x \relarrow{\impGSI.\po} b$.
%That is, we have  $(c, b) \in \impGSI.\hb$ as required. 

In case (3), we then know there exists $w \in \absGSI.\Transactions_{\txid_a}$ and $w' \in \absGSI.\Transactions_{\txid_b}$ such that $(w, w') \in \absGSI.\co$. 
Let $\loc w = \loc{w'} = \x$. 
From \cref{lem:si_alt_lock_hb} we then have $\txid_a.\mathit{wu}_\x  \relarrow{\impGSI.\hb} \txid_b.\mathit{rl}_\x$.
Pick an arbitrary $c \in \absGSI.\Transactions_{\txid_a}$.
As we also have $\txid_b.\mathit{rl}_\x  \relarrow{\impGSI.\po} w'$, 
and $c \relarrow{\impGSI.\po} \txid_a.\mathit{wu}_\x$, we then have 
$(c,  w'), (\txid_a.\mathit{wu}_\x, w') \in \impGSI.\hb$, as required.\\
%There are now two cases to consider: 
%i) $b \in \Writes$; or ii) $b \in \Reads$. 
%In case (3.i) we then have 
%$c \relarrow{\impGSI.\po} 
%\Transactions_{\txid_a}.\mathit{wu}_\x  \relarrow{\impGSI.\hb} 
%\Transactions_{\txid_b}.\mathit{rl}_\x  \relarrow{\impGSI.\po} b$.
%That is, we have $(c, b) \in \impGSI.\hb$, as required. 
%In case (3.ii) we have 

\noindent \textbf{Inductive case $i {=}  n {+} 1$}
\begin{align}
& \begin{array}{@{} l @{}}
	\for{j \in \Nats} 
	\for{\txid_a, \txid_b}
	\for{a \in \absGSI.\Transactions_{\txid_a}, b  \in \absGSI.\Transactions_{\txid_b}} \\
	\quad 
	(a, b) \in S_j \land j \leq n
	\Rightarrow \\
	\hspace{50pt}
	\exsts{d \in  \absGSI.\Transactions_{\txid_b}} 
	\for{c \in \absGSI.\Transactions_{\txid_a}}
	(c, d) \in \impGSI.\hb \\  % \land (b {=} d \lor b \in \Reads) \\
	\hspace{50pt} \land (\exsts{\y} (\txid_a.\mathit{wu}_\y, d) \in \impGSI.\hb
	\lor \for{c \in \txid_a.\Events} (c, d) \in \impGSI.\hb)
\end{array}	
\tag{I.H.}
\label{IH:si_alt_soundnenss}
\end{align}
Pick arbitrary $\txid_a, \txid_b$, $a \in \absGSI.\Transactions_{\txid_a}, b  \in \absGSI.\Transactions_{\txid_b}$ such that $(a, b) \in  S_i$.
From the definition of $S_i$ we then know there exist $e, \txid_e$ such that $e \in \absGSI.\Transactions_{\txid_e}$, $(a, e) \in S_0$ and $(e, b) \in S_n$. 
Since $(e, b) \in S_n$, from \eqref{IH:si_alt_soundnenss} we know there exists $d \in  \absGSI.\Transactions_{\txid_b}$ such that $\for{f' \in  \absGSI.\Transactions_{\txid_e}} (f', d) \in \impGSI.\hb$.
On the other hand, from the proof of the base case we know there exists $f \in  \absGSI.\Transactions_{\txid_e}$ such that $\for{c' \in  \absGSI.\Transactions_{\txid_a}} (c', f) \in \impGSI.\hb$; 
and that $\exsts{\y} (\txid_e.\mathit{wu}_\y, f) \in \impGSI.\hb \lor \for{c' \in  \txid_a.\Events} (c', f) \in \impGSI.\hb$. 
Pick an arbitrary $c \in \absGSI.\Transactions_{\txid_a}$. We thus know that $(c, f) \in \impGSI.\hb$.
As $f \in \absGSI.\Transactions_{\txid_e}$, we thus have $(f, d) \in \impGSI.\hb$. 
As $\impGSI.\hb$ is transitively closed and we have  $(c, f), (f, d) \in \impGSI.\hb$ and 
$\exsts{\y} (\txid_e.\mathit{wu}_\y, f) \in \impGSI.\hb \lor \for{c' \in  \txid_a.\Events} (c', f) \in \impGSI.\hb$, we then have $(c, d) \in \impGSI.\hb$, and 
$\exsts{\y} (\txid_e.\mathit{wu}_\y, d) \in \impGSI.\hb \lor \for{c' \in  \txid_a.\Events} (c', d) \in \impGSI.\hb$, as required.

\end{proof}
\end{lemma}

\begin{theorem}[Soundness]
For all execution graphs $\impGSI$ of the implementation and their counterpart SI execution graphs $\absGSI$ constructed as above,
\[
	\consistent{\impGSI} \Rightarrow \sicon
\]
\begin{proof}
Pick an arbitrary execution graph $\impGSI$ of the implementation such that $\consistent{\impGSI}$ holds, and its associated SI execution graph $\absGSI$ constructed as described above. \\

\noindent \textbf{RTS. $\irr{\sihb}$}\\ 
We proceed by contradiction. Let us assume $\neg \irr{\absGSI.(\pot \cup \rft \cup \cot \cup \sifr)}$. 
Let $S = \transC{\absGSI.(\pot \cup \rft \cup \cot)}$.
There are ow two cases to consider: either there is a cycle without a $\sifr$ edge; or there is a cycle with one or more $\sifr$ edges. 
That is, either
1) there exists $a$ such that $(a, a) \in  S$; or 
2) there exist $a_1, b_1, \cdots, a_n, b_n$  such that 
$a_1 \relarrow{\absGSI.\sifr} b_1 \relarrow{S} a_2 \relarrow{\absGSI.\sifr} b_2 \relarrow{S} \cdots \relarrow{S} a_n \relarrow{\absGSI.\sifr} b_n \relarrow{S} a_1$. 

In case (1) we then know there exists $\txid$ such that $a \in \absGSI.\Transactions_{\txid_a}$.
As such, from \cref{lem:si_alt_soundness} we know that there exists $d \in \absGSI.\Transactions_{\txid_a}$ such that for all $c \in \absGSI.\Transactions_{\txid_a}$, $(c, d) \in \impGSI.\hb$. As such, we have $(d, d) \in \impGSI.\hb$, contradicting the assumption that $\impGSI$ is RA-consistent. 

In case (2), for an arbitrary  $i \in \{1 \cdots n\}$, let $j = i {+} 1$ when $i \ne n$; and $j = 1$ when $i = n$. 
As $a_i \relarrow{\absGSI.\sifr} b_i$, we know there exists $\txid_{a_i}, \txid_{b_i}$ such that $a_i \in \Transactions_{\txid_{a_i}}$, $b_i \in \Transactions_{\txid_{b_i}}$, and that there exist $r_i \in \Transactions_{\txid_{a_i}} \cap \EReads$ and $w_i \in \Transactions_{\txid_{b_i}} \cap \Writes$ such that $(r_i, w_i) \in \absGSI.\fr$. 
Let $\loc{r_i} = \loc{w_i} = \x_i$. 
From \cref{lem:si_alt_lock_hb} we then know that either 
i) $\txid_{a_i}.\mathit{ru}_{\x_i} \relarrow{\impGSI.\hb} \txid_{b_i}.\mathit{pl}_{\x_i}$; or 
ii) $\txid_{a_i}.\mathit{wu}_{\x_i} \relarrow{\impGSI.\hb} \txid_{b_i}.\mathit{rl}_{\x_i}$.
Note that for all $\y$ such that $\txid_{b_i}.\mathit{wu}_{\y}$ exists, we know $\txid_{b_i}.\mathit{rl}_{\x_i} \relarrow{\impGSI.\po} \txid_{b_i}.\mathit{wu}_{\y}$ and $\txid_{b_i}.\mathit{pl}_{\x_i} \relarrow{\impGSI.\po} \txid_{b_i}.\mathit{wu}_{\y}$.
As such, as $(b_i, a_j) \in S$, from \cref{lem:si_alt_soundness} and since $\impGSI.\hb$ is transitively closed, we know there exists $d_j \in \Transactions_{\txid_{a_j}}$ such that either $\txid_{a_i}.\mathit{ru}_{\x_i} \relarrow{\impGSI.\hb} d_j$, or $\txid_{a_i}.\mathit{wu}_{\x_i} \relarrow{\impGSI.\hb} d_j$.
That is, $\txid_{a_i}.\mathit{u}_{\x_i} \relarrow{\impGSI.\hb} d_j$, where either $\mathit{u}_{\x_i} = \mathit{ru}_{\x_i}$ or $\mathit{u}_{\x_i} = \mathit{wu}_{\x_i}$.
On the other hand, observe that for all $d_i \in \Transactions_{\txid_{a_i}}$ we have 
$d_i \relarrow{\impGSI.\po} \txid_{a_i}.\mathit{u}_{\x_i}$ 
i.e.\ $d_i \relarrow{\impGSI.\hb} \txid_{a_i}.\mathit{u}_{\x_i}$.
As such, we have $d_j \relarrow{\impGSI.\hb} \txid_{a_j}.\mathit{u}_{\x_j}$.
As we also have $\txid_{a_i}.\mathit{u}_{\x_i} \relarrow{\impGSI.\hb} d_j$
and $\impGSI.\hb$ is transitively closed, we have $\txid_{a_i}.\mathit{u}_{\x_i} \relarrow{\impGSI.\hb} \txid_{a_j}.\mathit{u}_{\x_j}$.
We then have 
$\txid_{a_1}.\mathit{u}_{\x_1} \relarrow{\impGSI.\hb} \txid_{a_2}.\mathit{u}_{\x_2} \relarrow{\impGSI.\hb} \cdots \relarrow{\impGSI.\hb} \txid_{a_n}.\mathit{u}_{\x_n} \relarrow{\impGSI.\hb} \txid_{a_1}.\mathit{u}_{\x_1}$. 
That is, $\txid_{a_1}.\mathit{u}_{\x_1} \relarrow{\impGSI.\hb} \txid_{a_1}.\mathit{u}_{\x_1}$, 
contradicting the assumption that $\impGSI$ is RA-consistent. 
\\

\noindent \textbf{RTS. $\rfi \cup \coi \cup \fri \subseteq \po$}\\
Follows immediately from the construction of $\absGSI$.

\end{proof}

\end{theorem}

\subsection{Implementation Completeness}
In order to establish the completeness of our implementation, it suffices to show that given an SI-consistent execution graph $\absGSI = (\Events, \po, \rf, \co, \silo)$, we can construct a corresponding RA-consistent execution graph $\impGSI$ of the implementation.

%Before proceeding with the construction of a corresponding implementation graph, we describe several auxiliary definitions.
%
%
%Given an abstract transaction class $\Transactions_{\txid} \in \absGSI.\Transactions/\st$, we write $\writeset_{\txid}$ for the set of locations written to by $\Transactions_{\txid}$: $\writeset_{\txid} = \bigcup_{e \in \Transactions_{\txid}\cap \Writes} \loc{e}$.
%Similarly, we write $\readset_{\txid}$ for the set of locations read from by $\Transactions_{\txid}$,
%\emph{prior to} being written by $\Transactions_{\txid}$. 
%For each location \code x read from by $\Transactions_{\txid}$, we additionally record the first read event in $\Transactions_{\txid}$ that retrieved the value of \code x.
%That is, 
%\[
%\readset_{\Transactions_{\txid}} \eqdef
%\setcomp{
%	(\code x, r)
%}{
%	r \in \Transactions_i \cap \Reads_{\code x}
%	\land \neg\exsts{e \in \Transactions_{\txid} \cap \Events_{\code x}}
%	e \relarrow{\po} r
%}
%\]
%%

Note that the execution trace for each transaction $\Transactions_{\txid} \in \absGSI.\Transactions/\st$ is of the form 
$\trace'_{\txid} = \mathit{t}'_1 \relarrow{\imm \po} \cdots \relarrow{\imm \po} \mathit{t}'_k$ for some $k$, where each $\mathit{t}'_i$ is a read or write event.
As such, we have $\absGSI.\Events = \absGSI.\Transactions = \bigcup_{\Transactions_{\txid} \in \absGSI.\Transactions/\st} \Transactions_{\txid} = \trace'_{\txid}.\Events$.
%
%Let $\readset_{\Transactions_{\txid}} \cup \writeset_{\Transactions_{\txid}} = \{\code x_1 \cdots \code x_n\}$.
%
We thus construct the implementation trace $\trace_{\txid}$ as follows:
\[
	\mathit{Is}
	\relarrow{\imm \po} \mathit{Ts}
	\relarrow{\imm \po} \mathit{RUs}
	\relarrow{\imm \po} \mathit{PLs}
	\relarrow{\imm \po} \mathit{Ws}
	\relarrow{\imm \po} \mathit{Us}
\]
where:
\begin{itemize}
	\item $\mathit{Is}$ denotes the sequence of events initialising the values of \readset and  \writeset with $\emptyset$, and initialising \code{s[x]} with $\bot$ for each location \x;
	\item $\mathit{Ts}$  is of the form $\mathit{t}_1 \relarrow{\imm \po} \cdots \relarrow{\imm \po} \mathit{t}_k$, where for all $m \in \{1 \cdots k\}$:
	\[
	\mathit{t}_m = 
	\begin{cases}
		\rseq{\code x_m}{v_m}{\mathit{RS}_{m {-} 1}, \mathit{WS}_{m {-} 1}}
		\relarrow{\imm \po} \mathit{lr}_{\x_m}
		& \text{if } t'_m {=} \readE{-}{\code x_m}{v_m} \\
		\wseq{\code x_m}{v_m}{\mathit{RS}_{m {-} 1}, \mathit{WS}_{m {-} 1}} 
		\relarrow{\imm \po} \mathit{wws}_{\x_m}
		\relarrow{\imm \po} \mathit{lw}_{\x_m}
		& \text{if } t'_m {=} \writeE{\rel}{\code x_m}{v_m} \\
	\end{cases}
	\]
	where  $\mathit{lr}_{\x_m} \eqdef \readE{}{\code{s[x}_m\code ]}{v_m}$; 
	$\mathit{lw}_{\x_m} \eqdef \writeE{}{\code{s[x}_m\code ]}{v_m}$; 
	the identifiers of $\mathit{lr}_{\x_m}$ and $\mathit{lw}_{\x_m}$ are those of $t'_m$, whilst the identifiers of other events in $t_m$ are picked fresh; 
\[
\small
	\rseq{\code x_m}{v_m}{\mathit{RS}_{m {-} 1}, \mathit{WS}_{m {-} 1}}
	\eqdef
	\begin{cases}
		\readE{}{\code{s[x}_m\code ]}{\bot}
		& \text{if } \code x_m \not\in \mathit{RS}_{m {-} 1} \cup  \mathit{WS}_{m {-} 1} \\
		\relarrow{\imm \po} \mathit{wrs}_{\x_m} \\
		\relarrow{\imm \po} \mathit{fs_m} \\
		\relarrow{\imm \po} \mathit{rl}_{\x_m} \\
		\relarrow{\imm \po} \mathit{rs}_{\x_m} \\
		\relarrow{\imm \po} \mathit{ws}_{\x_m}
		& \\\\
		\emptyset
		& \text{otherwise} 
	\end{cases}
\]	
$\mathit{wrs}_{\x_m} \eqdef \writeE{}{\readset}{\mathit{RS}_m}$, 
$\mathit{fs}_m$ denotes the sequence of events attempting (but failing) to acquire the read lock on $\code x_m$, 
$\mathit{rl}_{\code x_m} \eqdef \rlockE{\xl_m}$, 
$\mathit{rs}_{\code x_m} \eqdef \readE{\acq}{\code x_m}{v_m}$, 
and $\mathit{ws}_{\code x_m} \eqdef \writeE{}{\code{s[x}_m\code ]}{v_m}$; 
$\mathit{RS}_0 {=} \emptyset$ and for all $m > 0$:
\[
	\mathit{RS}_{m {+} 1} \eqdef 
	\begin{cases}
		\mathit{RS}_m \cup \{\code x_m\} & \text{if } t'_m {=} \readE{-}{\code x_m}{-} \\
		\mathit{RS}_m & \text{otherwise}
	\end{cases}
\]
and
\[
\small
	\wseq{\code x_m}{v_m}{\mathit{RS}_{m {-} 1},  \mathit{WS}_{m {-} 1}}
	\eqdef
	\begin{cases}
		\readE{}{\code{s[x}_m\code ]}{\bot}
		& \text{if } \code x_m \not\in \mathit{RS}_{m {-} 1} \cup  \mathit{WS}_{m {-} 1} \\
		\relarrow{\imm \po} \mathit{fs_m}  \\
		\relarrow{\imm \po} \mathit{rl}_{\x_m}
		\emptyset
		& \text{otherwise} 
	\end{cases}
\]	
$\mathit{lw}_{\x_m} {=} \writeE{}{\code{s[x}_m\code ]}{v_m}$,
$\mathit{wws}_{\x_m} \eqdef \writeE{}{\writeset}{\mathit{WS}_m}$; 
$\mathit{fs}_m$and $\mathit{rl}_{\code x_m}$ are as defined above; and 
\[
	\mathit{WS}_{m {+} 1} \eqdef 
	\begin{cases}
		\mathit{WS}_m \cup \{\code x_m\} & \text{if } t'_m {=} \writeE{-}{\code x_m}{-} \\
		\mathit{WS}_m & \text{otherwise}
	\end{cases}
\]	
Let $\readset_\txid = \mathit{RS}_m$ and  
$\writeset_\txid = \mathit{WS}_m$; 
let $\readset_\txid \cup \writeset_\txid$ be enumerated as $\{\x_1 \cdots \x_i\}$ for some $i$.
	\item $\mathit{RUs}$ denotes the sequence of events releasing the reader locks (when the given location is in the read set only) and is of the form $\mathit{ru}_{\code x_1} \relarrow{\imm \po} \cdots \relarrow{\imm \po} \mathit{ru}_{\code x_i}$, where for all $n \in \{1 \cdots i\}$:
	\[
	\begin{array}{l}
		\mathit{ru}_{\code x_n} = 
		\begin{cases}
			\runlockE{\xl_n}
			& \text{ if } \code x_n \not\in \writeset_{\txid}  \\
			\emptyset
			& \text{ otherwise}
		\end{cases}
	\end{array}	
	\]
	with the identifier of each $\mathit{ru}_{\x_n}$ picked fresh; 
	\item $\mathit{PLs}$ denotes the sequence of events promoting the reader locks to writer ones (when the given location is in the write set), and is of the form $\mathit{pl}_{\code x_1} \relarrow{\imm \po} \cdots \relarrow{\imm \po} \mathit{pl}_{\code x_i}$, where for all $n \in \{1 \cdots i\}$:
	\[
	\begin{array}{l}
		\mathit{pl}_{\code x_n} = 
		\begin{cases}
			\plockE{\xl_n}
			& \text{if }  \code x_n \in \writeset_{\txid} \\
			\emptyset
			& \text{ otherwise } 
		\end{cases}
	\end{array}	
	\]
	with the identifier of each $\mathit{ru}_{\x_n}$ picked fresh; 
	\item $\mathit{Ws}$ denotes the sequence of events committing the writes of \denot{\code{T}} and is of the form $\mathit{c}_{\x_1} \relarrow{\imm \po} \cdots \relarrow{\imm \po} \mathit{c}_{\x_i}$, where for all $n \in \{1 \cdots i\}$:
\[
	\mathit{c}_{\x_n} = 
	\begin{cases}
		\readE{}{\code{s[x}_n \code{]}}{v_n}  \relarrow{\imm \po} 
		\mathit{w}_{\x_n} {=} \writeE{}{\x_n}{v_n} 
		& \text{if } \x_n \in \writeset_\txid \\
		\emptyset 
		& \text{otherwise}
	\end{cases}
	\]
	with the identifiers of events in each $\mathit{c}_{\x_n}$ picked fresh; 
	\item $\mathit{Us}$ denotes the sequence of events releasing the locks on the write set, and is of the form $\mathit{wu}_{\code x_1} \relarrow{\imm \po} \cdots \relarrow{\imm \po} \mathit{wu}_{\code x_i}$, where for all $n \in \{1 \cdots i\}$:
	\[
		\mathit{wu}_{\code x_n} = 
		\begin{cases}
			\wunlockE{\code{xl}_n} & \text{if } \code x_n \in \writeset_{\txid} \\
			\emptyset & \text{otherwise}
		\end{cases}				
	\]
	with the identifier of each $\mathit{wu}_{\x_n}$ picked fresh.
\end{itemize}
We use the $\txid.$ prefix to project the various events of the implementation trace $\trace_\txid$ (e.g.~$\txid.\mathit{rl}_{\code x_j}$). 
Note that for each transaction $\txid$ and each location $\x$, the $\txid.\mathit{rl}_\x$, $\txid.\mathit{rs}_\x$, $\txid.\mathit{ru}_\x$, $\txid.\mathit{pl}_\x$, $\txid.\mathit{wu}_\x$ and $\txid.\mathit{w}_\x$ are uniquely identified when they exist. 

For each location $\x \in \writeset_\txid$, let $\fw \x$ denote the maximal write (in $\po$ order within $\txid$) logging a write for \x in \code{s[x]}. That is,  when $\trace_{\txid}=  t_1 \relarrow{\imm \po} \cdots \relarrow{\imm \po} t_m$, let $\fw \x = \func{wmax}{\x, [t_1 \cdots t_m]}$, where
\[
\begin{array}{@{} c @{}}
	 \func{wmax}{\x, [\,]} \text{ undefined} \\
	  \func{wmax}{\x, L {++} [t]}
	  \eqdef
	  \begin{cases}
	  	\mathit{lw}_{\x}
	  	& \text{if }  t {=} \wseq{\x}{-}{-, -} \relarrow{\po} \mathit{lw}_{\x}  \relarrow{\po} \mathit{wws}_{\x} \\
	  	 \func{wmax}{\x, L}
	  	 & \text{otherwise}
	  \end{cases}
\end{array}	  
\]
Similarly, for each location $\x \in \writeset_\txid$, let $\iw \x$ denote the minimal write (in $\po$ order within $\txid$) logging a write for \x in \code{s[x]}. That is,  when $\trace_{\txid}=  t_1 \relarrow{\imm \po} \cdots \relarrow{\imm \po} t_m$, let $\iw \x = \func{wmin}{\x, [t_1 \cdots t_m]}$, where
\[
\begin{array}{@{} c @{}}
	 \func{wmin}{\x, [\,]} \text{ undefined} \\
	  \func{wmin}{\x, [t] {++} L}
	  \eqdef
	  \begin{cases}
	  	\mathit{lw}_{\x}
	  	& \text{if }  t {=} \wseq{\x}{-}{-, -} \relarrow{\po} \mathit{lw}_{\x}  \relarrow{\po} \mathit{wws}_{\x} \\
	  	 \func{wmin}{\x, L}
	  	 & \text{otherwise}
	  \end{cases}
\end{array}
\]
Analogously, for each location $\x \in \readset_\txid$, let $\ir \x$ denote the minimal read (in $\po$ order within $\txid$) reading the value of \x. That is,  when $\trace_{\txid}=  t_1 \relarrow{\imm \po} \cdots \relarrow{\imm \po} t_m$, let $\ir \x = \func{rmin}{\x, [t_1 \cdots t_m]}$, where
\[
\begin{array}{@{} c @{}}
	 \func{rmin}{\x, [\,]} \text{ undefined} \\
	  \func{rmin}{\x, [t] {++} L}
	  \eqdef
	  \begin{cases}
	  	\mathit{lr}_{\x}
	  	& \text{if } t {=} \rseq{\x}{-}{-, -} \relarrow{\po} \mathit{lr}_{\x} \\
	  	 \func{rmin}{\x, L}
	  	 & \text{otherwise}
	  \end{cases}
\end{array}
\]

Given the $\absGSI$ classes $\TClasses$, let us construct the strict total order $\TCO: \TClasses \times \TClasses$ as described in \cref{subapp:si_completeness}.
For each location \code x and $i \in \Nats$, let us similarly define $\sphase[\code x]{i}$, $\phase{i}$, $\writer{\phase{i}}$, $\LO_{\x}$, $\RLO_{\x}$, $\lo_1$ and $\lo_2$.
\begin{remark}
Recall that both $\lo_1$ and $\lo_2$ satisfy the conditions stated in \cref{def:si_implementation_consistency};
%Note that both $\wsync[\silo_{\x}]$ and $\wsync[\silor_\x]$ hold. 
$\lo_2$ additionally satisfies the `read-read-synchronisation' property in (\ref{ax:lock_rsync}).
As before, we demonstrate that given an SI-consistent execution graph $\absGSI$, it is always possible to construct an RA-consistent execution graph $\impGSI$ of the implementation with its lock order defined as $\lo_2$. Recall that as $\lo_1 \subseteq \lo_2$, it is straightforward to show that replacing $\lo_2$ in such a $\impGSI$ with $\lo_1$, preserves the RA-consistency of $\impGSI$, as defined in \cref{def:si_implementation_consistency}. In other words, as $\lo_1 \subseteq \lo_2$, we have:
\[
\begin{array}{@{} l @{}}
	\acyc{\hbloc \cup \co \cup \fr} \text{ with } \hb \eqdef \transC{(\po \cup \rf \cup \lo_2)} 
	\Rightarrow \\
	\hspace{100pt}
	\acyc{\hbloc \cup \co \cup \fr} \text{ with } \hb \eqdef \transC{(\po \cup \rf \cup \lo_1)} 
\end{array}	
\]
As such, by establishing the completeness of our implementation with respect to $\lo_2$, we also establish its completeness with respect to $\lo_1$. 
In other words, we demonstrate the completeness of our implementation with respect to both lock implementations presented earlier in \cref{app:lock_implementations}.
\end{remark}
We now demonstrate the completeness of our implementation. 
Given an SI-consistent graph,
we construct an implementation graph $\impGSI$ as follows and demonstrate that it is RA-consistent. 
\begin{itemize}
	\item $\impGSI.\Events = \bigcup\limits_{\Transactions_\txid \in \absGSI.\Transactions/\st} \trace_\txid.\Events$, with the $\tx{e} = 0$, for all $e \in \impGSI.\Events$. \\
%	Observe that 
%	$\absGSI.\Events \subseteq \impGSI.\Events$.
%
	\item $\impGSI.\po$ is defined as $\absGSI.\po$ extended by the $\po$ for the additional events of $\impGSI$, given by each $\trace_\txid$ trace defined above. Note that $\impGSI.\po$ does not introduce additional orderings between events of $\absGSI.\Events$. That is, $\for{a, b \in \absGSI.\Events} (a, b) \in \absGSI.\po \Leftrightarrow (a, b) \in \impGSI.\po$.
	\item 
	$\impGSI.\rf = 
	\bigcup_{\txid \in \textsc{Tx}} \mathsf{RF}_{\txid}$
	with 
	$
	\mathsf{RF}_{\txid} \eqdef
	\setcomp{
		(w, \txid.\mathit{rs}_{\x}	)
	}{
		\x \in \Locs \land (w, \txid.\ir \x) \in \absGSI.\rf	
	}
	$.
	\item $\impGSI.\co = 
	\transC{
	\setcomp{
		(\txid_1.\mathit{w}_{\x}	, \txid_2.\mathit{w}_{\x}	)
	}{
		\txid_1, \txid_2,  \in \textsc{Tx}
		\land \x \in \Locs 
		\land (\txid_1.\fw \x, \txid_2.\iw \x) \in \absGSI.\co
	}
	}
	$.
	
	\item $\impGSI.\silo = \bigcup_{\x \in \Locs} \RLO_\x$, with $\RLO_\x$ as defined above.
\end{itemize}
\paragraph{Notation} 
Given an implementation graph $\impGSI$ as constructed above (with $\impGSI.\Transactions = \emptyset$), and a relation $\makerel r \subseteq \impGSI.\Events \times \impGSI.\Events$, we override the $\tout{\makerel r}$ notation and write $\tout{\makerel r}$ for:
\[
	\setcomp{(a, b) \in \makerel r}{\exsts{\txid_a, \txid_b} a \in \trace_{\txid_a}.\Events \land b \in \trace_{\txid_b}.\Events \land \txid_a \ne \txid_b}
\]	
Analogously, we write $\tin{\makerel r}$ for $\setcomp{(a, b) \in \makerel r}{\exsts{\txid} a, b \in \trace_{\txid}.\Events}$.
\begin{lemma}\label{lem:si_alt_completeness_TCO}
For all SI-consistent execution graphs $\absGSI = (\Events, \po, \rf, \co, \silo)$, and for all $\code x \in \Locs$, $\txid_a, \txid_b \in \TXIDs$, and $i, j \in \Nats$:
\[
	\txid_a \in \sphase{i} \land  \txid_b \in \sphase{j} \land i < j \Rightarrow (\txid_a,  \txid_b) \in \TCO
\]
\begin{proof}
Follows immediately from \cref{cor:si_completeness_TCO}. 
\end{proof}
\end{lemma}
%
%
%Given a relation $\makerel r \in \impGSI.\Events \times \impGSI.\Events$, we write $\makerel{r}_{\textsc I}$ for $\makerel r \cap \setcomp{(a, b)}{\exsts{\txid} a, b \in \trace_\txid.\Events}$ and write $\makerel{r}_{\textsf T}$ for $\makerel r \cap \setcomp{(a, b)}{\exsts{\txid_1, \txid_2} a \in \trace_{\txid_1}.\Events \land b \in \trace_{\txid_2}.\Events \land \txid_1 \ne \txid_2}$.
%

Given an implementation graph $\impGSI = (\Events, \po, \rf, \co, \silo)$, let us define $\hb'$ and $\hb'_i$ as described in \cref{subapp:si_completeness} for all $i \in \Nats$.
\begin{lemma}\label{lem:si_alt_completeness_aux2}
For all implementation graphs $\impGSI = (\Events, \po, \rf, \co, \silo)$ constructed as above, 
%and for all $a,b \in \Events$:
%
\[
\begin{array}{@{} l @{\hspace{2pt}} l @{}}
%	\for{i \in \Nats^{+}}
	\for{a, b, \txid_a, \txid_b} 
	& a \in \trace_{\txid_a}.\Events
	\land
	b \in \trace_{\txid_b}.\Events
	\land
	a \relarrow{\hb'} b  \\
	& \qquad 
	\Rightarrow (a, b) \in \impGSI.\poi \lor (\txid_a, \txid_b) \in \TCO
\end{array}
\]
\begin{proof} 
Pick an arbitrary implementation graph $\impGSI = (\Events, \po, \rf, \co, \silo)$ constructed as above.
Since $\hb'$ is a transitive closure, it is straightforward to demonstrate that $\hb' = \bigcup\limits_{i \in \Nats} \hb'_i$. 
It thus suffices to show:
\[
\begin{array}{@{} l @{\hspace{2pt}} l @{}}
	\for{i \in \Nats}
	\for{a, b, \txid_a, \txid_b} 
	& a \in \trace_{\txid_a}.\Events
	\land
	b \in \trace_{\txid_b}.\Events
	\land
	a \relarrow{\hb'_i} b \\ 
	& \qquad 
	\Rightarrow (a, b) \in \impGSI.\poi \lor (\txid_a, \txid_b) \in \TCO
\end{array}
\]
We proceed by induction over $i$.\\

\noindent \textbf{Base case $i = 0$}\\
Pick arbitrary $a, b, \txid_a, \txid_b$ such that $a \in \trace_{\txid_a}.\Events$ and $b \in \trace_{\txid_b}.\Events$, $a \relarrow{\hb'_0} b$.
There are then four cases to consider: 1) $a \relarrow{\impGSI.\poi \cup \rfi \cup \siloi} b$; or 2) $a \relarrow{\impGSI.\poe} b$; 3) $a \relarrow{\impGSI.\rfe} b$; or 4) $a \relarrow{\impGSI.\siloe} b$. 

In case (1), from the construction of $\impGSI.\rfi$ we have $\impGSI.\rfi \subseteq \impGSI.\poi$; moreover, from the construction of $\impGSI.\siloi$ we have $\impGSI.\siloi \subseteq \impGSI.\poi$. Consequently we have $a \relarrow{\impGSI.\poi} b$, as required.

In case (2) from the construction of $\impGSI.\po$ we then know that $\Transactions_{\txid_a} \times  \Transactions_{\txid_b} \subseteq \absGSI.\po$  and thus from the definition of $\TCO$ we have $(\txid_a, \txid_b) \in \TCO$, as required.

In case (3) from the construction of $\impGSI.\rf$ we then know that $\Transactions_{\txid_a} \times  \Transactions_{\txid_b} \subseteq \absGSI.\rft$  and thus from the definition of $\TCO$ we have $(\txid_a, \txid_b) \in \TCO$, as required.

In case (4) from the definition of $\hb'_0$ we know that there exists $i, j, i', j'$ such that $\txid_a {=} \itemAt{\phase{i}}{j}$, $\txid_b {=} \itemAt{\phase{i'}}{j'}$ and either i) $i < i'$; or ii) $i = i' \land j < j'$.

In case (4.i) from \cref{lem:si_alt_completeness_TCO} we have $(\txid_a, \txid_b) \in \TCO$, as required.
In case (4.ii) from the definitions of $\phase{i}$ we have $(\txid_a, \txid_b) \in \TCO$, as required. \\

\noindent \textbf{Inductive case $i = n {+} 1$}\\
Pick arbitrary $a, b, \txid_a, \txid_b$ such that $a \in \trace_{\txid_a}.\Events$, $b \in \trace_{\txid_b}.\Events$ and $a \relarrow{\hb'_i} b$. 
\begin{align}
\begin{array}{@{} l @{\hspace{2pt}} l @{}}
	\for{j \leq n}
	\for{c, d, \txid_c, \txid_d} 
	& c \in \trace_{\txid_c}.\Events
	\land
	d \in \trace_{\txid_d}.\Events
	\land
	c \relarrow{\hb'_j} d  \\
	& \qquad 
	\Rightarrow (c, d) \in \impGSI.\poi \lor (\txid_c, \txid_d) \in \TCO
\end{array}
\tag{I.H.}
\label{IH:si_alt_completeness_auxiliary}
\end{align}
From the definition of $\hb'_{n {+} 1}$ we then know there exists $e$ such that  $(a, e) \in\hb'_0$ and $(e, b) \in \hb'_n$. Let $e \in \txid_e.\Events$. 
Consequently, from the proof of the base case we then know that $(a, e) \in \impGSI.\poi \lor (\txid_a, \txid_e) \in \TCO$. 
Similarly, from (\ref{IH:si_alt_completeness_auxiliary}) we have $(e, b) \in \impGSI.\poi \lor (\txid_e, \txid_b) \in \TCO$. 
There are now four cases to consider: 
1) $(a, e) \in \impGSI.\poi$ and $(e, b) \in \impGSI.\poi$; 
2) $(a, e) \in \impGSI.\poi$ and $(\txid_e, \txid_b) \in \TCO$; 
3) $(\txid_a, \txid_e) \in \TCO$ and $(e, b) \in \impGSI.\poi$; 
4) $(\txid_a, \txid_e) \in \TCO$ and $(\txid_e, \txid_b) \in \TCO$.

In case (1) from the definition of $\impGSI.\poi$ we have $(a, b) \in \impGSI.\poi$.
In cases (2-4) from the definitions of $\impGSI.\poi$, $\TCO$ we have $(\txid_a, \txid_b) \in \TCO$, as required. 

\end{proof}

\end{lemma}

Given an implementation graph $\impGSI = (\Events, \po, \rf, \co, \silo)$ constructed as above, let $\stage{.}: \impGSI.\Events \rightarrow \{1, 2, 3\}$ denote the \emph{event stage} as follows:
\[
\stage{e} \eqdef
\begin{cases}
	1 & \exsts{\txid} e \in \txid.\mathit{Ts} \\
	2 & \exsts{\txid, \code x} e = \txid.\mathit{ru}_{\code x} \\
	3 & \text{otherwise}
\end{cases}
\]
\begin{lemma}\label{lem:si_alt_completeness_aux3}
For all implementation graphs $\impGSI = (\Events, \po, \rf, \co, \silo)$ constructed as above, and for all $\code x \in \Locs$, $\txid_a, \txid_1, \txid_2 \in \TXIDs$ and $a \in  \trace_{\txid_a}.\Events$:
\[
\begin{array}{@{} l @{}}
	\left(
	\begin{array}{@{} l @{}}
		(\txid_1.\mathit{rl}_{\code x} \relarrow{\silo} \txid_2.\mathit{pl}_{\code x} 
		\land 
		\txid_2.\mathit{pl}_{\code x} \relarrow{\hb} a) \\
		\lor\,
		(\txid_1.\mathit{ru}_{\code x} \relarrow{\silo} \txid_2.\mathit{pl}_{\code x} 
		\land 
		\txid_2.\mathit{pl}_{\code x} \relarrow{\hb} a) \\
		\lor\,
		(\txid_1.\mathit{rl}_{\code x} \relarrow{\silo} \txid_2.\mathit{wu}_{\code x} 
		\land 
		\txid_2.\mathit{wu}_{\code x} \relarrow{\hb} a) \\
		\lor\,
		(\txid_1.\mathit{ru}_{\code x} \relarrow{\silo} \txid_2.\mathit{wu}_{\code x} 
		\land 
		\txid_2.\mathit{wu}_{\code x} \relarrow{\hb} a) 
	\end{array}	 	
	\right)
	\land
	\exsts{j, k, l} 
		\itemAt{\phase j}{k} {=} \txid_1 
		\land 
		\itemAt{\phase j}{l} {=} \txid_2  
		\land 
		k > l \\
	\hspace{150pt}
	\Rightarrow 
%	\qquad 
	\stage{\txid_1.\mathit{ru}_{\code x}} < \stage{a}
	\lor
	(\txid_1, \txid_a) \in \TCO
\end{array}	
\]
\begin{proof}
Pick an arbitrary implementation graph $\impGSI = (\Events, \po, \rf, \co, \silo)$ constructed as above,
and pick arbitrary $\code x \in \Locs$, $\txid_a, \txid_1, \txid_2 \in \TXIDs$ and $a \in \trace_{\txid_a}.\Events$ such that 
$\big((\txid_1.\mathit{rl}_{\code x} \relarrow{\silo} \txid_2.\mathit{pl}_{\code x} 
	\land 
	\txid_2.\mathit{pl}_{\code x} \relarrow{\hb} a) 
   \lor
   (\txid_1.\mathit{ru}_{\code x} \relarrow{\silo} \txid_2.\mathit{pl}_{\code x} 
	\land 
	\txid_2.\mathit{pl}_{\code x} \relarrow{\hb} a) 
  \lor\,
  (\txid_1.\mathit{rl}_{\code x} \relarrow{\silo} \txid_2.\mathit{wu}_{\code x} 
	\land 
	\txid_2.\mathit{wu}_{\code x} \relarrow{\hb} a) 
  \lor\,
  (\txid_1.\mathit{ru}_{\code x} \relarrow{\silo} \txid_2.\mathit{wu}_{\code x} 
	\land 
	\txid_2.\mathit{wu}_{\code x} \relarrow{\hb} a)\big)$ 
and $\exsts{j, k, l} 
		\itemAt{\phase j}{k} {=} \txid_1 
		\land 
		\itemAt{\phase j}{l} {=} \txid_2  
		\land 
		k > l$.
In case of the first disjunct from the construction of $\silo$ we then also have $\txid_1.\mathit{ru}_{\code x} \relarrow{\silo} \txid_2.\mathit{pl}_{\code x} \land \txid_2.\mathit{pl}_{\code x} \relarrow{\hb} a$.
Similarly, in case of the third and fourth disjuncts from the construction of $\silo$ we also have $\txid_1.\mathit{ru}_{\code x} \relarrow{\silo} \txid_2.\mathit{pl}_{\code x} \land	\txid_2.\mathit{wu}_{\code x} \relarrow{\hb} a$. Moreover, since we have $\txid_2.\mathit{pl}_{\code x} \relarrow{\po} \txid_2.\mathit{wu}_{\code x}$ and $\po \subseteq \hb$, we also have $\txid_1.\mathit{ru}_{\code x} \relarrow{\silo} \txid_2.\mathit{pl}_{\code x} \land	\txid_2.\mathit{pl}_{\code x} \relarrow{\hb} a$.
It thus suffices to show:
\[
\begin{array}{@{} l @{}}
	\txid_1.\mathit{ru}_{\code x} \relarrow{\silo} \txid_2.\mathit{pl}_{\code x} 
	\land 
	\txid_2.\mathit{pl}_{\code x} \relarrow{\hb} a
	\land
	\exsts{j, k, l} 
		\itemAt{\phase j}{k} {=} \txid_1 
		\land 
		\itemAt{\phase j}{l} {=} \txid_2  
		\land 
		k > l \\
	\hspace{150pt}
	\Rightarrow 
%	\qquad 
	\stage{\txid_1.\mathit{ru}_{\code x}} < \stage{a}
	\lor
	(\txid_1, \txid_a) \in \TCO
\end{array}	
\]

Since $\impGSI.\hb$ is a transitive closure, it is straightforward to demonstrate that $\impGSI.\hb = \bigcup\limits_{i \in \Nats} \hb_i$, where $\hb_0 = \impGSI.\po \cup \impGSI.\rf \cup \impGSI.\silo$ and $\hb_{i{+}1} = \hb_0 ; \hb_i$. 
It thus suffices to show:
\[
\begin{array}{@{} l @{}}
	\for{i \in \Nats}
	\for{\code x \in \Locs, \txid_a, \txid_1, \txid_2 \in \TXIDs, a \in  \trace_{\txid_a}.\Events} \\
	\qquad 
	\txid_1.\mathit{ru}_{\code x} \relarrow{\silo} \txid_2.\mathit{pl}_{\code x} 
	\land 
	\txid_2.\mathit{pl}_{\code x} \relarrow{\hb_i} a
	\land
	\exsts{j, k, l} 
		\itemAt{\phase j}{k} {=} \txid_1 
		\land 
		\itemAt{\phase j}{l} {=} \txid_2  
		\land 
		k > l \\
	\hspace{150pt}
	\Rightarrow 
	\stage{\txid_1.\mathit{ru}_{\code x}} < \stage{a}
	\lor
	(\txid_1, \txid_a) \in \TCO
\end{array}
\]
We thus proceed by induction over $i$.\\

\noindent \textbf{Base case $i = 0$}\\
Pick arbitrary $\code x \in \Locs$, $\txid_a, \txid_1, \txid_2 \in \TXIDs$ and $a \in \trace_{\txid_a}.\Events$ such that 
$\txid_1.\mathit{ru}_{\code x} \relarrow{\silo} \txid_2.\mathit{pl}_{\code x}$, $\txid_2.\mathit{pl}_{\code x} \relarrow{\hb_0} a$ 
and $\exsts{j, k, l} 
		\itemAt{\phase j}{k} {=} \txid_1 
		\land 
		\itemAt{\phase j}{l} {=} \txid_2  
		\land 
		k > l$.
From the construction of $\impGSI$ we know that there exists $r \in \absGSI.\Transactions_{\txid_1}$ and $w \in \absGSI.\Transactions_{\txid_2}$ such that $(r, w) \in \absGSI.\fr$. 
Moreover, since $k \ne l$ we know that $(r, w) \in \absGSI.\frt$. 
Since we have $\txid_2.\mathit{pl}_{\code x} \relarrow{\hb_0} a$, there are four cases to consider: 
1) $\txid_2.\mathit{pl}_{\code x} \relarrow{\impGSI.\poi \cup \rfi \cup \siloi} a$; or 
%2) $\txid_2.\mathit{wu}_{\code x} \relarrow{\impGSI.\poi \cup \rfi \cup \siloi} a$; or 
2) $\txid_2.\mathit{pl}_{\code x} \relarrow{\impGSI.\poe} a$; or 
%4) $\txid_2.\mathit{wu}_{\code x} \relarrow{\impGSI.\pot} a$; or 
3) $\txid_2.\mathit{pl}_{\code x} \relarrow{\impGSI.\rfe} a$; or 
%6) $\txid_2.\mathit{wu}_{\code x} \relarrow{\impGSI.\rft} a$; or 
4) $\txid_2.\mathit{pl}_{\code x} \relarrow{\impGSI.\siloe} a$. %; or
%8) $\txid_2.\mathit{wu}_{\code x} \relarrow{\impGSI.\silot} a$. 

In case (1), from the construction of $\impGSI.\rfi$ we have $\impGSI.\rfi \subseteq \impGSI.\poi$; moreover, from the construction of $\impGSI.\siloi$ we have $\impGSI.\siloi \subseteq \impGSI.\poi$. We thus have $\txid_2.\mathit{pl}_{\code x} \relarrow{\impGSI.\poi} a$. On the other hand, from the definition of $\stage{.}$ we know $\stage{a} \geq \stage{\txid_2.\mathit{pl}_{\code x}} > \stage{\txid_1.\mathit{ru}_{\code x}}$, and thus  $\stage{a} > \stage{\txid_1.\mathit{ru}_{\code x}}$, as required. 

In case (2), since $\txid_2.\mathit{pl}_{\code x} \relarrow{\poe} a$, we know there exists $a' \in \Transactions_{\txid_a}$ such that $(w, a') \in \absGSI.\pot$. As such, we have $(r, a') \in \absGSI.(\frt; \pot)$. From the definition of $\TCO$ we thus have $(\txid_1, \txid_a) \in \TCO$, as required.

%Case (3) cannot happen as any $\rf$ edge from $\txid_2.\mathit{pl}_{\code x}$ is an internal edge included in $\impGSI.\poi$.\\

Case (3) cannot happen as there are no $\rf$ edges from lock events. 
In case (4), from the construction of $\silo$ we know there exists $m, n$ such that $m>j$, and $\txid_a = \itemAt{\phase{m}}{n}$.
As such, from \cref{lem:si_alt_completeness_TCO} we have $(\txid_1, \txid_a) \in \TCO$, as required.\\

\noindent \textbf{Inductive case $i = n {+} 1$}\\
Pick arbitrary $\code x \in \Locs$ and $\txid_a, \txid_1, \txid_2 \in \TXIDs$ and $a \in \trace_{\txid_a}.\Events$ such that $\txid_1.\mathit{ru}_{\code x} \relarrow{\lo} \txid_2.\mathit{pl}_{\code x}$, $\txid_2.\mathit{pl}_{\code x} \relarrow{\hb_i} a$ and $\exsts{j, k, l} 
		\itemAt{\phase j}{k} {=} \txid_1 
		\land 
		\itemAt{\phase j}{l} {=} \txid_2  
		\land 
		k > l$.
\begin{align}
\hspace*{-15pt}
\begin{array}{@{} l @{}}
	\for{j \leq n}
	\for{\code y \in \Locs, \txid_{a'}, \txid_3, \txid_4 \in \TXIDs, a' \in \trace_{\txid_{a'}}.\Events} \\
	\quad 
	(\txid_3.\mathit{ru}_{\code x} \relarrow{\silo} \txid_4.\mathit{pl}_{\code x} 
	\land 
	\txid_4.\mathit{pl}_{\code x} \relarrow{\hb_j} a') 
	\land
	\exsts{j, k, l} 
		\itemAt{\phase j}{k} {=} \txid_1 
		\land 
		\itemAt{\phase j}{l} {=} \txid_2  
		\land 
		k > l \\
	\hspace{140pt}
	\Rightarrow 
	\stage{\txid_3.\mathit{ru}_{\code y}} < \stage{a'}
	\lor
	(\txid_3,  \txid_{a'}) \in \TCO
\end{array}
\tag{I.H.}
\label{IH:si_alt_completeness_auxiliary2}
\end{align}
There are two cases to consider: \\
1) $\txid_2.\mathit{pl}_{\code x} \relarrow{\transC{(\poi \cup \rfi \cup \siloi)}} a$; or \\
2) $\txid_2.\mathit{pl}_{\code x} \relarrow{\reftransC{(\poi \cup \rfi \cup \siloi)}; (\poe \cup \rfe \cup \siloe); \hb_m} a$, where $m \leq n$. 

In case (1), from the construction of $\impGSI.\rfi$ we have $\impGSI.\rfi \subseteq \impGSI.\poi$; moreover, from the construction of $\impGSI.\siloi$ we have $\impGSI.\siloi \subseteq \impGSI.\poi$. 
We thus have $\txid_2.\mathit{pl}_{\code x} \relarrow{\transC{\poi }} a$, i.e. $\txid_2.\mathit{pl}_{\code x} \relarrow{\poi } a$. As such, from the proof of the base case  we have $\stage{a} > \stage{\txid_1.\mathit{ru}_{\code x}}$, as required. 

In case (2)  from the construction of $\impGSI$ we know that there exists $r \in\trace_{ \txid_1}.\Events$ and $w \in \trace_{\txid_2}.\Events$ such that $(r, w) \in \absGSI.\fr$. 
%Moreover, since $\txid_1.\mathit{ru}_{\code x} \relarrow{\rf} \txid_2.\mathit{pl}_{\code x}$, from the construction of $\impGSI$ we know that $\txid_1 \ne \txid_2$ and thus we have $(r, w) \in \absGSI.\frt$. 
Moreover, since $k \ne l$, from the construction of $\impGSI$ we know that $\txid_1 \ne \txid_2$ and thus we have $(r, w) \in \absGSI.\frt$. 
On the other hand, we know that there exist $b, c$ such that $\txid_2.\mathit{pl}_{\code x} \relarrow{\reftransC{(\poi \cup \rfi \cup \siloi)}} b \relarrow{\poe \cup \rfe \cup \siloe} c \relarrow{\hb_m} a$. 
From the construction of $\impGSI.\rfi$ we have $\impGSI.\rfi \subseteq \impGSI.\poi$; moreover, from the construction of $\impGSI.\siloi$ we have $\impGSI.\siloi \subseteq \impGSI.\poi$. 
We thus have $\txid_2.\mathit{pl}_{\code x} \relarrow{\reftransC{\poi }} b$. As such we have $\txid_2.\mathit{pl}_{\code x} \relarrow{\reftransC{\poi}} b \relarrow{\poe \cup \rfe \cup \siloe} c$. Let $c \in \trace_{\txid_c}$.
There are now three cases to consider: a) $b \relarrow{\poe} c$; or b) $b \relarrow{\rfe} c$; or c) $b \relarrow{\siloe} c$. We first demonstrate that in all three cases we have $(\txid_1, \txid_c) \in \TCO$.

In case (2.a), since $\txid_2.\mathit{pl}_{\code x} \relarrow{\poe} c$, we know there exists $c' \in \Transactions_{\txid_c}$ such that $(w, c') \in \absGSI.\pot$. As such, we have $(r, c') \in \absGSI.(\frt; \pot)$. From the definition of $\TCO$ we thus have $(\txid_1, \txid_c) \in \TCO$.
%In case (2.b), there are two additional cases to consider: i) $\loc b = \loc c = \code x$; or ii) $\loc b = \loc c = \x$. 
In case (2.b), from the definition of $\impGSI.\rf$ we know there exists $c' \in \Transactions_{\txid_c}$ such that $(w, c') \in \absGSI.\rft$. 
As such, we have $(r, c) \in \absGSI.(\frt; \rft)$. From the definition of $\TCO$ we thus have $(\txid_1, \txid_c) \in \TCO$.
In case (2.c) from the construction of $\lo$ we know there exists $z$ such that either $b = \txid_2.\mathit{pl}_z$ or $b = \txid_2.\mathit{wu}_z$,  and that there exist $p, q$ such that $\txid_2 \in \sphase[\code z]{p}$, $\txid_c \in \sphase[\code z]{q}$ and $p < q$. Given the definition of $\sphase[\code z]{}$ we then know that there exists $w_z \in \absGSI.\Transactions_{\txid_2}$ where either 
i) there exists $w'_z \in \Transactions_{\txid_c}$ such that $(w_z, w'_z) \in \absGSI.\cot$; or  
ii) there exists $r_z \in \Transactions_{\txid_c}$ such that $(w_z, r_z) \in \absGSI. \rft$; or
ii) there exists $r_z \in \Transactions_{\txid_c}$ such that $(w_z, r_z) \in \absGSI.(\cot; \rft)$.  
That is, we have either $(w, w'_z) \in \absGSI.\cot$, or $(w, r_z) \in \absGSI.\rft$, or $(w, r_z) \in \absGSI.(\cot; \rft)$. 
Moreover, since we have $(r, w) \in \absGSI.\frt$, we then have $(r, w'_z) \in \absGSI.(\frt; \cot)$, or $(r, r_z) \in \absGSI.(\frt; \rft)$, or $(r, r_z) \in \absGSI.(\frt; \cot; \rft)$. From the definition of $\TCO$ we thus have $(\txid_1, \txid_c) \in \TCO$.

Since $c \relarrow{\hb_m} a$, there are now two cases to consider: 
i) $c \relarrow{\hb'} a$; or ii) $c \not\relarrow{\hb'} a$.

In case (2.i), 
%we have $c \relarrow{\hb'} a$  and thus 
from \cref{lem:si_alt_completeness_aux2} we have $(c, a) \in \poi  \lor (\txid_c, \txid_a) \in \TCO$. As we have $(\txid_1, \txid_c) \in \TCO$, we thus have $(\txid_1, \txid_a) \in \TCO$, as required.

In case (2.ii) let us split the path from $c$ at the first occurrence of a non-$\hb'$ edge. 
That is, pick $\txid_3, \txid_4, g, h, \code y, p, q, s, k$ such that $c  \relarrow{\hb'} \txid_3.g  \relarrow{\silo} \txid_4.h \relarrow{\hb_k} a$, $k < m$, $\itemAt{\phase[y]{s}}{p} = \Transactions_{\txid_3}$,  $\itemAt{\phase[y]{s}}{q} = \Transactions_{\txid_4}$, $p > q$
and either a) $g = \mathit{rl}_{\code y} \land h = \mathit{pl}_{\code y}$; or b) $g = \mathit{ru}_{\code y} \land h = \mathit{pl}_{\code y}$; or c) $g = \mathit{rl}_{\code y} \land h = \mathit{wu}_{\code y}$; or d) $g = \mathit{ru}_{\code y} \land h = \mathit{wu}_{\code y}$.
From \cref{lem:si_alt_completeness_aux2} we then have $(c, \txid_3.g) \in \poi  \lor (\txid_c, \txid_3) \in \TCO$. As we have $(\txid_1, \txid_c) \in \TCO$, we also have $(\txid_1, \txid_3) \in \TCO$.
%
%We first demonstrate that in all cases (2.ii.a-2.ii.d) there exists $t$ such that $t \leq m$ and $c  \relarrow{\hb'} \txid_3.\mathit{ru}_{\code y}  \relarrow{\silo} \txid_4.\mathit{pl}_{\code y} \relarrow{\hb_t} a$.
We next demonstrate that in all cases (2.ii.a-2.ii.d) there exists $t$ such that $t \leq m$ and $\txid_3.\mathit{ru}_{\code y}  \relarrow{\silo} \txid_4.\mathit{pl}_{\code y} \relarrow{\hb_t} a$.

In case (2.ii.a) from the definition of $\silo$ we also have $\txid_3.\mathit{ru}_{\code y} \relarrow{\silo} \txid_4.\mathit{pl}_{\code x}$.  
%Moreover, since $\txid_3.\mathit{rl}_{\code y} \relarrow{\po} \txid_3.\mathit{ru}_{\code y}$ and $\po \subseteq \hb'$, we also have $c \relarrow{\hb'} \txid_3.\mathit{rl}_{\code y} \relarrow{\hb'} \txid_3.\mathit{ru}_{\code y}$, i.e.~$c \relarrow{\hb'} \txid_3.\mathit{ru}_{\code y}$. 
%As such, we have $c  \relarrow{\hb'} \txid_3.\mathit{ru}_{\code y}  \relarrow{\silo} \txid_4.\mathit{pl}_{\code y} \relarrow{\hb_k} a$.
As such, we have $\txid_3.\mathit{ru}_{\code y}  \relarrow{\silo} \txid_4.\mathit{pl}_{\code y} \relarrow{\hb_k} a$.
In case (2.ii.b) the desired result holds immediately.

In cases (2.ii.c-2.ii.d) from the construction of $\silo$ we have $\txid_3.\mathit{ru}_{\code y} \relarrow{\silo} \txid_4.\mathit{pl}_{\code y}$. Moreover, since we have $\txid_4.\mathit{pl}_{\code y} \relarrow{\po} \txid_4.\mathit{wu}_{\code y}$, we also have $\txid_4.\mathit{pl}_{\code y} \relarrow{\po} \txid_4.\mathit{wu}_{\code y}  \relarrow{\hb_k} a$. As $\po \subseteq \hb$ and $\hb$ is transitively closed, we have $\txid_4.\mathit{pl}_{\code y} \relarrow{\hb_{k{+} 1}} a$. 
As such, we have $\txid_3.\mathit{ru}_{\code y}  \relarrow{\silo} \txid_4.\mathit{pl}_{\code y} \relarrow{\hb_{k {+} 1}} a$. As $k < m$, the desired result holds immediately.

Consequently, from (\ref{IH:si_alt_completeness_auxiliary2}) we have $\stage{\txid_3.\mathit{ru}_{\code y}} < \stage{a} \lor (\txid_3,  \txid_a) \in \TCO$. 
In the case of the first disjunct we have $\stage{\txid_3.\mathit{ru}_{\code y}} {=} \stage{\txid_1.\mathit{ru}_{\code x}} < \stage{a}$, as required.
In the case of the second disjunct, since we also have $(\txid_1, \txid_3) \in \TCO$ and $\TCO$ is transitively closed, we have $(\txid_1, \txid_a) \in \TCO$, as required.

\end{proof}

\end{lemma}

\begin{theorem}[Completeness]
For all SI execution graphs $\absGSI$ and their counterpart implementation graphs $\impGSI$ constructed as above,
\[
	\sicon[\absGSI]  \Rightarrow \consistent{\impGSI}
\]
\begin{proof}
Pick an arbitrary SI execution graph $\absGSI$ and its counterpart implementation graph $\impGSI$ constructed as above and assume $\sicon$ holds.
From the definition of $\consistent{\impGSI}$ it suffices to show: 
\begin{enumerate}
	\item $\irr{\impGSI.\hb}$ \label{goal:si_alt_completeness_hb_irr}
	\item $\irr{\impGSI.\co ; \impGSI.\hb}$ \label{goal:si_alt_completeness_co_hb_irr}
	\item $\irr{\impGSI.\fr ; \impGSI.\hb}$ \label{goal:si_alt_completeness_fr_hb_irr}
\end{enumerate}
\textbf{RTS. part \ref{goal:si_alt_completeness_hb_irr}}\\
We proceed by contradiction. Let us assume that there exists $a, \trace_\txid$ such that $a \in  \trace_\txid.\Events $ and $(a, a) \in \impGSI.\hb$.
There are now two cases to consider: 1) $(a, a) \in \hb'$; or 2) $(a, a) \not\in \hb'$. 

In case (1), from \cref{lem:si_alt_completeness_aux2} we have $(a, a) \in \impGSI.\poi \lor (\Transactions_\txid, \Transactions_\txid) \in \TCO$. The first disjunct leads to a contradiction as the construction of $\impGSI.\po$ yields an acyclic relation. The second disjunct leads to a contradiction as $\TCO$ is a strict total order. 

In case (2), let us split the $a \relarrow{\hb} a$ at the first occurrence of a non-$\hb'$ edge. That is, pick $\txid_1, \txid_2, \code x, i, j, k, g, h$ such that $a \relarrow{\hb} \txid_1.g \relarrow{\lo \setminus \hb'}  \txid_2.h \relarrow{\hb} a$, $\itemAt{\phase[x]{i}}{j} = \Transactions_{\txid_1}$,  $\itemAt{\phase[x]{i}}{k} = \Transactions_{\txid_2}$ and $j > k$.
As we have $\txid_1.g \relarrow{\lo \setminus \hb'} \txid_2.h \relarrow{\hb} a \relarrow{\hb} \txid_1.g$, from \cref{lem:si_alt_completeness_aux3} and the definition of $\hb$ we then have $\stage{\txid_1.g } < \txid_1.g  \lor (\txid_1, \txid_1) \in \TCO$, leading to a contradiction in both disjuncts (the second disjunct yields a contradiction as $\TCO$ is a strict total order). \\
%and  $\neg\exsts{\txid_3, \txid_4, \code y} a  \relarrow{\hb} \txid_3.\mathit{ru}_{\code y}  \relarrow{\rf} \txid_4.\mathit{pl}_{\code y}  \relarrow{\hb} \txid_1.\mathit{ru}_{\code x} $.

\noindent \textbf{RTS. part \ref{goal:si_alt_completeness_co_hb_irr}}\\
We proceed by contradiction. Let us assume that there exists $a, \trace_{\txid_a}, b, \trace_{\txid_a}$ such that $a \in  \trace_{\txid_a}.\Events$, $b \in  \trace_{\txid_b}.\Events$, $(a, b) \in \impGSI.\hb$ and $(b, a) \in \impGSI.\co$.
Let $\loc a = \loc b = \code x$ for some shared location \code x.
There are now two cases to consider: 1) $(b, a) \in \impGSI.\coi$; or 2) $(b, a) \in \impGSI.\coe$.

In case (1) we then have $(b, a) \in \absGSI.\coi \subseteq \absGSI.\poi$. That is, we have $(b, a) \in \impGSI.\po \subseteq \impGSI.\hb$. We thus have $a \relarrow{\impGSI.\hb} b \relarrow{\impGSI.\hb} a$, contradicting our proof in part~\ref{goal:si_alt_completeness_hb_irr}.
In case (2), from the construction of $\impGSI.\co$ we have $(b, a) \in \absGSI.\cot$ and thus 
from the construction of $\impGSI.\lo$ we then have $(\txid_{b}.\mathit{wu}_{\code x}, \txid_{a}.\mathit{rl}_{\code x}) \in \impGSI.\lo$. 
As such we have $a \relarrow{\impGSI.\hb} b \relarrow{\impGSI.\po} \txid_{b}.\mathit{wu}_{\code x} \relarrow{\impGSI.\lo} \txid_{a}.\mathit{rl}_{\code x} \relarrow{\impGSI.\po} a$. That is, we have $a \relarrow{\impGSI.\hb} a$, contradicting our proof in part~\ref{goal:si_alt_completeness_hb_irr}.\\

\noindent \textbf{RTS. part \ref{goal:si_alt_completeness_fr_hb_irr}}\\
We proceed by contradiction. Let us assume that there exists $a, \trace_{\txid_a}, b, \trace_{\txid_a}$ such that $a \in  \trace_{\txid_a}.\Events$, $b \in  \trace_{\txid_b}.\Events$, $(a, b) \in \impGSI.\hb$ and $(b, a) \in \impGSI.\fr$.

Let $\loc a = \loc b = \code x$ for some shared location \code x.
There are now two cases to consider: 1) $(b, a) \in \impGSI.\fri$; or 2) $(b, a) \in \impGSI.\fre$.

In case (1) we then have $(b, a) \in \absGSI.\fri \subseteq \absGSI.\poi$. That is, we have $(b, a) \in \impGSI.\po \subseteq \impGSI.\hb$. We thus have $a \relarrow{\impGSI.\hb} b \relarrow{\impGSI.\hb} a$, contradicting our proof in part~\ref{goal:si_alt_completeness_hb_irr}.
In case (2), from the construction of $\impGSI.\fr$ we have $(b, a) \in \absGSI.\frt$ and thus 
from the construction of $\impGSI.\lo$ we then have either $(b, \txid_{b}.\mathit{wu}_{\code x}) \in \impGSI.\po$ and $(\txid_{b}.\mathit{wu}_{\code x}, \txid_{a}.\mathit{rl}_{\code x}) \in \impGSI.\lo$; or $(b, \txid_{b}.\mathit{ru}_{\code x}) \in \po$ and $(\txid_{b}.\mathit{ru}_{\code x}, \txid_{a}.\mathit{pl}_{\code x}) \in \impGSI.\lo$.
As such in both cases we have $(b, a) \in \impGSI.\hb$. 
Consequently, we have $a \relarrow{\impGSI.\hb} b \relarrow{\impGSI.\hb} a$, contradicting our proof in part~\ref{goal:si_alt_completeness_hb_irr}.

\end{proof}
\end{theorem}

\renewcommand{\absGRSI}{\ensuremath{\impG'}}
\renewcommand{\impGRSI}{\ensuremath{\impG}}
\newcommand{\stg}[2]{\func{stg}{#1, #2}}
\renewcommand{\rsifr}{\sifr}
\newpage
\section{Soundness and Completeness of the Eager RSI Implementation}\label{app:rsi}
\paragraph{Notation} Given an execution graph $(\Events, \po, \rf, \co, \lo)$ we write $\TClasses$ for the set of equivalence classes of $\Transactions$ induced by $\st$; $\class{a}{\st}$ for the equivalence class that contains $a$; and $\Transactions_\txid$ for the equivalence class of transaction $\txid\in \TXIDs$: $\Transactions_\txid \eqdef \setcomp{a}{\tx{a} {=} \txid}$.
We write $\rsicon$ to denote that $\absGSI$ is RSI-consistent; and write $\consistent{\impG}$ to denote that $\impG$ is RA-consistent.

Given an execution graph $\impGRSI$ of the RSI implementation in \cref{fig:si_implementation}, let us assign a transaction identifier to each transaction executed by the program; and given a transaction $\txid$, let $\readset_{\txid}$ and $\writeset_{\txid}$ denote its read and write sets, respectively.
Observe that given a transaction $\txid$ of the RSI implementation in \cref{fig:si_implementation} with $\readset_{\txid} \cup \writeset_{\txid} = \simpleset{\code{x}_1, \cdots, \code{x}_i}$, the trace of $\txid$, written $\trace_{\txid}$, is of the following form: 
\[
	\trace_{\txid} = 
	\mathit{FS}^{*} 
	\relarrow{\imm \po} \mathit{Rs} 
	\relarrow{\imm \po} \mathit{RUs}
	\relarrow{\imm \po} \mathit{PLs}
	\relarrow{\imm \po} \mathit{Ts}
	\relarrow{\imm \po} \mathit{Us}
\]
where
\begin{itemize}
	\item $\mathit{FS}^*_{\txid}$ denotes the sequence of events failing to obtain a valid snapshot.% That is, the $\mathit{FS}^*_{\txid}$ captures the events resulting from the execution of all-but-last iteration of the while loop in \code{snapshot(\readset{+}\writeset)}.
	\item $\mathit{Rs}$ denotes the sequence of events acquiring a valid snapshot,
	% i.e.~the last iteration of the while loop in \code{snapshot(\readset{+}\writeset)}, 
	and is of the form 
	$\mathit{rl}_{\code x_1} \relarrow{\imm \po} \cdots \relarrow{\imm \po} \mathit{rl}_{\code x_i} \relarrow{\imm \po} \mathit{S}_{\code x_1} \relarrow{\imm \po} \cdots \relarrow{\imm \po} \mathit{S}_{\code x_i} \relarrow{\imm \po} \mathit{V}_{\code x_1} \relarrow{\imm \po} \cdots \relarrow{\imm \po} \mathit{V}_{\code x_i}$, where for all $n \in \{1 \cdots i\}$:
	\[
		\begin{array}{l}
			\mathit{rl}_{\code x_n} {=} \rlockE{\x_n} 
			\quad
			\mathit{S}_{\code x_n}{=}
			\begin{cases}
				\mathit{rs}_{\code x_n}
				\relarrow{\imm \po}
				\mathit{ws}_{\code x_n}
				& \text{if } \code x \in \readset_{\txid} \\
				
				\emptyset & \text{otherwise}
			\end{cases} 
			\quad
			\mathit{V}_{\code x_n} {=}
			\begin{cases}
				\mathit{vs}_{\code x_n}
				& \text{if } \code x \in \readset_{\txid} \\
				
				\emptyset & \text{otherwise}
			\end{cases} 			
		\end{array}
	\]	
	with $\mathit{rs}_{\code x_n} \eqdef \readE{\acq}{\code x_n}{v_n}$,
	$\mathit{ws}_{\code x_n} \eqdef \writeE{\rel}{\code{s[x}_n\code ]}{v_n} $
	and $\mathit{vs}_{\code x_n} \eqdef \readE{\acq}{\code x_n}{v_n}$,  for some $v_n$.
	\item $\mathit{RUs}$ denotes the sequence of events releasing the reader locks (when the given location is in the read set only), and is of the form $\mathit{ru}_{\code x_1} \relarrow{\imm \po} \cdots \relarrow{\imm \po} \mathit{ru}_{\code x_i}$, where for all $n \in \{1 \cdots i\}$:
	\[
	\begin{array}{l}
		\mathit{ru}_{\code x_n} = 
		\begin{cases}
			\runlockE{\x_n}
			& \text{ if } \code x_n \not\in \writeset_{\txid}  \\
			\emptyset
			& \text{ otherwise}
		\end{cases}
	\end{array}	
	\]
	\item $\mathit{PLs}$ denotes the sequence of events promoting the reader locks to writer ones (when the given location is in the write set), and is of the form $\mathit{pl}_{\code x_1} \relarrow{\imm \po} \cdots \relarrow{\imm \po} \mathit{pl}_{\code x_i}$, where for all $n \in \{1 \cdots i\}$:
	\[
	\begin{array}{l}
		\mathit{pl}_{\code x_n} = 
		\begin{cases}
			\plockE{\x_n}
			& \text{if }  \code x_n \in \writeset_{\txid} \\
			\emptyset
			& \text{ otherwise } 
		\end{cases}
	\end{array}	
	\]
	\item $\mathit{Ts}$ denotes the sequence of events corresponding to the execution of \denot{\code{T}} in \cref{fig:si_implementation} and is of the form $\mathit{t}_1 \relarrow{\imm \po} \cdots \relarrow{\imm \po} \mathit{t}_k$, where for all $m \in \{1 \cdots k\}$:
	\[
	\mathit{t}_m = 
	\begin{cases}
		\readE{-}{\code{s[x}_n \code{]}}{v_n} & \text{if } O_m {=} \readE{-}{\code x_n}{v_n} \\
		\writeE{\rel}{\code x_n}{v_n} \relarrow{\imm \po} \writeE{-}{\code{s[x}_n \code ]}{v_n}
		& \text{if } O_m {=} \writeE{\rel}{\code x_n}{v_n} \\
	\end{cases}
	\]
	where $O_m$ denotes the $m$th event in the trace of the original $\code T$;
	\item $\mathit{Us}$ denotes the sequence of events releasing the locks on the write set. That is, the events in $\mathit{Us}$ correspond to the execution of the last line of the implementation in \cref{fig:si_implementation}, and is of the form $\mathit{wu}_{\code x_1} \relarrow{\imm \po} \cdots \relarrow{\imm \po} \mathit{wu}_{\code x_i}$, where for all $n \in \{1 \cdots i\}$:
	\[
		\mathit{wu}_{\code x_n} = 
		\begin{cases}
			\wunlockE{\code{yl}_n} & \text{if } \code x_n \in \writeset_{\txid} \\
			\emptyset & \text{otherwise}
		\end{cases}				
	\]
\end{itemize}

Given a transaction trace $\trace_{\txid}$, we write e.g.~$\txid.\mathit{Ls}$ to refer to its constituent $\mathit{Ls}$ sub-trace and write $\mathit{Ls}.\Events$ for the set of events related by \po in $\mathit{Ls}$. Similarly, we write $\txid.\Events$ for the set of events related by \po in $\trace_{\txid}$.
Note that $\impGRSI.\Events = \bigcup\limits_{\txid \in \sort{Tx}}  \txid.\Events $.

\subsection{Implementation Soundness}
In order to establish the soundness of our implementation, it suffices to show that given an RA-consistent execution graph of the implementation $\impGRSI = (\Events, \po, \rf, \co, \lo)$, we can construct a corresponding RSI-consistent execution graph $\absGRSI$ with the same outcome.

Given a transaction $\txid \in \sort{Tx}$ with $\readset_{\txid} \cup \writeset_{\txid}=\{\code x_1 \cdots \code x_i\}$
%, $\writeset_{\txid_t}=\{\code y_1 \cdots \code y_j\}$ 
and trace 
$ \trace_{\txid} = \mathit{Fs}^{*}
	\relarrow{\imm \po} \mathit{Rs}
	\relarrow{\imm \po} \mathit{RUs}
	\relarrow{\imm \po} \mathit{PLs}
	\relarrow{\imm \po} \mathit{Ts}
	\relarrow{\imm \po} \mathit{Us}$, 
with $\mathit{Ts}= \mathit{t}_1 \relarrow{\imm \po} \cdots \relarrow{\imm \po} \mathit{t}_k$, 
we define $\trace'_{\txid} \eqdef \mathit{t}'_1 \relarrow{\imm \po} \cdots \relarrow{\imm \po} \mathit{t}'_k$, such that for all $m \in \{1 \cdots k\}$:
\[
\begin{array}{l c l}
	\mathit{t}'_m {=} \readE{\acq}{\code x_n}{rb_n}
	& \text{when} & 
	\mathit{t}_m = \readE{-}{\code{s[x}_n \code{]}}{rb_n} \\
	
	\mathit{t}'_m {=} \writeE{\rel}{\code x_n}{rb_n}
	& \text{when} & 
	\mathit{t}_m = \writeE{\rel}{\code x_n}{rb_n} \relarrow{\imm \po} \writeE{-}{\code{s[x}_n \code{]}}{rb_n} 
\end{array}	
\]
such that in the first case the identifier of $\mathit{t}'_m$ is that of $\mathit{t}_m$; and in the second case the identifier of $\mathit{t}'_m$ is that of the first event in $\mathit{t}_m$.
%
%Note that for each $\trace'_{i}$ we have $\trace'_{\txid_t}.\Reads = \trace'_{i}.\mathit{Ts}'_i.\Reads$. We thus define the function, $\src{.} : \trace'_{\txid_t}.\Reads \rightarrow \Writes$, for each $\mathit{T'}_m \in Ts'.\Events$ as:
%
We then define:
\[
	\mathsf{RF}_{\txid} \eqdef
	\setcomp{
		(w, t'_j)	
	}{
		t'_j \in \mathit{Ts}'_{\txid} \land \exsts{\code x, v} t'_j {=} \readE{\acq}{\code x}{v} \land w {=} \writeE{\rel}{\code x}{v} \\
		\land (w \in \txid.\Events \Rightarrow 
		\begin{array}[t]{@{} l @{}}
			w \relarrow{\po} t'_j \,\land\\
			(\for{e \in \txid.\Events } w \relarrow{\po} e \relarrow{\po} t'_j \Rightarrow (\loc e {\ne} \code x \lor e {\not\in} \Writes)))
		\end{array} \\				
		
		\land (w \not \in \txid.\Events \Rightarrow 
		\begin{array}[t]{@{} l @{}}
			(\for{e \in \txid.\Events} (e \relarrow{\po} t'_j \Rightarrow (\loc e \ne \code x \lor e \not\in \Writes)) \\
			\land\, (w, \txid.\mathit{rs}_\x), (w, \txid.\mathit{vs}_\x)  \in \impGRSI.\rf)
		\end{array}		
	}
%\begin{array}{@{} l @{}}
%	\src{\mathit{T'}_m} \eqdef w \iffdef \\
%	\quad \exsts{\code x, v} 
%	\mathit{T'}_k = \readE{\acq}{\code x}{v}	
%	\land w {=} \writeE{\rel}{\code x}{v}  \\
%	\quad 
%	\land \left(
%	\begin{array}{@{} l @{}}
%	
%	\left(\begin{array}{@{} l @{}}
%		\exsts{i} w {=} T'_i 
%		\land i < m 		
%		\land\, \for{j} i < j <m \Rightarrow  (\loc{T'_j} \ne \code x \lor T'_j \not\in \Writes)
%	\end{array}
%	\right)	\\
%	
%	\lor 
%	
%	\left(\begin{array}{@{} l @{}}
%		(\for{i < m} \loc{T'_i} {\ne} \code x \lor T'_i \not\in \Writes) 
%		\land \exsts{rx}  (w, rx) \in \rf \\
%		\land\, \mathit{Rs} {=} \cdots \relarrow{\po} rx {=} \readE{\acq}{\code x}{v} \relarrow{\imm{\po}}  \writeE{-}{\code{s[x]}}{v} \relarrow{\po} \cdots
%	\end{array}
%	\right)	
%	
%	\end{array}
%	\right)
%\end{array}
\]

We are now in a position to demonstrate the soundness of our implementation.
Let $\impGRSI.\NT \eqdef \impGRSI.\Events \setminus (\bigcup_{\txid \in \textsc{Tx}} \trace_{\txid}.\Events)$.
 Given an RA-consistent execution graph $\impGRSI$ of the implementation, we construct an RSI execution graph $\absGRSI$ as follows and demonstrate that $\rsicon$ holds.

\begin{itemize}
	\item $\absGRSI.\Events = \bigcup\limits_{\txid \in \sort{Tx}} \trace'_{\txid}.\Events$, with the $\tx{.}$ function defined as:
	\[
		\tx{a} \eqdef 
		\begin{cases}
			\txid & \text{if } a \in \trace'_{\txid} \\
			0 & \text{otherwise}
		\end{cases}
	\]
	\item $\absGRSI.\po = \coerce{\impGRSI.\po}{\absGRSI.\Events}$
	\item $\absGRSI.\rf = \big(\bigcup_{\txid \in \textsc{Tx}} \mathsf{RF}_{\txid}\big) \cup \big(\impGRSI.\rf \cap \impGRSI.\Events \times \impGRSI.\NT)$%\absGRSI.\Writes \times \absGRSI.\Reads \cap \setcomp{(a, b)}{\src{b} = a}$
	\item $\absGRSI.\co = \impGRSI.\co$ %\coerce{\impGRSI.\co}{\absGRSI.\Events}$
	\item $\absGRSI.\rsilo = \emptyset$
\end{itemize}
Observe that the events of each $\trace'_{\txid}$ trace coincides with those of the equivalence class  $\Transactions_{\txid}$ of $\absGSI$. That is,  $\trace'_{\txid}.\Events = \Transactions_{\txid}$. 
\begin{lemma}\label{lem:rsi_lock_hb}
Given an RA-consistent execution graph $\impGRSI$ of the implementation and its corresponding RSI execution graph $\absGRSI$ constructed as above, for all $a, b, \txid_a, \txid_b, \x$:
\small
\begin{align}
	& \hspace*{-15pt} 
	\txid_a \ne \txid_b
	\land a \in \txid_a.\Events 
	\land b \in \txid_b.\Events 
	\land \loc a = \loc b = \x
	\Rightarrow \nonumber \\
	& \hspace*{-15pt} \;\; 
	((a, b) \in \absGRSI.\rf \Rightarrow  \txid_a.\mathit{wu}_{\x} \relarrow{\impGRSI.\hb} \txid_b.\mathit{rl}_{\x} ) 
	\label{lem:rsi_lock_hb_rf} \\
	& \hspace*{-15pt} \;\; 
	\land ((a, b) \in \absGRSI.\co \Rightarrow  \txid_a.\mathit{wu}_{\x} \relarrow{\impGRSI.\hb} \txid_b.\mathit{rl}_{\x} ) 
	\label{lem:rsi_lock_hb_co}\\
	& \hspace*{-15pt} \;\; 
	\land ((a, b) \in \absGRSI.\co;\rf \Rightarrow  \txid_a.\mathit{wu}_{\x} \relarrow{\impGRSI.\hb} \txid_b.\mathit{rl}_{\x} ) 
	\label{lem:rsi_lock_hb_corf}\\
	& \hspace*{-15pt} \;\; 
	\land \big((a, b) \in \absGRSI.\fr \Rightarrow  
		(\x \in \writeset_{\txid_a} \land \txid_a.\mathit{wu}_{\x} \relarrow{\impGRSI.\hb} \txid_b.\mathit{rl}_{\x} ) 
		\lor 
		(\x \not\in \writeset_{\txid_a} \land \txid_a.\mathit{ru}_{\x} \relarrow{\impGRSI.\hb} \txid_b.\mathit{pl}_{\x} ) 
	\big)
	\label{lem:rsi_lock_hb_fr}
\end{align}	
\normalsize
\begin{proof}
Pick an arbitrary RA-consistent execution graph $\impGRSI$ of the implementation and its corresponding RSI execution graph $\absGRSI$ constructed as above. Pick an arbitrary $a, b, \txid_a, \txid_b, \x$ such that $\txid_a \ne \txid_b$, $a \in \txid_a.\Events$, $a \in \txid_a.\Events$, and $\loc a = \loc b = \x$.\\

\noindent \textbf{RTS. (\ref{lem:rsi_lock_hb_rf})}\\
Assume $(a, b) \in \absGRSI.\rf$. From the definition of $\absGRSI.\rf$ we then know $(a, \txid_a.\mathit{rs}_{\x}), (a, \txid_a.\mathit{vs}_{\x}) \in \impGRSI.\rf$.
On the other hand, from \cref{lem:lock-ordering} we  know that either i) $\x \in \writeset_{\txid_b}$ and $\txid_b.\mathit{wu}_{x} \relarrow {\impGRSI.\hb} \txid_a.\mathit{rl}_{x}$; or ii)  $\x \not\in \writeset_{\txid_b}$ and  $\txid_b.\mathit{ru}_{x} \relarrow {\impGRSI.\hb} \txid_a.\mathit{pl}_{x}$; or iii) $\txid_a.\mathit{wu}_{x} \relarrow {\impGRSI.\hb} \txid_b.\mathit{rl}_{x}$.
In case (i) we then have $a \relarrow{\impGRSI.\rf} \txid_a.\mathit{rs}_{\x} \relarrow{\impGRSI.\po} \txid_b.\mathit{wu}_{x}  \relarrow{\impGRSI.\hb} \txid_a.\mathit{rl}_{x}  \relarrow{\impGRSI.\po} a$. That is, we have $a \relarrow{\impGRSI.\hbloc} a$, contradicting the assumption that $\impGRSI$ is RA-consistent. 
Similarly in case (ii) we have $a \relarrow{\impGRSI.\rf} \txid_a.\mathit{rs}_{\x} \relarrow{\impGRSI.\po} \txid_b.\mathit{ru}_{x}  \relarrow{\impGRSI.\hb} \txid_a.\mathit{pl}_{x}  \relarrow{\impGRSI.\po} a$.  That is, we have $a \relarrow{\impGRSI.\hbloc} a$, contradicting the assumption that $\impGRSI$ is RA-consistent. 
In case (iii) the desired result holds trivially.\\

\noindent \textbf{RTS. (\ref{lem:rsi_lock_hb_co})}\\
Assume $(a, b) \in \absGRSI.\co$. From the definition of $\absGRSI.\co$ we then know $(a, b) \in \impGRSI.\co$.
On the other hand, from \cref{lem:lock-ordering}  we  know that either i) $\txid_b.\mathit{wu}_{x} \relarrow {\impGRSI.\hb} \txid_a.\mathit{rl}_{x}$; or ii) $\txid_a.\mathit{wu}_{x} \relarrow {\impGRSI.\hb} \txid_b.\mathit{rl}_{x}$.
In case (i) we then have $a \relarrow{\impGRSI.\co} b \relarrow{\impGRSI.\po} \txid_b.\mathit{wu}_{x}  \relarrow{\impGRSI.\hb} \txid_a.\mathit{rl}_{x}  \relarrow{\impGRSI.\po} a$. That is, we have $a \relarrow{\impGRSI.\hbloc} a$, contradicting the assumption that $\impGRSI$ is RA-consistent. 
In case (ii) the desired result holds trivially.\\

\noindent \textbf{RTS. (\ref{lem:rsi_lock_hb_corf})}\\
Assume $(a, b) \in \absGRSI.\co; \rf$. We then know there exists $w$ such that $(a, w) \in \absGRSI.\co$ and $(w, b) \in \absGRSI.\rf$. From the definition of $\absGRSI.\co$ we then know $(a, w) \in \impGRSI.\co$.
There are now two cases to consider: 1) $w \in \txid_b$; or 2) $w \not\in \txid_b$. 
In case (1) since $(a, w) \in \absGRSI.\co$ the desired result follows from part \ref{lem:rsi_lock_hb_co}.

In case (2) from the definition of $\absGRSI.\rf$ we know that $(w, \txid_b.\mathit{rs}_{\x}) \in \impGRSI.\rf$. 
From \cref{lem:lock-ordering} we  know that either i) $\x \in \writeset_{\txid_b}$ and $\txid_b.\mathit{wu}_{x} \relarrow {\impGRSI.\hb} \txid_a.\mathit{rl}_{x}$; or ii)  $\x \not\in \writeset_{\txid_b}$ and  $\txid_b.\mathit{ru}_{x} \relarrow {\impGRSI.\hb} \txid_a.\mathit{pl}_{x}$; or iii) $\txid_a.\mathit{wu}_{x} \relarrow {\impGRSI.\hb} \txid_b.\mathit{rl}_{x}$.
In case (i) we then have $a \relarrow{\impGRSI.\co} w \relarrow{\impGRSI.\rf} \txid_a.\mathit{rs}_{\x} \relarrow{\impGRSI.\po} \txid_b.\mathit{wu}_{x}  \relarrow{\impGRSI.\hb} \txid_a.\mathit{rl}_{x}  \relarrow{\impGRSI.\po} a$. That is, we have $a \relarrow{\impGRSI.\co} w\relarrow{\impGRSI.\hbloc} a$, contradicting the assumption that $\impGRSI$ is RA-consistent. 
Similarly in case (ii) we have $a  \relarrow{\impGRSI.\co} w \relarrow{\impGRSI.\rf} \txid_a.\mathit{rs}_{\x} \relarrow{\impGRSI.\po} \txid_b.\mathit{ru}_{x}  \relarrow{\impGRSI.\hb} \txid_a.\mathit{pl}_{x}  \relarrow{\impGRSI.\po} a$.  That is, we have $a \relarrow{\impGRSI.\co} w\relarrow{\impGRSI.\hbloc} a$, contradicting the assumption that $\impGRSI$ is RA-consistent. 
In case (iii) the desired result holds trivially.\\

\noindent \textbf{RTS. (\ref{lem:rsi_lock_hb_fr})}\\
Assume $(a, b) \in \absGRSI.\fr$. From the definition of $\absGRSI.\fr$ we know $(\txid_a.\mathit{rs}_{\x}, b), (\txid_a.\mathit{vs}_{\x}, b) \in \impGRSI.\fr$.
On the other hand, from \cref{lem:lock-ordering} we  know that either i) $\txid_b.\mathit{wu}_{x} \relarrow {\impGRSI.\hb} \txid_a.\mathit{rl}_{x}$; or ii)  $\x \not\in \writeset_{\txid_a}$ and $\txid_a.\mathit{ru}_{x} \relarrow {\impGRSI.\hb} \txid_a.\mathit{pl}_{x}$; or iii) $\x \in \writeset_{\txid_a}$ and  $\txid_a.\mathit{wu}_{x} \relarrow {\impGRSI.\hb} \txid_b.\mathit{rl}_{x}$.
In case (i) we then have $b \relarrow{\impGRSI.\po} \txid_b.\mathit{wu}_{x}  \relarrow{\impGRSI.\hb} \txid_a.\mathit{rl}_{x}  \relarrow{\impGRSI.\po} \txid_a.\mathit{rs}_{\x} \relarrow{\impGRSI.\fr} b$. That is, we have $b \relarrow{\impGRSI.\hbloc} \txid_a.\mathit{rs}_{\x} \relarrow{\impGRSI.\fr} b$, contradicting the assumption that $\impGRSI$ is RA-consistent. 
In cases (ii-iii) the desired result holds trivially.\\
\end{proof}
\end{lemma}
Let 
\[
	\rsihb' \eqdef \rsihb \setminus (\rsifr; \refC{\rsihb})
\]
\begin{lemma}\label{lem:rsi_soundness}
For all RA-consistent execution graphs $\impGRSI$ of the implementation and their counterpart RSI execution graphs $\absGRSI$ constructed as above, 
\begin{enumerate}
	\item $(\absGRSI.\rsipo \subseteq \impGRSI.\hb) \land (\absGRSI.(\rsipo; \rsifr) \subseteq \impGRSI.\hb)$
	\label{lem:rsi_soundness_po}
	
	\item $(\absGRSI.\rsico \subseteq \impGRSI.\hb)  \land (\absGRSI.(\rsico; \rsifr) \subseteq \impGRSI.\hb)$
	\label{lem:rsi_soundness_co}
	
	\item $(\absGRSI.\rsirf \subseteq \impGRSI.\hb)  \land (\absGRSI.(\rsirf; \rsifr) \subseteq \impGRSI.\hb)$
	\label{lem:rsi_soundness_rf}	
	
%	\item $(\absGRSI.\rsifr \subseteq \impGRSI.\hb)$
%	\label{lem:rsi_soundness_fr}	
	
	\item $(\absGRSI.\rsihb' \subseteq \impGRSI.\hb)$
	\label{lem:rsi_soundness_hb}	
%	\item $\absGRSI.((\po \cup \cot \rft); \frt) \subseteq \impGRSI.\hb$
%	\label{lem:rsi_soundness_fr}
\end{enumerate}
\begin{proof}[Proof (Part \ref{lem:rsi_soundness_po})]
Pick an arbitrary RA-consistent execution graph $\impGRSI$ of the implementation and its counterpart RSI execution graph $\absGRSI$ constructed as above.
Pick arbitrary $(a, b) \in  \absGRSI.\rsipo$ and $(b, c) \in  \absGRSI.\rsifr$. There are then two cases to consider: 1) $(a, b) \not\in  \absGRSI.\rsipoi$; or 2) $(a, b) \in  \absGRSI.\rsipoi$.

In case (1) from the construction of $\rsipo$ we then have $(a, b) \in \impGRSI.\po \subseteq \impGRSI.\hb$, as required  by the first conjunct.
For the second conjunct, we know there exists $\txid_b, \txid_c, r, w$ such that $b, r \in \txid_b.\EReads$, $c, w \in \txid_c.\Writes$,  and $(r, w) \in \absGRSI.\fr$. Let $\loc r = \loc w = \x$. Given the definition of $\EReads$ and the construction of $\absGRSI$ we know  $(\txid_b.\mathit{rs}_{\x}, w) \in \impGRSI.\fr$. 

As we have $(a, b) \not\in  \absGRSI.\rsipoi$ and $b$ and $\txid_b.\mathit{rs}_{\x}$ are both transactional events in $\txid_b$, from the definition of $\absGRSI.\po$ and $\impGRSI.\po$ we then know $(a, \txid_b.\mathit{rs}_{\x}) \in \impGRSI.\po$.
On the other hand, from \cref{lem:rsi_lock_hb} we have either i) $\x \in \writeset_{\txid_b}$ and $\txid_b.\mathit{wu}_{x} \relarrow {\impGRSI.\hb} \txid_c.\mathit{rl}_{x}$; or ii)  $\x \not\in \writeset_{\txid_b}$ and  $\txid_b.\mathit{ru}_{x} \relarrow {\impGRSI.\hb} \txid_c.\mathit{pl}_{x}$.
In case (i) we have $a \relarrow{\impGRSI.\po} \mathit{rs}_{\x} \relarrow{\impGRSI.\po} \txid_a.\mathit{wu}_{x}  \relarrow{\impGRSI.\hb} \txid_b.\mathit{rl}_{x}  \relarrow{\impGRSI.\po} c$, i.e.~$a \relarrow{\impGRSI.\hb} c$, as required.
Similarly, in case (ii) have $a \relarrow{\impGRSI.\po} \mathit{rs}_{\x} \relarrow{\impGRSI.\po} \txid_a.\mathit{ru}_{x}  \relarrow{\impGRSI.\hb} \txid_b.\mathit{pl}_{x}  \relarrow{\impGRSI.\po} c$, i.e.~$a \relarrow{\impGRSI.\hb} c$, as required.

In case (2) from the definition of $ \absGRSI.\rsipoi$ we then know that $a, b \in \absGRSI.\Writes$ and thus from the construction of  $\absGRSI.\rsipoi$ we have $(a, b) \in \impGRSI.\po \subseteq \impGRSI.\hb$, as required by the first conjunct. 
For the second conjunct, since we have $b \in \absGRSI.\Writes$, we cannot have $(b, c) \in \absGRSI.\rsifr$ and thus the desired result holds vacuously.\\

\renewcommand{\qed}{}
\end{proof}

\begin{proof}[Proof (Part \ref{lem:rsi_soundness_co})]
Pick an arbitrary RA-consistent execution graph $\impGRSI$ of the implementation and its counterpart RSI execution graph $\absGRSI$ constructed as above.
For the first conjunct pick arbitrary $(a, b) \in \absGRSI.\rsico$. From the definition of $\absGRSI.\rsico$ we then know there exists $w, w', \txid_a, \txid_b$ such that $\txid_a \ne \txid_b$, $w, a \in \txid_a.\Events$, $w', b \in \txid_b.\Events$ and $(w, w') \in \absGRSI.\co$. From the definition of $\absGRSI.\co$ we then have $(w, w') \in \impGRSI.\co$.
Let $\loc w = \loc {w'} = \x$.
From \cref{lem:rsi_lock_hb} we  then know that $\txid_a.\mathit{wu}_{x} \relarrow {\impGRSI.\hb} \txid_b.\mathit{rl}_{x}$. We then have $a \relarrow{\impGRSI.\po} \txid_a.\mathit{wu}_{x}  \relarrow{\impGRSI.\hb} \txid_b.\mathit{rl}_{x}  \relarrow{\impGRSI.\po} b$, i.e.~$a \relarrow{\impGRSI.\hb} b$, as required.\\

For the second conjunct pick arbitrary $c$ such that $(b, c) \in \rsifr$. From the definition of $\rsifr$ and the construction of $\absGRSI$ we then know there exist $\y, r_y, w_y, \txid_c$ such that $\txid_c \ne \txid_b$, $\loc{w_y} = \loc{r_y} = \y$, $(r_y, w_y) \in \absGRSI.\fr$, and $w_y, c \in \txid_c.\Writes$. 
As we demonstrated for the first conjunct we have $a \relarrow{\impGRSI.\po} \txid_a.\mathit{wu}_{x}  \relarrow{\impGRSI.\hb} \txid_b.\mathit{rl}_{x}$. That is, $a  \relarrow{\impGRSI.\hb} \txid_b.\mathit{rl}_{x}$.
On the other hand, from  \cref{lem:rsi_lock_hb} we then know that either i) $\y \in \writeset_{\txid_b}$ and $\txid_b.\mathit{wu}_{y} \relarrow {\impGRSI.\hb} \txid_c.\mathit{rl}_{y}$; or ii)  $\y \not\in \writeset_{\txid_b}$ and  $\txid_b.\mathit{ru}_{y} \relarrow {\impGRSI.\hb} \txid_c.\mathit{pl}_{y}$.
In case (i) we have $a \relarrow{\impGRSI.\hb} \txid_b.\mathit{rl}_{x} \relarrow{\impGRSI.\po}  \txid_a.\mathit{wu}_{y}  \relarrow{\impGRSI.\hb} \txid_b.\mathit{rl}_{y}  \relarrow{\impGRSI.\po} c$, i.e.~$a \relarrow{\impGRSI.\hb} c$, as required.
Similarly, in case (ii) we have $a \relarrow{\impGRSI.\hb} \txid_b.\mathit{rl}_{x} \relarrow{\impGRSI.\po} \txid_a.\mathit{ru}_{y}  \relarrow{\impGRSI.\hb} \txid_b.\mathit{pl}_{y}  \relarrow{\impGRSI.\po} c$, i.e.~$a \relarrow{\impGRSI.\hb} c$, as required.\\

\renewcommand{\qed}{}
\end{proof}

\begin{proof}[Proof (Part \ref{lem:rsi_soundness_rf})]
Pick an arbitrary RA-consistent execution graph $\impGRSI$ of the implementation and its counterpart RSI execution graph $\absGRSI$ constructed as above.
It suffices to show that:
\begin{align}
	& \absGRSI.([\NT]; \rf; \st) \subseteq \impGRSI.\hb \label{lem:rsi_soundness_ntrf} \\
	& \absGRSI.([\NT]; \rf; \st); \absGRSI.\rsifr \subseteq \impGRSI.\hb \label{lem:rsi_soundness_ntrf_fr} \\
	& \absGRSI.\rft \subseteq \impGRSI.\hb \label{lem:rsi_soundness_rft} \\
	& \absGRSI.\rft; \absGRSI.\rsifr \subseteq \impGRSI.\hb \label{lem:rsi_soundness_rft_fr} \\
	& \absGRSI.(\tlift{\co; \rf)} \subseteq \impGRSI.\hb \label{lem:rsi_soundness_corf} \\
	& \absGRSI.(\tlift{\co; \rf)}; \absGRSI.\rsifr \subseteq \impGRSI.\hb \label{lem:rsi_soundness_corf_fr} \\
	& \absGRSI.(\rf; [\NT]) \subseteq \impGRSI.\hb \label{lem:rsi_soundness_rf2} \\
	& \absGRSI.(\rf; [\NT]); \absGRSI.\rsifr \subseteq \impGRSI.\hb \label{lem:rsi_soundness_rf2_fr} 
\end{align}

\noindent \textbf{RTS. (\ref{lem:rsi_soundness_ntrf})}\\
Pick arbitrary $(w, a) \in  \absGRSI.([\NT]; \rf; \st)$ where $a \in \txid_a.\Events$. 
Let $\loc w = \x$. From the construction of $\absGRSI.\rf$ we then know $(w, \txid_a.\mathit{rs}_{\x}), (w, \txid_a.\mathit{vs}_{\x}) \in \impGRSI.\rf$ and  $(\txid_a.\mathit{rs}_{\x}, a), (\txid_a.\mathit{vs}_{\x}, a) \in \impGRSI.\po$. As such we have $w \relarrow{\impGRSI.\rf} \txid_a.\mathit{rs}_{\x} \relarrow{\impGRSI.\po} a$, i.e.~$w \relarrow{\impGRSI.\hb} a$, as required.\\

\noindent \textbf{RTS. (\ref{lem:rsi_soundness_ntrf_fr})}\\
Pick arbitrary $(w, c) \in  \absGRSI.([\NT]; \rf; \st); \absGRSI.\rsifr$. We then know there exist $a, \txid_a , \txid_c$ such that $(w, a) \in  \absGRSI.([\NT]; \rf; \st)$, $(a, c) \in \absGRSI.\rsifr$. $a \in \txid_a.\Events$, $c \in \txid_c.\Events$ and $\txid_a \ne \txid_c$.
Let $\loc w = \x$.
As we demonstrated in the previous part we then know $(w, \txid_a.\mathit{rs}_{\x}), (w, \txid_a.\mathit{vs}_{\x}) \in \impGRSI.\rf$ and  $(\txid_a.\mathit{rs}_{\x}, a), (\txid_a.\mathit{vs}_{\x}, a) \in \impGRSI.\po$.
Moreover, from the definition of $\rsifr$ and the construction of $\absGRSI$ we know there exist $\y, r_y, w_y, \txid_c$ such that $\loc{w_y} = \loc{r_y} = \y$, $(r_y, w_y) \in \absGRSI.\fr$, and $w_y, c \in \txid_c.\Writes$. 
From the construction of $\absGRSI.\fr$ we then know that $(\txid_a.\mathit{rs}_{\y}, w_y), (\txid_a.\mathit{vs}_{\y}, w_y) \in \impGRSI.\fr$. 
On the other hand, from  \cref{lem:rsi_lock_hb}  we know that either i) $\y \in \writeset_{\txid_a}$ and $\txid_a.\mathit{wu}_{y} \relarrow {\impGRSI.\hb} \txid_c.\mathit{rl}_{y}$; or ii)  $\y \not\in \writeset_{\txid_a}$ and  $\txid_a.\mathit{ru}_{y} \relarrow {\impGRSI.\hb} \txid_c.\mathit{pl}_{y}$.
In case (i) we have $w \relarrow{\impGRSI.\rf} \txid_a.\mathit{rs}_{\x}  \relarrow{\impGRSI.\po} \txid_a.\mathit{vs}_{\y} \relarrow{\impGRSI.\po} \txid_a.\mathit{wu}_{y}  \relarrow{\impGRSI.\hb} \txid_b.\mathit{rl}_{y}  \relarrow{\impGRSI.\po} c$, i.e.~$a \relarrow{\impGRSI.\hb} c$, as required.
Similarly, in case (ii) we have $w \relarrow{\impGRSI.\rf} \txid_a.\mathit{rs}_{\x}  \relarrow{\impGRSI.\po} \txid_a.\mathit{vs}_{\y} \relarrow{\impGRSI.\po} \txid_a.\mathit{ru}_{y}  \relarrow{\impGRSI.\hb} \txid_b.\mathit{pl}_{y}  \relarrow{\impGRSI.\po} c$, i.e.~$w \relarrow{\impGRSI.\hb} c$, as required.\\

\noindent \textbf{RTS. (\ref{lem:rsi_soundness_rft})}\\
Pick arbitrary $(a, b) \in \absGRSI.\rft$. We then know there exist $w, r, \txid_a, \txid_b$ such that$w, a \in \txid_a.\Events$, $r, b \in \txid_b.\Events$, $\txid_a \ne \txid_b$ and $(w, r) \in \absGRSI.\rf$. 
Let $\loc w = \loc r = x$. 
From \cref{lem:rsi_lock_hb} we then know that $\txid_a.\mathit{wu}_{x} \relarrow {\impGRSI.\hb} \txid_b.\mathit{rl}_{x}$.
Consequently from the structure of $\impGRSI$ and the construction of $\absGRSI$ we have $a \relarrow{\impGRSI.\po} \txid_a.\mathit{wu}_{x}  \relarrow{\impGRSI.\hb} \txid_b.\mathit{rl}_{x}  \relarrow{\impGRSI.\po} b$, i.e.~$a \relarrow{\impGRSI.\hb} b$, as required.\\

\noindent \textbf{RTS. (\ref{lem:rsi_soundness_rft_fr})}\\
Pick arbitrary $(a, c) \in \absGRSI.\rft; \absGRSI.\rsifr$. $b$ such that $(a, b) \in \absGRSI.\rft$ and $(b, c) \in \absGRSI.\rsifr$. 
From the definition of $\absGRSI.\rft$ we then know there exist $w, r, \txid_a, \txid_b$ such that$w, a \in \txid_a.\Events$, $r, b \in \txid_b.\Events$, $\txid_a \ne \txid_b$ and $(w, r) \in \absGRSI.\rf$. 
Let $\loc w = \loc r = x$. 
From \cref{lem:rsi_lock_hb} we then know that $\txid_a.\mathit{wu}_{x} \relarrow {\impGRSI.\hb} \txid_b.\mathit{rl}_{x}$.

On the other hand, from the definition of $\absGRSI.\rsifr$ we know there exist $r', w', \txid_c$ such that $r' \in \txid_b.\Events$, $w', c \in \txid_c.\Writes$, $\txid_c \ne \txid_b$ and $(r', w') \in \absGRSI.\fr$. 
Let $\loc{w'} = \loc{r'} = \y$. 
From \cref{lem:rsi_lock_hb} we then know that either i) $\y \in \writeset_{\txid_b}$ and $\txid_b.\mathit{wu}_{y} \relarrow {\impGRSI.\hb} \txid_c.\mathit{rl}_{y}$; or ii)  $\y \not\in \writeset_{\txid_b}$ and  $\txid_b.\mathit{ru}_{y} \relarrow {\impGRSI.\hb} \txid_c.\mathit{pl}_{y}$.

In case (i) we have $a \relarrow{\impGRSI.\po} \txid_a.\mathit{wu}_{x} \relarrow {\impGRSI.\hb} \txid_b.\mathit{rl}_{x} \relarrow{\impGRSI.\po} \txid_b.\mathit{wu}_{y} \relarrow {\impGRSI.\hb} \txid_c.\mathit{rl}_{y} \relarrow{\impGRSI.\po} c$. That is, $(a, c) \in \impGRSI.\hb$, as required.
Similarly, In case (ii) we have $a \relarrow{\impGRSI.\po} \txid_a.\mathit{wu}_{x} \relarrow {\impGRSI.\hb} \txid_b.\mathit{rl}_{x} \relarrow{\impGRSI.\po} \txid_b.\mathit{ru}_{y} \relarrow {\impGRSI.\hb} \txid_c.\mathit{pl}_{y} \relarrow{\impGRSI.\po} c$. That is, $(a, c) \in \impGRSI.\hb$, as required.\\

\noindent \textbf{RTS. (\ref{lem:rsi_soundness_corf})}\\
Pick arbitrary $(a, b) \in \absGRSI.\tlift{(\co; \rf)}$. We then know there exist $\txid_a, \txid_b, w, w', r$ such that $w, a \in \txid_a.\Events$, $r, b \in \txid_b.\Events$, $\txid_a \ne \txid_b$, $(w, w') \in \absGRSI.\co$ and $(w', r) \in \absGRSI.\rf$.

From \cref{lem:rsi_lock_hb} we then have $\txid_a.\mathit{wu}_{x}  \relarrow{\impGRSI.\hb} \txid_b.\mathit{rl}_{x}$. Consequently from the structure of $\impGRSI$ and the construction of $\absGRSI$ we have $a \relarrow{\impGRSI.\po} \txid_a.\mathit{wu}_{x}  \relarrow{\impGRSI.\hb} \txid_b.\mathit{rl}_{x}  \relarrow{\impGRSI.\po} b$, i.e.~$a \relarrow{\impGRSI.\hb} b$, as required.\\

\noindent \textbf{RTS. (\ref{lem:rsi_soundness_corf_fr})}\\
Pick arbitrary $(a, c) \in \absGRSI.\tlift{(\co; \rf)}; \absGRSI.\rsifr$. We then know there exist $b$ such that $(a, b) \in  \absGRSI.\tlift{(\co; \rf)}$ and $(b, c) \in \absGRSI.\rsifr$.
From the definition of $\absGRSI.\tlift{(\co; \rf)}$ we then know there exist
$\txid_a, \txid_b, w, w', r$ such that $w, a \in \txid_a.\Events$, $r, b \in \txid_b.\Events$, $\txid_a \ne \txid_b$, $(w, w') \in \absGRSI.\co$ and $(w', r) \in \absGRSI.\rf$.
From \cref{lem:rsi_lock_hb} we then have $\txid_a.\mathit{wu}_{x}  \relarrow{\impGRSI.\hb} \txid_b.\mathit{rl}_{x}$.

On the other hand, from the definition of $\absGRSI.\rsifr$ we know there exist $r'', w'', \txid_c$ such that $r'' \in \txid_b.\Events$, $w'', c \in \txid_c.\Writes$, $\txid_c \ne \txid_b$ and $(r'', w'') \in \absGRSI.\fr$. 
Let $\loc{w''} = \loc{r''} = \y$. 
From \cref{lem:rsi_lock_hb} we then know that either i) $\y \in \writeset_{\txid_b}$ and $\txid_b.\mathit{wu}_{y} \relarrow {\impGRSI.\hb} \txid_c.\mathit{rl}_{y}$; or ii)  $\y \not\in \writeset_{\txid_b}$ and  $\txid_b.\mathit{ru}_{y} \relarrow {\impGRSI.\hb} \txid_c.\mathit{pl}_{y}$.

In case (i) we then have $a \relarrow{\impGRSI.\po} \txid_a.\mathit{wu}_{x} \relarrow {\impGRSI.\hb} \txid_b.\mathit{rl}_{x} \relarrow{\impGRSI.\po} \txid_b.\mathit{wu}_{y} \relarrow {\impGRSI.\hb} \txid_c.\mathit{rl}_{y} \relarrow{\impGRSI.\po} c$. That is, $(a, c) \in \impGRSI.\hb$, as required.
Similarly, In case (ii) we then have $a \relarrow{\impGRSI.\po} \txid_a.\mathit{wu}_{x} \relarrow {\impGRSI.\hb} \txid_b.\mathit{rl}_{x} \relarrow{\impGRSI.\po} \txid_b.\mathit{ru}_{y} \relarrow {\impGRSI.\hb} \txid_c.\mathit{pl}_{y} \relarrow{\impGRSI.\po} c$. That is, $(a, c) \in \impGRSI.\hb$, as required.\\

\noindent \textbf{RTS. (\ref{lem:rsi_soundness_rf2})}\\
Pick arbitrary $(w, r) \in \absGRSI.(\rf; [\NT])$.
As $r \in \absGRSI.\NT$, from the construction of $\absGRSI.\rf$ we have $(w, r) \in \impGRSI.\rf \subseteq \impGRSI.\hb$, as required.\\
%there are then five cases to consider: 
%1) there exist $\txid_a$ such that $w, r \in \txid_a.\Events$; or 
%2)  there exist $\txid_a, \txid_b$ such that $\txid_a \not= \txid_b$, $w \in \txid_a.\Events$ and $r \in \txid_b.\Events$; or 
%3) there exist $\txid_a$ such that $w \in \txid_a.\Events$ and $r \in \absGRSI.\NT$; or 
%4) there exist $\txid_b$ such that $w \in \absGRSI.\NT$ and $r \in \txid_a.\Events$; or 
%5) $w, r \in \absGRSI.\NT$.\\

%\noindent Case (1): from the definition of $\absGRSI.\rf$ we have $(w, r) \in \impGRSI.\poi \subseteq \impGRSI.\hb$, as required.
%
%\noindent Case (2): we have $(w, r) \in \absGRSI.\rft$ and thus from part (\ref{lem:rsi_soundness_rft}) we have $(w, r) \in \impGRSI.\hb$, as required. 

%\noindent Case (3): from the construction of $\absGRSI.\rf$ we have $(w, r) \in \impGRSI.\rf \subseteq \impGRSI.\hb$, as required.
%
%\noindent Case (4): we have $(w, r) \in \absGRSI.([\NT]; \rf; \st)$ and thus from the proof of part (\ref{lem:rsi_soundness_ntrf}) we have $(w, r) \in \impGRSI.\hb$, as required. 
%
%\noindent Case (5): from the construction of $\absGRSI.\rf$  we have $(w, r) \in \impGRSI.\rf \subseteq \impGRSI.\hb$, as required.\\

\noindent \textbf{RTS. (\ref{lem:rsi_soundness_rf2_fr})}\\
The desired result holds trivially as $\absGRSI.(\rf; [\NT]); \absGRSI.\rsifr = \emptyset$.\\

\renewcommand{\qed}{}
\end{proof}

\begin{proof}[Proof (part \ref{lem:rsi_soundness_hb})]
Let $\rsihb_0 \eqdef \rsipo \cup \rsirf \cup \rsico \cup \rsifr$ and $\rsihb_{n {+} 1} \eqdef \rsihb_0; \rsihb_n$,  for all $n \geq 0$.
Similarly, let $\rsihb'_0 \eqdef \rsihb_0 \setminus \rsifr$ and $\rsihb'_{n {+} 1} \eqdef \rsihb'_0; \rsihb_n$, for all $n \geq 0$. 
It is then straightforward to demonstrate that $\rsihb' \eqdef \bigcup \limits_{n \in \Nats} \rsihb'_n$. 
It thus suffices to show that:
\[
	\for{n \in \Nats} \rsihb'_n \subseteq \impGRSI.\hb
\]
We proceed by induction on $n$.\\

\noindent\textbf{Base case $n = 0$}\\
Follows immediately from the definition of $\rsihb'_0$ and the results established in \ref{lem:rsi_soundness_po}-\ref{lem:rsi_soundness_rf}.\\

\noindent\textbf{Inductive case $n = m {+} 1$}
\begin{align}
	\for{i \in \Nats} i < n \Rightarrow \rsihb'_i \subseteq \impGRSI.\hb
	\tag{I.H.}
	\label{IH:rsi_soundness_lemma}
\end{align}
Pick arbitrary $(a, b) \in \rsihb'_{n}$. From the definition of $\rsihb'_n$ we then know there exists $c$ such that $(a, c) \in \rsihb'_0$ and  $(c, b) \in \rsihb_m$. 
Let $\rsihb_{-1} \eqdef \makerel{id}$. There are now two cases to consider: 1) $(c, b) \in \rsihb'_0; \rsihb_{m{-}1}$; or $(c, b) \in \rsifr; \rsihb_{m {-} 1}$. 
In case (1) from the proof of base case we have $(a, c) \in \impGRSI.\hb$. On the other hand from the definition of $\rsihb'_m$ we have $(c, b) \in \rsihb'_m$ and thus from (\ref{IH:rsi_soundness_lemma}) we have $(c, b) \in \impGRSI.\hb$. Consequently, since $\impGRSI.\hb$ is transitively closed we have $(a, b) \in \impGRSI.\hb$ as required. 

In case (2)  we know there exists $d$ such that $(c, d) \in \rsifr$; and $(d, b) \in \rsihb_{m {-} 1}$. Since we have $(a, c) \in \rsihb'_0$ and $(c, d) \in \rsifr$, from the definition of $\rsihb'_0$ and the proofs of parts \ref{lem:rsi_soundness_po}-\ref{lem:rsi_soundness_rf} we have $(a, d) \in \impGRSI.\hb$. Moreover we either have i) $m=0$; or ii) $m > 0$. In case  (2.i) since $(d, b) \in \rsihb_{m {-} 1}$, from the definition of $\rsihb_{m {-} 1}$  we have $b=d$ and thus $(a, b) \in \impGRSI.\hb$, as required. 
In case (2.ii) from (\ref{IH:rsi_soundness_lemma}) we then have $(d, b) \in \impGRSI.\hb$. As such, since $\impGRSI.\hb$ is transitively closed, we have $(a, b) \in \impGRSI.\hb$, as required. 

\end{proof}

\end{lemma}
\begin{lemma}\label{lem:rsi_soundness_aux}
For all RA-consistent execution graphs $\impGRSI$ of the implementation and their counterpart RSI execution graphs $\absGRSI$ constructed as above:
\[
\begin{array}{@{} l @{}}
	\for{n \in \Nats} 
	\for{a, b, w, \txid, \x} \\
	\quad (a, b) \in \absGRSI.\rsihb_n
	\land a \in \absGRSI.\Writes
	\land b \in \absGRSI.(\Reads \cap \Transactions_{\txid}) \Rightarrow \\
	\qquad  
	(w \in \absGRSI.(\Writes \cap \Transactions_{\txid})  \Rightarrow
		 (a, w) \in \impGRSI.\hb)
	\land 
	(\txid.\mathit{vs}_{\x} \text{ defined }  \Rightarrow
		(a, \txid.\mathit{vs}_{\x}) \in \impGRSI.\hb)
\end{array}	
\]
where $\rsihb_0 \eqdef \rsipo \cup \rsirf \cup \rsico \cup \rsifr$ and $\rsihb_{n {+} 1} \eqdef \rsihb_n; \rsihb_0$,  for all $n \geq 0$.
\begin{proof}
Pick an arbitrary RA-consistent execution graph $\impGRSI$ of the implementation and its counterpart RSI execution graph $\absGRSI$ constructed as above.
We then proceed by induction on $n$.\\

\noindent\textbf{Base case $n = 0$}\\
Pick arbitrary $a, b, \txid, \x$ such that $(a, b) \in (\rsihb_0)$, $a \in \absGRSI.\Writes$, $b \in \absGRSI.(\Reads \cap\Transactions_{\txid})$, $\txid.\mathit{vs}_{\x} \text{ defined}$ and $w \in \absGRSI.(\Writes \cap \Transactions_{\txid})$.
From the definition of  $\rsihb_0$ we then have either 
1) $(a, b) \in \absGRSI.\rsipo \land a \not\in \absGRSI.\Transactions_{\txid}$; or 2) $(a, b) \in \absGRSI.(\cot \cup \rft \cup  \tlift{(\co; \rf)})$; or 3) $(a, b) \in \absGRSI.([\NT]; \rf; \st)$.

In case (1) from the definition of $\absGRSI.\rsipo$ we then know that $\{a\} \times \trace_{\txid}.\Events \subseteq \impGRSI.\po$ and thus $(a, \txid.\mathit{vs}_{\x}), (a, w) \in \impGRSI.\po \subseteq \impGRSI.\hb$, as required.

In case (2) from the definitions of $\rft$, $\cot$, $\tlift{(\co; \rf)}$ we know there exists $\txid', c, d, \y$ such that $a, c \in \absGRSI.\Transactions_{\txid'}$, $a, c \in \txid'.\Events$, $b, d \in \absGRSI.\Transactions_{\txid}$, $b, d \in \txid.\Events$, $(c, d) \in \absGRSI.(\co \cup \rf \cup  (\co; \rf))$ and $\loc c = \loc d = \y$. 
As such, from \cref{lem:rsi_lock_hb} we know $\txid'.\mathit{wu}_{\y} \relarrow{\impGRSI.\hb} \txid.\mathit{rl}_{\y}$. 
On the other hand we have $a \relarrow{\impGSI.\po} \txid'.\mathit{wu}_{\y}$, $\txid.\mathit{rl}_{\y}  \relarrow{\impGSI.\po} \txid.\mathit{vs}_{\x}$ and $\txid.\mathit{rl}_{\y}  \relarrow{\impGSI.\po} w$. As such we have $a \relarrow{\impGSI.\po} \txid'.\mathit{wu}_{\y} \relarrow{\impGSI.\hb} \txid.\mathit{rl}_{\y}  \relarrow{\impGSI.\po} \txid.\mathit{vs}_{\x}$, i.e.~$a  \relarrow{\impGSI.\hb}  \txid.\mathit{vs}_{\x}$, as required.
Similarly we have $a \relarrow{\impGSI.\po} \txid'.\mathit{wu}_{\y} \relarrow{\impGSI.\hb} \txid.\mathit{rl}_{\y}  \relarrow{\impGSI.\po} w$, i.e.~$a  \relarrow{\impGSI.\hb}  w$, as required.

In case (3) from the construction of $\absGRSI$ we know there exists $\y$ such that $\loc a = \y$, and $(a, \txid.\mathit{rs}_{\y}), (a, \txid.\mathit{vs}_{\y}) \in \impGRSI.\rf$. As such we have $a \relarrow{\impGSI.\rf} \txid.\mathit{rs}_{\y}  \relarrow{\impGSI.\po} \txid.\mathit{vs}_{\x}$, i.e.~$a  \relarrow{\impGSI.\hb}  \txid.\mathit{vs}_{\x}$, as required.
Moreover, since we have $\txid.\mathit{rs}_{\y} \relarrow{\impGSI.\po} w$, we have $a  \relarrow{\impGSI.\hb}  \txid.\mathit{rs}_{\y} \relarrow{\impGSI.\po} w$, i.e.~$a  \relarrow{\impGSI.\hb} w$, as required.\\

\noindent\textbf{Inductive case $n = m {+} 1$}\\
Pick arbitrary $a, b, \txid, \x$ such that $(a, b) \in (\rsihb_n)$, $a \in \absGRSI.\Writes$, $b \in \absGRSI.(\Reads \cap\Transactions_{\txid})$, $\txid.\mathit{vs}_{\x} \text{ defined}$ and $w \in \absGRSI.(\Writes \cap \Transactions_{\txid})$.
%%
%\begin{align}
%\begin{array}{@{} l @{}}
%	\for{n \in \Nats} 
%	\for{a, b, \txid, \x} \\
%	\quad (a, b) \in \absGRSI.\rsihb_n
%	\land a \in \absGRSI.\Writes
%	\land b \in \absGRSI.(\Reads \cap \Transactions_{\txid}) 
%	\land \txid.\mathit{vs}_{\x} \text{ defined } 
%	\land i \leq m
%	\Rightarrow
%	(a, \txid.\mathit{vs}_{\x}) \in \impGRSI.\hb
%	\tag{I.H.}
%\end{array}	
%	\label{IH:rsi_soundness_aux}	
%\end{align}
%%
From the definition of $\rsihb_n$ we then know there exists $c$ such that $(a, c) \in \rsihb_m$ and $(c, b) \in \rsihb_0$. There are then two cases to consider: 1) $c \in \Transactions_\txid$; or 2) $c \not\in \Transactions_\txid$. 

Case (1) leads to contradiction as $(c, b) \in \rsihb_0$, $c, b \in \Transactions_\txid$, $b \in \Reads$ and $\rsihb_0$ does not include any internal edges to read events. 

In case (2) since $a$ is a write event, from the definition of $\rsihb'$ we then have $(a, c) \in \rsihb'$ and thus from \cref{lem:rsi_soundness} we have  $(a, c) \in \impGRSI.\hb$. 
On the other hand, since $c \not\in \Transactions_\txid$, $(c, b) \in \rsihb_0$ and $b$ is a read event, from the definition of $\rsihb_0$ we know that either 1) $(c, b) \in \absGRSI.\rsipo \land c \not \in \Transactions_\txid$; or 2) $(c, b) \in \absGRSI.(\cot \cup \rft \cup  \tlift{(\co; \rf)})$; or 3) $(c, b) \in \absGRSI.([\NT]; \rf; \st)$.
Following an analogous argument as that in the base case, we then have $c \relarrow{\impGSI.\hb}  \txid.\mathit{vs}_{\x}$ and $c \relarrow{\impGSI.\hb} w$. As such, we have $a \relarrow{\impGSI.\hb} c \relarrow{\impGSI.\hb} \txid.\mathit{vs}_{\x}$, i.e.~$a \relarrow{\impGSI.\hb} \txid.\mathit{vs}_{\x}$, as required.
Similarly, we have $a \relarrow{\impGSI.\hb} c \relarrow{\impGSI.\hb} w$, i.e.~$a \relarrow{\impGSI.\hb} w$, as required.
\end{proof}
\end{lemma}

\begin{theorem}[Soundness]
For all execution graphs $\impGRSI$ of the implementation and their counterpart RSI execution graphs $\absGRSI$ constructed as above,
\[
	\consistent{\impGRSI} \Rightarrow \rsicon
\]
\begin{proof}
Pick an arbitrary execution graph $\impGRSI$ of the implementation such that $\consistent{\impGRSI}$, and its associated RSI execution graph $\absGRSI$ constructed as described above. It then suffices to show 1) $\rfi \cup \coi \cup \fri \subseteq \po$; 2) $\acyc{\absGRSI.\rsihbloc}$; 3) $\acyc{\absGRSI.(\rsihbloc; \co)}$; and 4) $\acyc{\absGRSI.(\rsihbloc; \fr)}$. \\

\noindent \textbf{RTS. $\rfi \cup \coi \cup \fri \subseteq \po$}\\
Follows immediately from the construction of $\absGRSI$ and the RA-consistency of $\impGRSI$.\\

\noindent \textbf{RTS. $\acyc{\absGRSI.\rsihbloc}$}\\ 
We proceed by contradiction. Let us assume there exists $a$ such that $(a, a) \in \absGRSI.\rsihbloc$. 
There are then two cases to consider: i) $(a, a) \in \rsihb'$; or ii) $(a, a) \in \rsihb \setminus \rsihb'$. 
In case (1) from \cref{lem:rsi_soundness} part \ref{lem:rsi_soundness_hb} we then have $(a, a) \in \impGRSI.\hb$, contradicting the assumption that $\impGRSI$ is RA-consistent. 
In case (2) we then know there exists $b$ such that $a \in \absGRSI.\EReads$, $b \in \absGRSI.\Writes$, $(a, b) \in \rsifr$ and $(b, a) \in \rsihb$. Moreover, since $b \in \absGRSI.\Writes$, from the definition of $\rsifr$ and $\rsihb'$ we have $(b, a) \in \rsihb'$. 
Since we also have $(a, b) \in \rsifr \subseteq \rsihb$, from the definition  $\rsihb'$ we have $(b, b) \in \rsihb'$. Consequently, from \cref{lem:rsi_soundness} part \ref{lem:rsi_soundness_hb} we have $(b, b) \in \impGRSI.\hb$, contradicting the assumption that $\impGRSI$ is RA-consistent. 
\\

\noindent \textbf{RTS. $\acyc{\absGRSI.(\rsihbloc; \co)}$}\\ 
We proceed by contradiction. Let us assume there exist $a, b$ such that $(a, b) \in \absGRSI.\rsihbloc$ and $(b, a) \in \absGRSI.\co$. 
From the definition of $\absGRSI.\co$ we then know that $a, b \in \absGRSI.\Writes$ and $(a, b) \in \impGRSI.\co$. 
On the other hand, since $a  \in \absGRSI.\Writes$ and $(a, b) \in \rsihb$, from the definition of $\rsihb'$ we have $(a, b) \in \rsihb'$. Consequently, from \cref{lem:rsi_soundness} part \ref{lem:rsi_soundness_hb} we have $(a, b) \in \impGRSI.\hb$. We then have $a \relarrow{\impGRSI.\hb} b \relarrow{\impGRSI.\co} a$, contradicting the assumption that $\impGRSI$ is RA-consistent. \\

\noindent \textbf{RTS. $\acyc{\absGRSI.(\rsihbloc; \fr)}$}\\ 
We proceed by contradiction. Let us assume there exist $w, r$ such that $(w, r) \in \absGRSI.\rsihbloc$ and $(r, w) \in \absGRSI.\fr$. 
From the definition of $\absGRSI.\fr$ we then know that $w \in \absGRSI.\Writes$, $r \in \absGRSI.\Reads$. Since $w \in \absGRSI.\Writes$, from the definition of $\rsihb'$ we have $(w, r) \in \rsihb'$. 
Consequently, from \cref{lem:rsi_soundness} part \ref{lem:rsi_soundness_hb} we have $(w, r) \in \impGRSI.\hb$. 
Let $\loc w = \loc r = \x$.
There are then  two cases to consider: 1) $r \in \absGRSI.\NT$; or 2) $\exsts{\txid_r} r \in \txid_r.\Events$.

In case (1) from the definition of $\absGRSI.\fr$ we know $(r, w) \in \impGRSI.\fr$.  
As such, we have $w \relarrow{\impGRSI.\hb} r \relarrow{\impGRSI.\fr} w$, contradicting the assumption that $\impGRSI$ is RA-consistent. 

In case (2) we then know there exists $w'$ such that $(w', r) \in \absGRSI.\rf$ $(w', w) \in \absGRSI.\co$.
There are now three cases to consider: a) $w \in \txid_r.\Events$; or 
b) $w \not\in \txid_r.\Events \land w' \in \txid_r.\Events$; or 
c) $w, w' \not\in \txid_r.\Events$. 
In case (2.a.) from the construction of $\absGRSI$ (in particular, $\absGRSI.\rf$) we know that $(r, w) \in \absGRSI.\poi \subseteq \impGRSI.\poi \subseteq \impGRSI.\hb$. As such, we have $w \relarrow{\impGRSI.\hb} r \relarrow{\impGRSI.\hb} w$, contradicting our assumption that $\impGRSI$ is RA-consistent. 

In case (2.b), 
%from  the construction of $\absGRSI.\rf$ we then know that $(w', r) \in \impGRSI.\po$. On the other hand, 
since $w'$ is a write event, $w' \in \txid_r.\Events$, $r$ is a read event, $w$ is a write event and $(w, r) \in \absGRSI.\rsihb$, from \cref{lem:rsi_soundness_aux} we have $(w, w') \in \impGRSI.\hb$. 
Moreover, since $(w', w) \in \absGRSI.\co$, from the definition of $\absGRSI.\co$ we also have $(w', w) \in \impGRSI.\co$.
As such we have $w \relarrow{\impGRSI.\hb} w' \relarrow{\impGRSI.\co} w$, contradicting the assumption that $\impGRSI$ is RA-consistent. 

In case (2.c), from the construction of $\absGRSI$ we then know $(\txid_r.\mathit{rs}_\x, w) \in \impGRSI.\fr$ and $(\txid_r.\mathit{vs}_\x, w) \in \impGRSI.\fr$.
On the other hand, since $r$ is a read event, $w$ is a write event and $(w, r) \in \absGRSI.\rsihb$, from \cref{lem:rsi_soundness_aux} we have $(w, \txid_r.\mathit{vs}_\x) \in \impGRSI.\hb$. 
As such we have $w \relarrow{\impGRSI.\hb} \txid_r.\mathit{vs}_\x \relarrow{\impGRSI.\fr} w$, contradicting the assumption that $\impGRSI$ is RA-consistent. 

\end{proof}
\end{theorem}
\subsection{Implementation Completeness}
In order to establish the completeness of our implementation, it suffices to show that given an RSI-consistent execution graph $\absGRSI = (\Events, \po, \rf, \co, \rsilo)$, we can construct a corresponding RA-consistent execution graph $\impGRSI$ of the implementation.
%%
%\[
%	\for{\absGRSI} \rsicon \Rightarrow \exsts{\impGRSI} \consistent{\impGRSI}
%\]
%%
Before proceeding with the construction of a corresponding implementation graph, we describe several auxiliary definitions.

Given a transaction class $\Transactions_{\txid} \in \absGRSI.\Transactions/\st$, we write $\writeset_{\txid}$ for the set of locations written to by $\Transactions_{\txid}$: $\writeset_{\txid} = \bigcup_{e \in \Transactions_{\txid}\cap \Writes} \loc{e}$.
Similarly, we write $\readset_{\txid}$ for the set of locations read from by $\Transactions_{\txid}$,
\emph{prior to} being written by $\Transactions_{\txid}$. 
For each location \code x read from by $\Transactions_{\txid}$, we additionally record the first read event in $\Transactions_{\txid}$ that retrieved the value of \code x.
That is, 
\[
\readset_{\Transactions_{\txid}} \eqdef
\setcomp{
	(\code x, r)
}{
	r \in \Transactions_i \cap \Reads_{\code x}
	\land \neg\exsts{e \in \Transactions_{\txid} \cap \Events_{\code x}}
	e \relarrow{\po} r
}
\]
Note that the execution trace for each transaction $\Transactions_{\txid} \in \absGRSI.\Transactions/\st$ is of the form $\trace'_{\txid} = \mathit{t}'_1 \relarrow{\imm \po} \cdots \relarrow{\imm \po} \mathit{t}'_k$, comprising a series of read or write events.
%$\mathit{Rs}_i = R_1 \relarrow{\imm \po} \cdots \relarrow{\imm \po} R_n$ for some $n$, where each $\mathit{R}_j$ is a read event, and $\mathit{Ws}_i =  W_1 \relarrow{\imm \po} \cdots \relarrow{\imm \po} W_m$ for some $m$, where each $W_j$ is a write event. 
As such, we have $\absGRSI.\Events = \absGRSI.\Transactions = \bigcup_{\Transactions_{\txid} \in \absGRSI.\Transactions/\st} \Transactions_{\txid} = \trace'_{\txid}.\Events$.
Let $\readset_{\Transactions_{\txid}} \cup \writeset_{\Transactions_{\txid}} = \{\code x_1 \cdots \code x_n\}$.
%and $\writeset_{\Transactions_{\txid}} = \{\code y_1 \cdots \code y_q\}$. 
We then construct the implementation trace $\trace_{\txid}$ as:
\[
	\trace_{\txid} = 
	\mathit{Rs}
	\relarrow{\imm \po} \mathit{RUs}
	\relarrow{\imm \po} \mathit{PLs}
	\relarrow{\imm \po} \mathit{Ts}
	\relarrow{\imm \po} \mathit{Us}
\]
where
\begin{itemize}
	\item $\mathit{Rs} = \mathit{rl}_{\code x_1} \relarrow{\imm \po} \cdots \relarrow{\imm \po} \mathit{rl}_{\code x_n} \relarrow{\imm \po} \mathit{S}_{\code x_1} \relarrow{\imm \po} \cdots \relarrow{\imm \po} \mathit{S}_{\code x_n} \relarrow{\imm \po} \mathit{V}_{\code x_1} \relarrow{\imm \po} \cdots \relarrow{\imm \po} \mathit{V}_{\code x_n}$, where the identifiers of all constituent events of $\mathit{Rs}$ are picked fresh, and 
%	
%	\[
%	\begin{array}{l}
%		\mathit{IR}_{\code x_j} = 
%		\mathit{rl}_{\code x_j} 
%		\relarrow{\imm \po} 
%		\mathit{S}_{\code x_j}	
%		\qquad 
%		\mathit{FR}_{\code x_j} = 
%		\mathit{V}_{\code x_j}		
%	\end{array}
%	\]	
%%	\normalsize
%	where 
	\[
		\begin{array}{c}
			\mathit{rl}_{\code x_j} {=} \rlockE{\x_j} 
			\quad
			\mathit{S}_{\code x_j}{=}
			\begin{cases}
				\mathit{rs}_{\code x_j}
				\relarrow{\imm \po}
				\mathit{ws}_{\code x_j}
				&  \text{if } \exsts{r} (\code x_j, r) \in \readset_{\txid} \land \rval{r} = v_j \\
				
				\emptyset & \text{otherwise}
			\end{cases} 
			\\
			\mathit{V}_{\code x_j} {=}
			\begin{cases}
				\mathit{vs}_{\code x_j}
				&  \text{if } \exsts{r} (\code x_j, r) \in \readset_{\txid} \land \rval{r} = v_j \\
				
				\emptyset & \text{otherwise}
			\end{cases} 			
		\end{array}
	\]	
with $\mathit{rs}_{\code x_j} \eqdef \readE{\acq}{\code x_j}{v_j}$, $\mathit{vs}_{\code x_j} \eqdef \readE{\acq}{\code x_j}{v_j}$
and $\mathit{ws}_{\code x_j} \eqdef \writeE{\rel}{\code{s[x}_j\code ]}{v_j}$.
	\item $\mathit{RUs} = \mathit{ru}_{\code x_1} \relarrow{\imm \po} \cdots \relarrow{\imm \po} \mathit{ru}_{\code x_n}$, where the identifiers of all constituent events of $\mathit{RUs}$ are picked fresh, and for all $j \in \{1 \cdots n\}$:
	\[
	\begin{array}{l}
		\mathit{ru}_{\code x_j} = 
		\begin{cases}
			\begin{array}{@{} l @{}}
				\runlockE{\x_j}
			\end{array}
			& \text{ if } \code x_j \not\in \writeset_{\txid}  \\\\
			\emptyset
			& \text{ otherwise}
		\end{cases}
	\end{array}	
	\]
%
%
	%\updateE{\acqrel}{\x_j}{v}{v {-} 2}$ such that $\exsts{w} (w, \mathit{ru}_{\code x_j}) \in \RF[]{} \land \wval{w} = v$.
%
%
	\item $\mathit{PLs} = \mathit{pl}_{\code x_1} \relarrow{\imm \po} \cdots \relarrow{\imm \po} \mathit{pl}_{\code x_n}$, where the identifiers of all constituent events of $\mathit{PLs}$ are picked fresh, and for all $j \in \{1 \cdots n\}$:
	\[
	\begin{array}{l}
		\mathit{pl}_{\code x_j} = 
		\begin{cases}
			\begin{array}{@{} l @{}}
%				\mathit{pr}_{\code x_j}
%				\relarrow{\imm \po}
%				\mathit{pl}_{\code x_j}
%				\relarrow{\imm \po}
%%				\mathit{s}_{\code x_j}
%%				\relarrow{\imm \po}
%				\mathit{wl}_{\code x_j}
				\plockE{\x_j}
			\end{array}
			&  \text{ if } \code x_j \in \writeset_{\txid} \\\\
			\emptyset
			& \text{otherwise}			
		\end{cases}
	\end{array}	
	\]
%
%
%	where 
%	\[
%	\begin{array}{@{} r @{\hspace{2pt}} l @{}}
%		\mathit{pr}_{\code x_j} =
%		& \readE{\acq}{\x_j}{v} 
%		\quad \text{s.t. } 
%		\exsts{w} (w, \mathit{pr}_{\code x_j}) \in \RF[]{} \land \wval{w} = v \\
%		
%		\mathit{pl}_{\code x_j} =
%		& \updateE{\acqrel}{\x_j}{v}{v {-} 1} 
%		\quad \text{s.t. } 
%		\exsts{w} (w, \mathit{pl}_{\code x_j}) \in \RF[]{} \land \wval{w} = v \\
%
%%		\mathit{s}_{\code x_j} =
%%		&  \readE{\acq}{\x_n}{c_n^1} \relarrow{\imm \po} \cdots \relarrow{\imm \po}  \readE{\acq}{\x_n}{c_n^m} \\
%			
%		\mathit{wl}^{\code x_j} =
%		& \readE{\acq}{\x_j}{1}		
%	\end{array}
%	\]
	%
%
%
	\item $\mathit{Ts} = \mathit{t}_1 \relarrow{\imm \po} \cdots \relarrow{\imm \po} \mathit{t}_k$, where for all $m \in \{1 \cdots k\}$:
	\[
	\mathit{t}_m = 
	\begin{cases}
		\readE{-}{\code{s[x}_n \code{]}}{v_n} & \text{if } \mathit{t}'_m {=} \readE{-}{\code x_n}{v_n} \\
		\writeE{\rel}{\code x_n}{v_n} \relarrow{\imm \po} \writeE{-}{\code{s[x}_n \code ]}{v_n}
		& \text{if } \mathit{t}'_m {=} \writeE{\rel}{\code x_n}{v_n} \\
	\end{cases}
	\]
	such that in the first case the identifier of $\mathit{t}_m$ is that of $\mathit{t}'_m$; and in the second case the identifier of the first event in $\mathit{t}_m$ is that of $\mathit{t}'_m$ and the identifier of the second event is picked fresh.
	\item $\mathit{Us} = \mathit{wu}_{\code x_1} \relarrow{\imm \po} \cdots \relarrow{\imm \po} \mathit{wu}_{\code x_n}$, where the identifiers of all constituent events of $\mathit{Us}$ are picked fresh, and 
	\[
		\mathit{wu}_{\code x_j} = 
		\begin{cases}
%			\writeE{\rel}{\x_j}{0}
			\wunlockE{\x_j}
			& \text{ if } \code x_j \in \writeset_{\txid}  \\
			\emptyset
			& \text{ otherwise}
		\end{cases}
	\]
\end{itemize}
In what follows we write $\Events_\txid$ as a shorthand for the events in the implementation trace of $\trace_\txid$, i.e.~$\Events_\txid \eqdef \setcomp{a}{a \in \trace_\txid.\Events}$. We use the $\txid.$ prefix to project the various events of the implementation trace $\trace_\txid$ (e.g.~$\txid.\mathit{rl}_{\code x_j}$). 

For each location \x we then define:
\[
\begin{array}{@{} r @{\hspace{2pt}} l @{}}
\LO_{\x} \eqdef
& \begin{array}[t]{@{} l @{\hspace{1pt}} l @{}} 
	& \setcomp{
		(\txid.\mathit{rl}_{\code x}, \txid.\mathit{pl}_{\code x}), 
		(\txid.\mathit{rl}_{\code x}, \txid.\mathit{wu}_{\code x}), 
		(\txid.\mathit{pl}_{\code x}, \txid.\mathit{wu}_{\code x})
	} 
	{ 
		\absGRSI.\Transactions_\txid \cap \Writes_\x \ne \emptyset
	} \\
	\cup 
	& \setcomp{
		(\txid.\mathit{rl}_{\code x}, \txid'.\mathit{pl}_{\code x}), 
		(\txid.\mathit{rl}_{\code x}, \txid'.\mathit{wu}_{\code x}), \\
		(\txid.\mathit{ru}_{\code x}, \txid'.\mathit{pl}_{\code x}), 	
		(\txid.\mathit{ru}_{\code x}, \txid'.\mathit{wu}_{\code x})
	} 
	{ 
		\txid \ne \txid'
		\land \exsts{a, b, x} \\
			\quad a \in \absGRSI.\Transactions_\txid \land b \in \absGRSI.\Transactions_\txid' \\
			\land \loc{a} = \loc{b} = \x \\
			\quad \land\, a \in \absGRSI.\EReads
			\land (a, b) \in \absGRSI.\fr
	} \\
	\cup & \setcomp{
		(\txid.\mathit{rl}_{\code x}, \txid'.\mathit{pl}_{\code x}), 
		(\txid.\mathit{rl}_{\code x}, \txid'.\mathit{wu}_{\code x}), \\
		(\txid.\mathit{pl}_{\code x}, \txid'.\mathit{rl}_{\code x}), 
		(\txid.\mathit{pl}_{\code x}, \txid'.\mathit{pl}_{\code x}), \\
		(\txid.\mathit{pl}_{\code x}, \txid'.\mathit{wu}_{\code x}), \\
		(\txid.\mathit{wu}_{\code x}, \txid'.\mathit{rl}_{\code x}), 
		(\txid.\mathit{wu}_{\code x}, \txid'.\mathit{pl}_{\code x}), \\
		(\txid.\mathit{wu}_{\code x}, \txid'.\mathit{wu}_{\code x})
	} 
	{ 
		\txid \ne \txid'
		\land \exsts{a, b, x} \\
			\quad a \in \absGRSI.\Transactions_\txid \land b \in \absGRSI.\Transactions_{\txid'} \\
			\quad \land\, \loc{a} = \loc{b} = \x \\
			\quad \land\, (a, b) \in \absGRSI.\co
	} \\
	\cup & \setcomp{
		(\txid.\mathit{pl}_{\code x}, \txid'.\mathit{rl}_{\code x}), 
		(\txid.\mathit{pl}_{\code x}, \txid'.\mathit{ru}_{\code x}), \\
		(\txid.\mathit{wu}_{\code x}, \txid'.\mathit{rl}_{\code x}), 
		(\txid.\mathit{wu}_{\code x}, \txid'.\mathit{ru}_{\code x})
	} 
	{ 
		\txid \ne \txid'
		\land \exsts{a, b, x} \\
			\quad a \in \absGRSI.\Transactions_\txid \land b \in \absGRSI.\Transactions_{\txid'}  \\
			\quad \land \absGRSI.\Transactions_{\txid'} \cap \Writes_\x = \emptyset \\
			\quad \land\, \loc{a} {=} \loc{b} {=} \x \\
			\quad \land (a, b) \in \absGRSI.(\refC{\co}; \rf)
	} 
\end{array} 
\end{array}
\]
Note that each $\LO_\x$ satisfies the conditions in \cref{def:si_implementation_consistency}.

We are now in a position to demonstrate the completeness of our implementation. Given an RSI-consistent execution graph $\absGRSI$,
% $\absGRSI$  such that $\IRC{\absGRSI}$ holds, 
 we construct an execution graph $\impGRSI$ of the implementation as follows and demonstrate that $\consistent{\impGRSI}$ holds. 
\begin{itemize}
	\item $\impGRSI.\Events = \bigcup\limits_{\Transactions_\txid \in \absGRSI.\Transactions/\st} \trace_\txid.\Events \cup \absGRSI.\NT$. Observe that 
%	$\impGRSI.\Events$ is an extension of $\absGRSI.\Events$: 
	$\absGRSI.\Events \subseteq \impGRSI.\Events$.
	\item $\impGRSI.\po$ is defined as $\absGRSI.\po$ extended by the $\po$ for the additional events of $\impGRSI$, given by each $\trace_\txid$ trace defined above. Note that $\impGRSI.\po$ does not introduce additional orderings between events of $\absGRSI.\Events$. That is, $\for{a, b \in \absGRSI.\Events} (a, b) \in \absGRSI.\po \Leftrightarrow (a, b) \in \impGRSI.\po$.
	\item 
	$\impGRSI.\rf = 
	\bigcup_{\code x \in \textsc{Locs}} 
	\setcomp{
		(w, \txid.\mathit{rs}_{\code x}), \\
		(w, \txid.\mathit{vs}_{\code x})		
	}{
		\exsts{r} (\code x, r) \in \readset_{\txid} \\
		\quad \land\, (w, r) \in \absGRSI.\rf	
	}
	\cup \big(\absGRSI.\rf \cap \absGRSI.\Events \times \absGRSI.\NT)
	$.
	\item $\impGRSI.\co = \absGRSI.\co$. %\coerce{\absGRSI.\co}{\absGRSI.\Events_\code x}.
	\item 
	$\impGRSI.\rsilo = 
	\bigcup_{\code x \in \textsc{Locs}} \LO_{\x}
	$, with $\LO_\x$ as defined above.
\end{itemize}
%
%In what follows we write $\impGRSI.\NT$ for the transactional events of $\impGRSI$, i.e.~$\impGRSI.\Transactions \eqdef \setcomp{a}{\exsts{\txid} a \in \impGRSI.\Events_{\txid}}$. Similarly, we write $\impGRSI.\NT$ for the non-transactional events of $\impGRSI$, i.e.~$\impGRSI.\NT \eqdef \impGRSI.\Events \setminus \impGRSI.\Transactions$. 

\paragraph{Notation} 
In what follows, given an RSI implementation graph $\impGSI$ as constructed above we write $\impGRSI.\NT$ for the non-transactional events of $\impGRSI$, i.e.~$\impGRSI.\NT \eqdef \setcomp{a}{\neg\exsts{\txid} a \in \impGRSI.\Events_{\txid}}$.
Moreover, as before, given a relation $\makerel r \subseteq \impGRSI.\Events \times \impGRSI.\Events$, we override the $\tin{\makerel r}$ notation and write $\tin{\makerel r}$ for $\setcomp{(a, b) \in \makerel r}{\exsts{\txid} a, b \in \trace_{\txid}.\Events}$.

\begin{lemma}\label{lem:rsi_completeness}
Given an RSI-consistent execution graph $\absGRSI$ and its corresponding implementation graph $\impGRSI$ constructed as above, for all $a, b, \txid_a, \txid_b$:
\[
\begin{array}{@{} l @{}}
	(a, b) \in \impGRSI.\hb \Rightarrow \\
	\qquad\! (\exsts{\txid} a, b \in  \impGRSI.\Events_{\txid} \Rightarrow (a, b) \in \impGRSI.\po) \\
	\quad \land 
	\Big( \neg \exsts{\txid} a, b \in \impGRSI.\Events_{\txid} \Rightarrow \\
	\qqqquad 
	\begin{array}[t]{@{} l @{}}
		 \exsts{A, B} \emptyset \subset A \times B \subseteq \absGRSI.\rsihb \\
		 \land\ (a \in \impGRSI.\NT \Rightarrow A {=} \{a\}) \land (b \in \impGRSI.\NT \Rightarrow B {=} \{b\}) \\
		\land\ [a \in \impGRSI.\Events_{\txid_a} \Rightarrow \\
			\qquad A {=} \absGRSI.\Transactions_{\txid_a} \lor (\stg{a}{\txid_a} \leq 2 \land A {=} \absGRSI.\Transactions_{\txid_a} \cap \EReads) \\
			\qquad \lor\ (\stg{a}{\txid_a} \leq 4 \land \exsts{d \in \txid_a.\Writes} a \relarrow{\impGRSI.\refC{\po}} d \land A = \{d\}) \;] \\
		\land\ [b \in \impGRSI.\Events_{\txid_b} \Rightarrow \\
			\qquad B {=} \absGRSI.\Transactions_{\txid_b} \lor (\stg{b}{\txid_b} \geq 3 \land B {=} \absGRSI.\Transactions_{\txid_b} \cap \Writes) \;]	\quad \Big)
\end{array}		
\end{array}
\]
where
\[
	\stg{a}{\txid_a} \eqdef 
	\begin{cases}
		1 & \text{if } a \in  \txid_a.\mathit{Rs}_{\txid_a} \\
		2 & \text{if } a \in  \txid_a.\mathit{RUs}_{\txid_a} \\
		3 & \text{if } a \in  \txid_a.\mathit{PLs}_{\txid_a} \\
		4 & \text{otherwise}
	\end{cases}
\]
\begin{proof}
Pick an arbitrary RSI-consistent execution graph $\absGRSI$ and its corresponding implementation graph $\impGRSI$ constructed as above.
Let $\hb_0 \eqdef \impGRSI.(\po \cup \rf \cup \rsilo)$ and $\hb_{n {+} 1} \eqdef \hb_0; \hb_n$, for all $n \in \Nats$. It is then straightforward to demonstrate that $\impGRSI.\hb = \bigcup_{i \in \Nats} \hb_i$.
We thus demonstrate instead that:
\[
\begin{array}{@{} l @{}}
	\for{n \in \Nats} \for{a, b, \txid_a, \txid_b}
	(a, b) \in \impGRSI.\hb_n \Rightarrow \\
	\qqquad\! (\exsts{\txid} a, b \in  \impGRSI.\Events_{\txid} \Rightarrow (a, b) \in \impGRSI.\po) \\
	\qquad \land 
	\Big( \neg \exsts{\txid} a, b \in \impGRSI.\Events_{\txid} \Rightarrow \\
	\qqqquad 
	\begin{array}[t]{@{} l @{}}
		 \exsts{A, B} \emptyset \subset A \times B \subseteq \absGRSI.\rsihb \\
		 \land\ (a \in \impGRSI.\NT \Rightarrow A {=} \{a\}) \land (b \in \impGRSI.\NT \Rightarrow B {=} \{b\}) \\
		\land\ [a \in \impGRSI.\Events_{\txid_a} \Rightarrow \\
			\qquad A {=} \absGRSI.\Transactions_{\txid_a} \lor (\stg{a}{\txid_a} \leq 2 \land A {=} \absGRSI.\Transactions_{\txid_a} \cap \EReads) \\
			\qquad \lor\ (\stg{a}{\txid_a} \leq 4 \land \exsts{d \in \txid_a.\Writes} a \relarrow{\impGRSI.\refC{\po}} d \land A = \{d\}) \;] \\
		\land\ [b \in \impGRSI.\Events_{\txid_b} \Rightarrow \\
			\qquad B {=} \absGRSI.\Transactions_{\txid_b} \lor (\stg{b}{\txid_b} \geq 3 \land B {=} \absGRSI.\Transactions_{\txid_b} \cap \Writes) \;]	\quad \Big)
\end{array}		
\end{array}
\]
We proceed by induction on $n$.\\

\noindent \textbf{Base case $n = 0$}\\
There are three cases to consider: 1) $(a, b) \in \impGRSI.\po$; or 2) $(a, b) \in \impGRSI.\rf$; or 3) $(a, b) \in \impGRSI.\rsilo$.

In case (1) there are five cases to consider: a) $\exsts{\txid} (a, b) \in \impGRSI.\Events_{\txid}$; or b) $a, b \in \impGRSI.\NT$; or c) $a \in \impGRSI.\NT$ and $b \in \impGRSI.\Events_{\txid_b}$; or d) $a \in \impGRSI.\Events_{\txid_a}$ and $b \in \impGRSI.\NT$; or e) $a \in \impGRSI.\Events_{\txid_a}$, $b \in \impGRSI.\Events_{\txid_b}$ and $\txid_a \ne \txid_b$.
In case (1.a) we then have $(a, b) \in \impGRSI.\poi$, as required.
In case (1.b) from the construction of $\impGRSI.\po$ we have $(a, b) \in \absGRSI.\po$, as required.
In case (1.c) from the construction of $\impGRSI.\po$ we have $(\{a\} \times \absGRSI.\Transactions_{\txid_b}) \in \absGRSI.\po$, as required.
In case (1.d) from the construction of $\impGRSI.\po$ we have $(\absGRSI.\Transactions_{\txid_a} \times \{b\}) \in \absGRSI.\po$, as required.
In case (1.e) from the construction of $\impGRSI.\po$ we have $(\absGRSI.\Transactions_{\txid_a} \times \absGRSI.\Transactions_{\txid_b}) \in \absGRSI.\po$, as required.

In case (2) there are five cases to consider: 
a) $\exsts{\txid} (a, b) \in \impGRSI.\Events_{\txid}$; or 
b) $a, b \in \impGRSI.\NT$; or 
c) $a \in \impGRSI.\NT$ and $b \in \impGRSI.\Events_{\txid_b}$; or 
d) $a \in \impGRSI.\Events_{\txid_a}$ and $b \in \impGRSI.\NT$; or e) $a \in \impGRSI.\Events_{\txid_a}$, $b \in \impGRSI.\Events_{\txid_b}$ and $\txid_a \ne \txid_b$. 
Case (2.a) holds vacuously as $(a, b) \in \impGRSI.\rfi = \emptyset$.
In case (2.b) from the construction of $\impGRSI.\rf$ we have $(a, b) \in \absGRSI.\rf$, as required.
In case (2.c) from the construction of $\impGRSI.\rf$ we know that there exists $r \in \absGRSI.\Transactions_{\txid_b}$ such that $(a, r) \in \absGRSI.\rf$. 
As such we have $(\{w\} \times \absGRSI.\Transactions_{\txid_b}) \subseteq \absGRSI.([\NT]; \rf; \st) \subseteq \absGRSI.\rsihb$, as required.
In case (2.d) from the construction of $\impGRSI.\rf$ we then have $(a, b) \in \absGRSI.\rf$; that is $a \relarrow{\impGRSI.\refC{\po}} a$ and $(a, b) \in \absGRSI.\rf$, as required.
In case (2.e) from the construction of $\impGRSI.\rf$ we know that there exists $r \in \absGRSI.\Transactions_{\txid_b}$ such that $(a, r) \in \absGRSI.\rf$. As such we have $(\absGRSI.\Transactions_{\txid_a} \times \absGRSI.\Transactions_{\txid_b}) \subseteq \absGRSI.\rft \subseteq \absGRSI.\rsihb$, as required.

In case (3) there are five cases to consider: 
a) $\exsts{\txid} (a, b) \in \impGRSI.\Events_{\txid}$; or 
b) $a, b \in \impGRSI.\NT$; or 
c) $a \in \impGRSI.\NT$ and $b \in \impGRSI.\Events_{\txid_b}$; or 
d) $a \in \impGRSI.\Events_{\txid_a}$ and $b \in \impGRSI.\NT$; or 
e) $a \in \impGRSI.\Events_{\txid_a}$, $b \in \impGRSI.\Events_{\txid_b}$ and $\txid_a \ne \txid_b$. 
In case (3.a) from the construction of $\lo$ we have $(a, b) \in \impGRSI.\po$, as required.
Cases (3.b-3.d) hold vacuously as there are no $\rsilo$ edge 
%between events of the same transaction, neither are there any $\rsilo$ edges 
to or from non-transactional events.
In case (3.e) from the construction of $\rsilo$ we know there exist $\x, c, d$ such that  $c \in \absGRSI.\Transactions_{\txid_a}$, $d \in \absGRSI.\Transactions_{\txid_b}$, and either 
%i) $a = \txid_a.\mathit{wu}_{\x}$, $\txid_a.\mathit{rl}_{\x}$ and $(c, d) \in \absGRSI.(\refC{\co}; \rf)$; or 
%ii) $a = \txid_a.\mathit{wu}_{\x}$, $\txid_a.\mathit{rl}_{\x}$ and $(c, d) \in \absGRSI.\co$; or
%iii) $a = \txid_a.\mathit{ru}_{\x}$, $\txid_a.\mathit{pl}_{\x}$ and $c \in \absGRSI.\EReads$ and $(c, d) \in \absGRSI.\fr$.
%In case (3.e.i) we have $\absGRSI.\Transactions_{\txid_a} \times \absGRSI.\Transactions_{\txid_b} \subseteq \absGRSI.\tlift{(\refC{\co}; \rf)} \subseteq \rsirf  \subseteq \rsihb$, as required.
%In case (3.e.ii) we have $\absGRSI.\Transactions_{\txid_a} \times \absGRSI.\Transactions_{\txid_b} \subseteq \absGRSI.\cot \subseteq \rsico  \subseteq \rsihb$, as required.
%In case (3.e.iii) we then have $\stg{a}{\txid_a} = 2$, $\stg{b}{\txid_b} = 3$, and $\absGRSI.(\Transactions_{\txid_a} \cap \EReads) \times \absGRSI.(\Transactions_{\txid_b} \cap \Writes) \subseteq \absGRSI.\rsifr  \subseteq \rsihb$, as required.\\
i) $(c, d) \in \absGRSI.\rf$; or 
ii) $(c, d) \in \absGRSI.(\co; \rf)$; or 
iii) $(c, d) \in \absGRSI.\co$; or
iv) $\loc c = \loc d = \x$, $a = \txid_a.\mathit{rl}_{\x} \lor a = \txid_a.\mathit{ru}_{\x}$, $b = \txid_a.\mathit{pl}_{\x} \lor b = \txid_a.\mathit{wu}_{\x}$ and $c \in \absGRSI.\EReads$ and $(c, d) \in \absGRSI.\fr$.
In cases (3.e.i, 3.e.ii) we have $\absGRSI.\Transactions_{\txid_a} \times \absGRSI.\Transactions_{\txid_b} \subseteq \absGRSI.(\rft \cup \tlift{(\co; \rf)}) \subseteq \rsirf  \subseteq \rsihb$, as required.
In case (3.e.iii) we have $\absGRSI.\Transactions_{\txid_a} \times \absGRSI.\Transactions_{\txid_b} \subseteq \absGRSI.\cot  \subseteq \rsihb$, as required.
In case (3.e.iv) we then have $\stg{a}{\txid_a} \leq 2$, $\stg{b}{\txid_b} \geq 3$, and $\absGRSI.(\Transactions_{\txid_a} \cap \EReads) \times \absGRSI.(\Transactions_{\txid_b} \cap \Writes) \subseteq \absGRSI.\rsifr  \subseteq \rsihb$, as required.\\

\noindent \textbf{Inductive case $n = m {+} 1$}\\
\begin{align}
\begin{array}{@{} l @{}}
	\for{i \in \Nats} \for{a, b, \txid_a, \txid_b}
	i \leq m \land (a, b) \in \impGRSI.\hb_i \Rightarrow \\
	\qqquad\! (\exsts{\txid} a, b \in  \impGRSI.\Events_{\txid} \Rightarrow (a, b) \in \impGRSI.\po) \\
	\qquad \land 
	\Big( \neg \exsts{\txid} a, b \in \impGRSI.\Events_{\txid} \Rightarrow \\
	\qqqquad 
	\begin{array}[t]{@{} l @{}}
		 \exsts{A, B} \emptyset \subset A \times B \subseteq \absGRSI.\rsihb \\
		 \land\ (a \in \impGRSI.\NT \Rightarrow A {=} \{a\}) \land (b \in \impGRSI.\NT \Rightarrow B {=} \{b\}) \\
		\land\ [a \in \impGRSI.\Events_{\txid_a} \Rightarrow \\
			\qquad A {=} \absGRSI.\Transactions_{\txid_a} \lor (\stg{a}{\txid_a} \leq 2 \land A {=} \absGRSI.\Transactions_{\txid_a} \cap \EReads) \\
			\qquad \lor\ (\stg{a}{\txid_a} \leq 4 \land \exsts{d \in \txid_a.\Writes} a \relarrow{\impGRSI.\refC{\po}} d \land A = \{d\}) \;] \\
		\land\ [b \in \impGRSI.\Events_{\txid_b} \Rightarrow \\
			\qquad B {=} \absGRSI.\Transactions_{\txid_b} \lor (\stg{b}{\txid_b} \geq 3 \land B {=} \absGRSI.\Transactions_{\txid_b} \cap \Writes) \;]	\quad \Big)
\end{array}		
\end{array}
\tag{I.H.}
\label{IH:rsi_completeness}
\end{align}
Since $(a, b) \in \hb_n$, from the definition of $\hb_n$ we know there exists $c$ such that $(a, c) \in \hb_0$ and $(c, b) \in \hb_m$. 
There are then five cases to consider: 
1) $\exsts{\txid} (a, b) \in \impGRSI.\Events_{\txid}$; or 
2) $a, b \in \impGRSI.\NT$; or 
3) $a \in \impGRSI.\NT$ and $b \in \impGRSI.\Events_{\txid_b}$; or 
4) $a \in \impGRSI.\Events_{\txid_a}$ and $b \in \impGRSI.\NT$; or 
5) $a \in \impGRSI.\Events_{\txid_a}$, $b \in \impGRSI.\Events_{\txid_b}$ and $\txid_a \ne \txid_b$. \\

\noindent \textbf{Case 1}\\
In case (1) pick arbitrary $\txid$ such that $a, b \in \impGRSI.\Events_{\txid}$. There are then three additional cases to consider: 
a) $c \in \impGRSI.\NT$; 
b) $c \in \impGRSI.\Events_{\txid}$; or 
c) there exists $\txid' \ne \txid$ such that $c \in \impGRSI.\Events_{\txid'}$.

In case (1.a) from the proof of the base case and the \eqref{IH:rsi_completeness} we know there exist $A, B, C \ne \emptyset$ such that $C = \{c\}$, $A \times C \subseteq \absGRSI.\rsihb$, $C \times B \subseteq \absGRSI.\rsihb$ and thus $A \times B \subseteq \absGRSI.\rsihb$ and either: 
i) $A = \absGRSI.\Transactions_{\txid}$ and  $B = \absGRSI.\Transactions_{\txid}$; or 
ii) $A = \absGRSI.\Transactions_{\txid}$ and  $B = (\absGRSI.(\Transactions_{\txid} \cap  \Writes)$; or 
iii) $\stg{a}{\txid} \leq 2$, $A = \absGRSI.(\Transactions_{\txid} \cap \EReads)$ and $B = \absGRSI.\Transactions_{\txid}$; or
iv) $\stg{a}{\txid} \leq 2$, $A = \absGRSI.(\Transactions_{\txid} \cap \EReads)$, $\stg{b}{\txid} \geq 3$ and $B = \absGRSI.(\Transactions_{\txid} \cap \Writes)$; or
v) $\stg{a}{\txid} \leq 4$, $\exsts{d \in \txid.\Writes} a \relarrow{\impGRSI.\refC{\po}} d \land A = \{d\}$ and $B = \absGRSI.\Transactions_{\txid}$; or
vi) $\stg{a}{\txid} \leq 4$, $\exsts{d \in \txid.\Writes} a \relarrow{\impGRSI.\refC{\po}} d \land A = \{d\}$ and $B = \absGRSI.(\Transactions_{\txid} \cap \Writes)$.

Case (i) cannot arise as we would have $A = B$ and thus $A \times A \subseteq \absGRSI.\rsihb$, contradicting the assumption that $\absGRSI$ is RSI-consistent. 
Case (ii) cannot arise as we would have $(\absGRSI.(\Transactions_{\txid} \cap  \Writes) \times (\absGRSI.(\Transactions_{\txid} \cap  \Writes) \subseteq \absGRSI.\rsihb$, contradicting the assumption that $\absGRSI$ is RSI-consistent. 
Case (iii) cannot arise as we would have $(\absGRSI.(\Transactions_{\txid} \cap  \EReads) \times (\absGRSI.(\Transactions_{\txid} \cap  \EReads) \subseteq \absGRSI.\rsihb$, contradicting the assumption that $\absGRSI$ is RSI-consistent. 
In case (iv) since we have $\stg{a}{\txid} \leq 2$ and $\stg{b}{\txid} \geq 3$, from the definition of $\stg{.}{.}$ and the construction of $\impGRSI$ we have $(a, b) \in \impGRSI.\po$, as required.
Cases (v-vi) cannot arise as we would have $\exsts{d \in \txid.\Writes} (d, d)  \in \absGRSI.\rsihb$, contradicting the assumption that $\absGRSI$ is RSI-consistent. \\

In case (1.b) from the proof of the base case we have $(a, c) \in \impGRSI.\po$. On the other hand from (\ref{IH:rsi_completeness}) we have $(c, b) \in \impGRSI.\po$. As $\impGRSI.\po$ is transitively closed, we have $(a, b) \in \impGRSI.\po$, as required.\\

In case (1.c)  from the proof of the base case (in cases 1.e, 2.e and 3.e) we have either 
A) $\absGRSI.\Transactions_{\txid} \times \absGRSI.\Transactions_{\txid'} \subseteq \absGRSI.\rsihb$; or 
B) $\stg{a}{\txid} \leq 2$, $\stg{c}{\txid'} \geq 3$,  and $\emptyset \subset \absGRSI.(\Transactions_{\txid} \cap \EReads) \times \absGRSI.(\Transactions_{\txid'} \cap \Writes) \subseteq \absGRSI.\rsihb$.
On the other hand from \eqref{IH:rsi_completeness} we know there exist $C, B \ne \emptyset$ such that $C \times B \subseteq \absGRSI.\rsihb$ and either: 
i) $C = \absGRSI.\Transactions_{\txid'}$ and  $B = \absGRSI.\Transactions_{\txid}$; or 
ii) $C = \absGRSI.\Transactions_{\txid'}$,  $\stg{b}{\txid} \geq 3$ and $B = \absGRSI.(\Transactions_{\txid} \cap  \Writes)$; or 
iii) $\stg{c}{\txid'} \leq 2$, $C = \absGRSI.(\Transactions_{\txid'} \cap \EReads)$ and $B = \absGRSI.\Transactions_{\txid}$; or
iv) $\stg{c}{\txid'} \leq 2$, $C = \absGRSI.(\Transactions_{\txid'} \cap \EReads)$, $\stg{b}{\txid} \geq 3$ and $B = \absGRSI.(\Transactions_{\txid} \cap \Writes)$; or
v) $\stg{c}{\txid'} \leq 4$, $\exsts{d \in \txid'.\Writes} c \relarrow{\impGRSI.\refC{\po}} d \land C = \{d\}$ and $B = \absGRSI.\Transactions_{\txid}$; or
vi) $\stg{c}{\txid'} \leq 4$, $\exsts{d \in \txid'.\Writes} c \relarrow{\impGRSI.\refC{\po}} d \land C = \{d\}$, $\stg{b}{\txid} \geq 3$ and $B = \absGRSI.(\Transactions_{\txid} \cap \Writes)$.

Cases (A.i-A.vi) lead to a cycle in $\absGRSI.\rsihb$, contradicting the assumption that $\absGRSI$ is RSI-consistent.
In cases (B.ii, B.iv, B.vi) we then have $\stg{a}{\txid} \leq 2$ and $\stg{b}{\txid} \geq 3$. Consequently from the definition of $\stg{.}{.}$ and the construction of $\impGRSI$ we have $(a, b) \in \impGRSI.\po$, as required.
Case (B.i) leads to a cycle in $\absGRSI.\rsihb$, contradicting the assumption that $\absGRSI$ is RSI-consistent.
In cases (B.iii) we have $\stg{c}{\txid'} \geq 3$  and $\stg{c}{\txid'} \leq 2$, leading to a contradiction.
In case (B.v) we then know  $\exsts{d  \in \txid'.\Writes} \emptyset \subset \absGRSI.(\Transactions_{\txid} \cap \EReads) \times \{d\} \subseteq \absGRSI.\rsihb \land \{d\} \times \absGRSI.\Transactions_\txid \subseteq \absGRSI.\rsihb$. That is, we have  $\emptyset \subset \absGRSI.(\Transactions_{\txid} \cap \EReads) \times  \absGRSI.(\Transactions_{\txid} \cap \EReads) \subseteq \absGRSI.\rsihb$, contradicting the assumption that $\absGRSI$ is RSI-consistent.
\\

\noindent \textbf{Case 2}\\
There are two additional cases to consider: 
a) $c \in \impGRSI.\NT$; or
b) there exists $\txid$ such that $c \in \impGRSI.\Events_{\txid}$.

In case (2.a) from the proof of the base case we have $(a, c) \in \absGRSI.\rsihb$. On the other hand from (\ref{IH:rsi_completeness}) we have $(c, b) \in \absGRSI.\rsihb$. As $\absGRSI.\rsihb$ is transitively closed, we have $(a, b) \in \absGRSI.\rsihb$, as required.

In case (2.b) from the proof of the base case we know there exists $C_1 \ne \emptyset$ such that $\{a\} \times C_1 \subseteq \absGRSI.\rsihb$ and either: 
A) $C_1 = \absGRSI.\Transactions_{\txid}$; or 
B) $\stg{c}{\txid} \geq 3$ and $C_1 = \absGRSI.(\Transactions_{\txid} \cap \Writes)$.
On the other hand, from (\ref{IH:rsi_completeness}) we know there exists $C_2 \ne \emptyset$ such that $C_2 \times \{b\} \in \absGRSI\rsihb$ and either:
i) $C_2 = \absGRSI.\Transactions_{\txid}$; or
ii) $\stg{c}{\txid} \leq 2$ and $C_2 = \absGRSI.(\Transactions_{\txid} \cap \EReads)$; or
iii) $\stg{c}{\txid} \leq 4$ and $\exsts{d \in \txid.\Writes} c \relarrow{\impGRSI.\refC{\po}} d \land C_2 = \{d\}$.

In cases (A.i-A.iii) and (B.i, B.iii) from the transitivity of $\absGRSI.\rsihb$ we have $(a, b) \in \absGRSI.\rsihb$, as required.
Case (B.ii) cannot arise as otherwise we would have $3 \leq \stg{c}{\txid} \leq 2$, leading to a contradiction.\\

\noindent \textbf{Case 3}\\
There are three additional cases to consider: 
a) $c \in \impGRSI.\NT$; 
b) $c \in \impGRSI.\Events_{\txid_b}$; or 
c) there exists $\txid_c \ne \txid_b$ such that $c \in \impGRSI.\Events_{\txid_c}$.

In case (3.a) from the proof of the base case we have $(a, c) \in \absGRSI.\rsihb$. On the other hand from (\ref{IH:rsi_completeness}) we know there exists $B \ne \emptyset$ such that $\{c\} \times B \subseteq \absGRSI.\rsihb$ and $B = \absGRSI.\Transactions_{\txid_b} \lor (\stg{b}{\txid_b} \geq 3 \land B = \absGRSI.(\Transactions_{\txid_b} \cap \Writes))$.
As $\absGRSI.\rsihb$ is transitively closed we then know there exists $B \ne \emptyset$ such that $\{a\} \times B \subseteq \absGRSI.\rsihb$ and $B = \absGRSI.\Transactions_{\txid_b} \lor (\stg{b}{\txid_b} \geq 3 \land B = \absGRSI.(\Transactions_{\txid_b} \cap \Writes))$, as required.

In case (3.b) from the proof of the base case we know there exists $B \ne \emptyset$ such that $\{a\} \times B \subseteq \absGRSI.\rsihb$ and $B = \absGRSI.\Transactions_{\txid_b} \lor (\stg{c}{\txid_b} \geq 3 \land B = \absGRSI.(\Transactions_{\txid_b} \cap \Writes))$.
On the other hand, from (\ref{IH:rsi_completeness}) we have $c \relarrow{\impGRSI.\po} b$ and thus from the definition of $\stg{.}{.}$ and the construction of $\impGRSI$ we have $\stg{b}{\txid_b} \geq \stg{c}{\txid_b}$. As such, we we know there exists $B \ne \emptyset$ such that $\{a\} \times B \subseteq \absGRSI.\rsihb$ and $B = \absGRSI.\Transactions_{\txid_b} \lor (\stg{b}{\txid_b} \geq 3 \land B = \absGRSI.(\Transactions_{\txid_b} \cap \Writes))$, as required.

In case (3.c) from the proof of the base case we know there exists $C_1 \ne \emptyset$ such $\{a\} \times C_1 \subseteq \absGRSI.\rsihb$ and either 
A) $C_1 = \absGRSI.\Transactions_{\txid_c}$; or
B)$(\stg{c}{\txid_c} \geq 3 \land C_1 = \absGRSI.(\Transactions_{\txid_c} \cap \Writes))$.
On the other hand, from (\ref{IH:rsi_completeness}) we know there exist $C_2, B \ne \emptyset$ such that $C_2 \times B \in \absGRSI.\rsihb$ and either:
i) $C_2 = \absGRSI.\Transactions_{\txid_c}$ and $B = \absGRSI.\Transactions_{\txid_b}$; or
ii) $C_2 = \absGRSI.\Transactions_{\txid_c}$ and $(\stg{b}{\txid_b} \geq 3 \land B = \absGRSI.(\Transactions_{\txid_b} \cap \Writes))$; or
iii) $\stg{c}{\txid_c} \leq 2 \land C_2 = \absGRSI.(\Transactions_{\txid_c} \cap \EReads)$ and $B = \absGRSI.\Transactions_{\txid_b}$; or
iv) $\stg{c}{\txid_c} \leq 2 \land C_2 = \absGRSI.(\Transactions_{\txid_c} \cap \EReads)$ and $(\stg{b}{\txid_b} \geq 3 \land B = \absGRSI.(\Transactions_{\txid_b} \cap \Writes))$; or
v) $\stg{c}{\txid_c} \leq 4 \land \exsts{d \in \txid_c.\Writes} c \relarrow{\impGRSI.\refC{\po}} d \land C_2 = \{d\}$ and $B = \absGRSI.\Transactions_{\txid_b}$; or
vi) $\stg{c}{\txid_c} \leq 4 \land \exsts{d \in \txid_c.\Writes} c \relarrow{\impGRSI.\refC{\po}} d \land C_2 = \{d\}$ and $(\stg{b}{\txid_b} \geq 3 \land B = \absGRSI.(\Transactions_{\txid_b} \cap \Writes))$. 

In cases (A.i-A.vi) and (B.i, B.ii, B.v, B.vi) from the transitivity of $\absGRSI.\rsihb$ we know there exists $B \ne \emptyset$ such that $\{a\} \times B \subseteq \absGRSI.\rsihb$ and  $B = \absGRSI.\Transactions_{\txid_b} \lor (\stg{b}{\txid_b} \geq 3 \land B = \absGRSI.(\Transactions_{\txid_b} \cap \Writes))$, as required.
Cases (B.iii, B.iv) cannot arise as we would otherwise have $3 \leq \stg{c}{\txid_c} \leq 2$, leading to a contradiction.\\

\noindent \textbf{Case 4}\\
There are three additional cases to consider: 
a) $c \in \impGRSI.\NT$; 
b) $c \in \impGRSI.\Events_{\txid_a}$; or 
c) there exists $\txid_c \ne \txid_a$ such that $c \in \impGRSI.\Events_{\txid_c}$.

In case (4.a) from (\ref{IH:rsi_completeness}) we have $(c, b) \in \absGRSI.\rsihb$. On the other hand from the base case we know there exists $A \ne \emptyset$ such that $A \times \{c\}  \subseteq \absGRSI.\rsihb$ and $A = \absGRSI.\Transactions_{\txid_a} \lor (\stg{a}{\txid_a} \leq 2 \land A = \absGRSI.(\Transactions_{\txid_a} \cap \EReads)) \lor \stg{a}{\txid_a} \leq 4 \land \exsts{d \in \txid_a.\Writes} a \relarrow{\impGRSI.\refC{\po}} d \land A = \{d\}$.
As $\absGRSI.\rsihb$ is transitively closed we then know we know there exists $A \ne \emptyset$ such that $A \times \{b\} \subseteq \absGRSI.\rsihb$ and $A = \absGRSI.\Transactions_{\txid_a} \lor (\stg{a}{\txid_a} \leq 2 \land A = \absGRSI.(\Transactions_{\txid_a} \cap \EReads)) \lor \stg{a}{\txid_a} \leq 4 \land \exsts{d \in \txid_a.\Writes} a \relarrow{\impGRSI.\refC{\po}} d \land A = \{d\}$, as required.

In case (4.b) from \eqref{IH:rsi_completeness} we know there exists $A \ne \emptyset$ such that $A \times \{b\}  \subseteq \absGRSI.\rsihb$ and $A = \absGRSI.\Transactions_{\txid_a} \lor (\stg{c}{\txid_a} \leq 2 \land A = \absGRSI.(\Transactions_{\txid_a} \cap \EReads)) \lor \stg{c}{\txid_a} \leq 4 \land \exsts{d \in \txid_a.\Writes} a \relarrow{\impGRSI.\refC{\po}} d \land A = \{d\}$.
On the other hand, from the proof of the base case we have $a \relarrow{\impGRSI.\po} c$ and thus from the definition of $\stg{.}{.}$ and the construction of $\impGRSI$ we have $\stg{a}{\txid_a} \leq \stg{c}{\txid_a}$.
As such, we know there exists $A \ne \emptyset$ such that $A \times \{b\}  \subseteq \absGRSI.\rsihb$ and $A = \absGRSI.\Transactions_{\txid_a} \lor (\stg{a}{\txid_a} \leq 2 \land A = \absGRSI.(\Transactions_{\txid_a} \cap \EReads)) \lor \stg{a}{\txid_a} \leq 4 \land \exsts{d \in \txid_a.\Writes} a \relarrow{\impGRSI.\refC{\po}} d \land A = \{d\}$, as required.

In case (4.c) from the proof of the base case we know there exist $A, C_1 \ne \emptyset$ such that $A \times C_1  \subseteq \absGRSI.\rsihb$ and either:
i) $A = \absGRSI.\Transactions_{\txid_a}$ and $C_1 = \absGRSI.\Transactions_{\txid_c}$; or
ii) $A = \absGRSI.\Transactions_{\txid_a}$ and $(\stg{c}{\txid_c} \geq 3 \land C_1 = \absGRSI.(\Transactions_{\txid_c} \cap \Writes))$; or
iii) $\stg{a}{\txid_a} \leq 2 \land A = \absGRSI.(\Transactions_{\txid_a} \cap \EReads)$ and $C_1 = \absGRSI.\Transactions_{\txid_c}$; or
iv) $\stg{a}{\txid_a} \leq 2 \land A = \absGRSI.(\Transactions_{\txid_a} \cap \EReads)$ and $(\stg{c}{\txid_c} \geq 3 \land C_1 = \absGRSI.(\Transactions_{\txid_c} \cap \Writes))$; or
v) $\stg{a}{\txid_a} \leq 4 \land \exsts{d \in \txid_a.\Writes} a \relarrow{\impGRSI.\refC{\po}} d \land A = \{d\}$ and $C_1 = \absGRSI.\Transactions_{\txid_c}$; or
vi) $\stg{a}{\txid_a} \leq 4 \land \exsts{d \in \txid_a.\Writes} a \relarrow{\impGRSI.\refC{\po}} d \land A = \{d\}$ and $(\stg{c}{\txid_c} \geq 3 \land C_1 = \absGRSI.(\Transactions_{\txid_c} \cap \Writes))$.

On the other hand, from (\ref{IH:rsi_completeness}) we know there exists $C_2 \ne \emptyset$ such $C_2 \times \{b\} \subseteq \absGRSI.\rsihb$ and either 
A) $C_2 = \absGRSI.\Transactions_{\txid_c}$; or
B) $(\stg{c}{\txid_c} \leq 2 \land C_2 = \absGRSI.(\Transactions_{\txid_c} \cap \EReads))$; or
C) $(\stg{c}{\txid_c} \leq 4 \land \exsts{e \in \txid_c.\Writes} c \relarrow{\impGRSI.\refC{\po}} e \land C_2 = \{e\})$.

In cases (A.i-A.vi), (C.i-C.vi) and (B.i, B.iii, B.v) from the transitivity of $\absGRSI.\rsihb$ we know there exists $A \ne \emptyset$ such that $A \times \{b\} \subseteq \absGRSI.\rsihb$ and $A = \absGRSI.\Transactions_{\txid_a} \lor (\stg{a}{\txid_a} \leq 2 \land A = \absGRSI.(\Transactions_{\txid_a} \cap \EReads)) \lor (\stg{a}{\txid_a} \leq 4 \land \exsts{d \in \txid_a.\Writes} a \relarrow{\impGRSI.\refC{\po}} d \land A = \{d\})$, as required.
Cases (B.ii, B.iv, B.vi) cannot arise as we would otherwise have $3 \leq \stg{c}{\txid_c} \leq 2$, leading to a contradiction.\\

\noindent \textbf{Case 5}\\
There are four additional cases to consider: 
a) $c \in \impGRSI.\NT$; 
b) $c \in \impGRSI.\Events_{\txid_a}$; or 
c) $c \in \impGRSI.\Events_{\txid_b}$; or 
d) there exists $\txid_c$ such that $\txid_c \ne \txid_a$, $\txid_c \ne \txid_b$ and $c \in \impGRSI.\Events_{\txid_c}$.

In case (5.a) from the proof of the base case we know there exists $A \ne \emptyset$ such that $A \times \{c\} \subseteq \absGRSI.\rsihb$ and $A = \absGRSI.\Transactions_{\txid_a} \lor (\stg{a}{\txid_a} \leq 2 \land A = \absGRSI.(\Transactions_{\txid_a} \cap \EReads)) \lor (\stg{a}{\txid_a} \leq 4 \land \exsts{d \in \txid_a.\Writes} a \relarrow{\impGRSI.\refC{\po}} d \land A = \{d\})$.
On the other hand from (\ref{IH:rsi_completeness}) we know there exists $B \ne \emptyset$ such that $\{c\} \times B \subseteq \absGRSI.\rsihb$ and  $B = \absGRSI.\Transactions_{\txid_b} \lor (\stg{b}{\txid_b} \geq 3 \land B = \absGRSI.(\Transactions_{\txid_b} \cap \Writes))$.
As such, sine $\absGRSI.\rsihb$ is transitive we know there exist $A, B \ne \emptyset$ such that $A \times B \subseteq \absGRSI.\rsihb$; that $A = \absGRSI.\Transactions_{\txid_a} \lor (\stg{a}{\txid_a} \leq 2 \land A = \absGRSI.(\Transactions_{\txid_a} \cap \EReads)) \lor (\stg{a}{\txid_a} \leq 4 \land \exsts{d \in \txid_a.\Writes} a \relarrow{\impGRSI.\refC{\po}} d \land A = \{d\})$; 
and that $B = \absGRSI.\Transactions_{\txid_b} \lor (\stg{b}{\txid_b} \geq 3 \land B = \absGRSI.(\Transactions_{\txid_b} \cap \Writes))$, as required.

In case (5.b) from (\ref{IH:rsi_completeness}) we know there exist $A, B \ne \emptyset$ such that $A \times B \subseteq \absGRSI.\rsihb$; that $A = \absGRSI.\Transactions_{\txid_a} \lor (\stg{c}{\txid_a} \leq 2 \land A = \absGRSI.(\Transactions_{\txid_a} \cap \EReads)) \lor (\stg{c}{\txid_a} \leq 4 \land \exsts{d \in \txid_a.\Writes} a \relarrow{\impGRSI.\refC{\po}} d \land A = \{d\})$; 
and that $B = \absGRSI.\Transactions_{\txid_b} \lor (\stg{b}{\txid_b} \geq 3 \land B = \absGRSI.(\Transactions_{\txid_b} \cap \Writes))$.
On the other hand from the proof of the base case we have $(a, c) \in \absGRSI.\po$ and thus from the definition of $\stg{.}{.}$ and the construction of $\impGRSI$ we have $\stg{a}{\txid_a} \leq \stg{c}{\txid_a}$.
As such we know there exist $A, B \ne \emptyset$ such that $A \times B \subseteq \absGRSI.\rsihb$; that $A = \absGRSI.\Transactions_{\txid_a} \lor (\stg{a}{\txid_a} \leq 2 \land A = \absGRSI.(\Transactions_{\txid_a} \cap \EReads)) \lor (\stg{a}{\txid_a} \leq 4 \land \exsts{d \in \txid_a.\Writes} a \relarrow{\impGRSI.\refC{\po}} d \land A = \{d\})$ and that $B = \absGRSI.\Transactions_{\txid_b} \lor (\stg{b}{\txid_b} \geq 3 \land B = \absGRSI.(\Transactions_{\txid_b} \cap \Writes))$, as required.

In case (5.c) from the proof of the base case we know there exist $A, B \ne \emptyset$ such that $A \times B \subseteq \absGRSI.\rsihb$; that $A = \absGRSI.\Transactions_{\txid_a} \lor (\stg{a}{\txid_a} \leq 2 \land A = \absGRSI.(\Transactions_{\txid_a} \cap \EReads)) \lor (\stg{a}{\txid_a} \leq 4 \land \exsts{d \in \txid_a.\Writes} a \relarrow{\impGRSI.\refC{\po}} d \land A = \{d\})$, and that $B = \absGRSI.\Transactions_{\txid_b} \lor (\stg{c}{\txid_b} \geq 3 \land B = \absGRSI.(\Transactions_{\txid_b} \cap \Writes))$.
On the other hand from (\ref{IH:rsi_completeness}) we have $(c, b) \in \absGRSI.\po$ and thus from the definition of $\stg{.}{.}$ and the construction of $\impGRSI$ we have $\stg{c}{\txid_a} \leq \stg{b}{\txid_a}$.
As such we know there exist $A, B \ne \emptyset$ such that $A \times B \subseteq \absGRSI.\rsihb$; that $A = \absGRSI.\Transactions_{\txid_a} \lor (\stg{a}{\txid_a} \leq 2 \land A = \absGRSI.(\Transactions_{\txid_a} \cap \EReads)) \lor (\stg{a}{\txid_a} \leq 4 \land \exsts{d \in \txid_a.\Writes} a \relarrow{\impGRSI.\refC{\po}} d \land A = \{d\})$ and that $B = \absGRSI.\Transactions_{\txid_b} \lor (\stg{b}{\txid_b} \geq 3 \land B = \absGRSI.(\Transactions_{\txid_b} \cap \Writes))$, as required.

In case (5.d) from the proof of the base case we know there exist $A, C_1 \ne \emptyset$ such that $A \times C_1 \subseteq \absGRSI.\rsihb$; that $A = \absGRSI.\Transactions_{\txid_a} \lor (\stg{a}{\txid_a} \leq 2 \land A = \absGRSI.(\Transactions_{\txid_a} \cap \EReads)) \lor (\stg{a}{\txid_a} \leq 4 \land \exsts{d \in \txid_a.\Writes} a \relarrow{\impGRSI.\refC{\po}} d \land A = \{d\})$, and that either 
A) $C_1 = \absGRSI.\Transactions_{\txid_c}$; or
B) $(\stg{c}{\txid_c} \geq 3 \land C_1 = \absGRSI.(\Transactions_{\txid_c} \cap \Writes))$.

On the other hand, from (\ref{IH:rsi_completeness}) we know there exist $C_2, B \ne \emptyset$ such that $C_2 \times B \subseteq \absGRSI.\rsihb$; that $B = \absGRSI.\Transactions_{\txid_b} \lor (\stg{b}{\txid_b} \geq 3 \land B = \absGRSI.(\Transactions_{\txid_b} \cap \Writes))$; and that either:
i) $C_2 = \absGRSI.\Transactions_{\txid_c}$; or
ii) $\stg{c}{\txid_c} \leq 2 \land C_2 = \absGRSI.(\Transactions_{\txid_c} \cap \EReads)$; 
iii) $\stg{c}{\txid_c} \leq 4 \land \exsts{d \in \txid_c.\Writes} c \relarrow{\impGRSI.\refC{\po}} d \land C_2 = \{d\})$.

In cases (A.i-A.iii) and (B.i, B.iii) from the transitivity of $\absGRSI.\rsihb$ we know there exists $A, B \ne \emptyset$ such that $A \times B \subseteq \absGRSI.\rsihb$; that $A = \absGRSI.\Transactions_{\txid_a} \lor (\stg{a}{\txid_a} \leq 2 \land A = \absGRSI.(\Transactions_{\txid_a} \cap \EReads)) \lor (\stg{a}{\txid_a} \leq 4 \land \exsts{d \in \txid_a.\Writes} a \relarrow{\impGRSI.\refC{\po}} d \land A = \{d\})$; and that $B = \absGRSI.\Transactions_{\txid_b} \lor (\stg{b}{\txid_b} \geq 3 \land B = \absGRSI.(\Transactions_{\txid_b} \cap \Writes))$, as required.
Case (B.ii) cannot arise as we would otherwise have $3 \leq \stg{c}{\txid_c} \leq 2$, leading to a contradiction.

\end{proof}
\end{lemma}
%
%Given a relation $\makerel r \in \impGRSI.\Events \times \impGRSI.\Events$, we write $\makerel{r}_{\textsc I}$ for $\makerel r \cap \setcomp{(a, b)}{\exsts{\txid} a, b \in \trace_\txid.\Events}$ and write $\makerel{r}_{\textsf T}$ for $\makerel r \cap \setcomp{(a, b)}{\exsts{\txid_1, \txid_2} a \in \trace_{\txid_1}.\Events \land b \in \trace_{\txid_2}.\Events \land \txid_1 \ne \txid_2}$.
%
%
%
\begin{theorem}[Completeness]
For all RSI execution graphs $\absGRSI$ and their counterpart implementation graphs $\impGRSI$ constructed as above,
\[
	\rsicon  \Rightarrow \consistent{\impGRSI}
\]
\begin{proof}
Pick an arbitrary RSI execution graph $\absGRSI$ and its counterpart implementation graph $\impGRSI$ constructed as above and let us assume that $\rsicon$ holds.
From the \noshade{definition of $\consistent{\impGRSI}$} it then suffices to show: 
\begin{enumerate}
	\item $\irr{\impGRSI.\hb}$ \label{goal:rsi_completeness_hb_irr}
	\item $\irr{\impGRSI.\co ; \impGRSI.\hb}$ \label{goal:rsi_completeness_co_hb_irr}
	\item $\irr{\impGRSI.\fr ; \impGRSI.\hb}$ \label{goal:rsi_completeness_fr_hb_irr} \\
\end{enumerate}

\noindent \textbf{RTS. part \ref{goal:rsi_completeness_hb_irr}}\\
We proceed by contradiction. Let us assume that there exists $a$ such that $(a, a) \in \impGRSI.\hb$.
There are now two cases to consider: 1) $a \in \impGRSI.\NT$; or 2) $\exsts{\txid} a \in \impGRSI.\Events_{\txid}$
In case (1) from \cref{lem:rsi_completeness} we have $(a, a) \in \absGRSI.\rsihb$, contradicting the assumption that $\absGRSI$ is RSI-consistent. 
Similarly, in case (2) from \cref{lem:rsi_completeness} we have $(a, a) \in \impGRSI.\po$, leading to a contradiction as $\impGRSI.\po$ is acyclic by construction. \\

\noindent \textbf{RTS. part \ref{goal:rsi_completeness_co_hb_irr}}\\
We proceed by contradiction. Let us assume that there exist $a, b$ such that $(a, b) \in \impGRSI.\hb$ and $(b, a) \in \impGRSI.\co$.
From the construction of $\impGRSI.\co$ we then know that $(b, a) \in \absGRSI.\co$. Let $\loc a = \loc b = \x$. 
There are then five cases to consider: 
1) $\exsts{\txid} (a, b) \in \impGRSI.\Events_{\txid}$; or 
2) $a, b \in \impGRSI.\NT$; or 
3) $a \in \impGRSI.\NT$ and $b \in \impGRSI.\Events_{\txid_b}$; or 
4) $a \in \impGRSI.\Events_{\txid_a}$ and $b \in \impGRSI.\NT$; or 
5) $a \in \impGRSI.\Events_{\txid_a}$, $b \in \impGRSI.\Events_{\txid_b}$ and $\txid_a \ne \txid_b$. 

In case (1) from \cref{lem:rsi_completeness} we have $(a, b) \in \impGRSI.\po$, and since $a, b \in \impGRSI.\Writes = \absGRSI.\Writes$, from the construction of $\impGRSI.\po$ we have $(a, b) \in \absGRSI.\po$. Moreover, since $a, b$ are write events in the same transaction $\txid$, $(a, b) \in \absGRSI.\poi \cap \Writes^2 \subseteq \absGRSI.\rsihb$. As such we have $a \relarrow{\absGRSI.\rsihb} b \relarrow{\absGRSI.\co} a$, contradicting the assumption that $\absGRSI$ is RSI-consistent. 

In cases (2, 3) since $a, b \in \impGRSI.\Writes = \absGRSI.\Writes$, from \cref{lem:rsi_completeness} we have $(a, b) \in \absGRSI.\rsihb$.
As such we have $a \relarrow{\absGRSI.\rsihb} b \relarrow{\absGRSI.\co} a$, contradicting the assumption that $\absGRSI$ is RSI-consistent. 

Similarly, in cases (4, 5) since $a, b \in \impGRSI.\Writes = \absGRSI.\Writes$, from \cref{lem:rsi_completeness} we have either i) $(a, b) \in \absGRSI.\rsihb$; or ii) $\exsts{d \in \txid_a.\Writes} a \relarrow{\impGRSI.\refC{\po}} d \land (d, b) \in \absGRSI.\rsihb$.
On the other hand, since in case (ii) $a, d \in \impGRSI.\Events_{\txid_a}$, $a \relarrow{\impGRSI.\refC{\po}} d$ and $a, d$ are both write events, we also have $a \relarrow{\absGRSI.\refC{\po}} d$. Moreover, since $\absGRSI.\poi \cap \Writes^2 \subseteq \absGRSI.\rsihb$, we have $(a, d) \in \absGRSI.\refC{\rsihb}$. 
As such we have $a \relarrow{\absGRSI.\refC{\rsihb}} d \relarrow{\absGRSI.\rsihb} b$ and from the transitivity of $\absGRSI.\rsihb$ in both cases we have $(a, b) \in \absGRSI.\rsihb$.
Consequently, we have $a \relarrow{\absGRSI.\rsihb} b \relarrow{\absGRSI.\co} a$, contradicting the assumption that $\absGRSI$ is RSI-consistent. \\

\noindent \textbf{RTS. part \ref{goal:rsi_completeness_fr_hb_irr}}\\
We proceed by contradiction. Let us assume that there exists $a, b$ such that $(a, b) \in \impGRSI.\hb$ and $(b, a) \in \impGRSI.\fr$.
Let $\loc a = \loc b = \x$. 
There are then five cases to consider: 
1) $\exsts{\txid} (a, b) \in \impGRSI.\Events_{\txid}$; or 
2) $a, b \in \impGRSI.\NT$; or 
3) $a \in \impGRSI.\NT$ and $b \in \impGRSI.\Events_{\txid_b}$; or 
4) $a \in \impGRSI.\Events_{\txid_a}$ and $b \in \impGRSI.\NT$; or 
5) $a \in \impGRSI.\Events_{\txid_a}$, $b \in \impGRSI.\Events_{\txid_b}$ and $\txid_a \ne \txid_b$. \\

In case (1) from \cref{lem:rsi_completeness} we have $(a, b) \in \impGRSI.\po$.
On the other hand since $(b, a) \in \impGRSI.\fr$, from the construction of $\impGRSI$ we know that $(b, a) \in \impGRSI.\po$. As such, since $\impGRSI.\po$ is transitively closed and $a \relarrow{\impGRSI.\po} b \relarrow{\impGRSI.\po} a$, we have $(a, a) \in \impGRSI.\po$, leading to a contradiction as $\impGRSI$ is acyclic by construction. 

In case (2) from the construction of $\impGRSI.\fr$ we know that $(a, b) \in \absGRSI.\fr$. 
On the other hand, from \cref{lem:rsi_completeness} we have $(a, b) \in \absGRSI.\rsihb$.
As such we have $a \relarrow{\absGRSI.\rsihb} b \relarrow{\absGRSI.\fr} a$, contradicting the assumption that $\absGRSI$ is RSI-consistent. 

In case (3) from the construction of $\impGRSI.\fr$ we know that $b = \txid_b.\mathit{rs}_{\x} \lor b = \txid_b.\mathit{vs}_{\x}$. 
Consequently, from \cref{lem:rsi_completeness} we have $\{a\} \times \absGRSI.\Transactions_{\txid_b} \subseteq \absGRSI.\rsihb$.
On the other hand, from the construction of $\impGRSI.\fr$ we know there exists $r \in \absGRSI.\Transactions_{\txid_b}$ such that $(r, a) \in \absGRSI.\fr$. 
As such, we have $a \relarrow{\absGRSI.\rsihb} r \relarrow{\absGRSI.\fr} a$, contradicting the assumption that $\absGRSI$ is RSI-consistent. 

In case (4) from the construction of $\impGRSI.\fr$ we know that $(a, b) \in \absGRSI.\fr$. 
Moreover, since $a \in \impGRSI.\Writes$, from \cref{lem:rsi_completeness} we have either i) $(a, b) \in \absGRSI.\rsihb$; or ii) $\exsts{d \in \txid_a.\Writes} a \relarrow{\impGRSI.\refC{\po}} d \land (d, b) \in \absGRSI.\rsihb$.
As such, since $a, d$ are both write events and $a \relarrow{\impGRSI.\refC{\po}} d$, we also have $a \relarrow{\absGRSI.\refC{\po}} d$. Moreover, since $\absGRSI.\poi \cap \Writes^2 \subseteq \absGRSI.\rsihb$, we have $(a, d) \in \absGRSI.\refC{\rsihb}$. 
As such we have $a \relarrow{\absGRSI.\refC{\rsihb}} d \relarrow{\absGRSI.\rsihb} b$ and from the transitivity of $\absGRSI.\rsihb$ in both cases (i, ii) we have $(a, b) \in \absGRSI.\rsihb$.
Consequently, we have $a \relarrow{\absGRSI.\rsihb} b \relarrow{\absGRSI.\fr} a$, contradicting the assumption that $\absGRSI$ is RSI-consistent.

In case (5) from the construction of $\impGRSI.\fr$ we know that $b = \txid_b.\mathit{rs}_{\x} \lor b = \txid_b.\mathit{vs}_{\x}$. 
Consequently, from \cref{lem:rsi_completeness} and since $a$ is a write event, we have either
i) $\absGRSI.\Transactions_{\txid_a} \times \absGRSI.\Transactions_{\txid_b} \subseteq \absGRSI.\rsihb$; or
ii) $\exsts{d \in \txid_a.\Writes} a \relarrow{\impGRSI.\refC{\po}} d \land\{d\} \times \absGRSI.\Transactions_{\txid_b} \subseteq \absGRSI.\rsihb$. 
As such, since $a, d$ are both write events and $a \relarrow{\impGRSI.\refC{\po}} d$, we also have $a \relarrow{\absGRSI.\refC{\po}} d$. Moreover, since $\absGRSI.\poi \cap \Writes^2 \subseteq \absGRSI.\rsihb$, we have $(a, d) \in \absGRSI.\refC{\rsihb}$. 
As such we have $a \relarrow{\absGRSI.\refC{\rsihb}} d \relarrow{\absGRSI.\rsihb} b$ and from the transitivity of $\absGRSI.\rsihb$ in both cases (i, ii) we have $\{a\} \times \absGRSI.\Transactions_{\txid_b} \subseteq \absGRSI.\rsihb$.
On the other hand, from the construction of $\impGRSI.\fr$ we know there exists $r \in \absGRSI.\Transactions_{\txid_b}$ such that $(r, a) \in \absGRSI.\fr$. 
Consequently, we have $a \relarrow{\absGRSI.\rsihb} b \relarrow{\absGRSI.\fr} a$, contradicting the assumption that $\absGRSI$ is RSI-consistent. 
\end{proof}
\end{theorem}

\newcommand{\ws}{\ensuremath{ws}}
\newcommand{\mw}[1]{\func{mw}{#1}}
\newcommand{\mwi}[1]{\funcFont{mw}^{-1}(#1)}
\newpage
\section{Soundness and Completeness of the Lazy RSI Implementation}\label{app:rsi_alternative}
Given an execution graph $\impGRSI$ of the lazy RSI implementation, let us assign a transaction identifier to each transaction executed by the program; and given a transaction $\txid$. 
Let $\mathit{RS}^0_{\txid} = \mathit{WS}^0_{\txid} = \emptyset$. 
Let us write $\impGRSI.\NT$ for those events in $\impGRSI.\Events$ that do not occur in a transaction.
% and let $\mathit{Ts'} = t'_1 \cdots t'_k$ denote its read and write sets, respectively.
Observe that given a transaction $\txid$ of the lazy RSI implementation, the trace of $\txid$, written $\trace_{\txid}$, is of the form: 
\[
	\mathit{Fs}
	\relarrow{\imm \po} \mathit{Is}
	\relarrow{\imm \po} \mathit{Ts}
	\relarrow{\imm \po} \mathit{VRs}
	\relarrow{\imm \po} \mathit{RUs}
	\relarrow{\imm \po} \mathit{PLs}
	\relarrow{\imm \po} \mathit{Ws}
	\relarrow{\imm \po} \mathit{Us}
\]
where:
\begin{itemize}
	\item $\mathit{Fs}$ denotes the sequence of events failing to obtain the necessary locks (i.e.\ those iterations that do not succeed in promoting the writer locks) or validate the snapshot;
	\item $\mathit{Is}$ denotes the sequence of events initialising the values of \code{LS}, \readset and  \writeset with $\emptyset$, initialising \code{ws} with \code{[]} and initialising \code{s[x]} with $(\bot, \bot)$ for each location \x;
	\item $\mathit{Ts}$ denotes the sequence of events corresponding to the execution of \denot{\code{T}} and is of the form $\mathit{t}_1 \relarrow{\imm \po} \cdots \relarrow{\imm \po} \mathit{t}_k$, where for all $m \in \{1 \cdots k\}$:
	\[
	\mathit{t}_m {=} 
	\begin{cases}
		\rseq{\code x_m}{v_m}{\mathit{RS}_{m {-} 1},  \mathit{WS}_{m {-} 1}}
		\relarrow{\imm \po} \mathit{lr}_{\x_m}
		& \hspace{-5pt}\text{if } O_m {=} \readE{-}{\code x_m}{v_m} \\
		\wseq{\code x_m}{v_m}{\mathit{RS}_{m {-} 1},  \mathit{WS}_{m {-} 1}} 
		\relarrow{\imm \po} \mathit{wws}_{\x_m} 
		& \hspace{-5pt} \text{if } O_m {=} \writeE{\rel}{\code x_m}{v_m} \\
		\relarrow{\imm \po} \readE{}{\code{s[x}_m\code ]}{(v'_m, -)}
		\relarrow{\imm \po} \mathit{lw}_{\x_m}
		\relarrow{\imm \po} \writeE{}{\code{ws}}{\mathit{ws}_m}
	\end{cases}
	\]
	where $O_m$ denotes the $m$th event in the trace of $\code T$; 
	$\mathit{lr}_{\x_m} \eqdef \readE{}{\code{s[x}_m\code ]}{(-,v_m)}$; 
\[
\small
	\rseq{\code x_m}{v_m}{\mathit{RS}_{m {-} 1},  \mathit{WS}_{m {-} 1}}
	\eqdef
	\begin{cases}
		\readE{}{\code{s[x}_m\code ]}{(\bot,\bot)}
		& \text{if } \code x_m \not\in \mathit{RS}_{m {-} 1} \cup  \mathit{WS}_{m {-} 1} \\
		\relarrow{\imm \po} \mathit{fs_m}  \\
		\relarrow{\imm \po} \mathit{rl}_{\x_m} \\
		\relarrow{\imm \po} \mathit{wrs}_{\x_m} \\
		\relarrow{\imm \po} \mathit{rs}_{\x_m} \\
		\relarrow{\imm \po} \mathit{ws}_{\x_m}
		& \\\\
		\emptyset
		& \text{otherwise} 
	\end{cases}
\]	
$\mathit{fs}_m$ denotes the sequence of events attempting (but failing) to acquire the read lock on $\code x_m$,  
$\mathit{rl}_{\code x_m} \eqdef \rlockE{\xl_m}$, 
$\mathit{wrs}_{\x_m} \eqdef \writeE{}{\readset}{\mathit{RS}_m}$, 
$\mathit{rs}_{\code x_m} \eqdef \readE{}{\code x_m}{v_{\x_m}^0}$, 
$\mathit{ws}_{\code x_m} \eqdef \writeE{}{\code{s[x}_m\code ]}{(v_{\x_m}^0, v_{\x_m}^0)}$; 
and for all $m > 0$:
\[
	\mathit{RS}_{m {+} 1} \eqdef 
	\begin{cases}
		\mathit{RS}_m \cup \{\code x_m\} & \text{if } O_m {=} \readE{-}{\code x_m}{-} \\
		\mathit{RS}_m & \text{otherwise}
	\end{cases}
\]
and
\[
\small
	\wseq{\code x_m}{v_m}{\mathit{RS}_{m {-} 1},  \mathit{WS}_{m {-} 1}}
	\eqdef
	\begin{cases}
		\readE{}{\code{s[x}_m\code ]}{(\bot,\bot)}
		& \text{if } \code x_m \not\in \mathit{RS}_{m {-} 1} \cup  \mathit{WS}_{m {-} 1} \\
		\relarrow{\imm \po} \mathit{fs_m}  \\
		\relarrow{\imm \po} \mathit{rl}_{\x_m}
		\emptyset
		& \text{otherwise} 
	\end{cases}
\]	
$\mathit{wws}_{\x_m} \eqdef \writeE{}{\writeset}{\mathit{WS}_m}$; 
$\mathit{lw}_{\x_m} \eqdef \writeE{}{\code{s[x}_m\code ]}{(v'_m, v_m)}$; 
$\mathit{fs}_m$ and $\mathit{rl}_{\code x_m}$ are as defined above; and for all $m > 0$:
\[
\begin{array}{@{} c @{}}
	\mathit{WS}_{m {+} 1} \eqdef 
	\begin{cases}
		\mathit{WS}_m \cup \{\code x_m\} & \text{if } O_m {=} \writeE{-}{\code x_m}{-} \\
		\mathit{WS}_m & \text{otherwise}
	\end{cases} \\
	\mathit{ws}_{m {+} 1} \eqdef 
	\begin{cases}
		\mathit{ws}_m {++} [(\x_m, v_m)] & \text{if } O_m {=} \writeE{-}{\code x_m}{v_m} \\
		\mathit{ws}_m & \text{otherwise}
	\end{cases}
\end{array}
\]	
Let $\readset_\txid = \mathit{RS}_m$,  
$\writeset_\txid = \mathit{WS}_m$, and
$\ws_{\txid} = \ws_m$;
let $\readset_\txid \cup \writeset_\txid$ be enumerated as $\{\x_1 \cdots \x_i\}$ for some $i$, 
and $\ws_\txid$ be enumerated as $\{(\x_1, v_1) \cdots (\x_j, v_j)\}$ for some $j$.
	\item $\mathit{VRs}$ denotes the sequence of events validating the reads and is of the  form $\mathit{v}_{\code x_1} \relarrow{\imm \po} \cdots \relarrow{\imm \po} \mathit{v}_{\code x_i}$, where for all $n \in \{1 \cdots i\}$:
	\[
	\begin{array}{l}
		\mathit{v}_{\code x_n} = 
		\begin{cases}
			\readE{}{\code{s[x$_n$]}}{(v_{\x_n}^0, -)}	
			\relarrow{\imm \po} \mathit{vr}_{\code x_n} {=} \readE{}{\x_n}{v_{\x_n}^0}
			& \text{ if } \code x_n \in \readset_{\txid}  \\
			\emptyset
			& \text{ otherwise}
		\end{cases}
	\end{array}	
	\]
	\item $\mathit{RUs}$ denotes the sequence of events releasing the reader locks (when the given location is in the read set only) and is of the form $\mathit{ru}_{\code x_1} \relarrow{\imm \po} \cdots \relarrow{\imm \po} \mathit{ru}_{\code x_i}$, where for all $n \in \{1 \cdots i\}$:
	\[
	\begin{array}{l}
		\mathit{ru}_{\code x_n} = 
		\begin{cases}
			\runlockE{\xl_n}
			& \text{ if } \code x_n \not\in \writeset_{\txid}  \\
			\emptyset
			& \text{ otherwise}
		\end{cases}
	\end{array}	
	\]
	\item $\mathit{PLs}$ denotes the sequence of events promoting the reader locks to writer ones (when the given location is in the write set), and is of the form $\mathit{pl}_{\code x_1} \relarrow{\imm \po} \cdots \relarrow{\imm \po} \mathit{pl}_{\code x_i}$, where for all $n \in \{1 \cdots i\}$:
	\[
	\begin{array}{l}
		\mathit{pl}_{\code x_n} = 
		\begin{cases}
			\plockE{\xl_n}
			& \text{if }  \code x_n \in \writeset_{\txid} \\
			\emptyset
			& \text{ otherwise } 
		\end{cases}
	\end{array}	
	\]
	\item $\mathit{Ws}$ denotes the sequence of events committing the writes of \denot{\code{T}} and is of the form $\mathit{c}_{\x_1, v_1} \relarrow{\imm \po} \cdots \relarrow{\imm \po} \mathit{c}_{\x_j, v_j}$, where for all $n \in \{1 \cdots j\}$:
$
	\mathit{c}_{\x_n, v_n} =  \writeE{}{\x_n}{v_n} 
%	\begin{cases}
%		\readE{}{\code{s[x}_n \code{]}}{v_n}  \relarrow{\imm \po} 
%		\mathit{w}_{\x_n} {=} \writeE{}{\x_n}{v_n} 
%		& \text{if } \x_n \in \writeset_\txid \\
%		\emptyset 
%		& \text{otherwise}
%	\end{cases}
$
	\item $\mathit{Us}$ denotes the sequence of events releasing the locks on the write set, and is of the form $\mathit{wu}_{\code x_1} \relarrow{\imm \po} \cdots \relarrow{\imm \po} \mathit{wu}_{\code x_i}$, where for all $n \in \{1 \cdots i\}$:
	\[
		\mathit{wu}_{\code x_n} = 
		\begin{cases}
			\wunlockE{\code{xl}_n} & \text{if } \code x_n \in \writeset_{\txid} \\
			\emptyset & \text{otherwise}
		\end{cases}				
	\]
\end{itemize}

%and $\accset = \simpleset{\code z_1, \cdots, \code z_k}$, 
%where $\accset$ denotes the set of locations both read from and written to by \code T; that is, $\accset \subseteq \readset \land \accset \subseteq \writeset$. 

Given a transaction trace $\trace_{\txid}$, we write e.g.~$\txid.\mathit{Us}$ to refer to its constituent $\mathit{Us}$ sub-trace and write $\mathit{Us}.\Events$ for the set of events related by \po in $\mathit{Us}$. Similarly, we write $\txid.\Events$ for the set of events related by \po in $\trace_{\txid}$.
Note that $\impGRSI.\Events = \bigcup\limits_{\txid \in \sort{Tx}}  \txid.\Events $.
%As such, the $\sti$ relation induces a set of equivalence classes on $\absGRSI.\Events$, written $\absGRSI.\Events/\sti$. As before, we write $\class{a}{\sti}$ for the equivalence class in $\absGRSI.\Events/\sti$ that contains $a$.

Note that for each transaction $\txid$ and each location $\x$, the $\txid.\mathit{rl}_\x$, $\txid.\mathit{rs}_\x$, $\txid.\mathit{vr}_\x$, $\txid.\mathit{ru}_\x$, $\txid.\mathit{pl}_\x$ and $\txid.\mathit{wu}_\x$ are uniquely identified when they exist. 
Let $\txid.\mathit{w}_\x$ the last (in \po order) write to \x in $\mathit{Ws}$, when it exists.

For each location $\x \in \writeset_\txid$, let $\fw \x$ denote the maximal write (in $\po$ order within $\txid$) logging a write for \x in \code{s[x]}. That is,  when $\trace_{\txid}=  t_1 \relarrow{\imm \po} \cdots \relarrow{\imm \po} t_m$, let $\fw \x = \func{wmax}{\x, [t_1 \cdots t_m]}$, where
\[
\begin{array}{@{} c @{}}
	 \func{wmax}{\x, [\,]} \text{ undefined} \\
	  \func{wmax}{\x, L {++} [t]}
	  \eqdef
	  \begin{cases}
	  	\mathit{lw}_{\x}
	  	& \text{if } t {=} \wseq{\x}{-}{-, -} \relarrow{\po} \mathit{lw}_{\x}  \relarrow{\po} \mathit{wws}_{\x} \\
	  	 \func{wmax}{\x, L}
	  	 & \text{otherwise}
	  \end{cases}
\end{array}
\]
%
%Similarly, for each location $\x \in \writeset_\txid$, let $\iw \x$ denote the minimal write (in $\po$ order within $\txid$) logging a write for \x in \code{s[x]}. That is,  when $\trace_{\txid}=  t_1 \relarrow{\imm \po} \cdots \relarrow{\imm \po} t_m$, let $\iw \x = \func{wmin}{\x, [t_1 \cdots t_m]}$, where
%\[
%	 \func{wmin}{\x, [\,]} \text{ undefined}
%	 \quad
%	  \func{wmin}{\x, [t] {++} L}
%	  \eqdef
%	  \begin{cases}
%	  	\mathit{lw}_{\x}
%	  	& \text{if } t {=} \wseq{\x}{-}{-, -} \relarrow{\po} \mathit{lw}_{\x}  \relarrow{\po} \mathit{wws}_{\x} \\
%	  	 \func{wmin}{\x, L}
%	  	 & \text{otherwise}
%	  \end{cases}
%\]
%%
%%
%
%
%
%
\subsection{Implementation Soundness}
In order to establish the soundness of our implementation, it suffices to show that given an RA-consistent execution graph of the implementation $\impGRSI$, we can construct a corresponding RSI-consistent execution graph $\absGRSI$ with the same outcome.
%
%\[
%	\for{\impGRSI} \consistent{\impGRSI} \Rightarrow \exsts{\absGRSI} \rsicon
%\]
%%

Given a transaction $\txid \in \sort{Tx}$ with $\readset_{\txid} \cup \writeset_{\txid}=\{\code x_1 \cdots \code x_i\}$
%, $\writeset_{\txid_t}=\{\code y_1 \cdots \code y_j\}$ 
and trace 
$ \trace_{\txid}$ as above with  $\mathit{Ts}= \mathit{t}_1 \relarrow{\imm \po} \cdots \relarrow{\imm \po} \mathit{t}_k$, 
we construct the corresponding RSI execution trace $\trace'_{\txid}$ as follows:
\[
	\trace'_{\txid} \eqdef  \mathit{t}'_1 \relarrow{\imm \po} \cdots \relarrow{\imm \po} \mathit{t}'_k
\]
where  for all $m \in \{1 \cdots k\}$:
\[
\begin{array}{l c l}
	\mathit{t}'_m {=} \readE{}{\code x_m}{v_m}
	& \text{when} & 
	\mathit{t}_m = 
	\rseq{\x_m}{v_m}{S_m}
	\relarrow{\imm \po} \mathit{lr}_{\x_m} \\
	
	\mathit{t}'_m {=} \writeE{}{\code x_m}{v_m}
	& \text{when} & 
	\mathit{t}_m = 
	\wseq{\x_m}{v_m}{S_m}
	\relarrow{\imm \po} \mathit{wws}_{\x_m}
		\relarrow{\imm \po} \cdots
\end{array}	
\]
such that in the first case the identifier of $\mathit{t}'_m$ is that of $\mathit{lr}_{\x_m}$; 
and in the second case the identifier of $\mathit{t}'_m$ is that of $\mathit{lw}_{\x_m}$.
Note that for each write operation $w$ in $\trace'_{\txid}$, there exists a \emph{matching} write operation in $\trace_{\txid}.\mathit{Ws}$, denoted by $\mw{w}$. 
That is, 
\[
	\mw{w} {=} w'
	\iffdef
	\exsts{i}
	\land \itemAt{(\trace'_{\txid}.\Events \cap \Writes)}{i} = w
	\land  \itemAt{(\trace_{\txid}.\mathit{Ws})}{i} = w'
\]
We then define:
\[
	\mathsf{RF}_{\txid} \eqdef
	\begin{array}[t]{@{} l @{}}
	\setcomp{
		(w, t'_j)	
	}{
		t'_j \in \trace'_{\txid}.\Events \land \exsts{\code x, v} t'_j {=} \readE{\acq}{\code x}{v} \land w {=} \writeE{\rel}{\code x}{v} \\
		\land (w \in \txid.\Events \Rightarrow 
		\begin{array}[t]{@{} l @{}}
			w \relarrow{\po} t'_j \,\land\\
			(\for{e \in \txid.\Events } w \relarrow{\po} e \relarrow{\po} t'_j \Rightarrow (\loc e {\ne} \code x \lor e {\not\in} \Writes)))
		\end{array} \\				
		
		\land (w \not \in \txid.\Events \Rightarrow 
		\begin{array}[t]{@{} l @{}}
			(\for{e \in \txid.\Events} (e \relarrow{\po} t'_j \Rightarrow (\loc e \ne \code x \lor e \not\in \Writes)) \\
			\land\, 
			\big(
			(\exsts{\txid'} 
		 	(\txid'.w_\x, \txid.\mathit{rs}_{\x})  \in \impGRSI.\rf)
			\land w {=} \txid'.\fw \x) \\
			\;\quad \lor 
			(w \in \impGRSI.\NT \land (w, \mathit{rs}_{\x})  \in \impGRSI.\rf)
			\big)
		\end{array}		
	} \\
	\cup
	\setcomp{
		(w, r)	
	}{
		r \in \impGRSI.\NT \,\land \\
		\big(
			(w \in \impGRSI.\NT \land (w, r) \in \impGRSI.\rf)  \\
			\quad \lor (\exsts{w'} 
			w' = \mw{w}
			\land (w', r) \in \impGRSI.\rf) 
		\big)	
	}
	\end{array}	
\]
Similarly, we define: 
\[
	\mathsf{MO} \eqdef
	\begin{array}[t]{@{} l @{}}
	\setcomp{
		(w_1, w_2)
	}{
		\exsts{w'_1, w'_2} 
		w'_1 = \mw{w_1}
		\land w'_2 = \mw{w_2}
		\land (w'_1, w'_2) \in \impGRSI.\co 
	} \\
	\cup
	\setcomp{
		(w, w')	
	}{
		w, w' \in \impGRSI.\NT \land (w, w') \in \impGRSI.\co	
	}
	 \\
	\cup
	\setcomp{
		(w, w')	
	}{
		w' \in \impGRSI.\NT \land 
		\exsts{w''} 
		w'' = \mw{w}
		\land (w'', w') \in \impGRSI.\co 
	} \\
	\cup 
	\setcomp{
		(w', w)	
	}{
		w' \in \impGRSI.\NT \land 
		\exsts{w''} 
		w'' = \mw{w}
		\land (w', w'') \in \impGRSI.\co 
	}
\end{array}	
\]
We are now in a position to demonstrate the soundness of our implementation. Given an RA-consistent execution graph $\impGRSI$ of the implementation, we construct an RSI execution graph $\absGRSI$ as follows and demonstrate that $\rsicon$ holds.

\begin{itemize}
	\item $\absGRSI.\Events = \bigcup\limits_{\txid \in \sort{Tx}} \trace'_{\txid}.\Events \cup \impGRSI.\NT$, with the $\tx{.}$ function defined as:
	\[
		\tx{a} \eqdef \txid \quad \text{ when } \quad a \in \trace'_{\txid}
		\qquad 
		\tx{a} \eqdef 0 \quad \text{ when } \quad a \in \impGRSI.\NT
	\]
	\item $\absGRSI.\po = \coerce{\impGRSI.\po}{\absGRSI.\Events}$
	\item $\absGRSI.\rf = \bigcup_{\txid \in \textsc{Tx}} \mathsf{RF}_{\txid}$%\absGRSI.\Writes \times \absGRSI.\Reads \cap \setcomp{(a, b)}{\src{b} = a}$
	\item $\absGRSI.\co = \mathsf{MO}$
	\item $\absGRSI.\silo = \emptyset$
%	\item $\absGRSI.\Transactions = \absGRSI.\Events$
%	with the $\tx{.}$ function defined as:
%	\[
%		\tx{a} \eqdef \txid \quad \text{ where } \quad a \in \trace'_{\txid}
%	\]
\end{itemize}
Observe that the events of each $\trace'_{\txid}$ trace coincides with those of the equivalence class  $\Transactions_{\txid}$ of $\absGRSI$. That is,  $\trace'_{\txid}.\Events = \Transactions_{\txid}$. 

\begin{lemma}\label{lem:rsi_alt_lock_hb}
Given an RA-consistent execution graph $\impGRSI$ of the implementation and its corresponding RSI execution graph $\absGRSI$ constructed as above, for all $a, b, \txid_a, \txid_b, \x$:
\small
\begin{align}
	&  \hspace*{-15pt}\txid_a \ne \txid_b
	\land a \in \txid_a.\Events 
	\land b \in \txid_b.\Events 
	\land \loc a = \loc b = \x
	\Rightarrow \nonumber \\
	& \hspace*{-15pt} \;\;  
	((a, b) \in \absGRSI.\rf \Rightarrow  \txid_a.\mathit{wu}_{\x} \relarrow{\impGRSI.\hb} \txid_b.\mathit{rl}_{\x} ) 
	\label{lem:rsi_alt_lock_hb_rf} \\
	& \hspace*{-15pt}\;\; \land ((a, b) \in \absGRSI.\co \Rightarrow  \txid_a.\mathit{wu}_{\x} \relarrow{\impGRSI.\hb} \txid_b.\mathit{rl}_{\x} ) 
	\label{lem:rsi_alt_lock_hb_co}\\
%	& \quad \land ((a, b) \in \absGRSI.\co;\rf \Rightarrow  \txid_a.\mathit{wu}_{\x} \relarrow{\impGRSI.\hb} \txid_b.\mathit{rl}_{\x} ) 
%	\label{lem:rsi_alt_lock_hb_corf}\\
	& \hspace*{-15pt}\;\; \land \big((a, b) \in \absGRSI.\fr \Rightarrow  
		(\x \in \writeset_{\txid_a} \land \txid_a.\mathit{wu}_{\x} \relarrow{\impGRSI.\hb} \txid_b.\mathit{rl}_{\x} ) 
		\lor 
		(\x \not\in \writeset_{\txid_a} \land \txid_a.\mathit{ru}_{\x} \relarrow{\impGRSI.\hb} \txid_b.\mathit{pl}_{\x} ) 
	\big)
	\label{lem:rsi_alt_lock_hb_fr} \\
	& \hspace*{-15pt} \;\; \land ((a, b) \in \absGRSI.(\co;\rf) \Rightarrow  \txid_a.\mathit{wu}_{\x} \relarrow{\impGRSI.\hb} \txid_b.\mathit{rl}_{\x} ) 
	\label{lem:rsi_alt_lock_hb_co_rf} 
\end{align}	
\normalsize
\begin{proof}
Pick an arbitrary RA-consistent execution graph $\impGRSI$ of the implementation and its corresponding RSI execution graph $\absGRSI$ constructed as above. Pick an arbitrary $a, b, \txid_a, \txid_b, \x$ such that $\txid_a \ne \txid_b$, $a \in \txid_a.\Events$, $a \in \txid_a.\Events$, and $\loc a = \loc b = \x$.\\

\noindent \textbf{RTS. (\ref{lem:rsi_alt_lock_hb_rf})}\\
Assume $(a, b) \in \absGRSI.\rf$. From the definition of $\absGRSI.\rf$ we then know $(\txid_a.w_\x, \txid_b.\mathit{rs}_{\x}) \in \impGRSI.\rf$.
On the other hand, from \cref{lem:lock-ordering} we  know that either i) $\x \in \writeset_{\txid_b}$ and $\txid_b.\mathit{wu}_{x} \relarrow {\impGRSI.\hb} \txid_a.\mathit{rl}_{x}$; or ii)  $\x \not\in \writeset_{\txid_b}$ and  $\txid_b.\mathit{ru}_{x} \relarrow {\impGRSI.\hb} \txid_a.\mathit{pl}_{x}$; or iii) $\txid_a.\mathit{wu}_{x} \relarrow {\impGRSI.\hb} \txid_b.\mathit{rl}_{x}$.
In case (i) we then have $\txid_a.w_\x \relarrow{\impGRSI.\rf} \txid_b.\mathit{rs}_{\x} \relarrow{\impGRSI.\po} \txid_b.\mathit{wu}_{x}  \relarrow{\impGRSI.\hb} \txid_a.\mathit{rl}_{x}  \relarrow{\impGRSI.\po} \txid_a.w_\x$. That is, we have $\txid_a.w_\x \relarrow{\impGRSI.\hbloc} \txid_a.w_\x$, contradicting the assumption that $\impGRSI$ is RA-consistent. 
Similarly in case (ii) we have $\txid_a.w_\x \relarrow{\impGRSI.\rf} \txid_b.\mathit{rs}_{\x} \relarrow{\impGRSI.\po} \txid_b.\mathit{ru}_{x}  \relarrow{\impGRSI.\hb} \txid_a.\mathit{pl}_{x}  \relarrow{\impGRSI.\po} \txid_a.w_\x$.  That is, we have $\txid_a.w_\x \relarrow{\impGRSI.\hbloc} \txid_a.w_\x$, contradicting the assumption that $\impGRSI$ is RA-consistent. 
In case (iii) the desired result holds trivially.\\

\noindent \textbf{RTS. (\ref{lem:rsi_alt_lock_hb_co})}\\
Assume $(a, b) \in \absGRSI.\co$. From the definition of $\absGRSI.\co$ we then know there exist $w_1 \in \txid_a.\mathit{Ws}$ and $w_2 \in \txid_b.\mathit{Ws}$ such that $(w_1, w_2) \in \impGRSI.\co$ and $\loc{w_1} = \loc{w_2} = \x$.
On the other hand, from \cref{lem:lock-ordering}  we  know that either i) $\txid_b.\mathit{wu}_{x} \relarrow {\impGRSI.\hb} \txid_a.\mathit{rl}_{x}$; or ii) $\txid_a.\mathit{wu}_{x} \relarrow {\impGRSI.\hb} \txid_b.\mathit{rl}_{x}$.
In case (i) we then have $w_1 \relarrow{\impGRSI.\co} w_2 \relarrow{\impGRSI.\po} \txid_b.\mathit{wu}_{x}  \relarrow{\impGRSI.\hb} \txid_a.\mathit{rl}_{x}  \relarrow{\impGRSI.\po} w_1$. That is, we have $w_1 \relarrow{\impGRSI.\co} w_2 \relarrow{\impGRSI.\hbloc} w_1$, contradicting the assumption that $\impGRSI$ is RA-consistent. 
In case (ii) the desired result holds trivially.\\

%\noindent \textbf{RTS. (\ref{lem:rsi_alt_lock_hb_corf})}\\
%Assume $(a, b) \in \absGRSI.\co; \rf$. We then know there exists $w$ such that $(a, w) \in \absGRSI.\co$ and $(w, b) \in \absGRSI.\rf$. From the definition of $\absGRSI.\co$ we then know $(a, w) \in \impGRSI.\co$.
%There are now three cases to consider: 1) $w \in \txid_a$; or 2) $w \in \txid_b$; or 3) $w \in \txid_c \land \txid_c \ne \txid_a \land \txid_c \ne \txid_b$. 
%In case (1) the desired result follows from part \ref{lem:rsi_alt_lock_hb_rf}.
%In case (2) since $(a, w) \in \absGRSI.\co$ the desired result follows from part \ref{lem:rsi_alt_lock_hb_co}.
%
%In case (3) from the proof of part \ref{lem:rsi_alt_lock_hb_co} we have $\txid_a.\mathit{wu}_{x} \relarrow {\impGRSI.\hb} \txid_c.\mathit{rl}_{x}$.
%Moreover, from the shape of $\impGRSI$ traces we have $\txid_c.\mathit{rl}_{x} \relarrow {\impGRSI.\po} \txid_c.\mathit{wu}_{x}$.
%On the other hand, from the proof of part \ref{lem:rsi_alt_lock_hb_rf} we have $\txid_c.\mathit{wu}_{x} \relarrow {\impGRSI.\hb} \txid_b.\mathit{rl}_{x}$.
%We thus have 
%$\txid_a.\mathit{wu}_{x} \relarrow {\impGRSI.\hb} \txid_c.\mathit{rl}_{x} \relarrow {\impGRSI.\po} \txid_c.\mathit{wu}_{x} \relarrow {\impGRSI.\hb} \txid_b.\mathit{rl}_{x}$.
%As $\impGRSI.\po \subseteq \impGRSI.\hb$ and $\impGRSI.\hb$ is transitively closed, we have $\txid_a.\mathit{wu}_{x} \relarrow {\impGRSI.\hb} \txid_b.\mathit{rl}_{x}$, as required. \\

\noindent \textbf{RTS. (\ref{lem:rsi_alt_lock_hb_fr})}\\
Assume $(a, b) \in \absGRSI.\fr$. From the definition of $\absGRSI.\fr$ we then know that 
there exist $w \in \txid_b.\mathit{Ws}$ such that $\loc w = \x$ and $(\txid_a.\mathit{rs}_\x, w) \in \impGRSI.\fr$. 
%either 
%$(\txid_a.w_{\x}, \txid_b.w_\x) \in \impGRSI.\co$
%or $(\txid_a.\mathit{rs}_{\x}, \txid_b.w_\x) \in \impGRSI.\fr$.
%%, (\txid_a.\mathit{vs}_{\x}, b)
%In the former case the desired result follows immediately from the proof of part \eqref{lem:rsi_alt_lock_hb_co}.
From \cref{lem:lock-ordering} we then know that either i) $\txid_b.\mathit{wu}_{x} \relarrow {\impGRSI.\hb} \txid_a.\mathit{rl}_{x}$; or ii)  $\x \not\in \writeset_{\txid_a}$ and $\txid_a.\mathit{ru}_{x} \relarrow {\impGRSI.\hb} \txid_a.\mathit{pl}_{x}$; or iii) $\x \in \writeset_{\txid_a}$ and  $\txid_a.\mathit{wu}_{x} \relarrow {\impGRSI.\hb} \txid_b.\mathit{rl}_{x}$.
In case (i) we then have $w \relarrow{\impGRSI.\po} \txid_b.\mathit{wu}_{x}  \relarrow{\impGRSI.\hb} \txid_a.\mathit{rl}_{x}  \relarrow{\impGRSI.\po} \txid_a.\mathit{rs}_{\x} \relarrow{\impGRSI.\fr} w$. That is, we have $w \relarrow{\impGRSI.\hbloc} \txid_a.\mathit{rs}_{\x} \relarrow{\impGRSI.\fr} w$, contradicting the assumption that $\impGRSI$ is RA-consistent. 
In cases (ii-iii) the desired result holds trivially.
\end{proof}
\end{lemma}

\begin{lemma}
\label{lem:rsi_alt_soundness}
For all RA-consistent execution graphs $\impGRSI$ of the implementation and their counterpart RSI execution graphs $\absGRSI$ constructed as above and $S \eqdef \transC{\absGRSI.(\rsipo \cup \rsirf \cup \cot)}$:
\[
\begin{array}{@{} r @{\hspace{2pt}} l @{}}
	\for{\txid_a, \txid_b, a, b}
	\for{a, b} \\
	\quad 
	(a, b) \in  S
	\Rightarrow
	& 
	\begin{array}[t]{@{} l @{}}
		(a, b \in \absGRSI.\NT \land (a, b) \in \impGRSI.\hb) \\
		\lor
		(a \in \absGRSI.\Transactions_{\txid_a} 
		\land b  \in \absGRSI.\Transactions_{\txid_b}
		\land \txid_a = \txid_b
		\land (a,b) \in \impGRSI.\po
		) \\
		\lor 
		\left(
		\begin{array}{@{} l @{}}
			a \in \absGRSI.\Transactions_{\txid_a} \land b  \in \absGRSI.\Transactions_{\txid_b} 
			\land \txid_a \ne \txid_b\\
			\land \exsts{d \in  \absGRSI.\Transactions_{\txid_b}} 
			\for{c \in \absGRSI.\Transactions_{\txid_a}} (c, d) \in \impGRSI.\hb \\
			\land\, a \in \Writes \Rightarrow (\mw{a}, d) \in \impGRSI.\hb
%			\land\, (\exsts{e \in \txid_a.\mathit{Ws}}  (e, d) \in \impGRSI.\hb
%				\lor \for{c \in \txid_a.\Events} (c, d) \in \impGRSI.\hb)
		\end{array}
		\right) \\
		\lor
		\left(
		\begin{array}{@{} l @{}}
			a \in \absGRSI.\NT \land b  \in \absGRSI.\Transactions_{\txid_b} \\
			\land \exsts{d \in  \absGRSI.\Transactions_{\txid_b}} 
			(a, d) \in \impGRSI.\hb 
		\end{array}
		\right) \\
		\lor
		\left(
		\begin{array}{@{} l @{}}
			a \in \absGRSI.\Transactions_{\txid_a} \land b  \in \absGRSI.\NT \\
			\land \for{c \in \absGRSI.\Transactions_{\txid_a}} (c, b) \in \impGRSI.\hb \\
			\land\, (\exsts{e \in \txid_a.\mathit{Ws}}  (e, b) \in \impGRSI.\hb
				\lor \for{c \in \txid_a.\Events} (c, b) \in \impGRSI.\hb)
		\end{array}
		\right) \\
	\end{array}		
\end{array}	
\]
\begin{proof}
Let $S_0 = \absGRSI.(\rsipo \cup \rsirf \cup \cot)$, and $S_{n {+} 1} = S_0; S_n$, for all $n >=0$. 
It is straightforward to demonstrate that $S = \bigcup\limits_{i \in \Nats} S_i$. 
We thus demonstrate instead that:
\[
\begin{array}{@{} r @{\hspace{2pt}} l @{}}
	\for{i \in \Nats} \for{\txid_a, \txid_b, a, b}
	\for{a, b} \\
	\quad 
	(a, b) \in  S_i
	\Rightarrow
	& 
	\begin{array}[t]{@{} l @{}}
		(a, b \in \absGRSI.\NT \land (a, b) \in \impGRSI.\hb) \\
		\lor
		(a \in \absGRSI.\Transactions_{\txid_a} 
		\land b  \in \absGRSI.\Transactions_{\txid_b}
		\land \txid_a = \txid_b
		\land (a,b) \in \impGRSI.\po
		) \\
		\lor 
		\left(
		\begin{array}{@{} l @{}}
			a \in \absGRSI.\Transactions_{\txid_a} \land b  \in \absGRSI.\Transactions_{\txid_b} 
			\land \txid_a \ne \txid_b\\
			\land \exsts{d \in  \absGRSI.\Transactions_{\txid_b}} 
			\for{c \in \absGRSI.\Transactions_{\txid_a}} (c, d) \in \impGRSI.\hb \\
			\land\, (\exsts{e \in \txid_a.\mathit{Ws}}  (e, d) \in \impGRSI.\hb \\
				\quad \lor \for{c \in \txid_a.\Events} (c, d) \in \impGRSI.\hb)
		\end{array}
		\right) \\
		\lor
		\left(
		\begin{array}{@{} l @{}}
			a \in \absGRSI.\NT \land b  \in \absGRSI.\Transactions_{\txid_b} \\
			\land \exsts{d \in  \absGRSI.\Transactions_{\txid_b}} 
			(a, d) \in \impGRSI.\hb 
		\end{array}
		\right) \\
		\lor
		\left(
		\begin{array}{@{} l @{}}
			a \in \absGRSI.\Transactions_{\txid_a} \land b  \in \absGRSI.\NT \\
			\land \for{c \in \absGRSI.\Transactions_{\txid_a}} (c, b) \in \impGRSI.\hb \\
			\land\, (\exsts{e \in \txid_a.\mathit{Ws}}  (e, b) \in \impGRSI.\hb \\
				\quad \lor \for{c \in \txid_a.\Events} (c, b) \in \impGRSI.\hb)
		\end{array}
		\right) \\
	\end{array}		
\end{array}	
\]
We proceed by induction on $i$.\\

\noindent \textbf{Base case $i = 0$}\\
Pick arbitrary $\txid_a, \txid_b$, $a, b$ such that $\txid_a \ne \txid_b$ and $(a, b) \in  S_0$. 
There are now four cases to consider: 
A) $a \in \absGRSI.\Transactions_{\txid_a}, b  \in \absGRSI.\Transactions_{\txid_b}$; or
B) $a, b \in \absGRSI.\NT$; or 
C) $a \in \absGRSI.\NT, b \in \absGRSI.\Transactions_{\txid_b}$; or 
D) $a \in \absGRSI.\Transactions_{\txid_a}, b  \in \absGRSI.\NT$.

In case (A) there are three additional cases to consider: 
1) $(a, b) \in \absGRSI.\rsipo$; or 
2) $(a, b) \in \absGRSI.\rsirf$; or 
3) $(a, b) \in \absGRSI.\cot$.

In case (A.1), pick an arbitrary $c \in \absGRSI.\Transactions_{\txid_a}$. From the definition of $\absGRSI.\pot$ we have $(c, b) \in \impGRSI.\po \subseteq \impGRSI.\hb$, as required. 
%Pick an arbitrary $c \in \txid_a.\Events$. From the definition of $\absGRSI.\pot$ we have $(c, b) \in \impGRSI.\po \subseteq \impGRSI.\hb$, as required. 
Now assume that $a \in \Writes$. From the definition of $\absGRSI.\po$ we then have $(\mw{a}, d) \in \impGRSI.\po \subseteq \impGRSI.\hb$, as required. 

In case (A.2), we then know there exists $w \in \absGRSI.\Transactions_{\txid_a}$ and $r \in \absGRSI.\Transactions_{\txid_b}$ such that $(w, r) \in \absGRSI.\rf$. 
Let $\loc w = \loc r = \x$. 
%From the definition of $\absGRSI.\rf$ we know $\mw{w} \in \txid_a.\mathit{Ws}$ and $(\mw{w}, \txid_a.\mathit{rs}_\x) \in \impGRSI.\rf$.
From \cref{lem:rsi_alt_lock_hb} we then have $\txid_a.\mathit{wu}_\x  \relarrow{\impGRSI.\hb} \txid_b.\mathit{rl}_\x$.
Pick an arbitrary $c \in \absGRSI.\Transactions_{\txid_a}$.
As such we have 
$c \relarrow{\impGRSI.\po} 
\txid_a.\mathit{wu}_\x  \relarrow{\impGRSI.\hb} 
\txid_b.\mathit{rl}_\x  \relarrow{\impGRSI.\po} \txid_b.\mathit{rs}_\x$. 
That is, we have $(c, \txid_b.\mathit{rs}_\x) \in \impGRSI.\hb$, as required.
Now assume that $a \in \Writes$. We then have $(\mw{a}, \txid_a.\mathit{wu}_\x) \in \impGRSI.\po$.
As such, from the transitivity of $\impGRSI.\hb$ we also have $(\mw{a}, \txid_b.\mathit{rs}_\x) \in \impGRSI.\hb$, as required. 

%Similarly, we have $w' \relarrow{\impGRSI.\rf} \txid_b.\mathit{rs}_\x$. 
%That is, we have $(w', \txid_b.\mathit{rs}_\x) \in \impGRSI.\hb$, as required.
%(\txid_a.\mathit{wu}_\x, \txid_b.\mathit{rs}_\x)

In case (A.3), we then know there exists $w \in \absGRSI.\Transactions_{\txid_a}$ and $w' \in \absGRSI.\Transactions_{\txid_b}$ such that $(w, w') \in \absGRSI.\co$. 
Let $\loc w = \loc{w'} = \x$. 
%From the definition of $\absGRSI.\co$ we know there exists $w_1 \in \txid_a.\mathit{Ws}$, $w_2 \in \txid_b.\mathit{Ws}$ such that $(w_1, w_2) \in \impGRSI.\co$ and $\loc{w_1} = \loc{w_2} = \x$.
From \cref{lem:rsi_alt_lock_hb} we then have $\txid_a.\mathit{wu}_\x  \relarrow{\impGRSI.\hb} \txid_b.\mathit{rl}_\x$.
Pick an arbitrary $c \in \absGRSI.\Transactions_{\txid_a}$.
As we also have $\txid_b.\mathit{rl}_\x  \relarrow{\impGRSI.\po} w_2$,
$w_1 \relarrow{\impGRSI.\po} \txid_a.\mathit{wu}_\x  $, 
and $c \relarrow{\impGRSI.\po} \txid_a.\mathit{wu}_\x$, we then have 
$(c,  w_2) \in \impGRSI.\hb$ and $(w_1, w_2) \in \impGRSI.\hb$, as required.
Now assume that $a \in \Writes$. We then have $(\mw{a}, \txid_a.\mathit{wu}_\x) \in \impGRSI.\po$.
As such, from the transitivity of $\impGRSI.\hb$ we also have $(\mw{a}, w_2) \in \impGRSI.\hb$, as required. 
%, (\txid_a.\mathit{wu}_\x, w')

In case (B) there are three additional cases to consider: 
1) $(a, b) \in \absGRSI.\rsipo$; or 
2) $(a, b) \in \absGRSI.\rsirf$; or 
3) $(a, b) \in \absGRSI.\cot$.

In case (B.1) we then have $(a, b) \in \absGRSI.\po$ and from the definition of $\absGRSI.\po$ we have $(a, b) \in \impGRSI.\po \subseteq \impGRSI.\hb$, as required.
In case (B.2) we then have $(a, b) \in \absGRSI.\rf$ and from the definition of $\absGRSI.\rf$ we have $(a, b) \in \impGRSI.\rf \subseteq \impGRSI.\hb$, as required.
Case (B.3) does not apply as $a, b \in \absGRSI.\NT$. 

In case (C) there are three additional cases to consider: 
1) $(a, b) \in \absGRSI.\rsipo$; or 
2) $(a, b) \in \absGRSI.\rsirf$; or 
3) $(a, b) \in \absGRSI.\cot$.

In case (C.1) we then have $(a, b) \in \absGRSI.\po$ and from the definition of $\absGRSI.\po$ we have $(a, b) \in \impGRSI.\po \subseteq \impGRSI.\hb$, as required.
In case (C.2) we know there exists $r \in \absGRSI.\Transactions_{\txid_b}$ such that $(a, r) \in \absGRSI.\rf$. 
Let $\loc a = \loc r = \x$. 
From  the definition of $\absGRSI.\rf$ we then have $(a, \txid_b.\mathit{rs}_\x) \in \impGRSI.\rf \subseteq \impGRSI.\hb$, as required.
Case (C.3) does not apply as $a \in \absGRSI.\NT$.

In case (D) there are three additional cases to consider: 
1) $(a, b) \in \absGRSI.\rsipo$; or 
2) $(a, b) \in \absGRSI.\rsirf$; or 
3) $(a, b) \in \absGRSI.\cot$.

In case (D.1) we then have $(a, b) \in \absGRSI.\po$.
Pick an arbitrary $c \in \absGRSI.\Transactions_{\txid_a}$. From the definition of $\absGRSI.\po$ we have $(c, b) \in \absGRSI.\po$.
As such, from the definition of $\absGRSI.\po$ we have $(c, b) \in \impGRSI.\po \subseteq \impGRSI.\hb$, as required.
Now assume that $a \in \Writes$. We then have $(\mw{a}, b) \in \impGRSI.\po \subseteq \impGRSI.\hb$, as required.
%Pick an arbitrary $c \in \txid_a.\Events$. From the definition of $\absGRSI.\po$ we have $(c, b) \in \impGRSI.\po \subseteq \impGRSI.\hb$, as required. 

In case (D.2) we then have $(a, b) \in \absGRSI.\rf$ and from the definition of $\absGRSI.\rf$ we know that $\mw{a} \in \txid_a.\mathit{Ws}$ and $(\mw{a}, b) \in \impGRSI.\rf \subseteq \impGRSI.\hb$, as required. 
Pick an arbitrary $c \in \absGRSI.\Transactions_{\txid_a}$. 
We then know $(c, \mw{a}) \in \impGRSI.\po$. 
We then have $c \relarrow{\impGRSI.\po} \mw{a} \relarrow{\impGRSI.\rf} b$.
As such, from the definition of $\impGRSI.\hb$ we have $(c, b) \in \impGRSI.\hb$, as required.
Case (D.3) does not apply as $b \in \absGRSI.\NT$. \\

\noindent \textbf{Inductive case $i {=}  n {+} 1$}
\begin{align}
\hspace*{-15pt}
\begin{array}{@{} r @{\hspace{2pt}} l @{}}
	\for{j \in \Nats} \for{\txid_a, \txid_b, a, b}
	\for{a, b} \\
	\quad 
	(a, b) \in  S_j
	\land j \leq n
	\Rightarrow
	& 
	\begin{array}[t]{@{} l @{}}
		(a, b \in \absGRSI.\NT \land (a, b) \in \impGRSI.\hb) \\
		\lor
		(a \in \absGRSI.\Transactions_{\txid_a} 
		\land b  \in \absGRSI.\Transactions_{\txid_b}
		\land \txid_a = \txid_b
		\land (a,b) \in \impGRSI.\po
		) \\
		\lor 
		\left(
		\begin{array}{@{} l @{}}
			a \in \absGRSI.\Transactions_{\txid_a} \land b  \in \absGRSI.\Transactions_{\txid_b} 
			\land \txid_a \ne \txid_b\\
			\exsts{d \in  \absGRSI.\Transactions_{\txid_b}} 
			\for{c \in \absGRSI.\Transactions_{\txid_a}} (c, d) \in \impGRSI.\hb \\
			\land\, (\exsts{e \in \txid_a.\mathit{Ws}}  (e, d) \in \impGRSI.\hb \\
				\quad \lor \for{c \in \txid_a.\Events} (c, d) \in \impGRSI.\hb)
		\end{array}
		\right) \\
		\lor
		\left(
		\begin{array}{@{} l @{}}
			a \in \absGRSI.\NT \land b  \in \absGRSI.\Transactions_{\txid_b} \\
			\land \exsts{d \in  \absGRSI.\Transactions_{\txid_b}} 
			(a, d) \in \impGRSI.\hb 
		\end{array}
		\right) \\
		\lor
		\left(
		\begin{array}{@{} l @{}}
			a \in \absGRSI.\Transactions_{\txid_a} \land b  \in \absGRSI.\NT \\
			\land \for{c \in \absGRSI.\Transactions_{\txid_a}} (c, b) \in \impGRSI.\hb \\
			\land\, (\exsts{e \in \txid_a.\mathit{Ws}}  (e, b) \in \impGRSI.\hb \\
				\quad \lor \for{c \in \txid_a.\Events} (c, b) \in \impGRSI.\hb)
		\end{array}
		\right) \\
	\end{array}		
\end{array}	
\tag{I.H.}
\label{IH:rsi_alt_soundnenss}
\end{align}
Pick arbitrary $\txid_a, \txid_b$, $a \in \absGRSI.\Transactions_{\txid_a}, b  \in \absGRSI.\Transactions_{\txid_b}$ such that $(a, b) \in  S_i$.
From the definition of $S_i$ we then know there exist $e, \txid_e$ such that $e \in \absGRSI.\Transactions_{\txid_e}$, $(a, e) \in S_0$ and $(e, b) \in S_n$. 
The desired result then follows from the inductive hypothesis and case analysis on $a, b$ annd $c$. 
%Since $(e, b) \in S_n$, from \eqref{IH:rsi_alt_soundnenss} we know there exists $d \in  \absGRSI.\Transactions_{\txid_b}$ such that $\for{f' \in  \absGRSI.\Transactions_{\txid_e}} (f', d) \in \impGRSI.\hb$.
%On the other hand, from the proof of the base case we know there exists $f \in  \absGRSI.\Transactions_{\txid_e}$ such that $\for{c' \in  \absGRSI.\Transactions_{\txid_a}} (c', f) \in \impGRSI.\hb$; 
%and that $\exsts{\y} (\txid_e.\mathit{wu}_\y, f) \in \impGRSI.\hb \lor \for{c' \in  \txid_a.\Events} (c', f) \in \impGRSI.\hb$. 
%Pick an arbitrary $c \in \absGRSI.\Transactions_{\txid_a}$. We thus know that $(c, f) \in \impGRSI.\hb$.
%As $f \in \absGRSI.\Transactions_{\txid_e}$, we thus have $(f, d) \in \impGRSI.\hb$. 
%As $\impGRSI.\hb$ is transitively closed and we have  $(c, f), (f, d) \in \impGRSI.\hb$ and 
%$\exsts{\y} (\txid_e.\mathit{wu}_\y, f) \in \impGRSI.\hb \lor \for{c' \in  \txid_a.\Events} (c', f) \in \impGRSI.\hb$, we then have $(c, d) \in \impGRSI.\hb$, and 
%$\exsts{\y} (\txid_e.\mathit{wu}_\y, d) \in \impGRSI.\hb \lor \for{c' \in  \txid_a.\Events} (c', d) \in \impGRSI.\hb$, as required. 

\end{proof}
\end{lemma}

\begin{theorem}[Soundness]
For all execution graphs $\impGRSI$ of the implementation and their counterpart RSI execution graphs $\absGRSI$ constructed as above,
\[
	\consistent{\impGRSI} \Rightarrow \rsicon
\]
\begin{proof}
Pick an arbitrary execution graph $\impGRSI$ of the implementation such that $\consistent{\impGRSI}$ holds, and its associated RSI execution graph $\absGRSI$ constructed as described above. \\

\noindent \textbf{RTS. $\irr{\absGRSI.\rsihb}$}\\ 
We proceed by contradiction. Let us assume $\neg \irr{\absGRSI.\rsihb}$ 
Let $S = \transC{\absGRSI.(\rsipo \cup \rsirf \cup \cot)}$.
There are now two cases to consider: either there is an $\rsihb$ cycle without a $\rsifr$ edge; or there is a cycle with one or more $\rsifr$ edges. 
That is, either
1) there exists $a$ such that $(a, a) \in  S$; or 
2) there exist $a_1, b_1, \cdots, a_n, b_n$  such that 
$a_1 \relarrow{\absGRSI.\rsifr} b_1 \relarrow{S} a_2 \relarrow{\absGRSI.\rsifr} b_2 \relarrow{S} \cdots \relarrow{S} a_n \relarrow{\absGRSI.\rsifr} b_n \relarrow{S} a_1$. 

In case (1) we then know that either i) there exists $\txid$ such that $a \in \absGRSI.\Transactions_{\txid_a}$; or ii) $a \in \absGRSI.\NT$. 
In case (1.i) from \cref{lem:rsi_alt_soundness} we then have $(a, a) \in \impGRSI.\po$, contradicting the assumption that $\impGRSI$ is RA-consistent. 
Similarly, in case (1.ii) from \cref{lem:rsi_alt_soundness} we know that $(a, a) \in \impGRSI.\hb$, contradicting the assumption that $\impGRSI$ is RA-consistent. 

In case (2), for an arbitrary  $i \in \{1 \cdots n\}$, let $j = i {+} 1$ when $i \ne n$; and $j = 1$ when $i = n$. 
As $a_i \relarrow{\absGRSI.\rsifr} b_i$, we know there exists $\txid_{a_i}, \txid_{b_i}$ such that $a_i \in \Transactions_{\txid_{a_i}}$, $b_i \in \Transactions_{\txid_{b_i}}$, and that there exist $r_i \in \Transactions_{\txid_{a_i}} \cap \EReads$ and $w_i \in \Transactions_{\txid_{b_i}} \cap \Writes$ such that $(r_i, w_i) \in \absGRSI.\fr$. 
Let $\loc{r_i} = \loc{w_i} = \x_i$. 
From \cref{lem:rsi_alt_lock_hb} we then know that either 
i) $\txid_{a_i}.\mathit{ru}_{\x_i} \relarrow{\impGRSI.\hb} \txid_{b_i}.\mathit{pl}_{\x_i}$; or 
ii) $\txid_{a_i}.\mathit{wu}_{\x_i} \relarrow{\impGRSI.\hb} \txid_{b_i}.\mathit{rl}_{\x_i}$.
Note that for all $w' \in \txid_{b_i}.\mathit{Ws}$, we know $\txid_{b_i}.\mathit{rl}_{\x_i} \relarrow{\impGRSI.\po} w'$ and $\txid_{b_i}.\mathit{pl}_{\x_i} \relarrow{\impGRSI.\po} w'$.
As such, as $(b_i, a_j) \in S$, from \cref{lem:rsi_alt_soundness} and since $\impGRSI.\hb$ is transitively closed, we know there exists $d_j \in \Transactions_{\txid_{a_j}}$ such that either $\txid_{a_i}.\mathit{ru}_{\x_i} \relarrow{\impGRSI.\hb} d_j$, or $\txid_{a_i}.\mathit{wu}_{\x_i} \relarrow{\impGRSI.\hb} d_j$.
That is, $\txid_{a_i}.\mathit{u}_{\x_i} \relarrow{\impGRSI.\hb} d_j$, where either $\mathit{u}_{\x_i} = \mathit{ru}_{\x_i}$ or $\mathit{u}_{\x_i} = \mathit{wu}_{\x_i}$.
On the other hand, observe that for all $d_i \in \Transactions_{\txid_{a_i}}$ we have 
$d_i \relarrow{\impGRSI.\po} \txid_{a_i}.\mathit{u}_{\x_i}$ 
i.e.\ $d_i \relarrow{\impGRSI.\hb} \txid_{a_i}.\mathit{u}_{\x_i}$.
As such, we have $d_j \relarrow{\impGRSI.\hb} \txid_{a_j}.\mathit{u}_{\x_j}$.
As we also have $\txid_{a_i}.\mathit{u}_{\x_i} \relarrow{\impGRSI.\hb} d_j$
and $\impGRSI.\hb$ is transitively closed, we have $\txid_{a_i}.\mathit{u}_{\x_i} \relarrow{\impGRSI.\hb} \txid_{a_j}.\mathit{u}_{\x_j}$.
We then have 
$\txid_{a_1}.\mathit{u}_{\x_1} \relarrow{\impGRSI.\hb} \txid_{a_2}.\mathit{u}_{\x_2} \relarrow{\impGRSI.\hb} \cdots \relarrow{\impGRSI.\hb} \txid_{a_n}.\mathit{u}_{\x_n} \relarrow{\impGRSI.\hb} \txid_{a_1}.\mathit{u}_{\x_1}$. 
That is, $\txid_{a_1}.\mathit{u}_{\x_1} \relarrow{\impGRSI.\hb} \txid_{a_1}.\mathit{u}_{\x_1}$, 
contradicting the assumption that $\impGRSI$ is RA-consistent. 
\\

\noindent \textbf{RTS. $\rfi \cup \coi \cup \fri \subseteq \po$}\\
Follows immediately from the construction of $\absGRSI$.\\

\noindent \textbf{RTS. $\irr{\absGRSI.(\rsihb; \co)}$}\\ 
We proceed by contradiction. Let us assume $\neg \irr{\absGRSI.(\rsihb; \co)}$.\\
That is, there exists $a, b$ such that $(a, b) \in \rsihb$ and $(b, a) \in \co$. 
There are now five cases to consider: 
1) there exists $\txid$ such that $a, b \in \Transactions_\txid$; or
2) there exists $\txid_a, \txid_b$ such that $\txid_a \ne \txid_b$, $a \in \Transactions_{\txid_a}$ and $b \in \Transactions_{\txid_b}$; or
3) there exists $\txid$ such that $a \in \Transactions_\txid$ and $b \in \absGRSI.\NT$; or
4) there exists $\txid$ such that $b \in \Transactions_\txid$ and $a \in \absGRSI.\NT$; or
5) $a, b \in \absGRSI.\NT$.

In case (1) we then have $(b, a) \in \coi \subseteq \poi$ (from the proof of the previous part) and since $a,b$ are both write events, we have $(b,a) \subseteq \rsihb$. 
We then have $a \relarrow{\rsihb} b \relarrow{\rsihb} a$, contradicting our proof above that $\rsihb$ is irreflexive. 

In case (2) we then have $(b, a) \in \cot \subseteq \rsihb$.
We then have $a \relarrow{\rsihb} b \relarrow{\rsihb} a$, contradicting our proof above that $\rsihb$ is irreflexive.

In case (3) from \cref{lem:rsi_alt_soundness} we know $(\mw{a}, b) \in \impGRSI.\hb$. 
From the definition of $\absGRSI.\co$ we also have $(b, \mw{a}) \in \impGRSI.\co$.
We then have $\mw a \relarrow{\impGRSI.\hb} b \relarrow{\impGRSI.\co} \mw a$, contradicting the assumption that $\impGRSI$ is RA-consistent.

In case (4) from \cref{lem:rsi_alt_soundness} we know there exists $d \in \Transactions_{\txid}$ such that $(a, d) \in \impGRSI.\hb$. 
From the definition of $\absGRSI.\co$ we also have $(\mw b, a) \in \impGRSI.\co$.
Moreover, from the construction of $\absGRSI$ we know $(d, \mw b) \in \impGRSI.\po$
We then have $a \relarrow{\impGRSI.\hb} d \relarrow{\impGRSI.\po} \mw b \relarrow{\impGRSI.\co} a$, contradicting the assumption that $\impGRSI$ is RA-consistent.

In case (5) from \cref{lem:rsi_alt_soundness} we know $(a, b) \in \impGRSI.\hb$. 
From the definition of $\absGRSI.\co$ we also have $(b, a) \in \impGRSI.\co$.
We then have $a \relarrow{\impGRSI.\hb} b \relarrow{\impGRSI.\co} a$, contradicting the assumption that $\impGRSI$ is RA-consistent.\\

\noindent \textbf{RTS. $\irr{\absGRSI.(\rsihb; \fr)}$}\\ 
We proceed by contradiction. Let us assume $\neg \irr{\absGRSI.(\rsihb; \fr)}$.\\
That is, there exists $a, b$ such that $(a, b) \in \rsihb$ and $(b, a) \in \fr$. 
There are now five cases to consider: 
1) there exists $\txid$ such that $a, b \in \Transactions_\txid$; or
2) there exists $\txid_a, \txid_b$ such that $\txid_a \ne \txid_b$, $a \in \Transactions_{\txid_a}$ and $b \in \Transactions_{\txid_b}$; or
3) there exists $\txid$ such that $a \in \Transactions_\txid$ and $b \in \absGRSI.\NT$; or
4) there exists $\txid$ such that $b \in \Transactions_\txid$ and $a \in \absGRSI.\NT$; or
5) $a, b \in \absGRSI.\NT$.

In case (1) we then have $(b, a) \in \fri \subseteq \poi$ (from the proof of the earlier part) and thus  $(b, a) \in \impGRSI.\po$.
Moreover, from  \cref{lem:rsi_alt_soundness} we know $(a, b) \in \impGRSI.\po$
We then have $a \relarrow{\impGRSI.\po} b \relarrow{\impGRSI.\po} a$, contradicting the assumption that $\impGRSI$ is RA-consistent.

In case (2) we then have $(b, a) \in \sifr \cup \cot \subseteq \rsihb$.
We then have $a \relarrow{\rsihb} b \relarrow{\rsihb} a$, contradicting our proof above that $\rsihb$ is irreflexive.

In case (3) from \cref{lem:rsi_alt_soundness} we know $(\mw{a}, b) \in \impGRSI.\hb$. 
From the definition of $\absGRSI.\fr$ we also have $(b, \mw{a}) \in \impGRSI.\fr$.
We then have $\mw a \relarrow{\impGRSI.\hb} b \relarrow{\impGRSI.\fr} \mw a$, contradicting the assumption that $\impGRSI$ is RA-consistent.

In case (4) from \cref{lem:rsi_alt_soundness} we know there exists $d \in \Transactions_{\txid}$ such that $(a, d) \in \impGRSI.\hb$. 
Let $\loc a = \loc b = \x$.  
From the definition of $\absGRSI.\fr$ and our race-freedom stipulation of non-transactional writes with the same transaction we also have $(\txid_b.\mathit{vr}_\x, a) \in \impGRSI.\fr$.
Moreover, from the construction of $\absGRSI$ we know $(d, \txid_b.\mathit{vr}_\x,b) \in \impGRSI.\po$
We then have $a \relarrow{\impGRSI.\hb} d \relarrow{\impGRSI.\po} \txid_b.\mathit{vr}_\x, \relarrow{\impGRSI.\fr} a$, contradicting the assumption that $\impGRSI$ is RA-consistent.

In case (5) from \cref{lem:rsi_alt_soundness} we know $(a, b) \in \impGRSI.\hb$. 
From the definition of $\absGRSI.\fr$ we also have $(b, a) \in \impGRSI.\fr$.
We then have $a \relarrow{\impGRSI.\hb} b \relarrow{\impGRSI.\fr} a$, contradicting the assumption that $\impGRSI$ is RA-consistent.
\end{proof}
\end{theorem}
\subsection{Implementation Completeness}
In order to establish the completeness of our implementation, it suffices to show that given an RSI-consistent execution graph $\absGRSI = (\Events, \po, \rf, \co, \rsilo)$, we can construct a corresponding RA-consistent execution graph $\impGRSI$ of the implementation.
%
%Before proceeding with the construction of a corresponding implementation graph, we describe several auxiliary definitions.
%
%Given a transaction class $\Transactions_{\txid} \in \absGRSI.\Transactions/\st$, we write $\writeset_{\txid}$ for the set of locations written to by $\Transactions_{\txid}$: $\writeset_{\txid} = \bigcup_{e \in \Transactions_{\txid}\cap \Writes} \loc{e}$.
%Similarly, we write $\readset_{\txid}$ for the set of locations read from by $\Transactions_{\txid}$,
%\emph{prior to} being written by $\Transactions_{\txid}$. 
%For each location \code x read from by $\Transactions_{\txid}$, we additionally record the first read event in $\Transactions_{\txid}$ that retrieved the value of \code x.
%That is, 
%\[
%\readset_{\Transactions_{\txid}} \eqdef
%\setcomp{
%	(\code x, r)
%}{
%	r \in \Transactions_i \cap \Reads_{\code x}
%	\land \neg\exsts{e \in \Transactions_{\txid} \cap \Events_{\code x}}
%	e \relarrow{\po} r
%}
%\]
%
%

Note that the execution trace for each transaction $\Transactions_{\txid} \in \absGSI.\Transactions/\st$ is of the form 
$\trace'_{\txid} = \mathit{t}'_1 \relarrow{\imm \po} \cdots \relarrow{\imm \po} \mathit{t}'_k$ for some $k$, where each $\mathit{t}'_i$ is a read or write event.
As such, we have $\absGSI.\Transactions = \bigcup_{\Transactions_{\txid} \in \absGSI.\Transactions/\st} \Transactions_{\txid} = \trace'_{\txid}.\Events$.
For each transaction $\txid$, we construct the implementation trace $\trace_{\txid}$ as follows:
\[
	\mathit{Fs}
	\relarrow{\imm \po} \mathit{Is}
	\relarrow{\imm \po} \mathit{Ts}
	\relarrow{\imm \po} \mathit{VRs}
	\relarrow{\imm \po} \mathit{RUs}
	\relarrow{\imm \po} \mathit{PLs}
	\relarrow{\imm \po} \mathit{Ws}
	\relarrow{\imm \po} \mathit{Us}
\]
where:
\begin{itemize}
	\item $\mathit{Fs}$ denotes the sequence of events failing to obtain the necessary locks (i.e.\ those iterations that do not succeed in promoting the writer locks) or validate the snapshot;
	\item $\mathit{Is}$ denotes the sequence of events initialising the values of \code{LS}, \readset and  \writeset with $\emptyset$, initialising \code{ws} with \code{[]} and initialising \code{s[x]} with $(\bot, \bot)$ for each location \x;
	\item $\mathit{Ts}$ denotes the sequence of events corresponding to the execution of \denot{\code{T}} and is of the form $\mathit{t}_1 \relarrow{\imm \po} \cdots \relarrow{\imm \po} \mathit{t}_k$, where for all $m \in \{1 \cdots k\}$:
	\[
	\mathit{t}_m {=} 
	\begin{cases}
		\rseq{\code x_m}{v_m}{\mathit{RS}_{m {-} 1},  \mathit{WS}_{m {-} 1}}
		\relarrow{\imm \po} \mathit{lr}_{\x_m}
		& \hspace{-10pt} \text{if } O_m {=} \readE{-}{\code x_m}{v_m} \\
		\wseq{\code x_m}{v_m}{\mathit{RS}_{m {-} 1},  \mathit{WS}_{m {-} 1}} 
		\relarrow{\imm \po} \mathit{wws}_{\x_m} 
		& \hspace{-10pt} \text{if } O_m {=} \writeE{\rel}{\code x_m}{v_m} \\
		\relarrow{\imm \po} \readE{}{\code{s[x}_m\code ]}{(v'_m, -)}
		\relarrow{\imm \po} \mathit{lw}_{\x_m}
		\relarrow{\imm \po} \writeE{}{\code{ws}}{\mathit{ws}_m}
	\end{cases}
	\]
	where $O_m$ denotes the $m$th event in the trace of $\code T$; 
	$\mathit{lr}_{\x_m} \eqdef \readE{}{\code{s[x}_m\code ]}{(-,v_m)}$; 
\[
\small
	\rseq{\code x_m}{v_m}{\mathit{RS}_{m {-} 1},  \mathit{WS}_{m {-} 1}}
	\eqdef
	\begin{cases}
		\readE{}{\code{s[x}_m\code ]}{(\bot,\bot)}
		& \text{if } \code x_m \not\in \mathit{RS}_{m {-} 1} \cup  \mathit{WS}_{m {-} 1} \\
		\relarrow{\imm \po} \mathit{fs_m} \\
		\relarrow{\imm \po} \mathit{rl}_{\x_m} \\
		\relarrow{\imm \po} \mathit{wrs}_{\x_m} \\
		\relarrow{\imm \po} \mathit{rs}_{\x_m} \\
		\relarrow{\imm \po} \mathit{ws}_{\x_m}
		& \\\\
		\emptyset
		& \text{otherwise} 
	\end{cases}
\]	
$\mathit{fs}_m$ denotes the sequence of events attempting (but failing) to acquire the read lock on $\code x_m$,  
$\mathit{rl}_{\code x_m} \eqdef \rlockE{\xl_m}$, 
$\mathit{wrs}_{\x_m} \eqdef \writeE{}{\readset}{\mathit{RS}_m}$, 
$\mathit{rs}_{\code x_m} \eqdef \readE{}{\code x_m}{v_{\x_m}^0}$, 
$\mathit{ws}_{\code x_m} \eqdef \writeE{}{\code{s[x}_m\code ]}{(v_{\x_m}^0, v_{\x_m}^0)}$; 
and for all $m > 0$:
\[
	\mathit{RS}_{m {+} 1} \eqdef 
	\begin{cases}
		\mathit{RS}_m \cup \{\code x_m\} & \text{if } O_m {=} \readE{-}{\code x_m}{-} \\
		\mathit{RS}_m & \text{otherwise}
	\end{cases}
\]
and
\[
\small
	\wseq{\code x_m}{v_m}{\mathit{RS}_{m {-} 1},  \mathit{WS}_{m {-} 1}}
	\eqdef
	\begin{cases}
		\readE{}{\code{s[x}_m\code ]}{(\bot,\bot)}
		&  \hspace{-5pt} \text{if } \code x_m \not\in \mathit{RS}_{m {-} 1} \cup  \mathit{WS}_{m {-} 1} \\
		\relarrow{\imm \po} \mathit{fs_m} \\
		\relarrow{\imm \po} \mathit{rl}_{\x_m}
		\emptyset
		& \hspace{-5pt} \text{otherwise} 
	\end{cases}
\]	
$\mathit{wws}_{\x_m} \eqdef \writeE{}{\writeset}{\mathit{WS}_m}$; 
$\mathit{lw}_{\x_m} \eqdef \writeE{}{\code{s[x}_m\code ]}{(v'_m, v_m)}$; 
$\mathit{fs}_m$ and $\mathit{rl}_{\code x_m}$ are as defined above; and for all $m > 0$:
\[
\begin{array}{@{} c @{}}
	\mathit{WS}_{m {+} 1} \eqdef 
	\begin{cases}
		\mathit{WS}_m \cup \{\code x_m\} & \text{if } O_m {=} \writeE{-}{\code x_m}{-} \\
		\mathit{WS}_m & \text{otherwise}
	\end{cases} \\
	\mathit{ws}_{m {+} 1} \eqdef 
	\begin{cases}
		\mathit{ws}_m {++} [(\x_m, v_m)] & \text{if } O_m {=} \writeE{-}{\code x_m}{v_m} \\
		\mathit{ws}_m & \text{otherwise}
	\end{cases}
\end{array}	
\]	
Let $\readset_\txid = \mathit{RS}_m$,  
$\writeset_\txid = \mathit{WS}_m$, and
$\ws_{\txid} = \ws_m$;
let $\readset_\txid \cup \writeset_\txid$ be enumerated as $\{\x_1 \cdots \x_i\}$ for some $i$, 
and $\ws_\txid$ be enumerated as $\{(\x_1, v_1) \cdots (\x_j, v_j)\}$ for some $j$.
	\item $\mathit{VRs}$ denotes the sequence of events validating the reads and is of the  form $\mathit{v}_{\code x_1} \relarrow{\imm \po} \cdots \relarrow{\imm \po} \mathit{v}_{\code x_i}$, where for all $n \in \{1 \cdots i\}$:
	\[
	\begin{array}{l}
		\mathit{v}_{\code x_n} = 
		\begin{cases}
			\readE{}{\code{s[x$_n$]}}{(v_{\x_n}^0, -)}	
			\relarrow{\imm \po} \mathit{vr}_{\code x_n} {=} \readE{}{\x_n}{v_{\x_n}^0}
			& \text{ if } \code x_n \in \readset_{\txid}  \\
			\emptyset
			& \text{ otherwise}
		\end{cases}
	\end{array}	
	\]
	\item $\mathit{RUs}$ denotes the sequence of events releasing the reader locks (when the given location is in the read set only) and is of the form $\mathit{ru}_{\code x_1} \relarrow{\imm \po} \cdots \relarrow{\imm \po} \mathit{ru}_{\code x_i}$, where for all $n \in \{1 \cdots i\}$:
	\[
	\begin{array}{l}
		\mathit{ru}_{\code x_n} = 
		\begin{cases}
			\runlockE{\xl_n}
			& \text{ if } \code x_n \not\in \writeset_{\txid}  \\
			\emptyset
			& \text{ otherwise}
		\end{cases}
	\end{array}	
	\]
	\item $\mathit{PLs}$ denotes the sequence of events promoting the reader locks to writer ones (when the given location is in the write set), and is of the form $\mathit{pl}_{\code x_1} \relarrow{\imm \po} \cdots \relarrow{\imm \po} \mathit{pl}_{\code x_i}$, where for all $n \in \{1 \cdots i\}$:
	\[
	\begin{array}{l}
		\mathit{pl}_{\code x_n} = 
		\begin{cases}
			\plockE{\xl_n}
			& \text{if }  \code x_n \in \writeset_{\txid} \\
			\emptyset
			& \text{ otherwise } 
		\end{cases}
	\end{array}	
	\]
	\item $\mathit{Ws}$ denotes the sequence of events committing the writes of \denot{\code{T}} and is of the form $\mathit{c}_{\x_1, v_1} \relarrow{\imm \po} \cdots \relarrow{\imm \po} \mathit{c}_{\x_j, v_j}$, where for all $n \in \{1 \cdots j\}$:
$
	\mathit{c}_{\x_n, v_n} =  \writeE{}{\x_n}{v_n} 
%	\begin{cases}
%		\readE{}{\code{s[x}_n \code{]}}{v_n}  \relarrow{\imm \po} 
%		\mathit{w}_{\x_n} {=} \writeE{}{\x_n}{v_n} 
%		& \text{if } \x_n \in \writeset_\txid \\
%		\emptyset 
%		& \text{otherwise}
%	\end{cases}
$
	\item $\mathit{Us}$ denotes the sequence of events releasing the locks on the write set, and is of the form $\mathit{wu}_{\code x_1} \relarrow{\imm \po} \cdots \relarrow{\imm \po} \mathit{wu}_{\code x_i}$, where for all $n \in \{1 \cdots i\}$:
	\[
		\mathit{wu}_{\code x_n} = 
		\begin{cases}
			\wunlockE{\code{xl}_n} & \text{if } \code x_n \in \writeset_{\txid} \\
			\emptyset & \text{otherwise}
		\end{cases}				
	\]
\end{itemize}
Given a transaction trace $\trace_{\txid}$, we write e.g.~$\txid.\mathit{Us}$ to refer to its constituent $\mathit{Us}$ sub-trace and write $\mathit{Us}.\Events$ for the set of events related by \po in $\mathit{Us}$. Similarly, we write $\txid.\Events$ for the set of events related by \po in $\trace_{\txid}$.
Note that $\impGRSI.\Events = \bigcup\limits_{\txid \in \sort{Tx}}  \txid.\Events $.
%As such, the $\sti$ relation induces a set of equivalence classes on $\absGRSI.\Events$, written $\absGRSI.\Events/\sti$. As before, we write $\class{a}{\sti}$ for the equivalence class in $\absGRSI.\Events/\sti$ that contains $a$.

Note that for each transaction $\txid$ and each location $\x$, the $\txid.\mathit{rl}_\x$, $\txid.\mathit{rs}_\x$, $\txid.\mathit{vr}_\x$, $\txid.\mathit{ru}_\x$, $\txid.\mathit{pl}_\x$ and $\txid.\mathit{wu}_\x$ are uniquely identified when they exist. 
Let $\txid.\mathit{w}_\x$ the last (in \po order) write to \x in $\mathit{Ws}$, when it exists.

%For each location $\x \in \writeset_\txid$, let $\fw \x$ denote the maximal write (in $\po$ order within $\txid$) logging a write for \x in \code{s[x]}. That is,  when $\trace_{\txid}=  t_1 \relarrow{\imm \po} \cdots \relarrow{\imm \po} t_m$, let $\fw \x = \func{wmax}{\x, [t_1 \cdots t_m]}$, where
%\[
%	 \func{wmax}{\x, [\,]} \text{ undefined}
%	 \quad
%	  \func{wmax}{\x, L {++} [t]}
%	  \eqdef
%	  \begin{cases}
%	  	\mathit{lw}_{\x}
%	  	& \text{if } t {=} \wseq{\x}{-}{-, -} \relarrow{\po} \mathit{lw}_{\x}  \relarrow{\po} \mathit{wws}_{\x} \\
%	  	 \func{wmax}{\x, L}
%	  	 & \text{otherwise}
%	  \end{cases}
%\]
%%
%Similarly, for each location $\x \in \writeset_\txid$, let $\iw \x$ denote the minimal write (in $\po$ order within $\txid$) logging a write for \x in \code{s[x]}. That is,  when $\trace_{\txid}=  t_1 \relarrow{\imm \po} \cdots \relarrow{\imm \po} t_m$, let $\iw \x = \func{wmin}{\x, [t_1 \cdots t_m]}$, where
%\[
%	 \func{wmin}{\x, [\,]} \text{ undefined}
%	 \quad
%	  \func{wmin}{\x, [t] {++} L}
%	  \eqdef
%	  \begin{cases}
%	  	\mathit{lw}_{\x}
%	  	& \text{if } t {=} \wseq{\x}{-}{-, -} \relarrow{\po} \mathit{lw}_{\x}  \relarrow{\po} \mathit{wws}_{\x} \\
%	  	 \func{wmin}{\x, L}
%	  	 & \text{otherwise}
%	  \end{cases}
%\]
%%
%%
%
%
%
%
%
Note that for each write operation $w$ in $\trace'_{\txid}$, there exists a \emph{matching} write operation in $\trace_{\txid}.\mathit{Ws}$, denoted by $\mw{w}$ is a one-to-one function such that: 
That is, 
\[
	\mw{w} {=} w'
	\iffdef
	\exsts{i}
	\land \itemAt{(\trace'_{\txid}.\Events \cap \Writes)}{i} = w
	\land  \itemAt{(\trace_{\txid}.\mathit{Ws})}{i} = w'
\]
We then define:
\[
	\mathsf{RF} \eqdef
	\begin{array}[t]{@{} l @{}}
		\setcomp{
			(w, r)	
		}{
			r \in \absGRSI.\NT \\ 	
			\land\, \big(
				 (w \in \impGRSI.\NT \land (w, r)  \in \absGRSI.\rf) \\
		 		\qquad \lor\,
				(\exsts{\txid'} \exsts{w' \in \txid'.\Events \cap \Writes} 
		 		(w', r)  \in \absGRSI.\rf
		 		\land w {=} \mw{w'} ) 
		 	\big)	
		} \\
		\cup
		\setcomp{
			(w, \txid.\mathit{rs}_\x)	, \\
			(w, \txid.\mathit{vr}_\x)	
		}{
			
				(w \in \impGRSI.\NT \land (w, \txid.\mathit{rs}_\x)  \in \absGRSI.\rf) \\
		 		\lor\,
				(\exsts{\txid'} \exsts{w' \in \txid'.\Events \cap \Writes} 
		 		(w', \txid.\mathit{rs}_\x)  \in \absGRSI.\rf
		 		\land w {=} \mw{w'} ) 
		} 
	\end{array}	
\]
Similarly, we define: 
\[
	\mathsf{MO} \eqdef
	\begin{array}[t]{@{} l @{}}
	\setcomp{
		(w_1, w_2)
	}{
		\exsts{w'_1, w'_2} 
		w_1 = \mw{w'_1}
		\land w_2 = \mw{w'_2}
		\land (w'_1, w'_2) \in \absGRSI.\co 
	} \\
	\cup
	\setcomp{
		(w, w')	
	}{
		w, w' \in \impGRSI.\NT \land (w, w') \in \absGRSI.\co	
	}
	 \\
	\cup
	\setcomp{
		(w, w')	
	}{
		w' \in \impGRSI.\NT \land 
		\exsts{w''} 
		w = \mw{w''}
		\land (w'', w') \in \absGRSI.\co 
	} \\
	\cup 
	\setcomp{
		(w', w)	
	}{
		w' \in \impGRSI.\NT \land 
		\exsts{w''} 
		w = \mw{w''}
		\land (w', w'') \in \absGRSI.\co 
	}
\end{array}	
\]
For each location \x we then define:
\[
\begin{array}{@{} r @{\hspace{2pt}} l @{}}
\LO_{\x} \eqdef
& \begin{array}[t]{@{} l @{\hspace{1pt}} l @{}} 
	& \setcomp{
		(\txid.\mathit{rl}_{\code x}, \txid.\mathit{pl}_{\code x}), 
		(\txid.\mathit{rl}_{\code x}, \txid.\mathit{wu}_{\code x}), 
		(\txid.\mathit{pl}_{\code x}, \txid.\mathit{wu}_{\code x})
	} 
	{ 
		\absGRSI.\Transactions_\txid \cap \Writes_\x \ne \emptyset
	} \\
	\cup 
	& \setcomp{
		(\txid.\mathit{rl}_{\code x}, \txid'.\mathit{pl}_{\code x}), 
		(\txid.\mathit{rl}_{\code x}, \txid'.\mathit{wu}_{\code x}), \\
		(\txid.\mathit{ru}_{\code x}, \txid'.\mathit{pl}_{\code x}), 	
		(\txid.\mathit{ru}_{\code x}, \txid'.\mathit{wu}_{\code x})
	} 
	{ 
		\txid \ne \txid'
		\land \exsts{a, b, x} \\
			\quad a \in \absGRSI.\Transactions_\txid \land b \in \absGRSI.\Transactions_\txid' \\
			\quad \land \loc{a} = \loc{b} = \x \\
			\quad \land\, a \in \absGRSI.\EReads
			\land (a, b) \in \absGRSI.\fr
	} \\
	\cup & \setcomp{
		(\txid.\mathit{rl}_{\code x}, \txid'.\mathit{pl}_{\code x}), 
		(\txid.\mathit{rl}_{\code x}, \txid'.\mathit{wu}_{\code x}), \\
		(\txid.\mathit{pl}_{\code x}, \txid'.\mathit{rl}_{\code x}), 
		(\txid.\mathit{pl}_{\code x}, \txid'.\mathit{pl}_{\code x}), 
		(\txid.\mathit{pl}_{\code x}, \txid'.\mathit{wu}_{\code x}), \\
		(\txid.\mathit{wu}_{\code x}, \txid'.\mathit{rl}_{\code x}), 
		(\txid.\mathit{wu}_{\code x}, \txid'.\mathit{pl}_{\code x}), 
		(\txid.\mathit{wu}_{\code x}, \txid'.\mathit{wu}_{\code x})
	} 
	{ 
		\txid \ne \txid'
		\land \exsts{a, b, x} \\
			a \in \absGRSI.\Transactions_\txid \land b \in \absGRSI.\Transactions_{\txid'} \\
			\land\, \loc{a} = \loc{b} = \x \\
			\land\, (a, b) \in \absGRSI.\co
	} \\
	\cup & \setcomp{
		(\txid.\mathit{pl}_{\code x}, \txid'.\mathit{rl}_{\code x}), 
		(\txid.\mathit{pl}_{\code x}, \txid'.\mathit{ru}_{\code x}), \\
		(\txid.\mathit{wu}_{\code x}, \txid'.\mathit{rl}_{\code x}), 
		(\txid.\mathit{wu}_{\code x}, \txid'.\mathit{ru}_{\code x})
	} 
	{ 
		\txid \ne \txid'
		\land \exsts{a, b, x} \\
			\quad a \in \absGRSI.\Transactions_\txid \land b \in \absGRSI.\Transactions_{\txid'} \\
			\quad \land \absGRSI.\Transactions_{\txid'} \cap \Writes_\x = \emptyset \\
			\quad \land\, \loc{a} {=} \loc{b} {=} \x \\
			\quad \land (a, b) \in \absGRSI.(\refC{\co}; \rf)
	} 
\end{array} 
\end{array}
\]
Note that each $\LO_\x$ satisfies the conditions in \cref{def:si_implementation_consistency}.

We are now in a position to demonstrate the completeness of our implementation. Given an RSI-consistent execution graph $\absGRSI$,
% $\absGRSI$  such that $\IRC{\absGRSI}$ holds, 
 we construct an execution graph $\impGRSI$ of the implementation as follows and demonstrate that $\consistent{\impGRSI}$ holds. 
\begin{itemize}
	\item $\impGRSI.\Events = \bigcup\limits_{\Transactions_\txid \in \absGRSI.\Transactions/\st} \trace_\txid.\Events \cup \absGRSI.\NT$. Observe that 
%	$\impGRSI.\Events$ is an extension of $\absGRSI.\Events$: 
	$\absGRSI.\Events \subseteq \impGRSI.\Events$.
	\item $\impGRSI.\po$ is defined as $\absGRSI.\po$ extended by the $\po$ for the additional events of $\impGRSI$, given by each $\trace_\txid$ trace defined above. Note that $\impGRSI.\po$ does not introduce additional orderings between events of $\absGRSI.\Events$. That is, $\for{a, b \in \absGRSI.\Events} (a, b) \in \absGRSI.\po \Leftrightarrow (a, b) \in \impGRSI.\po$.
	\item 
	$\impGRSI.\rf = \mathsf{RF}$.
	\item $\impGRSI.\co = \mathsf{MO}$. 
	\item 
	$\impGRSI.\rsilo = 
	\bigcup_{\code x \in \textsc{Locs}} \LO_{\x}
	$, with $\LO_\x$ as defined above.
\end{itemize}
\paragraph{Notation} 
In what follows, given an RSI implementation graph $\impGSI$ as constructed above we write $\impGRSI.\NT$ for the non-transactional events of $\impGRSI$, i.e.~$\impGRSI.\NT \eqdef \setcomp{a}{\neg\exsts{\txid} a \in \impGRSI.\Events_{\txid}}$.
Moreover, as before, given a relation $\makerel r \subseteq \impGRSI.\Events \times \impGRSI.\Events$, we override the $\tin{\makerel r}$ notation and write $\tin{\makerel r}$ for $\setcomp{(a, b) \in \makerel r}{\exsts{\txid} a, b \in \trace_{\txid}.\Events}$.

\begin{lemma}\label{lem:rsi_alt_completeness}
Given an RSI-consistent execution graph $\absGRSI$ and its corresponding implementation graph $\impGRSI$ constructed as above, for all $a, b, \txid_a, \txid_b$:
\[
\begin{array}{@{} l @{}}
	(a, b) \in \impGRSI.\hb \Rightarrow \\
	\qquad\! (\exsts{\txid} a, b \in  \impGRSI.\Events_{\txid} \Rightarrow (a, b) \in \impGRSI.\po) \\
	\quad \land 
	\Big( \neg \exsts{\txid} a, b \in \impGRSI.\Events_{\txid} \Rightarrow \\
	\qqqquad 
	\begin{array}[t]{@{} l @{}}
		 \exsts{A, B} \emptyset \subset A \times B \subseteq \absGRSI.\rsihb \\
		 \land\ (a \in \impGRSI.\NT \Rightarrow A {=} \{a\}) \land (b \in \impGRSI.\NT \Rightarrow B {=} \{b\}) \\
		\land\ [a \in \impGRSI.\Events_{\txid_a} \Rightarrow \\
			\qquad A {=} \absGRSI.\Transactions_{\txid_a} \lor (\stg{a}{\txid_a} \leq 2 \land A {=} \absGRSI.\Transactions_{\txid_a} \cap \EReads) \\
			\qquad \lor\ (\stg{a}{\txid_a} \leq 4 \land \exsts{d \in \txid_a.\Writes} a \relarrow{\impGRSI.\refC{\po}} d \land A = \{\mwi{d}\}) \;] \\
		\land\ [b \in \impGRSI.\Events_{\txid_b} \Rightarrow \\
			\qquad B {=} \absGRSI.\Transactions_{\txid_b} \lor (\stg{b}{\txid_b} \geq 3 \land B {=} \absGRSI.\Transactions_{\txid_b} \cap \Writes) \;]	\quad \Big)
\end{array}		
\end{array}
\]
where
\[
	\stg{a}{\txid_a} \eqdef 
	\begin{cases}
		1 & \text{if } a \not\in \in 
		\txid_a.\mathit{RUs}_{\txid_a} 
		\cup \txid_a.\mathit{PLs}_{\txid_a}
		\cup \txid_a.\mathit{Ws}_{\txid_a}
		\cup \txid_a.\mathit{Us}_{\txid_a}\\
		2 & \text{if } a \in  \txid_a.\mathit{RUs}_{\txid_a} \\
		3 & \text{if } a \in  \txid_a.\mathit{PLs}_{\txid_a} \\
		4 & \text{otherwise}
	\end{cases}
\]
\begin{proof}
Pick an arbitrary RSI-consistent execution graph $\absGRSI$ and its corresponding implementation graph $\impGRSI$ constructed as above.
Let $\hb_0 \eqdef \impGRSI.(\po \cup \rf \cup \rsilo)$ and $\hb_{n {+} 1} \eqdef \hb_0; \hb_n$, for all $n \in \Nats$. It is then straightforward to demonstrate that $\impGRSI.\hb = \bigcup_{i \in \Nats} \hb_i$.
We thus demonstrate instead that:
\[
\begin{array}{@{} l @{}}
	\for{n \in \Nats} \for{a, b, \txid_a, \txid_b}
	(a, b) \in \impGRSI.\hb_n \Rightarrow \\
	\qqquad\! (\exsts{\txid} a, b \in  \impGRSI.\Events_{\txid} \Rightarrow (a, b) \in \impGRSI.\po) \\
	\qquad \land 
	\Big( \neg \exsts{\txid} a, b \in \impGRSI.\Events_{\txid} \Rightarrow \\
	\qqqquad 
	\begin{array}[t]{@{} l @{}}
		 \exsts{A, B} \emptyset \subset A \times B \subseteq \absGRSI.\rsihb \\
		 \land\ (a \in \impGRSI.\NT \Rightarrow A {=} \{a\}) \land (b \in \impGRSI.\NT \Rightarrow B {=} \{b\}) \\
		\land\ [a \in \impGRSI.\Events_{\txid_a} \Rightarrow \\
			\qquad A {=} \absGRSI.\Transactions_{\txid_a} \lor (\stg{a}{\txid_a} \leq 2 \land A {=} \absGRSI.\Transactions_{\txid_a} \cap \EReads) \\
			\qquad \lor\ (\stg{a}{\txid_a} \leq 4 \land \exsts{d \in \txid_a.\Writes} a \relarrow{\impGRSI.\refC{\po}} d \land A = \{\mwi d\}) \;] \\
		\land\ [b \in \impGRSI.\Events_{\txid_b} \Rightarrow \\
			\qquad B {=} \absGRSI.\Transactions_{\txid_b} \lor (\stg{b}{\txid_b} \geq 3 \land B {=} \absGRSI.\Transactions_{\txid_b} \cap \Writes) \;]	\quad \Big)
\end{array}		
\end{array}
\]
We proceed by induction on $n$.\\

\noindent \textbf{Base case $n = 0$}\\
There are three cases to consider: 1) $(a, b) \in \impGRSI.\po$; or 2) $(a, b) \in \impGRSI.\rf$; or 3) $(a, b) \in \impGRSI.\rsilo$.

In case (1) there are five cases to consider: a) $\exsts{\txid} (a, b) \in \impGRSI.\Events_{\txid}$; or b) $a, b \in \impGRSI.\NT$; or c) $a \in \impGRSI.\NT$ and $b \in \impGRSI.\Events_{\txid_b}$; or d) $a \in \impGRSI.\Events_{\txid_a}$ and $b \in \impGRSI.\NT$; or e) $a \in \impGRSI.\Events_{\txid_a}$, $b \in \impGRSI.\Events_{\txid_b}$ and $\txid_a \ne \txid_b$.
In case (1.a) we then have $(a, b) \in \impGRSI.\poi$, as required.
In case (1.b) from the construction of $\impGRSI.\po$ we have $(a, b) \in \absGRSI.\po$, as required.
In case (1.c) from the construction of $\impGRSI.\po$ we have $(\{a\} \times \absGRSI.\Transactions_{\txid_b}) \in \absGRSI.\po$, as required.
In case (1.d) from the construction of $\impGRSI.\po$ we have $(\absGRSI.\Transactions_{\txid_a} \times \{b\}) \in \absGRSI.\po$, as required.
In case (1.e) from the construction of $\impGRSI.\po$ we have $(\absGRSI.\Transactions_{\txid_a} \times \absGRSI.\Transactions_{\txid_b}) \in \absGRSI.\po$, as required.

In case (2) there are five cases to consider: 
a) $\exsts{\txid} (a, b) \in \impGRSI.\Events_{\txid}$; or 
b) $a, b \in \impGRSI.\NT$; or 
c) $a \in \impGRSI.\NT$ and $b \in \impGRSI.\Events_{\txid_b}$; or 
d) $a \in \impGRSI.\Events_{\txid_a}$ and $b \in \impGRSI.\NT$; or e) $a \in \impGRSI.\Events_{\txid_a}$, $b \in \impGRSI.\Events_{\txid_b}$ and $\txid_a \ne \txid_b$. 
Case (2.a) holds vacuously as $(a, b) \in \impGRSI.\rfi = \emptyset$.
In case (2.b) from the construction of $\impGRSI.\rf$ we have $(a, b) \in \absGRSI.\rf$, as required.
In case (2.c) from the construction of $\impGRSI.\rf$ we know that there exists $r \in \absGRSI.\Transactions_{\txid_b}$ such that $(a, r) \in \absGRSI.\rf$. 
As such we have $(\{w\} \times \absGRSI.\Transactions_{\txid_b}) \subseteq \absGRSI.([\NT]; \rf; \st) \subseteq \absGRSI.\rsihb$, as required.
In case (2.d) from the construction of $\impGRSI.\rf$ we then have $(\mwi a, b) \in \absGRSI.\rf$; that is $a \relarrow{\impGRSI.\refC{\po}} a$ and $(\mwi a, b) \in \absGRSI.\rf$, as required.
In case (2.e) from the construction of $\impGRSI.\rf$ we know that there exists $w \in \absGRSI.\Transactions_{\txid_a}$, $r \in \absGRSI.\Transactions_{\txid_b}$ such that $(w, r) \in \absGRSI.\rf$. As such we have $(\absGRSI.\Transactions_{\txid_a} \times \absGRSI.\Transactions_{\txid_b}) \subseteq \absGRSI.\rft \subseteq \absGRSI.\rsihb$, as required.

In case (3) there are five cases to consider: 
a) $\exsts{\txid} (a, b) \in \impGRSI.\Events_{\txid}$; or 
b) $a, b \in \impGRSI.\NT$; or 
c) $a \in \impGRSI.\NT$ and $b \in \impGRSI.\Events_{\txid_b}$; or 
d) $a \in \impGRSI.\Events_{\txid_a}$ and $b \in \impGRSI.\NT$; or 
e) $a \in \impGRSI.\Events_{\txid_a}$, $b \in \impGRSI.\Events_{\txid_b}$ and $\txid_a \ne \txid_b$. 
In case (3.a) from the construction of $\lo$ we have $(a, b) \in \impGRSI.\po$, as required.
Cases (3.b-3.d) hold vacuously as there are no $\rsilo$ edge 
%between events of the same transaction, neither are there any $\rsilo$ edges 
to or from non-transactional events.
In case (3.e) from the construction of $\rsilo$ we know there exist $\x, c, d$ such that  $c \in \absGRSI.\Transactions_{\txid_a}$, $d \in \absGRSI.\Transactions_{\txid_b}$, and either 
%i) $a = \txid_a.\mathit{wu}_{\x}$, $\txid_a.\mathit{rl}_{\x}$ and $(c, d) \in \absGRSI.(\refC{\co}; \rf)$; or 
%ii) $a = \txid_a.\mathit{wu}_{\x}$, $\txid_a.\mathit{rl}_{\x}$ and $(c, d) \in \absGRSI.\co$; or
%iii) $a = \txid_a.\mathit{ru}_{\x}$, $\txid_a.\mathit{pl}_{\x}$ and $c \in \absGRSI.\EReads$ and $(c, d) \in \absGRSI.\fr$.
%In case (3.e.i) we have $\absGRSI.\Transactions_{\txid_a} \times \absGRSI.\Transactions_{\txid_b} \subseteq \absGRSI.\tlift{(\refC{\co}; \rf)} \subseteq \rsirf  \subseteq \rsihb$, as required.
%In case (3.e.ii) we have $\absGRSI.\Transactions_{\txid_a} \times \absGRSI.\Transactions_{\txid_b} \subseteq \absGRSI.\cot \subseteq \rsico  \subseteq \rsihb$, as required.
%In case (3.e.iii) we then have $\stg{a}{\txid_a} = 2$, $\stg{b}{\txid_b} = 3$, and $\absGRSI.(\Transactions_{\txid_a} \cap \EReads) \times \absGRSI.(\Transactions_{\txid_b} \cap \Writes) \subseteq \absGRSI.\rsifr  \subseteq \rsihb$, as required.\\
i) $(c, d) \in \absGRSI.\rf$; or 
ii) $(c, d) \in \absGRSI.(\co; \rf)$; or 
iii) $(c, d) \in \absGRSI.\co$; or
iv) $\loc c = \loc d = \x$, $a = \txid_a.\mathit{rl}_{\x} \lor a = \txid_a.\mathit{ru}_{\x}$, $b = \txid_a.\mathit{pl}_{\x} \lor b = \txid_a.\mathit{wu}_{\x}$ and $c \in \absGRSI.\EReads$ and $(c, d) \in \absGRSI.\fr$.
In cases (3.e.i, 3.e.ii) we have $\absGRSI.\Transactions_{\txid_a} \times \absGRSI.\Transactions_{\txid_b} \subseteq \absGRSI.(\rft \cup \tlift{(\co; \rf)}) \subseteq \rsirf  \subseteq \rsihb$, as required.
In case (3.e.iii) we have $\absGRSI.\Transactions_{\txid_a} \times \absGRSI.\Transactions_{\txid_b} \subseteq \absGRSI.\cot  \subseteq \rsihb$, as required.
In case (3.e.iv) we then have $\stg{a}{\txid_a} \leq 2$, $\stg{b}{\txid_b} \geq 3$, and $\absGRSI.(\Transactions_{\txid_a} \cap \EReads) \times \absGRSI.(\Transactions_{\txid_b} \cap \Writes) \subseteq \absGRSI.\rsifr  \subseteq \rsihb$, as required.\\

\noindent \textbf{Inductive case $n = m {+} 1$}\\
\begin{align}
\begin{array}{@{} l @{}}
	\for{i \in \Nats} \for{a, b, \txid_a, \txid_b}
	i \leq m \land (a, b) \in \impGRSI.\hb_i \Rightarrow \\
	\qqquad\! (\exsts{\txid} a, b \in  \impGRSI.\Events_{\txid} \Rightarrow (a, b) \in \impGRSI.\po) \\
	\qquad \land 
	\Big( \neg \exsts{\txid} a, b \in \impGRSI.\Events_{\txid} \Rightarrow \\
	\qqqquad 
	\begin{array}[t]{@{} l @{}}
		 \exsts{A, B} \emptyset \subset A \times B \subseteq \absGRSI.\rsihb \\
		 \land\ (a \in \impGRSI.\NT \Rightarrow A {=} \{a\}) \land (b \in \impGRSI.\NT \Rightarrow B {=} \{b\}) \\
		\land\ [a \in \impGRSI.\Events_{\txid_a} \Rightarrow \\
			\qquad A {=} \absGRSI.\Transactions_{\txid_a} \lor (\stg{a}{\txid_a} \leq 2 \land A {=} \absGRSI.\Transactions_{\txid_a} \cap \EReads) \\
			\qquad \lor\ (\stg{a}{\txid_a} \leq 4 \land \exsts{d \in \txid_a.\Writes} a \relarrow{\impGRSI.\refC{\po}} d \land A = \{\mwi d\}) \;] \\
		\land\ [b \in \impGRSI.\Events_{\txid_b} \Rightarrow \\
			\qquad B {=} \absGRSI.\Transactions_{\txid_b} \lor (\stg{b}{\txid_b} \geq 3 \land B {=} \absGRSI.\Transactions_{\txid_b} \cap \Writes) \;]	\quad \Big)
\end{array}		
\end{array}
\tag{I.H.}
\label{IH:rsi_alt_completeness}
\end{align}
Since $(a, b) \in \hb_n$, from the definition of $\hb_n$ we know there exists $c$ such that $(a, c) \in \hb_0$ and $(c, b) \in \hb_m$. 
There are then five cases to consider: 
1) $\exsts{\txid} (a, b) \in \impGRSI.\Events_{\txid}$; or 
2) $a, b \in \impGRSI.\NT$; or 
3) $a \in \impGRSI.\NT$ and $b \in \impGRSI.\Events_{\txid_b}$; or 
4) $a \in \impGRSI.\Events_{\txid_a}$ and $b \in \impGRSI.\NT$; or 
5) $a \in \impGRSI.\Events_{\txid_a}$, $b \in \impGRSI.\Events_{\txid_b}$ and $\txid_a \ne \txid_b$. \\

\noindent \textbf{Case 1}\\
In case (1) pick arbitrary $\txid$ such that $a, b \in \impGRSI.\Events_{\txid}$. There are then three additional cases to consider: 
a) $c \in \impGRSI.\NT$; 
b) $c \in \impGRSI.\Events_{\txid}$; or 
c) there exists $\txid' \ne \txid$ such that $c \in \impGRSI.\Events_{\txid'}$.

In case (1.a) from the proof of the base case and the \eqref{IH:rsi_alt_completeness} we know there exist $A, B, C \ne \emptyset$ such that $C = \{c\}$, $A \times C \subseteq \absGRSI.\rsihb$, $C \times B \subseteq \absGRSI.\rsihb$ and thus $A \times B \subseteq \absGRSI.\rsihb$ and either: 
i) $A = \absGRSI.\Transactions_{\txid}$ and  $B = \absGRSI.\Transactions_{\txid}$; or 
ii) $A = \absGRSI.\Transactions_{\txid}$ and  $B = (\absGRSI.(\Transactions_{\txid} \cap  \Writes)$; or 
iii) $\stg{a}{\txid} \leq 2$, $A = \absGRSI.(\Transactions_{\txid} \cap \EReads)$ and $B = \absGRSI.\Transactions_{\txid}$; or
iv) $\stg{a}{\txid} \leq 2$, $A = \absGRSI.(\Transactions_{\txid} \cap \EReads)$, $\stg{b}{\txid} \geq 3$ and $B = \absGRSI.(\Transactions_{\txid} \cap \Writes)$; or
v) $\stg{a}{\txid} \leq 4$, $\exsts{d \in \txid.\Writes} a \relarrow{\impGRSI.\refC{\po}} d \land A = \{\mwi d\}$ and $B = \absGRSI.\Transactions_{\txid}$; or
vi) $\stg{a}{\txid} \leq 4$, $\exsts{d \in \txid.\Writes} a \relarrow{\impGRSI.\refC{\po}} d \land A = \{\mwi d\}$ and $B = \absGRSI.(\Transactions_{\txid} \cap \Writes)$.

Case (i) cannot arise as we would have $A = B$ and thus $A \times A \subseteq \absGRSI.\rsihb$, contradicting the assumption that $\absGRSI$ is RSI-consistent. 
Case (ii) cannot arise as we would have $(\absGRSI.(\Transactions_{\txid} \cap  \Writes) \times (\absGRSI.(\Transactions_{\txid} \cap  \Writes) \subseteq \absGRSI.\rsihb$, contradicting the assumption that $\absGRSI$ is RSI-consistent. 
Case (iii) cannot arise as we would have $(\absGRSI.(\Transactions_{\txid} \cap  \EReads) \times (\absGRSI.(\Transactions_{\txid} \cap  \EReads) \subseteq \absGRSI.\rsihb$, contradicting the assumption that $\absGRSI$ is RSI-consistent. 
In case (iv) since we have $\stg{a}{\txid} \leq 2$ and $\stg{b}{\txid} \geq 3$, from the definition of $\stg{.}{.}$ and the construction of $\impGRSI$ we have $(a, b) \in \impGRSI.\po$, as required.
Cases (v-vi) cannot arise as we would have $\exsts{d \in \txid.\Writes} (\mwi d, \mwi d)  \in \absGRSI.\rsihb$, contradicting the assumption that $\absGRSI$ is RSI-consistent. \\

In case (1.b) from the proof of the base case we have $(a, c) \in \impGRSI.\po$. On the other hand from (\ref{IH:rsi_alt_completeness}) we have $(c, b) \in \impGRSI.\po$. As $\impGRSI.\po$ is transitively closed, we have $(a, b) \in \impGRSI.\po$, as required.\\

In case (1.c)  from the proof of the base case (in cases 1.e, 2.e and 3.e) we have either 
A) $\absGRSI.\Transactions_{\txid} \times \absGRSI.\Transactions_{\txid'} \subseteq \absGRSI.\rsihb$; or 
B) $\stg{a}{\txid} \leq 2$, $\stg{c}{\txid'} \geq 3$,  and $\emptyset \subset \absGRSI.(\Transactions_{\txid} \cap \EReads) \times \absGRSI.(\Transactions_{\txid'} \cap \Writes) \subseteq \absGRSI.\rsihb$.
On the other hand from \eqref{IH:rsi_alt_completeness} we know there exist $C, B \ne \emptyset$ such that $C \times B \subseteq \absGRSI.\rsihb$ and either: 
i) $C = \absGRSI.\Transactions_{\txid'}$ and  $B = \absGRSI.\Transactions_{\txid}$; or 
ii) $C = \absGRSI.\Transactions_{\txid'}$,  $\stg{b}{\txid} \geq 3$ and $B = \absGRSI.(\Transactions_{\txid} \cap  \Writes)$; or 
iii) $\stg{c}{\txid'} \leq 2$, $C = \absGRSI.(\Transactions_{\txid'} \cap \EReads)$ and $B = \absGRSI.\Transactions_{\txid}$; or
iv) $\stg{c}{\txid'} \leq 2$, $C = \absGRSI.(\Transactions_{\txid'} \cap \EReads)$, $\stg{b}{\txid} \geq 3$ and $B = \absGRSI.(\Transactions_{\txid} \cap \Writes)$; or
v) $\stg{c}{\txid'} \leq 4$, $\exsts{d \in \txid'.\Writes} c \relarrow{\impGRSI.\refC{\po}} d \land C = \{\mwi d\}$ and $B = \absGRSI.\Transactions_{\txid}$; or
vi) $\stg{c}{\txid'} \leq 4$, $\exsts{d \in \txid'.\Writes} c \relarrow{\impGRSI.\refC{\po}} d \land C = \{\mwi d\}$, $\stg{b}{\txid} \geq 3$ and $B = \absGRSI.(\Transactions_{\txid} \cap \Writes)$.

Cases (A.i-A.vi) lead to a cycle in $\absGRSI.\rsihb$, contradicting the assumption that $\absGRSI$ is RSI-consistent.
In cases (B.ii, B.iv, B.vi) we then have $\stg{a}{\txid} \leq 2$ and $\stg{b}{\txid} \geq 3$. Consequently from the definition of $\stg{.}{.}$ and the construction of $\impGRSI$ we have $(a, b) \in \impGRSI.\po$, as required.
Case (B.i) leads to a cycle in $\absGRSI.\rsihb$, contradicting the assumption that $\absGRSI$ is RSI-consistent.
In cases (B.iii) we have $\stg{c}{\txid'} \geq 3$  and $\stg{c}{\txid'} \leq 2$, leading to a contradiction.
In case (B.v) we then know  $\exsts{d  \in \txid'.\Writes} \emptyset \subset \absGRSI.(\Transactions_{\txid} \cap \EReads) \times \{\mwi d\} \subseteq \absGRSI.\rsihb \land \{\mwi d\} \times \absGRSI.\Transactions_\txid \subseteq \absGRSI.\rsihb$. That is, we have  $\emptyset \subset \absGRSI.(\Transactions_{\txid} \cap \EReads) \times  \absGRSI.(\Transactions_{\txid} \cap \EReads) \subseteq \absGRSI.\rsihb$, contradicting the assumption that $\absGRSI$ is RSI-consistent.
\\

\noindent \textbf{Case 2}\\
There are two additional cases to consider: 
a) $c \in \impGRSI.\NT$; or
b) there exists $\txid$ such that $c \in \impGRSI.\Events_{\txid}$.

In case (2.a) from the proof of the base case we have $(a, c) \in \absGRSI.\rsihb$. On the other hand from (\ref{IH:rsi_alt_completeness}) we have $(c, b) \in \absGRSI.\rsihb$. As $\absGRSI.\rsihb$ is transitively closed, we have $(a, b) \in \absGRSI.\rsihb$, as required.

In case (2.b) from the proof of the base case we know there exists $C_1 \ne \emptyset$ such that $\{a\} \times C_1 \subseteq \absGRSI.\rsihb$ and either: 
A) $C_1 = \absGRSI.\Transactions_{\txid}$; or 
B) $\stg{c}{\txid} \geq 3$ and $C_1 = \absGRSI.(\Transactions_{\txid} \cap \Writes)$.
On the other hand, from (\ref{IH:rsi_alt_completeness}) we know there exists $C_2 \ne \emptyset$ such that $C_2 \times \{b\} \in \absGRSI\rsihb$ and either:
i) $C_2 = \absGRSI.\Transactions_{\txid}$; or
ii) $\stg{c}{\txid} \leq 2$ and $C_2 = \absGRSI.(\Transactions_{\txid} \cap \EReads)$; or
iii) $\stg{c}{\txid} \leq 4$ and $\exsts{d \in \txid.\Writes} c \relarrow{\impGRSI.\refC{\po}} d \land C_2 = \{\mwi d\}$.

In cases (A.i-A.iii) and (B.i, B.iii) from the transitivity of $\absGRSI.\rsihb$ we have $(a, b) \in \absGRSI.\rsihb$, as required.
Case (B.ii) cannot arise as otherwise we would have $3 \leq \stg{c}{\txid} \leq 2$, leading to a contradiction.\\

\noindent \textbf{Case 3}\\
There are three additional cases to consider: 
a) $c \in \impGRSI.\NT$; 
b) $c \in \impGRSI.\Events_{\txid_b}$; or 
c) there exists $\txid_c \ne \txid_b$ such that $c \in \impGRSI.\Events_{\txid_c}$.

In case (3.a) from the proof of the base case we have $(a, c) \in \absGRSI.\rsihb$. On the other hand from (\ref{IH:rsi_alt_completeness}) we know there exists $B \ne \emptyset$ such that $\{c\} \times B \subseteq \absGRSI.\rsihb$ and $B = \absGRSI.\Transactions_{\txid_b} \lor (\stg{b}{\txid_b} \geq 3 \land B = \absGRSI.(\Transactions_{\txid_b} \cap \Writes))$.
As $\absGRSI.\rsihb$ is transitively closed we then know there exists $B \ne \emptyset$ such that $\{a\} \times B \subseteq \absGRSI.\rsihb$ and $B = \absGRSI.\Transactions_{\txid_b} \lor (\stg{b}{\txid_b} \geq 3 \land B = \absGRSI.(\Transactions_{\txid_b} \cap \Writes))$, as required.

In case (3.b) from the proof of the base case we know there exists $B \ne \emptyset$ such that $\{a\} \times B \subseteq \absGRSI.\rsihb$ and $B = \absGRSI.\Transactions_{\txid_b} \lor (\stg{c}{\txid_b} \geq 3 \land B = \absGRSI.(\Transactions_{\txid_b} \cap \Writes))$.
On the other hand, from (\ref{IH:rsi_alt_completeness}) we have $c \relarrow{\impGRSI.\po} b$ and thus from the definition of $\stg{.}{.}$ and the construction of $\impGRSI$ we have $\stg{b}{\txid_b} \geq \stg{c}{\txid_b}$. As such, we we know there exists $B \ne \emptyset$ such that $\{a\} \times B \subseteq \absGRSI.\rsihb$ and $B = \absGRSI.\Transactions_{\txid_b} \lor (\stg{b}{\txid_b} \geq 3 \land B = \absGRSI.(\Transactions_{\txid_b} \cap \Writes))$, as required.

In case (3.c) from the proof of the base case we know there exists $C_1 \ne \emptyset$ such $\{a\} \times C_1 \subseteq \absGRSI.\rsihb$ and either 
A) $C_1 = \absGRSI.\Transactions_{\txid_c}$; or
B)$(\stg{c}{\txid_c} \geq 3 \land C_1 = \absGRSI.(\Transactions_{\txid_c} \cap \Writes))$.
On the other hand, from (\ref{IH:rsi_alt_completeness}) we know there exist $C_2, B \ne \emptyset$ such that $C_2 \times B \in \absGRSI.\rsihb$ and either:
i) $C_2 = \absGRSI.\Transactions_{\txid_c}$ and $B = \absGRSI.\Transactions_{\txid_b}$; or
ii) $C_2 = \absGRSI.\Transactions_{\txid_c}$ and $(\stg{b}{\txid_b} \geq 3 \land B = \absGRSI.(\Transactions_{\txid_b} \cap \Writes))$; or
iii) $\stg{c}{\txid_c} \leq 2 \land C_2 = \absGRSI.(\Transactions_{\txid_c} \cap \EReads)$ and $B = \absGRSI.\Transactions_{\txid_b}$; or
iv) $\stg{c}{\txid_c} \leq 2 \land C_2 = \absGRSI.(\Transactions_{\txid_c} \cap \EReads)$ and $(\stg{b}{\txid_b} \geq 3 \land B = \absGRSI.(\Transactions_{\txid_b} \cap \Writes))$; or
v) $\stg{c}{\txid_c} \leq 4 \land \exsts{d \in \txid_c.\Writes} c \relarrow{\impGRSI.\refC{\po}} d \land C_2 = \{\mwi d\}$ and $B = \absGRSI.\Transactions_{\txid_b}$; or
vi) $\stg{c}{\txid_c} \leq 4 \land \exsts{d \in \txid_c.\Writes} c \relarrow{\impGRSI.\refC{\po}} d \land C_2 = \{\mwi d\}$ and $(\stg{b}{\txid_b} \geq 3 \land B = \absGRSI.(\Transactions_{\txid_b} \cap \Writes))$. 

In cases (A.i-A.vi) and (B.i, B.ii, B.v, B.vi) from the transitivity of $\absGRSI.\rsihb$ we know there exists $B \ne \emptyset$ such that $\{a\} \times B \subseteq \absGRSI.\rsihb$ and  $B = \absGRSI.\Transactions_{\txid_b} \lor (\stg{b}{\txid_b} \geq 3 \land B = \absGRSI.(\Transactions_{\txid_b} \cap \Writes))$, as required.
Cases (B.iii, B.iv) cannot arise as we would otherwise have $3 \leq \stg{c}{\txid_c} \leq 2$, leading to a contradiction.\\

\noindent \textbf{Case 4}\\
There are three additional cases to consider: 
a) $c \in \impGRSI.\NT$; 
b) $c \in \impGRSI.\Events_{\txid_a}$; or 
c) there exists $\txid_c \ne \txid_a$ such that $c \in \impGRSI.\Events_{\txid_c}$.

In case (4.a) from (\ref{IH:rsi_alt_completeness}) we have $(c, b) \in \absGRSI.\rsihb$. On the other hand from the base case we know there exists $A \ne \emptyset$ such that $A \times \{c\}  \subseteq \absGRSI.\rsihb$ and $A = \absGRSI.\Transactions_{\txid_a} \lor (\stg{a}{\txid_a} \leq 2 \land A = \absGRSI.(\Transactions_{\txid_a} \cap \EReads)) \lor \stg{a}{\txid_a} \leq 4 \land \exsts{d \in \txid_a.\Writes} a \relarrow{\impGRSI.\refC{\po}} d \land A = \{\mwi d\}$.
As $\absGRSI.\rsihb$ is transitively closed we then know we know there exists $A \ne \emptyset$ such that $A \times \{b\} \subseteq \absGRSI.\rsihb$ and $A = \absGRSI.\Transactions_{\txid_a} \lor (\stg{a}{\txid_a} \leq 2 \land A = \absGRSI.(\Transactions_{\txid_a} \cap \EReads)) \lor \stg{a}{\txid_a} \leq 4 \land \exsts{d \in \txid_a.\Writes} a \relarrow{\impGRSI.\refC{\po}} d \land A = \{\mwi d\}$, as required.

In case (4.b) from \eqref{IH:rsi_alt_completeness} we know there exists $A \ne \emptyset$ such that $A \times \{b\}  \subseteq \absGRSI.\rsihb$ and $A = \absGRSI.\Transactions_{\txid_a} \lor (\stg{c}{\txid_a} \leq 2 \land A = \absGRSI.(\Transactions_{\txid_a} \cap \EReads)) \lor \stg{c}{\txid_a} \leq 4 \land \exsts{d \in \txid_a.\Writes} a \relarrow{\impGRSI.\refC{\po}} d \land A = \{\mwi d\}$.
On the other hand, from the proof of the base case we have $a \relarrow{\impGRSI.\po} c$ and thus from the definition of $\stg{.}{.}$ and the construction of $\impGRSI$ we have $\stg{a}{\txid_a} \leq \stg{c}{\txid_a}$.
As such, we know there exists $A \ne \emptyset$ such that $A \times \{b\}  \subseteq \absGRSI.\rsihb$ and $A = \absGRSI.\Transactions_{\txid_a} \lor (\stg{a}{\txid_a} \leq 2 \land A = \absGRSI.(\Transactions_{\txid_a} \cap \EReads)) \lor \stg{a}{\txid_a} \leq 4 \land \exsts{d \in \txid_a.\Writes} a \relarrow{\impGRSI.\refC{\po}} d \land A = \{\mwi d\}$, as required.

In case (4.c) from the proof of the base case we know there exist $A, C_1 \ne \emptyset$ such that $A \times C_1  \subseteq \absGRSI.\rsihb$ and either:
i) $A = \absGRSI.\Transactions_{\txid_a}$ and $C_1 = \absGRSI.\Transactions_{\txid_c}$; or
ii) $A = \absGRSI.\Transactions_{\txid_a}$ and $(\stg{c}{\txid_c} \geq 3 \land C_1 = \absGRSI.(\Transactions_{\txid_c} \cap \Writes))$; or
iii) $\stg{a}{\txid_a} \leq 2 \land A = \absGRSI.(\Transactions_{\txid_a} \cap \EReads)$ and $C_1 = \absGRSI.\Transactions_{\txid_c}$; or
iv) $\stg{a}{\txid_a} \leq 2 \land A = \absGRSI.(\Transactions_{\txid_a} \cap \EReads)$ and $(\stg{c}{\txid_c} \geq 3 \land C_1 = \absGRSI.(\Transactions_{\txid_c} \cap \Writes))$; or
v) $\stg{a}{\txid_a} \leq 4 \land \exsts{d \in \txid_a.\Writes} a \relarrow{\impGRSI.\refC{\po}} d \land A = \{\mwi d\}$ and $C_1 = \absGRSI.\Transactions_{\txid_c}$; or
vi) $\stg{a}{\txid_a} \leq 4 \land \exsts{d \in \txid_a.\Writes} a \relarrow{\impGRSI.\refC{\po}} d \land A = \{\mwi d\}$ and $(\stg{c}{\txid_c} \geq 3 \land C_1 = \absGRSI.(\Transactions_{\txid_c} \cap \Writes))$.

On the other hand, from (\ref{IH:rsi_alt_completeness}) we know there exists $C_2 \ne \emptyset$ such $C_2 \times \{b\} \subseteq \absGRSI.\rsihb$ and either 
A) $C_2 = \absGRSI.\Transactions_{\txid_c}$; or
B) $(\stg{c}{\txid_c} \leq 2 \land C_2 = \absGRSI.(\Transactions_{\txid_c} \cap \EReads))$; or
C) $(\stg{c}{\txid_c} \leq 4 \land \exsts{e \in \txid_c.\Writes} c \relarrow{\impGRSI.\refC{\po}} e \land C_2 = \{\mwi e\})$.

In cases (A.i-A.vi), (C.i-C.vi) and (B.i, B.iii, B.v) from the transitivity of $\absGRSI.\rsihb$ we know there exists $A \ne \emptyset$ such that $A \times \{b\} \subseteq \absGRSI.\rsihb$ and $A = \absGRSI.\Transactions_{\txid_a} \lor (\stg{a}{\txid_a} \leq 2 \land A = \absGRSI.(\Transactions_{\txid_a} \cap \EReads)) \lor (\stg{a}{\txid_a} \leq 4 \land \exsts{d \in \txid_a.\Writes} a \relarrow{\impGRSI.\refC{\po}} d \land A = \{d\})$, as required.
Cases (B.ii, B.iv, B.vi) cannot arise as we would otherwise have $3 \leq \stg{c}{\txid_c} \leq 2$, leading to a contradiction.\\

\noindent \textbf{Case 5}\\
There are four additional cases to consider: 
a) $c \in \impGRSI.\NT$; 
b) $c \in \impGRSI.\Events_{\txid_a}$; or 
c) $c \in \impGRSI.\Events_{\txid_b}$; or 
d) there exists $\txid_c$ such that $\txid_c \ne \txid_a$, $\txid_c \ne \txid_b$ and $c \in \impGRSI.\Events_{\txid_c}$.

In case (5.a) from the proof of the base case we know there exists $A \ne \emptyset$ such that $A \times \{c\} \subseteq \absGRSI.\rsihb$ and $A = \absGRSI.\Transactions_{\txid_a} \lor (\stg{a}{\txid_a} \leq 2 \land A = \absGRSI.(\Transactions_{\txid_a} \cap \EReads)) \lor (\stg{a}{\txid_a} \leq 4 \land \exsts{d \in \txid_a.\Writes} a \relarrow{\impGRSI.\refC{\po}} d \land A = \{\mwi d\})$.
On the other hand from (\ref{IH:rsi_alt_completeness}) we know there exists $B \ne \emptyset$ such that $\{c\} \times B \subseteq \absGRSI.\rsihb$ and  $B = \absGRSI.\Transactions_{\txid_b} \lor (\stg{b}{\txid_b} \geq 3 \land B = \absGRSI.(\Transactions_{\txid_b} \cap \Writes))$.
As such, sine $\absGRSI.\rsihb$ is transitive we know there exist $A, B \ne \emptyset$ such that $A \times B \subseteq \absGRSI.\rsihb$; that $A = \absGRSI.\Transactions_{\txid_a} \lor (\stg{a}{\txid_a} \leq 2 \land A = \absGRSI.(\Transactions_{\txid_a} \cap \EReads)) \lor (\stg{a}{\txid_a} \leq 4 \land \exsts{d \in \txid_a.\Writes} a \relarrow{\impGRSI.\refC{\po}} d \land A = \{\mwi d\})$; 
and that $B = \absGRSI.\Transactions_{\txid_b} \lor (\stg{b}{\txid_b} \geq 3 \land B = \absGRSI.(\Transactions_{\txid_b} \cap \Writes))$, as required.

In case (5.b) from (\ref{IH:rsi_alt_completeness}) we know there exist $A, B \ne \emptyset$ such that $A \times B \subseteq \absGRSI.\rsihb$; that $A = \absGRSI.\Transactions_{\txid_a} \lor (\stg{c}{\txid_a} \leq 2 \land A = \absGRSI.(\Transactions_{\txid_a} \cap \EReads)) \lor (\stg{c}{\txid_a} \leq 4 \land \exsts{d \in \txid_a.\Writes} a \relarrow{\impGRSI.\refC{\po}} d \land A = \{\mwi d\})$; 
and that $B = \absGRSI.\Transactions_{\txid_b} \lor (\stg{b}{\txid_b} \geq 3 \land B = \absGRSI.(\Transactions_{\txid_b} \cap \Writes))$.
On the other hand from the proof of the base case we have $(a, c) \in \absGRSI.\po$ and thus from the definition of $\stg{.}{.}$ and the construction of $\impGRSI$ we have $\stg{a}{\txid_a} \leq \stg{c}{\txid_a}$.
As such we know there exist $A, B \ne \emptyset$ such that $A \times B \subseteq \absGRSI.\rsihb$; that $A = \absGRSI.\Transactions_{\txid_a} \lor (\stg{a}{\txid_a} \leq 2 \land A = \absGRSI.(\Transactions_{\txid_a} \cap \EReads)) \lor (\stg{a}{\txid_a} \leq 4 \land \exsts{d \in \txid_a.\Writes} a \relarrow{\impGRSI.\refC{\po}} d \land A = \{\mwi d\})$ and that $B = \absGRSI.\Transactions_{\txid_b} \lor (\stg{b}{\txid_b} \geq 3 \land B = \absGRSI.(\Transactions_{\txid_b} \cap \Writes))$, as required.

In case (5.c) from the proof of the base case we know there exist $A, B \ne \emptyset$ such that $A \times B \subseteq \absGRSI.\rsihb$; that $A = \absGRSI.\Transactions_{\txid_a} \lor (\stg{a}{\txid_a} \leq 2 \land A = \absGRSI.(\Transactions_{\txid_a} \cap \EReads)) \lor (\stg{a}{\txid_a} \leq 4 \land \exsts{d \in \txid_a.\Writes} a \relarrow{\impGRSI.\refC{\po}} d \land A = \{\mwi d\})$, and that $B = \absGRSI.\Transactions_{\txid_b} \lor (\stg{c}{\txid_b} \geq 3 \land B = \absGRSI.(\Transactions_{\txid_b} \cap \Writes))$.
On the other hand from (\ref{IH:rsi_alt_completeness}) we have $(c, b) \in \absGRSI.\po$ and thus from the definition of $\stg{.}{.}$ and the construction of $\impGRSI$ we have $\stg{c}{\txid_a} \leq \stg{b}{\txid_a}$.
As such we know there exist $A, B \ne \emptyset$ such that $A \times B \subseteq \absGRSI.\rsihb$; that $A = \absGRSI.\Transactions_{\txid_a} \lor (\stg{a}{\txid_a} \leq 2 \land A = \absGRSI.(\Transactions_{\txid_a} \cap \EReads)) \lor (\stg{a}{\txid_a} \leq 4 \land \exsts{d \in \txid_a.\Writes} a \relarrow{\impGRSI.\refC{\po}} d \land A = \{\mwi d\})$ and that $B = \absGRSI.\Transactions_{\txid_b} \lor (\stg{b}{\txid_b} \geq 3 \land B = \absGRSI.(\Transactions_{\txid_b} \cap \Writes))$, as required.

In case (5.d) from the proof of the base case we know there exist $A, C_1 \ne \emptyset$ such that $A \times C_1 \subseteq \absGRSI.\rsihb$; that $A = \absGRSI.\Transactions_{\txid_a} \lor (\stg{a}{\txid_a} \leq 2 \land A = \absGRSI.(\Transactions_{\txid_a} \cap \EReads)) \lor (\stg{a}{\txid_a} \leq 4 \land \exsts{d \in \txid_a.\Writes} a \relarrow{\impGRSI.\refC{\po}} d \land A = \{\mwi d\})$, and that either 
A) $C_1 = \absGRSI.\Transactions_{\txid_c}$; or
B) $(\stg{c}{\txid_c} \geq 3 \land C_1 = \absGRSI.(\Transactions_{\txid_c} \cap \Writes))$.

On the other hand, from (\ref{IH:rsi_alt_completeness}) we know there exist $C_2, B \ne \emptyset$ such that $C_2 \times B \subseteq \absGRSI.\rsihb$; that $B = \absGRSI.\Transactions_{\txid_b} \lor (\stg{b}{\txid_b} \geq 3 \land B = \absGRSI.(\Transactions_{\txid_b} \cap \Writes))$; and that either:
i) $C_2 = \absGRSI.\Transactions_{\txid_c}$; or
ii) $\stg{c}{\txid_c} \leq 2 \land C_2 = \absGRSI.(\Transactions_{\txid_c} \cap \EReads)$; 
iii) $\stg{c}{\txid_c} \leq 4 \land \exsts{e \in \txid_c.\Writes} c \relarrow{\impGRSI.\refC{\po}} e \land C_2 = \{\mwi  e\})$.

In cases (A.i-A.iii) and (B.i, B.iii) from the transitivity of $\absGRSI.\rsihb$ we know there exists $A, B \ne \emptyset$ such that $A \times B \subseteq \absGRSI.\rsihb$; that $A = \absGRSI.\Transactions_{\txid_a} \lor (\stg{a}{\txid_a} \leq 2 \land A = \absGRSI.(\Transactions_{\txid_a} \cap \EReads)) \lor (\stg{a}{\txid_a} \leq 4 \land \exsts{d \in \txid_a.\Writes} a \relarrow{\impGRSI.\refC{\po}} d \land A = \{\mwi d\})$; and that $B = \absGRSI.\Transactions_{\txid_b} \lor (\stg{b}{\txid_b} \geq 3 \land B = \absGRSI.(\Transactions_{\txid_b} \cap \Writes))$, as required.
Case (B.ii) cannot arise as we would otherwise have $3 \leq \stg{c}{\txid_c} \leq 2$, leading to a contradiction.

\end{proof}
\end{lemma}

\begin{theorem}[Completeness]
For all RSI execution graphs $\absGRSI$ and their counterpart implementation graphs $\impGRSI$ constructed as above,
\[
	\rsicon  \Rightarrow \consistent{\impGRSI}
\]
\begin{proof}
Pick an arbitrary RSI execution graph $\absGRSI$ and its counterpart implementation graph $\impGRSI$ constructed as above and let us assume that $\rsicon$ holds.
From the \noshade{definition of $\consistent{\impGRSI}$} it then suffices to show: 
\begin{enumerate}
	\item $\irr{\impGRSI.\hb}$ \label{goal:rsi_alt_completeness_hb_irr}
	\item $\irr{\impGRSI.\co ; \impGRSI.\hb}$ \label{goal:rsi_alt_completeness_co_hb_irr}
	\item $\irr{\impGRSI.\fr ; \impGRSI.\hb}$ \label{goal:rsi_alt_completeness_fr_hb_irr} \\
\end{enumerate}

\noindent \textbf{RTS. part \ref{goal:rsi_alt_completeness_hb_irr}}\\
We proceed by contradiction. Let us assume that there exists $a$ such that $(a, a) \in \impGRSI.\hb$.
There are now two cases to consider: 1) $a \in \impGRSI.\NT$; or 2) $\exsts{\txid} a \in \impGRSI.\Events_{\txid}$
In case (1) from \cref{lem:rsi_alt_completeness} we have $(a, a) \in \absGRSI.\rsihb$, contradicting the assumption that $\absGRSI$ is RSI-consistent. 
Similarly, in case (2) from \cref{lem:rsi_alt_completeness} we have $(a, a) \in \impGRSI.\po$, leading to a contradiction as $\impGRSI.\po$ is acyclic by construction. \\

\noindent \textbf{RTS. part \ref{goal:rsi_alt_completeness_co_hb_irr}}\\
We proceed by contradiction. Let us assume that there exist $a, b$ such that $(a, b) \in \impGRSI.\hb$ and $(b, a) \in \impGRSI.\co$.
Let $\loc a = \loc b = \x$. 
There are then five cases to consider: 
1) $\exsts{\txid} (a, b) \in \impGRSI.\Events_{\txid}$; or 
2) $a, b \in \impGRSI.\NT$; or 
3) $a \in \impGRSI.\NT$ and $b \in \impGRSI.\Events_{\txid_b}$; or 
4) $a \in \impGRSI.\Events_{\txid_a}$ and $b \in \impGRSI.\NT$; or 
5) $a \in \impGRSI.\Events_{\txid_a}$, $b \in \impGRSI.\Events_{\txid_b}$ and $\txid_a \ne \txid_b$. 

In case (1) from the construction of $\impGRSI.\co$ we have $(\mwi  b, \mwi  a) \in \absGRSI.\co$. 
Moreover, from \cref{lem:rsi_alt_completeness} we have $(a, b) \in \impGRSI.\po$, and thus from the construction of $\impGRSI.\po$ we have $(\mwi a, \mwi b) \in \absGRSI.\po$. 
Moreover, since $\mwi a, \mwi b$ are write events in the same transaction $\txid$, we have $(\mwi a,$ $ \mwi b) \in \absGRSI.\poi \cap \Writes^2 \subseteq \absGRSI.\rsihb$. As such we have $\mwi a \relarrow{\absGRSI.\rsihb} \mwi b$ $\relarrow{\absGRSI.\co} \mwi a$, contradicting the assumption that $\absGRSI$ is RSI-consistent. 

In case (2) from the construction of $\impGRSI.\co$ we have $(b, a) \in \absGRSI.\co$. 
Moreover, from \cref{lem:rsi_alt_completeness} we have $(a, b) \in \absGRSI.\rsihb$.
As such we have $a \relarrow{\absGRSI.\rsihb} b \relarrow{\absGRSI.\co} a$, contradicting the assumption that $\absGRSI$ is RSI-consistent. 

In case (3) from the construction of $\impGRSI.\co$ we have $(\mwi b, a) \in \absGRSI.\co$. 
Moreover, since $b \in \impGRSI.\Writes$ and thus $\mwi b \in \absGRSI.\Writes$, from \cref{lem:rsi_alt_completeness} we have $(a, b) \in \absGRSI.\rsihb$.
As such we have $a \relarrow{\absGRSI.\rsihb} \mwi b \relarrow{\absGRSI.\co} a$, contradicting the assumption that $\absGRSI$ is RSI-consistent.

In cases (4, 5) from the construction of $\impGRSI$ we know that 
$\mwi a \relarrow{\impGRSI.\po} a \relarrow{\impGRSI.\hb} b$, and thus from the transitivity of $\impGRSI.\hb$ we have $(\mwi a, b) \in \impGRSI.\hb$.
As $\stg{\mwi a}{\txid_a} = 1$, from \cref{lem:rsi_alt_completeness} we have $(\mwi a, b) \in \absGRSI.\rsihb$.
On the other hand, from the construction of $\impGRSI.\co$ we have $(b, \mwi a) \in \absGRSI.\co$. 
Consequently, we have $\mwi a \relarrow{\absGRSI.\rsihb} b \relarrow{\absGRSI.\co} \mwi a$, contradicting the assumption that $\absGRSI$ is RSI-consistent. \\

\noindent \textbf{RTS. part \ref{goal:rsi_alt_completeness_fr_hb_irr}}\\
We proceed by contradiction. Let us assume that there exists $a, b$ such that $(a, b) \in \impGRSI.\hb$ and $(b, a) \in \impGRSI.\fr$.
Let $\loc a = \loc b = \x$. 
There are then five cases to consider: 
1) $\exsts{\txid} (a, b) \in \impGRSI.\Events_{\txid}$; or 
2) $a, b \in \impGRSI.\NT$; or 
3) $a \in \impGRSI.\NT$ and $b \in \impGRSI.\Events_{\txid_b}$; or 
4) $a \in \impGRSI.\Events_{\txid_a}$ and $b \in \impGRSI.\NT$; or 
5) $a \in \impGRSI.\Events_{\txid_a}$, $b \in \impGRSI.\Events_{\txid_b}$ and $\txid_a \ne \txid_b$. \\

Case (1) cannot arise as from the definition of $\impGRSI.\fr$ we know $a$ is a write event in $\txid.\mathit{Ws}$ while $b$ is a read event in $\txid.\mathit{Ts}$ and no $\po$ edge exists between the events of $\txid.\mathit{Ts}$ and $\txid.\mathit{Ws}$.

In case (1) from \cref{lem:rsi_alt_completeness} we have $(a, b) \in \impGRSI.\po$. 
This however leads to a contradiction as from the definition of $\impGRSI.\fr$ we know $a$ is a write event in $\txid.\mathit{Ws}$ while $b$ is a read event in $\txid.\mathit{Ts}$ and no $\po$ edge exists between the events of $\txid.\mathit{Ws}$ and $\txid.\mathit{Ts}$.

In case (2) from the construction of $\impGRSI.\fr$ we know that $(a, b) \in \absGRSI.\fr$. 
On the other hand, from \cref{lem:rsi_alt_completeness} we have $(a, b) \in \absGRSI.\rsihb$.
As such we have $a \relarrow{\absGRSI.\rsihb} b \relarrow{\absGRSI.\fr} a$, contradicting the assumption that $\absGRSI$ is RSI-consistent. 

In case (3), as $b$ is a read event and thus $\stg{b}{\txid_b} = 1$, from \cref{lem:rsi_alt_completeness} we have $(a, b) \in \absGRSI.\rsihb$.
On the other hand, from the construction of $\impGRSI.\fr$ we have $(b, a) \in \absGRSI.\fr$. 
Consequently, we have $a \relarrow{\absGRSI.\rsihb} b \relarrow{\absGRSI.\fr} a$, contradicting the assumption that $\absGRSI$ is RSI-consistent.

In case (4) since $a$ is a write, from the construction of $\impGRSI$ we know that 
$\mwi a \relarrow{\impGRSI.\po} a \relarrow{\impGRSI.\hb} b$, and thus from the transitivity of $\impGRSI.\hb$ we have $(\mwi a, b) \in \impGRSI.\hb$.
As $\stg{\mwi a}{\txid_a} = 1$, from \cref{lem:rsi_alt_completeness} we have $(\mwi a, b) \in \absGRSI.\rsihb$.
On the other hand, from the construction of $\impGRSI.\fr$ we have $(b, \mwi a) \in \absGRSI.\fr$. 
Consequently, we have $\mwi a \relarrow{\absGRSI.\rsihb} b \relarrow{\absGRSI.\fr} \mwi a$, contradicting the assumption that $\absGRSI$ is RSI-consistent.

Similarly, in case (5) since $a$ is a write event, from the construction of $\impGRSI$ we know that 
$\mwi a \relarrow{\impGRSI.\po} a \relarrow{\impGRSI.\hb} b$, and thus from the transitivity of $\impGRSI.\hb$ we have $(\mwi a, b) \in \impGRSI.\hb$.
As $\stg{\mwi a}{\txid_a} = \stg{b}{\txid_b} = 1$, from \cref{lem:rsi_alt_completeness} we have $(\mwi a, b) \in \absGRSI.\rsihb$.
On the other hand, from the construction of $\impGRSI.\fr$ we have $(b, \mwi a) \in \absGRSI.\fr$. 
Consequently, we have $\mwi a \relarrow{\absGRSI.\rsihb} b \relarrow{\absGRSI.\fr} \mwi a$, contradicting the assumption that $\absGRSI$ is RSI-consistent. \\

\end{proof}
\end{theorem}

\end{document}

%\NewEnviron{hideall}{}
%\renewenvironment{remark}{\hideall}{\endhideall}